\documentclass[leqno]{article}%
\usepackage{amsfonts}
\usepackage{amsmath}
\usepackage{amssymb}
\usepackage{graphicx}
\usepackage{makeidx}%
\newtheorem{theorem}{Theorem}

\newtheorem{lemma}[theorem]{Lemma}
\newtheorem{notation}[theorem]{Notation}
\newtheorem{problem}[theorem]{Problem}
\newtheorem{proposition}[theorem]{Proposition}
\newtheorem{remark}[theorem]{Remark}

\newenvironment{proof}[1][Proof]{\noindent\textbf{#1.} }{\ \rule{0.5em}{0.5em}}
\ifx\pdfoutput\relax\let\pdfoutput=\undefined\fi
\newcount\msipdfoutput
\ifx\pdfoutput\undefined\else
\ifcase\pdfoutput\else
\ifx\paperwidth\undefined\else
\ifdim\paperheight=0pt\relax\else\pdfpageheight\paperheight\fi
\ifdim\paperwidth=0pt\relax\else\pdfpagewidth\paperwidth\fi
\fi\fi\fi
\begin{document}

\title{Dynamic Mean-Field Theory for Continuous Random Networks}
\author{W. A. Z\'{u}\~{n}iga-Galindo\\University of Texas Rio Grande Valley\\School of Mathematical \& Statistical Sciences\\One West University Blvd\\Brownsville, TX 78520, United States}
\maketitle

\begin{abstract}
This article studies the dynamics of the mean-field approximation of
continuous random networks. These networks are stochastic integrodifferential
equations driven by Gaussian noise. The kernels in the integral operators are
realizations of generalized Gaussian random variables. The equation controls
the time evolution of a macroscopic state interpreted as neural activity,
which depends on position and time. The position is an element of a measurable
space. Such a network corresponds to a statistical field theory (STF) given by
a momenta-generating functional. Discrete versions of the mentioned networks
appeared in spin glasses and as models of artificial neural networks (NNs).
Each of these discrete networks corresponds to a lattice SFT, where the action
contains a finite number of neurons and two scalar fields for each neuron.
Recently, it has been proposed that these networks can be used as models for
deep learning. In this application, the number of neurons is astronomical;
consequently, continuous models are required. In this article, we develop
mathematically rigorous, continuous versions of the mean-field theory
approximation and the double-copy system that allow us to derive a condition
for the criticality of continuous stochastic networks via the largest Lyapunov
exponent. It is essential to mention that the classical methods for mean-field
theory approximation and the double-copy based on the stationary phase
approximation cannot be used here because we are dealing with oscillatory
integrals on infinite dimensional spaces. To our knowledge, the approach
presented here is completely new. We use two basic architectures; in the first
one, the space of neurons is the real line, and then the neurons are organized
in one layer; in the second one, the space of neurons is the $p$-adic line,
and then the neurons are organized in an infinite, fractal, tree-like
structure. We also studied a toy model of a continuous Gaussian network with a
continuous phase transition. This behavior matches the critical brain
hypothesis, which states that certain biological neuronal networks work near
phase transitions.

\end{abstract}
\tableofcontents

\section{\label{SECT_1}Introduction}

This article is framed into the correspondence between continuous random
neural networks (NNs) and statistical field theories (SFTs). A continuous
random NN is given by a stochastic integrodifferential equation of the form%
\begin{gather}
\frac{d\boldsymbol{h}(x,t)}{dt}=-\gamma\boldsymbol{h}\left(  x,t\right)
-\widetilde{\boldsymbol{j}}\left(  x,t\right)  +%
{\displaystyle\int\limits_{\Omega}}
\boldsymbol{A}(x,y)\boldsymbol{h}\left(  y,t\right)  d\mu\left(  y\right)
\label{Continuous_Network}\\%
{\displaystyle\int\limits_{\Omega}}
\boldsymbol{B}(x,y)\boldsymbol{x}\left(  y,t\right)  d\mu\left(  y\right)  +%
{\displaystyle\int\limits_{\Omega}}
\boldsymbol{J}\left(  x,y\right)  \phi\left(  \boldsymbol{h}\left(
y,t\right)  \right)  d\mu\left(  y\right) \nonumber\\
+%
{\displaystyle\int\limits_{\Omega}}
\boldsymbol{U}\left(  x,y\right)  \varphi\left(  \boldsymbol{x}\left(
y,t\right)  \right)  d\mu\left(  y\right)  +\eta\left(  x,t\right)  \text{,
}\nonumber
\end{gather}
where $x\in\Omega$,\ $t\in\mathbb{R}$. The space of neurons $\Omega$ is an
infinite set, with measure $\mu$, and the real-valued function $\boldsymbol{h}%
(x,t)$ is a neural activity at the location $x$ and time $t$; we assume that
$\boldsymbol{h}(x,t)\in L^{2}\left(  \Omega\times\mathbb{R},d\mu dt\right)  $.
The functions $\boldsymbol{A}(x,y)$, $\boldsymbol{B}(x,y)$, $\boldsymbol{J}%
\left(  x,y\right)  $, and $\boldsymbol{U}\left(  x,y\right)  $ are
realizations of generalized Gaussian random variables in $L^{2}\left(
\Omega\times\Omega\right)  $, with mean zero, and $\eta\left(  x,t\right)  $
is a realization of a generalized Gaussian noise in $L^{2}(\Omega
\times\mathbb{R}\mathbb{)}$. The functions $\phi$ and $\varphi$ are activation
functions, sigmoidal functions of a general type, $\boldsymbol{x}\left(
x,t\right)  $, $\widetilde{\boldsymbol{j}}\left(  x,t\right)  $ are input
signals, and $\gamma>0$. The system (\ref{Continuous_Network}) is a stochastic
version of a continuous cellular neural network (CNN), \cite{Chua-Tamas},
\cite{Slavova}, \cite{Zambrano-Zuniga-1}-\cite{Zambrano-Zuniga-2}.

In the case $\boldsymbol{A}(x,y)=\boldsymbol{B}(x,y)=\boldsymbol{U}\left(
x,y\right)  =0$, the SFT associated with the network (\ref{Continuous_Network}%
) corresponds to a moment-generating functional of\ the form%
\begin{equation}
{\LARGE Z}\left(  \boldsymbol{j},\widetilde{\boldsymbol{j}};\boldsymbol{J}%
\right)  =%
{\displaystyle\iint}
D\boldsymbol{h}D\widetilde{\boldsymbol{h}}\exp\left(  S_{0}\left[
\boldsymbol{h},\widetilde{\boldsymbol{h}}\right]  -S_{int}\left[
\boldsymbol{h},\widetilde{\boldsymbol{h}};\boldsymbol{J}\right]  +\left\langle
\boldsymbol{j},\boldsymbol{h}\right\rangle +\left\langle \widetilde
{\boldsymbol{j}},\widetilde{\boldsymbol{h}}\right\rangle \right)  ,
\label{Zeta_function}%
\end{equation}
where the term $S_{0}\left[  \boldsymbol{h},\widetilde{\boldsymbol{h}}\right]
$ plays the role of free action and the term
\[
S_{int}\left[  \boldsymbol{h},\widetilde{\boldsymbol{h}};\mathbf{J}\right]
=\left\langle \widetilde{\boldsymbol{h}},%
{\displaystyle\int\limits_{\Omega}}
\boldsymbol{J}\left(  x,y\right)  \phi\left(  \boldsymbol{h}\left(
y,t\right)  \right)  d\mu\left(  y\right)  \right\rangle
\]
controls the interaction between the fields $\widetilde{\boldsymbol{h}}%
$,$\boldsymbol{h}$; here $\left\langle \cdot,\cdot\right\rangle $ is the usual
inner product on the Hilbert space $L^{2}\left(  \Omega\times\mathbb{R},d\mu
dt\right)  $. This article has two goals. The first one is to develop
mathematically rigorous, continuous versions of the mean-field theory
approximation and the double-copy system. It is essential to mention that the
methods based on the stationary phase approximation cannot be used here
because we are dealing with oscillatory integrals on infinite dimensional
spaces like $L^{2}\left(  \Omega\times\Omega\right)  $. We used white noise
calculus techniques, the Bochner-Minlos theorem, and Gaussian measures in
Hilbert spaces. To the best of our knowledge the approach presented here is
completely new.

As a first application, we formally adapt the technique of Sompolinsky et al.,
\cite{Sompolinsky et al}, that allows us to derive a condition for the
criticality of networks of type (\ref{Continuous_Network}) via the largest
Lyapunov exponent. In the case of continuous networks, the mean-field theory
approximation and the double-copy system associated with networks of type
(\ref{Continuous_Network}) are determined by a system of differential
equations on the covariance functions of the fields. The covariance functions
depend on time and the neurons' spatial distribution/concentration. Only time
derivatives appear in the mentioned systems, and the spatial variables appear
as parameters. Several types of spatial averages of the covariance functions
also satisfy the mentioned system of differential equations. This new
"self-averaging" phenomenon plays a central role in the dynamics of continuous
networks. As a second application, we study the mean-field of a toy model of a
random network. We show that the toy network works in the second phase of the
transition. In the section `The road map,' we discuss the mathematical results
presented here and many open problems.

The Buice-Cowan STFs play a central role in modern theoretical neuroscience,
\cite{Buice and Cowan}-\cite{Chow et al}. These SFTs are the limits of
discrete SFTs, and the rigorous study of these theories is very relevant. The
SFTs considered in this article correspond to CNNs, artificial neural networks
bioinspired in the Wilson-Cowan model. Discrete versions of CNNs can be
considered models of deep neural networks \cite{Grosvenor-Jefferson}, and
consequently, the corresponding discrete STFs have natural, formal,
thermodynamic limits as the standard Buice-Cowan SFTs. The work aims to
contribute to a mathematically rigorous understanding of the Buice-Cowan STFs
and their deep learning. We warn the reader that our contributions focus on
fundamental mathematical-physical aspects rather than machine learning applications.

We use three\ basic spaces of neurons: $\mathbb{R}$, $\mathbb{Q}_{p}$,
$\mathbb{Z}_{p}$, where $\mathbb{Q}_{p}$ is the field of $p$-adic numbers, and
$\mathbb{Z}_{p}$ is the ring of $p$-adic integers. In the case $\Omega
=\mathbb{R}$, the neurons are organized in one layer containing infinitely
many neurons. In the case $\Omega=\mathbb{Q}_{p}$, the neurons are organized
in a tree-like hierarchical structure. In contrast, $\Omega=\mathbb{Z}_{p}$,
the neurons are organized in an infinite rooted tree having infinitely many
layers; each layer contains a finite number of neurons. It is worth to mention
that $\mathbb{Q}_{p}$, $\mathbb{Z}_{p}$ are fractals. Networks of type
(\ref{Continuous_Network}) appear in models of cortical neural networks as
well as in computing as CNNs. The author and his collaborators have developed
image-processing algorithms based on $p$-adic CNNs, \cite{Zambrano-Zuniga-2},
\cite{Zuniga-images}. On the other hand, in \cite{Zuniga-Entropy}, a new
$p$-adic Wilson-Cowan model was introduced. Extensive numerical simulations
show that the new model predicts the same phenomena as the standard one, with
the advantage that the $p$-adic model may incorporate experimental data from
connection matrices. In \cite{Zuniga et al}, new $p$-adic Boltzmann machines
were implanted. These NNs are the discretization of continuous STFs. This
shows that continuous STFs, similar to the ones studied here, have
applications in machine learning.

Several researchers have hypothesized that in the cortical neural networks,
the neurons are organized in self-similar (fractal) patterns, with the neural
connections having a hierarchical ordering, see, e.g., \cite{Sporns}. In the
case of artificial NNs, this hypothesis drives to the notion of deep
architecture, understood as hierarchical structure where the neurons are
organized in layers, with connections between neurons in different layers,
see, e.g., \cite{Lecun et al}. For this reason, we argue, the $p$-adic
versions of the network (\ref{Continuous_Network}) are a very promising model.

Assuming that $\boldsymbol{A}(x,y)=\boldsymbol{B}(x,y)=\boldsymbol{U}\left(
x,y\right)  =0$, the partition function $\overline{{\LARGE Z}}\left(
\boldsymbol{0},\boldsymbol{0}\right)  $ of the mean-field theory associated to
${\LARGE Z}\left(  \boldsymbol{0},\boldsymbol{0};\boldsymbol{J}\right)  $ is
given by
\[
\overline{{\LARGE Z}}\left(  \boldsymbol{0},\boldsymbol{0}\right)
=\left\langle {\LARGE Z}\left(  \boldsymbol{0},\boldsymbol{0};\boldsymbol{J}%
\right)  \right\rangle _{\boldsymbol{J}};
\]
which is the partition function of a Gaussian SFT, where the fields are
elements of an infinite-dimensional Hilbert space. In this framework, the
correlation functional of the Gaussian measure is determined by a trace class
operator on the space of fields. So, a Gaussian NN is determined by picking
such an operator.

The second goal of this article is the rigorous study of a toy model of a
Gaussian NN using the techniques introduced here. We set $\Omega
=\mathbb{Z}_{p}$, so the fields $\boldsymbol{h}\left(  x,t\right)  \in
L^{2}\left(  \mathbb{Z}_{p}\times\mathbb{R}\right)  $, and pick a trace class
operator on $L^{2}\left(  \mathbb{Z}_{p}\times\mathbb{Z}_{p}\right)  $
depending on a real parameter $\rho\in\left(  1,\infty\right)  $; we interpret
$\rho$ as the control parameter of our model. In this framework,
$\overline{{\LARGE Z}}_{M}\left(  \rho\right)  $ has a pole at $\rho=2$, i.e.,
$\overline{{\LARGE Z}}_{M}\left(  \rho\right)  =\overline{{\LARGE Z}}%
_{M}\left(  \frac{1}{1-p^{\rho-2}}\right)  $, so
\[
\left.  \frac{d^{k}}{d\rho^{k}}\ln\overline{{\LARGE Z}}_{M}\left(
\rho\right)  \right\vert _{\rho=2}=\infty\text{, for }k=1,2,\ldots
\]
consequently, the network has a continuous phase transition at $\rho=2$. We
compute $G_{\boldsymbol{hh}}^{\rho}(x,y,\tau)=\left\langle \boldsymbol{h}%
(x,t_{1}),\boldsymbol{h}(y,t_{2})\right\rangle _{\boldsymbol{h}}$, $\tau
=t_{1}-t_{2}$, for $\boldsymbol{h}(x,t)$ in a dense subset of $L^{2}\left(
\mathbb{Z}_{p}\times\mathbb{R}\right)  $. A key observation is that for $t$
fixed $\boldsymbol{h}(x,t)\in L^{2}\left(  \mathbb{Z}_{p}\right)  $, and then
it can be approximated for a linear combination of characteristic functions of
ball contained in $\mathbb{Z}_{p}$. For $\rho\in\left(  1,\infty\right)
\smallsetminus\left\{  2\right\}  $, $G_{\boldsymbol{hh}}^{\rho}(x,y,\tau)$ is
controlled by a term of the form $-\sigma^{2}\frac{e^{\rho\tau}}{2\rho}\left(
1-e^{-2\rho\tau}\right)  \delta\left(  x-y\right)  $, which corresponds a
white noise, which we interpret as the `baseline state/pattern' (or background
noise) of the network. For $\rho=2$,%
\[
G_{\boldsymbol{hh}}^{2}(x,y,\tau)\approx C(\tau,\phi)+f\left(  \left\vert
y-x\right\vert _{p}\right)  +\sigma^{2}\frac{e^{\rho\tau}}{2\rho}\left(
1-e^{-2\rho\tau}\right)  \delta\left(  x-y\right)  ,
\]
where $f$ is a positive, non-bounded, increasing function of the distance
($\left\vert y-x\right\vert _{p}$) between the neurons at positions $x$, $y$.
Interpreting $G_{\boldsymbol{hh}}^{\rho}(x,y,\tau),\rho\in\left(
1,\infty\right)  ,$ as an order parameter, we argue that $G_{\boldsymbol{hh}%
}^{\rho}(x,y,\tau)$, $\rho\in\left(  1,\infty\right)  \smallsetminus\left\{
2\right\}  $ describes a disordered phase, meaning that it is completely
controlled it by the background noise and that the short-range or large-range
neuronal connections are not relevant. The phase described by
$G_{\boldsymbol{hh}}^{2}(x,y,\tau)$ is ordered; the long-range interactions
between the neurons control the network behavior. In addition, for $\tau$
fixed, these interactions cannot be masked by the background noise. This
behavior matches the critical brain hypothesis, which states that certain
biological neuronal networks work near phase transitions, \cite{Chialvo},
\cite{Hesse et al}, see also \cite{Buice and Cowan}.

By a well-known paradigm in statistical physics, the phase transitions and
their associated non-analyticities appear only in NNs with infinitely many
neurons, i.e., in the thermodynamic limit when the number of neurons tends to
infinity. A mathematical framework is required to study NNs with infinitely
many neurons. Our toy model shows that our mathematical framework helps us
understand the organization of NNs with infinite neurons. Our mathematical
framework allows us to extend and adapt the techniques based on the largest
Lyapunov exponent of a random NN to assess the conditions under which the
network transitions into the chaotic regime; as a consequence, the continuous
NNs considered here exhibit a type edge of the chaos organization,
\cite{Sompolinsky et al}, \cite[Chapter 10]{Helias et al},
\cite{Grosvenor-Jefferson}. It is an interesting open question is to analyze
the organization of our toy NN from the perspective of the largest Lyapunov exponent.

The practical implications of the mathematical framework developed here in
deep learning and neuroscience applications have not been fully explored. Of
course, these matters are out of scope in this work, which is the development
of mathematical techniques to study random NNs with infinitely many neurons.
For instance, it is relevant to investigate the specific connections to
challenges in modern deep learning (e.g., training stability, expressivity, or
generalization) and to explore whether the framework introduced here could
help understand phenomena like gradient explosion/vanishing in deep networks
or assist in optimizing network depth.

The existence of a correspondence between large/deep neural networks and
statistical field theories (SFTs), also called Euclidean quantum field
theories, is a promising idea for understanding the dynamics of large/deep
neural networks. This correspondence can take several different forms, as
demonstrated by the extensive literature available, see, e.g., \cite{Buice and
Cowan}-\cite{Chow et al}, \cite{Demirtsas et al}-\cite{Erbin et al},
\cite{Grosvenor-Jefferson}-\cite{Helias et al}, \cite{Roberts et
al}-\cite{Segadlo et al}, \cite{Zuniga et al}-\cite{Zuniga-PhyA}, among many
references. In particular, the author and his collaborators have developed a
rigorous mathematical theory for the correspondence of deep Boltzmann machines
(DBNs) having a tree-like topology and certain SFTs, \cite{Zuniga-ATMP}%
-\cite{Zuniga-PhyA}. $p$-Adic numbers were used to encode this type of
tree-like topology. In this framework, a $p$-adic continuous DBM is a
statistical field theory defined by an energy functional on the space of
square-integrable functions defined on a $p$-adic $N$-dimensional ball. In
\cite{Zuniga et al}, some of these $p$-adic deep Boltzmann machines were
implemented. Furthermore, the discrete $p$-adic Boltzmann machines are
universal approximators, \cite{Zuniga-PhyA}.

The study of the correspondence between random networks NNs and SFT is a
relevant problem. For instance, in \cite{Schoenholz et al}, the authors show
that random neural networks (with a finite number of neurons) can be precisely
mapped onto lattice SFTs. We argue that the rigorous study of the
thermodynamic limit of these lattice STFs is important in understanding the
organization of random NNs.

The models presented in \cite{Roberts et al}, \cite{Segadlo et al} use
timeless lattice field theories with two variables ($\boldsymbol{h}%
_{i},\widetilde{\boldsymbol{h}}_{i}$) for each neuron, while in
\cite{Grosvenor-Jefferson} uses\ a lattice field theory with continuous time
and two variables ($\boldsymbol{h}_{i}\left(  t\right)  ,\widetilde
{\boldsymbol{h}}_{i}\left(  t\right)  $) for each neuron. In practical
applications (like ChatGPT), NNs with billions of neurons naturally happen;
the mentioned lattice SFTs use in the definition of the action billions of
variables of the\ types ($\boldsymbol{h}_{i},\widetilde{\boldsymbol{h}}_{i}$)
or ($\boldsymbol{h}_{i}\left(  t\right)  ,\widetilde{\boldsymbol{h}}%
_{i}\left(  t\right)  $). We argue that these networks should be considered as
having infinite neurons distributed continuously in a certain space. This
approach is based on the well-known principle in statistical physics: when a
system contains a very large number of particles, their individual behaviors
tend to average out, leading to a smooth, continuous behavior on a macroscopic
scale. The corresponding SFTs should use time and space continuous variables.

On the other hand, most of the STFs proposed for the deep NNs are lattice
SFTs, for\ instance, \cite{Roberts et al}, \cite{Segadlo et al},
\cite{Grosvenor-Jefferson}, among many works available. A central problem is
to find the thermodynamic limit of these lattice STFs. This article aims to
study the thermodynamic limit of the networks introduced by Grosvenor and
Jefferson in \cite{Grosvenor-Jefferson}.The continuous versions of this type
of network have the form (\ref{Continuous_Network}). The Buice-Cowan SFTs are
similar to (\ref{Zeta_function}); in these models, the time and space are
continuous variables, and the network has infinitely many neurons. These
models were obtained by taking the thermodynamic limit of certain lattice
SFTs, \cite{Buice-Cowean-2007}.

In this work, we use thermodynamic limit as the process for passing from a
lattice SFT to a continuous SFT. This thermodynamic limit approach is beyond
the central limit theorem (CLT). Consider a term of the form
\begin{equation}
\sum_{j=1}^{N}\mathbb{J}_{ij}\phi\left(  \boldsymbol{h}_{j}\left(  t\right)
\right)  , \label{Eq_term}%
\end{equation}
where $N$ is the number of neurons of the network. Assume that the
$\mathbb{J}_{ij}$ are independent identically Gaussian distributed\ random
coupling weights with zero mean and variance $g^{2}/N$. The CLT cannot be used
to conclude that the mentioned term converges to $\int_{\Omega}\boldsymbol{J}%
\left(  x,y\right)  \phi\left(  \boldsymbol{h}\left(  y,t\right)  \right)
d\mu\left(  y\right)  $; we can only expect that $\boldsymbol{J}$ be a
generalized Gaussian random variable with mean zero. The thermodynamic limit
of lattice STFs is not well-posed problem, which means that starting with
discrete action, we can attach to it several limits (continuos actions) by
interpreting the terms of the form (\ref{Eq_term}) as approximations of
integrals in some measurable space $(\Omega,\mathcal{B},d\mu)$.

The construction of continuous versions of the SFTs for NNs involves several
new and old challenges. This article assumes that the fields $\boldsymbol{h}%
(x,t)$ and $\widetilde{\boldsymbol{h}}(x,t)$ are elements from a Hilbert
subspace $\mathcal{H}$ of $L^{2}\left(  \Omega\times\mathbb{R},d\mu dt\right)
$. The first problem is to know if the fields are well-defined on all points
of $\Omega\times\mathbb{R}$; this is true because $\mathcal{H}$ admits an
orthonormal basis consisting of continuous functions (see Lemma \ref{Lemma_2}%
-(ii) in Appendix A). We also require that the evaluation of functions from
$\mathcal{H}$ at some fixed point $(x_{0},t_{0})\in\Omega\times\mathbb{R}$
gives rise to a usual Gaussian random variable with mean of zero. This
assertion is true, see Lemma \ref{Lemma_16} in Appendix E, when $(x_{0}%
,t_{0})$ belongs to $\Omega\times\mathbb{R}$ $\mathbb{\smallsetminus}$
$\mathcal{M}$, where $\mathcal{M}$ is $d\mu dt$-measure zero set. The
evaluation map $\boldsymbol{h}\rightarrow\boldsymbol{h}(x_{0},t_{0})$ agrees
with a linear map $\left\langle \boldsymbol{h},\boldsymbol{v}(x_{0}%
,t_{0})\right\rangle _{\mathcal{H}}$, where $\boldsymbol{v}(x_{0},t_{0}%
)\in\mathcal{H}$, and $\left\langle \boldsymbol{\cdot},\boldsymbol{\cdot
}\right\rangle _{\mathcal{H}}$\ is the inner product\ in $\mathcal{H}$, almost
everywhere in $(x_{0},t_{0})$. This means that $\boldsymbol{h}(x,t)$ can be
extended to a smeared field like in ordinary QFT. Then, given two functions in
$\boldsymbol{h}_{1},\boldsymbol{h}_{2}\in\mathcal{H}$, $\left\langle
\boldsymbol{h}_{1}(x,t)\boldsymbol{h}_{2}(x,t)\right\rangle _{\boldsymbol{h}%
_{1},\boldsymbol{h}_{2}\text{ }}$is the cross-correlation of two ordinary
Gaussian random variables with mean zero, for $(x,t)\in\Omega\times\mathbb{R}$
$\mathbb{\smallsetminus}$ $\mathcal{M}^{\prime}$, where $\mathcal{M}^{\prime}$
has measure zero. This fact allows us to use the techniques of \cite[Chapter
10]{Helias et al} to study the dynamics of continuous random NNs.

We systematically use Gel'fand triplets jointly with the Bochner-Minlos
theorem to construct measures on infinite dimensional spaces. The construction
of these triplets is well-known in the case $\Omega=\mathbb{R}$. Here, we give
some new triplets in the cases $\Omega=\mathbb{Z}_{p}$, $\mathbb{Q}_{p}$, see
Appendix\ C. Non-mathematically inclined readers may disregard the
construction of measures on the space of fields without any prejudice in
understanding the main results presented here. These measures can be
considered in a formal way, as happens in the standard literature of quantum
field theory. The construction of the toy model requires basic techniques of
$p$-adic analysis. Appendix H provides an overview of the essential aspects of
the $p$-adic analysis. For an in-depth exposition the reader may consult
\cite{A-K-S}, \cite{Taibleson}, \cite{V-V-Z}, \cite{Zuniga-Textbook}.

In the STFs used in modeling of NNs, the fields represent signals (or data).
Assuming that the fields are functions from $L^{2}\left(  \Omega
\times\mathbb{R},d\mu dt\right)  $ is quite convenient because the $L^{2}%
$-norm corresponds to the energy of a signal, and the physical signals have
finite energy. For a given NN, the signals it produces should have finite
energy bounded by some fixed positive constant $M$; this means that the
signals are elements from a ball in $L^{2}\left(  \Omega\times\mathbb{R},d\mu
dt\right)  $ with radius $M$. The characteristic function of this ball acts as
a cutoff function for the moment-generating functionals associated with the
NN, i.e., the integration over fields is restricted to the elements of a ball
with radius $M$. Then, our notion of cutoff differs radically from the
classical one, where the cutoff is used to restrict the domain of the fields
to a bounded subset $\Omega\times\mathbb{R}$. The existence of a natural
cutoff for STFs associated with NNs was already pointed out by the author in
\cite{Zuniga-ATMP}; this cutoff allows the rigorous construction of partition
functions for NNs. It is relevant to mention that using the mentioned cutoff
does not prevent the propagators computed from the partition function from
having singularities, which in turn may produce infinities in the perturbative calculations.

We have written the article for an interdisciplinary audience. For this
reason, all the mathematical proofs have been placed in appendices at the end
of the article. In this way, it is possible to read the article without
needing to go through all the mathematical proofs. We focus on the case
$\boldsymbol{A}(x,y)=\boldsymbol{B}(x,y)=\boldsymbol{U}\left(  x,y\right)
=0$. The extension to networks of the type (\ref{Continuous_Network})\ is
straightforward. The mean-field theory for the discrete counterparts of this
type networks is presented in the book \cite[Chapter 10]{Helias et al}. This
chapter exposes the relation between the mean-field theory of spin glasses and
neural networks \ from the perspective of the seminal work by Sompolinsky et
al., \cite{Sompolinsky et al}. We develop a mathematical framework so that the
rigorous study of the transition to chaos for networks of type
(\ref{Continuous_Network}) can be carried out using the physical reasoning
given in \cite{Grosvenor-Jefferson}, \cite{Helias et al}, \cite{Schuecker et
al}, \cite{Sompolinsky et al}.

\section{\label{SECT_2}Related work}

\subsection{Buice-Cowan statistical field theories}

We review here the correspondence NN-SFT from the perspective of the
Buice-Cowan work on SFTs for the cortex and the techniques to study stochastic
equations via path integrals, see \cite{Buice and Cowan}-\cite{Chow et al},
see also \cite{Neural-Fields}, and the references therein.

Consider a Langevin equation of the type,%
\begin{equation}
\left\{
\begin{array}
[c]{l}%
\frac{d\boldsymbol{\varphi}}{dt}=f(\boldsymbol{\varphi}%
,t)+g(\boldsymbol{\varphi},t)\eta\left(  t\right)  ,\\
\\
\boldsymbol{\varphi}\left(  0\right)  =\boldsymbol{y},
\end{array}
\right.  \label{Eq_20}%
\end{equation}
where the stochastic forcing term obeys to $\left\langle \eta\left(  t\right)
\right\rangle =0$ and $\left\langle \eta\left(  t\right)  \eta\left(
s\right)  \right\rangle =\delta\left(  t-s\right)  $. By interpreting this
equation as an It\^{o} stochastic differential equation, the probability
density of function of the trajectory $\boldsymbol{\varphi}\left(  t\right)  $
is given by%
\[
\boldsymbol{P}\left(  \boldsymbol{\varphi}\left(  t\right)  \mid
\boldsymbol{y}\right)  =%
{\displaystyle\int}
D\widetilde{\boldsymbol{\varphi}}\left(  t\right)  \exp\left(
-S(\boldsymbol{\varphi},\widetilde{\boldsymbol{\varphi}})\right)  ,
\]
where
\[
S(\boldsymbol{\varphi},\widetilde{\boldsymbol{\varphi}})=%
{\displaystyle\int}
\left[  \widetilde{\boldsymbol{\varphi}}\left(  t\right)  \left(  \frac{d}%
{dt}\boldsymbol{\varphi}\left(  t\right)  -f\left(  \boldsymbol{\varphi
}\left(  t\right)  ,t\right)  -\boldsymbol{y}\delta\left(  t\right)  \right)
+\frac{1}{2}\widetilde{\boldsymbol{\varphi}}^{2}\left(  t\right)  g^{2}\left(
\boldsymbol{\varphi}\left(  t\right)  ,t\right)  \right]  dt.
\]
The generating functional for $\boldsymbol{\varphi}\left(  t\right)  $,
$\widetilde{\boldsymbol{\varphi}}\left(  t\right)  $ is given by
\begin{equation}
{\LARGE Z}(\boldsymbol{j}\left(  t\right)  ,\widetilde{\boldsymbol{j}}\left(
t\right)  )=%
{\displaystyle\iint}
D\boldsymbol{\varphi}\left(  t\right)  D\widetilde{\boldsymbol{\varphi}%
}\left(  t\right)  e^{S(\boldsymbol{\varphi},\widetilde{\boldsymbol{\varphi}%
})+%
{\textstyle\int}
\boldsymbol{j}\left(  t\right)  \boldsymbol{\varphi}\left(  t\right)  dt+%
{\textstyle\int}
\widetilde{\boldsymbol{j}}\left(  t\right)  \widetilde{\boldsymbol{\varphi}%
}\left(  t\right)  }dt,\label{Eq_21}%
\end{equation}
see \cite[Section 3]{Chow et al}. So, we have a correspondence between the
network (\ref{Eq_20}) and the SFT (\ref{Eq_21}). In \cite{Buice-Cowean-2007},
Buice and Cowan introduced a general type of SFTs associated with the
Cowan-Wilson rate equations,
\begin{equation}
{\LARGE Z}(\boldsymbol{j}\left(  x,t\right)  ,\widetilde{\boldsymbol{j}%
}\left(  x,t\right)  )=%
{\displaystyle\iint}
D\boldsymbol{\varphi}\left(  x,t\right)  D\widetilde{\boldsymbol{\varphi}%
}\left(  x,t\right)  e^{S_{BC}(\boldsymbol{\varphi},\widetilde
{\boldsymbol{\varphi}})+\boldsymbol{j}\cdot\boldsymbol{\varphi}+\widetilde
{\boldsymbol{j}}\cdot\widetilde{\boldsymbol{\varphi}}},\label{Eq_22}%
\end{equation}
where the action is given by%
\[
S_{BC}(\boldsymbol{\varphi},\widetilde{\boldsymbol{\varphi}}):=S_{0}%
(\boldsymbol{\varphi},\widetilde{\boldsymbol{\varphi}})+S_{int}%
(\boldsymbol{\varphi},\widetilde{\boldsymbol{\varphi}})-W\left[
\widetilde{\boldsymbol{\varphi}}\left(  x,0\right)  \right]  ,
\]%
\[
S_{0}(\boldsymbol{\varphi},\widetilde{\boldsymbol{\varphi}})=%
{\textstyle\iint}
d^{N}xdt\text{ }\widetilde{\boldsymbol{\varphi}}\left(  x,t\right)  \left[
\frac{\partial}{\partial t}\boldsymbol{\varphi}\left(  x,t\right)
+\alpha\boldsymbol{\varphi}\left(  x,t\right)  \right]  ,
\]%
\begin{equation}
S_{int}(\boldsymbol{\varphi},\widetilde{\boldsymbol{\varphi}})=-%
{\textstyle\iint}
d^{N}xdt\left[  f\left(
{\textstyle\int}
w\left(  x-y\right)  \left[  \widetilde{\boldsymbol{\varphi}}\left(
y,t\right)  \boldsymbol{\varphi}\left(  y,t\right)  +\boldsymbol{\varphi
}\left(  y,t\right)  \right]  d^{N}y\right)  \right]  ,\label{Eq_21B}%
\end{equation}
here $f$ is the activation function of the network, a sigmoidal type function,
$w\left(  x-y\right)  $ gives the strength of the connection between the
neurons located at positions $x$ and $y$, $\boldsymbol{u}\cdot\boldsymbol{v}%
=\int\int$ $d^{N}xdt$ $\boldsymbol{u}(x,t)\boldsymbol{v}(x,t)$, and $W\left[
\widetilde{\boldsymbol{\varphi}}\left(  x,0\right)  \right]  $\ is the
cumulant generating function functional of the initial distribution. This
generating functional is obtained by taking the limit when the number of
neurons tends to infinity. This calculation assumes that the neurons are
organized in an $N$-dimensional lattice contained in $\mathbb{R}^{N}$ (or
$N$-dimensional cube); however, this hypothesis is not necessary or convenient.

In the lattice version of (\ref{Eq_22}), the fields $\widetilde
{\boldsymbol{\varphi}}\left(  x,t\right)  $, $\boldsymbol{\varphi}\left(
x,t\right)  $ are replaced by $\widetilde{\boldsymbol{\varphi}}_{i}\left(
t\right)  $, $\boldsymbol{\varphi}_{i}\left(  t\right)  $, for $i=1,\ldots,N$,
where $N$ is the number of neurons in the network. As mentioned in the
introduction of \cite{Buice and Cowan}, there are approximately $3\cdot
10^{10}$ neurons in the human neocortex, each supporting up to $10^{4}$
synaptic contacts, all packed in a volume of $3000cc$. By
statistical-mechanics considerations, a system with $6\cdot10^{10}$ degrees of
freedom should be studied through a continuous model.

In \cite{Buice-Cowan-Chow}, generating functionals of type (\ref{Eq_22}) were
heuristically derived starting with an equation of type Wilson-Cowan-Langevin,%
\begin{equation}
\frac{\partial}{\partial t}\boldsymbol{\varphi}\left(  x,t\right)
=-\alpha\boldsymbol{\varphi}\left(  x,t\right)  +F\left(  \boldsymbol{\varphi
}\left(  x,t\right)  \right)  +\boldsymbol{\varphi}_{0}\left(  x\right)
\delta\left(  t\right)  +\xi\left(  x,t\right)  , \label{Eq_general network}%
\end{equation}
were $\boldsymbol{\varphi}\left(  x,t\right)  $ is a neural activity at the
location $x$ and time $t$, $\xi\left(  x,t\right)  $ is a stochastic forcing
with probability density $\boldsymbol{P}(\xi)$, and the firing rate $F$ is not
necessarily the function $f$ appearing in (\ref{Eq_21B}). The probability
density for $\boldsymbol{\varphi}\left(  x,t\right)  $, up to multiplication
for a positive constant, is
\[
\boldsymbol{P}\left(  \boldsymbol{\varphi}\left(  x,t\right)  \mid
\boldsymbol{\varphi}_{0}\left(  x\right)  \right)  =%
{\displaystyle\int}
D\xi\text{ }\delta\left[  \frac{\partial}{\partial t}\boldsymbol{\varphi
}-\alpha\boldsymbol{\varphi}-F\left(  \boldsymbol{\varphi}\right)
-\boldsymbol{\varphi}_{0}\delta\left(  t\right)  -\xi\left(  x,t\right)
\right]  \boldsymbol{P}(\xi),
\]
where $\delta\left(  \cdot\right)  $ is a functional generalization of the
Dirac delta function. By using the generalized Fourier transform of
$\delta\left(  \cdot\right)  $,%
\begin{equation}
\boldsymbol{P}\left(  \boldsymbol{\varphi}\left(  x,t\right)  \mid
\boldsymbol{\varphi}_{0}\left(  x\right)  \right)  =%
{\displaystyle\iint}
D\widetilde{\boldsymbol{\varphi}}\left(  x,t\right)  \exp\left[  S_{0}\left(
\widetilde{\boldsymbol{\varphi}},\boldsymbol{\varphi}\right)  +S_{int}\left(
\widetilde{\boldsymbol{\varphi}},\boldsymbol{\varphi}\right)  -W(\widetilde
{\boldsymbol{\varphi}},\boldsymbol{\varphi}_{0})\right]  ), \label{Eq_24}%
\end{equation}
where%
\[
S_{0}\left(  \widetilde{\boldsymbol{\varphi}},\boldsymbol{\varphi}\right)  =%
{\displaystyle\int}
d^{N}xdt\text{ }\widetilde{\boldsymbol{\varphi}}\left(  x,t\right)
\widetilde{\boldsymbol{\varphi}}\left(  x,t\right)  \left\{  \frac{\partial
}{\partial t}\boldsymbol{\varphi}\left(  x,t\right)  +\alpha
\boldsymbol{\varphi}\left(  x,t\right)  -\boldsymbol{\varphi}_{0}\left(
x\right)  \delta\left(  t\right)  \right\}  ,
\]%
\[
S_{int}\left(  \widetilde{\boldsymbol{\varphi}},\boldsymbol{\varphi}\right)
=-%
{\displaystyle\int}
d^{N}xdt\text{ }\widetilde{\boldsymbol{\varphi}}\left(  x,t\right)  F\left(
\boldsymbol{\varphi}\left(  x,t\right)  \right)
\]

\[
W(\widetilde{\boldsymbol{\varphi}},\boldsymbol{\varphi}_{0})=%
{\displaystyle\int}
D\xi\left(  x,t\right)  \exp\left\{
{\displaystyle\int}
d^{N}xdt\text{ }\widetilde{\boldsymbol{\varphi}}\left(  x,t\right)
\boldsymbol{\varphi}_{0}\left(  x\right)  \right\}  \boldsymbol{P}(\xi\left(
x,t\right)  ,
\]
see \cite{Buice-Cowan-Chow}\ for details. The generating functional
corresponding to (\ref{Eq_24}) is a generalization of (\ref{Eq_22}).\ In the
works \cite{Buice and Cowan}-\cite{Chow et al}, $d^{N}x$ is intended to
denoted the Lebesgue measure of $\mathbb{R}^{N}$. However, this is not
necessary. The arguments given in these works are valid if we assume that the
neurons form an abstract space with measure $d^{N}x$. This work aims to start
the studying of the chaotic behavior of networks of type
(\ref{Eq_general network}) using the correspondence NN-SFT via the largest
Lyapunov exponent.

\subsection{Partial differential equation based models for neural fields}

The spatiotemporal continuum models for the dynamics of macroscopic activity
patterns in the cortex were introduced in the 1970s following the seminal work
by Wilson and Cowan, Amari, among others, see, e.g., \cite{Wilson-Cowan-1}%
-\cite{Wilson-Cowan-2}, \cite{Amari}, see also \cite{Neural-Fields} and the
references therein. Such models take the form of integrodifferential evolution
equations. The integration kernels represent the strength of the connection
between neurons, or more generally, the spatial distribution of synaptic
connections between different neural populations, and the state macroscopic
variables represent some average of neural activity.

The simplest continuous model in one spatial dimension is%
\begin{equation}
\frac{\partial u\left(  x,t\right)  }{\partial t}=-u\left(  x,t\right)  +%
{\displaystyle\int\nolimits_{-\infty}^{\infty}}
J\left(  x-y\right)  \phi\left(  u\left(  y,t\right)  \right)  dy,
\label{Model_0}%
\end{equation}
where at the position $x\in\mathbb{R}$ there is a group of neurons, and
$u\left(  x,t\right)  $ is some average network activity at the time
$t\in\mathbb{R}$. The kernel $J$ describes the strength of the connection
between the neurons located at the positions $x$ and $y$. We assume that the
neurons are organized in an infinite straight line. The non-linear function
$\phi$ is the activation (or firing) function. Models of type (\ref{Model_0}%
)\ provide an approximate description of the mean-field dynamics of a neural network.

There are two basic discrete models. The first one is the voltage-based rate
model:%
\begin{equation}
\tau_{m}\frac{\partial a_{i}\left(  t\right)  }{\partial t}=-a_{i}\left(
t\right)  +%
{\displaystyle\sum\limits_{j\in\Omega_{N}}}
\mathbb{J}_{ij}\phi\left(  a_{j}\left(  t\right)  \right)  , \label{Model_1}%
\end{equation}
where $\tau_{m}$ is a time constant and $\mathbb{J}_{ij}$ is the strength
between neurons $i$ and $j$, and $\Omega_{N}=\left\{  j_{1},\ldots
,j_{N}\right\}  $ is the space of neurons. We assume that $\Omega_{N}$ has a
`topology' which provides the architecture of the network. For instance, if
the neurons are organized in a hierarchical tree-like structure, $\Omega_{N}$
is a finite union of rooted trees. In the most common choice, $\Omega_{N}$ is
a finite subset of the lattice $\mathbb{Z}^{D}$.

The second one is called the activity-based model; it has the form%
\begin{equation}
\tau_{s}\frac{\partial v_{i}\left(  t\right)  }{\partial t}=-v_{i}\left(
t\right)  +\phi\left(
{\displaystyle\sum\limits_{j\in\Omega_{N}}}
\mathbb{J}_{ij}v_{j}\left(  t\right)  \right)  , \label{Model_2}%
\end{equation}
where $\tau_{s}$ is a time constant. The reader may consult \cite{Bressloff}
for an in-depth discussion about these two models. The continuous versions of
the models (\ref{Model_1})-(\ref{Model_2}) are obtained by taking the limit
when the number of neurons tends to infinity. In the limit, the discrete
variable $i$ becomes a continuous variable $x$. Almost all the published
literature assumes that the neurons are organized in a lattice contained in
$\mathbb{R}^{n}$. Here, we assume that there exists a measure space $\left(
\Omega,\Sigma,d\mu\right)  $ such that
\begin{align*}
\lim_{N\rightarrow\infty}%
{\displaystyle\sum\limits_{j\in\Omega_{N}}}
\mathbb{J}_{ij}\phi\left(  a_{j}\left(  t\right)  \right)   &  =%
{\displaystyle\int\limits_{\Omega}}
J(x,y)\phi\left(  a\left(  y,t\right)  \right)  d\mu\left(  y\right)  \text{
and }\\
& \\
\lim_{N\rightarrow\infty}\phi\left(
{\displaystyle\sum\limits_{j\in\Omega_{N}}}
\mathbb{J}_{ij}v_{j}\left(  t\right)  \right)   &  =\phi\left(
{\displaystyle\int\limits_{\Omega}}
J(x,y)v\left(  y,t\right)  d\mu\left(  y\right)  \right)  .
\end{align*}

Then, the continuous versions of the models (\ref{Model_1})-(\ref{Model_2})
are%
\[
\tau_{m}\frac{\partial a\left(  x,t\right)  }{\partial t}=-a\left(
x,t\right)  +%
{\displaystyle\int\limits_{\Omega}}
J(x,y)\phi\left(  a\left(  y,t\right)  \right)  d\mu\left(  y\right)  ,
\]%
\[
\tau_{s}\frac{\partial v\left(  x,t\right)  }{\partial t}=-v\left(
x,t\right)  +\phi\left(
{\displaystyle\int\limits_{\Omega}}
J(x,y)v\left(  y,t\right)  d\mu\left(  y\right)  \right)  ,
\]
where $x\in\Omega$, $t\in\mathbb{R}$.

\subsection{Cellular neural Networks as deep neural networks}

In the late 80s, Chua and Yang introduced a new natural computing paradigm
called the cellular neural networks (CNNs), which included the cellular
automata as a particular case. This paradigm has been highly successful in
various applications in vision, robotics, and remote sensing, among many
applications, \cite{Chua-Tamas}, \cite{Slavova}. Recently, the author and
Zambrano-Luna introduced $p$-adic CNNs, which are CNNs where the neurons are
organized in an infinite-tree-like structure, meaning that this type of neural
network has a natural deep architecture. The $p$-adic CNNs have been used in
image processing tasks, \cite{Zambrano-Zuniga-1}-\cite{Zambrano-Zuniga-2},
\cite{Zuniga-images}. In a recent article, Grosvenor and Jefferson studied a
class of stochastic CNNs using SFT techniques and proposed such type of
network as a deep neural network model, \cite{Grosvenor-Jefferson}.

Let $\Omega_{N}$ be the set neurons. We assume that the cardinality of
$\Omega$ is $N$. A discrete CNN is a dynamical system on $\Omega_{N}$. The
state $\boldsymbol{h}_{i}(t)\in\mathbb{R}$ of the neuron $i$ is described by
the following system of differential equations:
\begin{gather}
\frac{d\boldsymbol{h}_{i}(t)}{dt}=-\gamma\boldsymbol{h}_{i}\left(  t\right)
+\sum_{j\in\Omega_{N}}\mathbb{A}_{i,j}^{\left(  N\right)  }\boldsymbol{h}%
_{i}\left(  t\right)  +\sum_{j\in\Omega_{N}}\mathbb{B}_{i,j}^{\left(
N\right)  }\boldsymbol{x}_{j}(t)\label{Model_3}\\
+\sum_{j\in\Omega_{N}}\mathbb{J}_{i,j}^{\left(  N\right)  }\phi\left(
\boldsymbol{h}_{i}\left(  t\right)  \right)  +\sum_{j\in\Omega_{N}}%
\mathbb{U}_{i,j}^{N}\varphi\left(  \boldsymbol{x}_{i}\left(  t\right)
\right)  +\eta_{i}^{\left(  N\right)  }\left(  t\right)  \text{, }i\in
\Omega_{N}\text{,\ \ }i=1,\ldots,N\text{,}\nonumber
\end{gather}
where $\phi,\varphi:\mathbb{R}\rightarrow\mathbb{R}$ are firing (or
activation) functions, typically bounded, differentiable, Lipschitz functions.
The function $\boldsymbol{x}_{i}(t)\in\mathbb{R}$ is the input of the cell $i$
at time $t$, the real matrices $\mathbb{A}^{\left(  N\right)  }$,
$\mathbb{B}^{(N)}$, $\mathbb{J}^{\left(  N\right)  }$, $\mathbb{U}^{\left(
N\right)  }$ control the architecture of the network, and the function
$\eta_{i}^{(N)}\left(  t\right)  \in\mathbb{R}$ is the threshold of the neuron
$i$. In the classical CNNs, the output of the neuron $i$ at the time $t$ is
defined as $\boldsymbol{y}_{i}(t)=\phi(\boldsymbol{h}_{i}(t))$.

The model (\ref{Model_3}) is a generalization of the classical and $p$-adic
CNNs, as well as the recurrent neural networks, and multilayer perceptron.
This model is a generalization of the model proposed in
\cite{Grosvenor-Jefferson}. In the stochastic version, the $\mathbb{A}%
_{i,j}^{(N)}$, $\mathbb{B}_{i,j}^{(N)}$, $\mathbb{J}_{i,j}^{(N)}$,
$\mathbb{U}_{i,j}^{(N)}$ are identically distributed random variables, here,
we assume that they are Gaussian variables with mean zero, and $\eta_{i}%
^{(N)}\left(  t\right)  $ is a noise. In this article, we study the chaotic
behavior of\ continuous versions of NNs introduced in
\cite{Grosvenor-Jefferson}.

If we expect that model (\ref{Model_3}) can be helpful in the study of
large/deep neural networks, then $N$ should be allowed to take values of order
$10^{11}$. For instance, in the case of GPT-3, it is widely known that it
consists of 175 billion parameters, akin to the \textquotedblleft
neurons\textquotedblright\ in the AI model. Then, we require the thermodynamic
limit ($N\rightarrow\infty$) of model (\ref{Model_3}) in the sense used in
this article, which means to get a moment-generating functional involving an
action with time and space continuous variables, \cite[Section 7.1]{Buice and
Cowan}. In contrast, in \cite{Grosvenor-Jefferson}, \cite{Helias et al} and in
many other publications, a network is studied in the the thermodynamic limit
($N\rightarrow\infty$) when the assertions about averages of observables
become exact in the limit $N\rightarrow\infty$ due to the central limit theorem.

To a random NN of type (\ref{Model_3}) corresponds an SFT given by a
generating functional $Z(\boldsymbol{j},\widetilde{\boldsymbol{j}}%
;\mathbb{A}^{\left(  N\right)  },\mathbb{B}^{(N)},\mathbb{J}^{\left(
N\right)  },\mathbb{U}^{\left(  N\right)  })$; here we assume that the
activation functions $\phi,\varphi$ and noise $\eta_{i}^{(N)}\left(  t\right)
$ are fixed. This type of network has been studied extensively using
techniques like mean-field approximation and double-copy; see for instance,
the Helias \& Dahmen book \cite[Chapter 10]{Helias et al}, and
\cite{Grosvenor-Jefferson}. In particular, the study of the chaotic behavior
of such networks has received significant attention after the seminal work of
Sompolinsky et al., \cite{Sompolinsky et al}. In this type of model,
${\LARGE Z}(\boldsymbol{j},\widetilde{\boldsymbol{j}};\mathbb{A}^{\left(
N\right)  },\mathbb{B}^{(N)},\mathbb{J}^{\left(  N\right)  },\mathbb{U}%
^{\left(  N\right)  })$ contains variables ( $\widetilde{\boldsymbol{\varphi}%
}_{j}(t)$, $\boldsymbol{\varphi}_{j}(t)$) for each neuron $j$, so this
functional may depend on an astronomical number of variables, and thus one
should consider
\[
\lim_{N\rightarrow\infty}{\LARGE Z}(\boldsymbol{j},\widetilde{\boldsymbol{j}%
};\mathbb{A}^{\left(  N\right)  },\mathbb{B}^{(N)},\mathbb{J}^{\left(
N\right)  },\mathbb{U}^{\left(  N\right)  }).
\]
Formally, the limit generating functional looks like (\ref{Eq_22}), and then,
the techniques of the mean-field approximation and double-copy, as presented
in \cite[Chapter 10]{Helias et al}, cannot be used directly in study
continuous networks.

On the other hand, the expected value
\[
\left\langle Z(\boldsymbol{0},\widetilde{\boldsymbol{0}};\mathbb{A}^{\left(
N\right)  },\mathbb{B}^{(N)},\mathbb{J}^{\left(  N\right)  },\mathbb{U}%
^{\left(  N\right)  })\right\rangle _{\mathbb{A}^{\left(  N\right)
},\mathbb{B}^{(N)},\mathbb{J}^{\left(  N\right)  },\mathbb{U}^{\left(
N\right)  }}%
\]
corresponds to a partition function of a Gaussian theory, the mean-field
theory of the network, which provides a relevant approximation of the dynamics
of the network. Then, we have a correspondence between stochastic NNs and
Gaussian SFTs similar to the one discussed by Halverson et al. in
\cite{Halverson et al}.

\section{\label{SECT_3}Statistical field theory of random networks}

In this section, we provide a heuristic introduction to the correspondence
between random NNs and SFTs that we study in this article. A discrete random
network is a stochastic dynamical system defined as%
\begin{equation}
\frac{d\boldsymbol{h}_{i}\left(  t\right)  }{dt}=-\gamma\boldsymbol{h}%
_{i}\left(  t\right)  +%
{\displaystyle\sum\limits_{j\in\Omega_{N}}}
\mathbb{J}_{ij}\phi\left(  \boldsymbol{h}_{j}\left(  t\right)  \right)
+\eta_{i}\left(  t\right)  \text{, \ }i\in\Omega_{N}, \label{Eq_1_networks}%
\end{equation}
where $\Omega_{N}$ is the set of neurons, and $N$ is its cardinality,
$\boldsymbol{h}_{i}\left(  t\right)  \in\mathbb{R}$ is the state of the neuron
$i$, $\gamma>0$, the $\mathbb{J}_{ij}$ are independent identically Gaussian
distributed\ random coupling weights with zero mean and variance $g^{2}/N$.
The time-varying inputs $\eta_{i}\left(  t\right)  $\ are independent Gaussian
white-noise processes with mean zero and correlation functions
\[
\left\langle \eta_{i}\left(  t\right)  \eta_{j}\left(  s\right)  \right\rangle
=\sigma^{2}\delta_{ij}\delta\left(  t-s\right)  ,
\]
The activation function $\phi$ is a sigmoid, for instance $\phi\left(
x\right)  =\tanh(x)$.

By interpreting the stochastic differential equations in the It\^{o}
convention, the moment-generating functional ${\LARGE Z}\left(  \boldsymbol{j}%
;\left[  \boldsymbol{J}_{ij}\right]  \right)  $ for the probability density
\[
\boldsymbol{P}(\left[  \boldsymbol{h}_{i}\left(  t\right)  \right]
_{i\in\Omega_{N}}\mid\left[  \boldsymbol{h}_{i}\left(  0\right)  \right]
_{i\in\Omega_{N}}=\boldsymbol{0})
\]
can be obtained by using the Martin--Siggia--Rose--de Dominicis--Janssen path
integral formalism \cite{Martin et al}, \cite{de Domicis et al}, \cite{Altland
et al}, \cite{Chow et al}. Following the exposition in \cite[Chapter
10]{Helias et al}, this functional is given by%
\begin{equation}
{\LARGE Z}\left(  \boldsymbol{j};\left[  \boldsymbol{J}_{ij}\right]  \right)
=%
{\displaystyle\int}
D\boldsymbol{h}%
{\displaystyle\int}
D\widetilde{\boldsymbol{h}}\exp\left(  S_{0}\left[  \boldsymbol{h}%
,\widetilde{\boldsymbol{h}}\right]  -\widetilde{\boldsymbol{h}}^{T}\left[
\boldsymbol{J}_{ij}\right]  \phi\left(  \boldsymbol{h}\right)  +\boldsymbol{j}%
^{T}\boldsymbol{h}\right)  , \label{Eq_generrating_func}%
\end{equation}
where $\boldsymbol{h=}\left[  \boldsymbol{h}_{i}\left(  t\right)  \right]
_{i\in\Omega_{N}}$, $\widetilde{\boldsymbol{h}}\boldsymbol{=}\left[
\widetilde{\boldsymbol{h}}_{i}\left(  t\right)  \right]  _{i\in\Omega_{N}}$,
$\boldsymbol{j=}\left[  \boldsymbol{j}_{i}\left(  t\right)  \right]
_{i\in\Omega_{N}}$\ are column vectors in $\mathbb{R}^{N}$,
\[
\boldsymbol{j}^{T}\boldsymbol{h=}%
{\displaystyle\sum\limits_{i\in\Omega_{N}}}
\boldsymbol{j}_{i}\left(  t\right)  \boldsymbol{h}_{i}\left(  t\right)  ,
\]%
\[
S_{0}\left[  \boldsymbol{h},\widetilde{\boldsymbol{h}}\right]  =\widetilde
{\boldsymbol{h}}^{T}\left(  \partial_{t}+\gamma\right)  \boldsymbol{h}%
+\frac{1}{2}\sigma^{2}\widetilde{\boldsymbol{h}}^{T}\widetilde{\boldsymbol{h}%
}\text{.}%
\]
The vector $\boldsymbol{j}$ represents the source field, and the response
field $\widetilde{\boldsymbol{h}}$ appears as a result of the
Hubbard-Stratonovich transformation, representing a Dirac delta as
\[
\delta\left(  x\right)  =\frac{1}{2\pi i}%
{\displaystyle\int\nolimits_{a-i\infty}^{a+i\infty}}
d\widetilde{x}\exp\left(  x\widetilde{x}\right)  .
\]
In the thermodynamic limit ($N\rightarrow\infty$), the field components
$\boldsymbol{h}_{i}\left(  t\right)  $ become $\boldsymbol{h}\left(
x,t\right)  $, where $x$\ denotes the position of a neuron in a continuous
space $\Omega$, and $\widetilde{\boldsymbol{h}}_{i}\left(  t\right)
\rightarrow\widetilde{\boldsymbol{h}}\left(  x,t\right)  $.\ Similarly, in a
formal way, we assume that%
\begin{equation}
\boldsymbol{j}^{T}\boldsymbol{h\rightarrow}\text{ }\left\langle \boldsymbol{j}%
\left(  x,t\right)  ,\boldsymbol{h}\left(  x,t\right)  \right\rangle :=%
{\displaystyle\int\limits_{\Omega}}
{\displaystyle\int\limits_{\mathbb{R}}}
\boldsymbol{j}\left(  x,t\right)  \boldsymbol{h}\left(  x,t\right)
dtd\mu\left(  x\right)  , \label{inner_product}%
\end{equation}%
\[
\text{\ }\widetilde{\boldsymbol{h}}^{T}\left[  \boldsymbol{J}_{ij}\right]
\phi\left(  \boldsymbol{h}\right)  \rightarrow\left\langle \widetilde
{\boldsymbol{h}},%
{\displaystyle\int\limits_{\Omega}}
\boldsymbol{J}\left(  x,y\right)  \phi\left(  \boldsymbol{h}\left(
y,t\right)  \right)  d\mu\left(  y\right)  \right\rangle ,
\]
as $N\rightarrow\infty$, where $\mu\left(  x\right)  $ is a measure in the
space $\Omega$,\ and the action $S_{0}\left[  \boldsymbol{h},\widetilde
{\boldsymbol{h}}\right]  $ becomes
\begin{equation}
S_{0}\left[  \boldsymbol{h},\widetilde{\boldsymbol{h}}\right]  =\left\langle
\widetilde{\boldsymbol{h}},\left(  \partial_{t}+\gamma\right)  \boldsymbol{h}%
\right\rangle +\frac{1}{2}\sigma^{2}\left\langle \widetilde{\boldsymbol{h}%
},\widetilde{\boldsymbol{h}}\right\rangle . \label{action}%
\end{equation}
Then, formally, in the thermodynamic limit, the generating functional
(\ref{Eq_generrating_func}) becomes%
\begin{equation}
{\LARGE Z}\left(  \boldsymbol{j};\boldsymbol{J}\right)  =%
{\displaystyle\int}
D\boldsymbol{h}%
{\displaystyle\int}
D\widetilde{\boldsymbol{h}}\exp\left(  S_{0}\left[  \boldsymbol{h}%
,\widetilde{\boldsymbol{h}}\right]  -%
{\displaystyle\int\limits_{\Omega}}
\boldsymbol{J}\left(  x,y\right)  \phi\left(  \boldsymbol{h}\left(
y,t\right)  \right)  d\mu\left(  y\right)  +\left\langle \boldsymbol{j}%
,\boldsymbol{h}\right\rangle \right)  , \label{Eq_generrating_func_2}%
\end{equation}
where $S_{0}\left[  \boldsymbol{h},\widetilde{\boldsymbol{h}}\right]  $ is
defined in (\ref{action}). We identify $\boldsymbol{h}\left(  x,t\right)  $,
$\widetilde{\boldsymbol{h}}\left(  x,t\right)  $ with functions from
$L^{2}(\Omega\times\mathbb{R})$, which is a real Hilbert space with the inner
product $\left\langle \cdot,\cdot\right\rangle $ defined in
(\ref{inner_product}).

On the other hand, formally, in the thermodynamic limit, the system
(\ref{Eq_1_networks}) becomes%
\begin{equation}
\frac{d\boldsymbol{h}\left(  x,t\right)  }{dt}=-\gamma\boldsymbol{h}\left(
x,t\right)  +%
{\displaystyle\int\limits_{\Omega}}
\boldsymbol{J}\left(  x,y\right)  \phi\left(  \boldsymbol{h}\left(
y,t\right)  \right)  d\mu\left(  y\right)  +\eta\left(  x,t\right)
,\ \label{Continuous_system}%
\end{equation}
where $\boldsymbol{J}\left(  x,y\right)  $ is a realization of a generalized
Gaussian process in $L^{2}\left(  \Omega\times\Omega\right)  $, and
$\eta\left(  x,t\right)  $ is a realization of a generalized Gaussian process
in $L^{2}(\Omega\times\mathbb{R}\mathbb{)}$. The system (\ref{Eq_1_networks})
is a stochastic version of a continuous cellular neural network.

All the assertions about the thermodynamic limits for the generating
functional (\ref{Eq_generrating_func}) and the system (\ref{Eq_1_networks})
are just an ansatz. In this way, we obtain a correspondence between random
neural networks (\ref{Continuous_system}) and statistical field theories
(\ref{Eq_generrating_func_2}). This correspondence is similar to the one
introduced by Buice and Cowan in \cite[Section 7.1]{Buice and Cowan}.

It is more convenient for us to consider generating functionals of the form%
\begin{equation}
{\LARGE Z}\left(  \boldsymbol{j},\widetilde{\boldsymbol{j}};\boldsymbol{J}%
\right)  =%
{\displaystyle\iint}
D\boldsymbol{h}D\widetilde{\boldsymbol{h}}\exp\left(  S_{0}\left[
\boldsymbol{h},\widetilde{\boldsymbol{h}}\right]  -S_{int}\left[
\boldsymbol{h},\widetilde{\boldsymbol{h}};\boldsymbol{J}\right]  +\left\langle
\boldsymbol{j},\boldsymbol{h}\right\rangle +\left\langle \widetilde
{\boldsymbol{j}},\widetilde{\boldsymbol{h}}\right\rangle \right)  ,
\label{Generating_Funt_J_2}%
\end{equation}
where $S_{0}\left[  \boldsymbol{h},\widetilde{\boldsymbol{h}}\right]  $
defined as in (\ref{action}), and
\begin{equation}
S_{int}\left[  \boldsymbol{h},\widetilde{\boldsymbol{h}};\mathbf{J}\right]
=\left\langle \widetilde{\boldsymbol{h}},%
{\displaystyle\int\limits_{\Omega}}
\boldsymbol{J}\left(  x,y\right)  \phi\left(  \boldsymbol{h}\left(
y,t\right)  \right)  d\mu\left(  y\right)  \right\rangle .
\label{S_interaction}%
\end{equation}
The corresponding network is determined as%
\begin{equation}
\frac{d\boldsymbol{h}\left(  x,t\right)  }{dt}=-\gamma\boldsymbol{h}\left(
x,t\right)  -\widetilde{\boldsymbol{j}}\left(  x,t\right)  +%
{\displaystyle\int\limits_{\Omega}}
\boldsymbol{J}\left(  x,y\right)  \phi\left(  \boldsymbol{h}\left(
y,t\right)  \right)  d\mu\left(  y\right)  +\eta\left(  x,t\right)  .
\label{Eq_Network_gneral}%
\end{equation}
This article aims to show that the correspondence (ansatz) between the network
(\ref{Eq_Network_gneral}) and the STF (\ref{Generating_Funt_J_2}) is correct,
in the sense that the chaotic behavior of this type of network can be obtained
${\LARGE Z}\left(  \boldsymbol{j},\widetilde{\boldsymbol{j}};\boldsymbol{J}%
\right)  $ using techniques of SFT.

\section{\label{SECT_4}Some network prototypes}

In this section, we discuss some network prototypes whose chaotic behavior can
be studied using the techniques introduced in this article. The first choice
is the set of neurons $\Omega$, which is an infinite topological space
supporting a measure $\mu$, i.e., $\left(  \Omega,\mathcal{B}\left(
\Omega\right)  ,\mu\right)  $ is a measure space, where $\mathcal{B}\left(
\Omega\right)  $ is the Borel $\sigma$-algebra of $\Omega$. We assume that
$\mu$ is Radon measure; examples of such measures are the Lebesgue measure of
$\mathbb{R}^{n}$ restricted to some Borel subset, the Haar measure on any
locally compact topological group, the Dirac measure, and the Gaussian measure
in $\mathbb{R}^{n}$.

We assume the existence of a space of continuous functions $\mathcal{D}\left(
\Omega\right)  $ on $\Omega$, which plays the role of space of test functions.
This space is embedded in its topological dual $\mathcal{D}^{\prime}\left(
\Omega\right)  $, the space of distributions. This theory of distributions
should be sufficient to carry out most mathematical calculations required in
mathematical physics. In particular, we assume that $0\in\Omega$ so the Dirac
delta function can be defined. Also, we assume that $\Omega$ is contained in
an additive group, so the expression $x-y$ makes sense for $x,y\in\Omega$, but
$x-y$ is not necessarily an element from $\Omega$. This hypothesis is
necessary to define $\delta\left(  x-y\right)  $ and convolution of functions.

In addition, we also require the existence of a Gel'fand triplet%
\[
\mathcal{D}\left(  \Omega\right)  \hookrightarrow L^{2}\left(  \Omega\right)
\hookrightarrow\mathcal{D}^{\prime}\left(  \Omega\right)  ,
\]
where the arrow $\hookrightarrow$ denotes a continuous and dense embedding,
and $\mathcal{D}\left(  \Omega\right)  $ is a nuclear space. This condition
will be used to use white noise analysis techniques to construct measures in
infinite-dimensional spaces, as we will discuss later. The Appendices provide
further details and references.

\subsection{Archimedean CNNs}

We now take $\Omega=\mathbb{R}$, $\mathcal{D}\left(  \Omega\right)
=\mathcal{S}\left(  \mathbb{R}\right)  $ (the Schwartz space), $\mathcal{D}%
^{\prime}\left(  \Omega\right)  =\mathcal{S}^{\prime}\left(  \mathbb{R}%
\right)  $ (the space of tempered distributions), and the Gel'fand triplet is
\[
\mathcal{S}\left(  \mathbb{R}\right)  \hookrightarrow L^{2}\left(
\mathbb{R}\right)  \hookrightarrow\mathcal{S}^{\prime}\left(  \mathbb{R}%
\right)  ,
\]
where $L^{2}\left(  \mathbb{R},dt\right)  $ is the real vector space of square
integrable functions with respect to the Lebesgue measure in $\mathbb{R}$.
This triplet is well-known, see Appendix C.

In this case, the neurons are organized in an infinite straight line, with
only one layer. So, the network is `shallow.' Typically, it is assumed that
discrete network (\ref{Eq_1_networks}) has a hierarchical structure (i.e., the
network is deep) using, for instance, that the entries of the matrix
$\mathbb{J}_{ij}$ form a graph. However, if we assume that in the
thermodynamic limit the space of neurons is $\mathbb{R}$, then this geometric
information is not preserved.

\subsection{Non-Archimedean CNNs}

In this section, we use $p$ to denote a fixed prime number. Any non-zero
$p$-adic number $x$ has a unique expansion of the form%
\begin{equation}
x=x_{-k}p^{-k}+x_{-k+1}p^{-k+1}+\ldots+x_{0}+x_{1}p+\ldots,\text{ }
\label{p-adic-number}%
\end{equation}
with $x_{-k}\neq0$, where $k$ is an integer, and the $x_{j}$s \ are numbers
from the set $\left\{  0,1,\ldots,p-1\right\}  $. The set of all possible
sequences form the (\ref{p-adic-number}) constitutes the field of $p$-adic
numbers $\mathbb{Q}_{p}$. There are natural field operations, sum and
multiplication, on series of form (\ref{p-adic-number}). There is also a norm
in $\mathbb{Q}_{p}$ defined as $\left\vert x\right\vert _{p}=p^{k}$, for a
nonzero $p$-adic number $x$. The field of $p$-adic numbers with the distance
induced by $\left\vert \cdot\right\vert _{p}$ is a complete ultrametric space.
The ultrametric property refers to the fact that $\left\vert x-y\right\vert
_{p}\leq\max\left\{  \left\vert x-z\right\vert _{p},\left\vert z-y\right\vert
_{p}\right\}  $ for any $x$, $y$, $z$ in $\mathbb{Q}_{p}$. The $p$-adic
integers which are sequences of the form (\ref{p-adic-number}) with $-k\geq0$.
All these sequence constitute the unit ball $\mathbb{Z}_{p}$.

There is a natural truncation operation on $p$-adic integers:
\[
x=%
{\displaystyle\sum\limits_{k=0}^{\infty}}
x_{k}p^{k}\rightarrow x_{0}+x_{1}p+\ldots+x_{l-1}p^{l-1}\text{, \ }l\geq1.
\]
The set all truncated integers mod $p^{l}$ is denoted as $G_{l}=\mathbb{Z}%
_{p}/p^{l}\mathbb{Z}_{p}$. This set is \ a rooted tree with $l$ levels. The
unit ball $\mathbb{Z}_{p}$ (which is the inverse limit of the $G_{l}$s) is an
infinite rooted tree with fractal structure.

A function $\varphi:\mathbb{Q}_{p}\rightarrow\mathbb{R}$ is called locally
constant, if for any $a\in\mathbb{Q}_{p}$, there is an integer $l=l(a)$, such
that
\[
\varphi\left(  a+x\right)  =\varphi\left(  a\right)  \text{ for any
}\left\vert x\right\vert _{p}\leq p^{l}.
\]
We say $\varphi$ is a test function if it is locally constant with compact
support. We denote by $\mathcal{D}(\mathbb{Q}_{p})$ the real vector space of
test functions. There is a natural integration theory so that $\int
_{\mathbb{Q}_{p}}\varphi\left(  x\right)  dx$ gives a well-defined complex
number. The measure $dx$ is the Haar measure of $\mathbb{Q}_{p}$. Let
$U\subset\mathbb{Q}_{p}$ be an open subset, we denote by $\mathcal{D}(U)$ the
space of test functions with supports contained in $U$. Then $\mathcal{D}(U)$
is dense in
\[
L^{2}(U):=L^{2}(U,dx)=\left\{  f:U\rightarrow\mathbb{R};\left\Vert
f\right\Vert _{2}=\left(
{\displaystyle\int\limits_{U}}
\left\vert f\left(  x\right)  \right\vert ^{2}d^{N}x\right)  ^{\frac{1}{2}%
}<\infty\right\}  ,
\]
cf. \cite[Proposition 4.3.3]{A-K-S}. All the balls with respect to $\left\vert
\cdot\right\vert _{p}$ are open and compact.\ For further details about
$p$-adic analysis the reader may consult Appendix H.

We first take $\Omega=\mathbb{Q}_{p}$ and $\mathcal{D}\left(  \mathbb{Q}%
_{p}\right)  $, the space of real-valued test functions. Then
\[
\mathcal{D}\left(  \mathbb{Q}_{p}\right)  \hookrightarrow L^{2}\left(
\mathbb{Q}_{p}\right)  \hookrightarrow\mathcal{D}^{\prime}\left(
\mathbb{Q}_{p}\right)
\]
is a Gel'fand triplet. This triplet were introduced in \cite{Fuquen et al},
see also \cite{Arroyo et al}. In this case, the neurons are organized in an
infinite tree-like structure, which is the network is hierarchically
organized. We now take $\Omega=\mathbb{Z}_{p}$ and $\mathcal{D}\left(
\mathbb{Z}_{p}\right)  $, the space of real-valued test functions supported in
$\mathbb{Z}_{p}$.\ Then
\[
\mathcal{D}\left(  \mathbb{Z}_{p}\right)  \hookrightarrow L^{2}\left(
\mathbb{Z}_{p}\right)  \hookrightarrow\mathcal{D}^{\prime}\left(
\mathbb{Z}_{p}\right)
\]
is a Gel'fand triplet. In this case, the neurons are organized in an infinite
rooted tree.

\subsection{Further comments about the required hypotheses}

\textit{All the results announced in this article are valid for Archimedean
and non-Archimedean networks}. However, the results can be extended, in
several different ways, to a more general scope. The number fields
$\mathbb{R}$ and $\mathbb{Q}_{p}$ are examples of local number fields. By the
classification\ theorem, every local field is isomorphic (as a topological
field) to one of the following fields:

\begin{itemize}
\item the real numbers $\mathbb{R}$, and the complex numbers $\mathbb{C}$
(Archimedean local fields)

\item finite extensions of the $p$-adic numbers $\mathbb{Q}_{p}$
(Non-Archimedean local fields of characteristic zero)

\item the field of formal Laurent series $F_{q}((T))$ over a finite field
$F_{q}$, where $q$ is a power of $p$ (Non-Archimedean local fields of
characteristic $p$),
\end{itemize}

\noindent cf. \cite{Weil}. Any local number field supports a sufficiently
general theory of distributions so most of the calculations in mathematical
physics can be carried out on it. Also, it is possible to consider continuous
NNs with an adelic space of neurons:%
\[
\Omega=\mathbb{R\times}%
{\displaystyle\prod\limits_{2\leq p<M}}
\mathbb{Q}_{p},
\]
where $M=2,3,\ldots,\infty$. Another possible generalization is based on the
fact that $\mathbb{R}$ and $\mathbb{Q}_{p}$ are examples of locally compact
topological groups. This type of structure also supports a sufficiently rich
theory of distributions for applications in mathematical physics. The
extension of the results presented here to the case\ where $\Omega$ is a local
number field or a locally compact topological group is not given here. In most
of the cases, we give explicitly the hypotheses to perform a particular calculation.

\section{\label{SECT_5}The road map}

This work's organization has been strongly influenced by Chapter 10 in the
book \cite{Helias et al}. So, using this reference as companion material for
reading this work can be useful. A discrete random network, see
(\ref{Eq_1_networks}), has two sources of randomness: the quenched disorder
due to the random coupling weights and the temporal fluctuating drive. A
particular realization of the couplings $\mathbb{J}_{ij}$ defines a network
configuration, and the dynamical properties vary between different
realizations. For a large network size $N$ (number of neurons), certain
quantities are self-averaging, which means that their values for a typical
realization can be obtained by an average over the network configurations. In
this direction, it is natural to conjecture that
\begin{equation}
\lim_{N\rightarrow\infty}\ {\LARGE Z}\left(  \boldsymbol{j};\left[
\boldsymbol{J}_{ij}\right]  \right)  ={\LARGE Z}\left(  \boldsymbol{j}%
;\boldsymbol{J}\right)  , \label{LIMIT_conjecture_2}%
\end{equation}
see (\ref{Eq_generrating_func}) and (\ref{Eq_generrating_func_2}), be a random
variable; and that all observables that can be calculated from ${\LARGE Z}%
\left(  \boldsymbol{j};\boldsymbol{J}\right)  $ can be approximately
obtained\ from $\left\langle {\LARGE Z}\left(  \boldsymbol{j};\boldsymbol{J}%
\right)  \right\rangle _{\boldsymbol{J}}$.\ Each network is obtained as one
realization of the random variable $J$; the values of relevant observables are
expected to be independent of the particular realization of the randomness.

Hidden in (\ref{LIMIT_conjecture_2}) is the fact that\ for a very large number
of neurons, the lattice SFT corresponding to ${\LARGE Z}\left(  \boldsymbol{j}%
;\left[  \boldsymbol{J}_{ij}\right]  \right)  $ is approximated by continuous
SFT corresponding to ${\LARGE Z}\left(  \boldsymbol{j};\boldsymbol{J}\right)
$. It is more convenient to have two currents $\boldsymbol{j},\widetilde
{\boldsymbol{j}}$, so\ we work with ${\LARGE Z}\left(  \boldsymbol{j}%
,\widetilde{\boldsymbol{j}};\boldsymbol{J}\right)  $. Section \ref{SECT_6} is
dedicated to the construction of rigorous moment-generating functionals
${\LARGE Z}_{M}\left(  \boldsymbol{j},\widetilde{\boldsymbol{j}}%
;\boldsymbol{J}\right)  $ for\ networks of type (\ref{Eq_Network_gneral}). The
conclusion is stated in the Theorem \ref{Porp2} in Appendix C. The space of
fields is a Hilbert space; we use white noise calculus techniques to construct
probability measures on the space of fields. The mathematical details appear
in Appendices A, B, and C. The construction of the probability measures
requires using a specific Gel'fand triplets. In the case $\Omega=\mathbb{R}$,
we use well-know triplet; see subsection \ref{Section_Gelfan_triplets} in
Appendix C. In the cases $\Omega=\mathbb{Q}_{p},\mathbb{Z}_{p}$, we construct
some new Gel'fand triplets; see Theorem \ref{Theorem6} in Appendix C. The
author has studied extensively Euclidean quantum field theories on $p$-adic
spaces, \cite{Zuniga-ATMP}, \cite{Zuniga-JFAA}, \cite{Zuniga-RIM}. In
particular, the correspondence between $p$-adic continuous Boltzmann NNs and
SFTs, \cite{Zuniga-ATMP}.

There are many relevant open problems related to the rigorous construction of
moment-generating functionals for several types of continuous networks.

\begin{problem}
The rigorous mathematical formulation of the Buice-Cowan STFs describing the
cortex activity, \cite{Buice-Cowean-2007}-\cite{Buice and Cowan}. In the
original model $\Omega=\mathbb{R}^{n}$, so the neurons are not organized in
hierarchical structure. For instance, if $n=1$, the neurons are organized in a
straight line. The cases $\Omega=\mathbb{Q}_{p},\mathbb{Z}_{p}$ are relevant
since the neurons have a hierarchical structure. To the best of our knowledge,
this type of SFTs has not been studied in the literature.
\end{problem}

\begin{problem}
In \cite{Schoenholz et al}, a correspondence between random NNs, with finite
neurons, and lattice SFTs is developed. An interesting problem is finding and
studying the continuous version of these lattice SFTs.
\end{problem}

\begin{problem}
A relevant problem is to extend the theory developed here to
reaction-diffusion random NNs:
\end{problem}

\[
\frac{d\boldsymbol{h}\left(  x,t\right)  }{dt}=-\gamma\boldsymbol{h}\left(
x,t\right)  +\Delta\boldsymbol{h}\left(  x,t\right)  +%
{\displaystyle\int\limits_{\Omega}}
\boldsymbol{J}\left(  x,y\right)  \phi\left(  \boldsymbol{h}\left(
y,t\right)  \right)  d\mu\left(  y\right)  +\eta\left(  x,t\right)  ,
\]
where $\frac{d\boldsymbol{h}\left(  x,t\right)  }{dt}=\Delta\boldsymbol{h}%
\left(  x,t\right)  $ is a heat equation. For the $p$-adic theory of these
equations, the reader may consult \cite{KKZuniga}, \cite{Zuniga-Textbook}, and
the references therein. In the case $\Omega=\mathbb{Z}_{p}$, in
\cite{Zambrano-Zuniga-2}, we have developed algorithms for image denoising
using $p$-adic reaction-diffusion NNs.

The first consequence of Theorem \ref{Theorem6} is that ${\LARGE Z}_{M}\left(
\boldsymbol{j},\widetilde{\boldsymbol{j}};\boldsymbol{J}\right)  $ is a
well-defined generalized Gaussian random variable with mean zero; so the
average
\begin{equation}
\overline{{\LARGE Z}}_{M}\left(  \boldsymbol{j},\widetilde{\boldsymbol{j}%
}\right)  =\left\langle {\LARGE Z}_{M}\left(  \boldsymbol{j},\widetilde
{\boldsymbol{j}};\boldsymbol{J}\right)  \right\rangle _{\boldsymbol{J}}
\label{Average}%
\end{equation}
is the moment-generating functional of the MFT of a\ random NN. In Section
\ref{SECT_7}, we present all the details of the computation of $\overline
{{\LARGE Z}}_{M}\left(  \boldsymbol{j},\widetilde{\boldsymbol{j}}\right)  $.
Hidden in (\ref{Average}) is the calculation of an oscillatory integral in
infinite dimension. The white noise calculus allow us to compute this integral
in a simple and rigorous form. The coupling $\boldsymbol{J}$ is also a
generalized Gaussian random variable, which corresponds to a Gaussian measure
on\ $L^{2}(\Omega^{2})$. This type of measure is uniquely determined by the
trace class operators on $L^{2}(\Omega^{2})$.\ We select integral operators
with kernels in $L^{2}(\Omega^{4})$. This is a standard choice in quantum
field theory. The details are given in Appendix D. For our purposes, we just
need $\overline{{\LARGE Z}}_{M}\left(  \boldsymbol{0},0\right)  =:\overline
{{\LARGE Z}}_{M}$, which is the partition function. Theorem \ref{Porp3}, in
Appendix D, provide an explicit formula for $\overline{{\LARGE Z}}_{M}$, which
involves an action of the form%
\[
-\left\langle \widetilde{\boldsymbol{h}},\left(  \partial_{t}+\gamma\right)
\boldsymbol{h}\right\rangle +\frac{1}{2}\sigma^{2}\left\langle \widetilde
{\boldsymbol{h}},\widetilde{\boldsymbol{h}}\right\rangle +\left\langle
\widetilde{\boldsymbol{h}},C_{\phi\left(  \boldsymbol{h}\right)  \phi\left(
\boldsymbol{h}\right)  }\widetilde{\boldsymbol{h}}\right\rangle _{\left(
L^{2}(\Omega\times\mathbb{R})\right)  ^{2}},
\]
where the bilinear form $B(\widetilde{\boldsymbol{h}},\widetilde
{\boldsymbol{h}})=$ $\left\langle \widetilde{\boldsymbol{h}},C_{\phi\left(
\boldsymbol{h}\right)  \phi\left(  \boldsymbol{h}\right)  }\widetilde
{\boldsymbol{h}}\right\rangle _{\left(  L^{2}(\Omega\times\mathbb{R})\right)
^{2}}$ is the correlation functional of a Gaussian noise, with mean zero. In
addition, Theorem \ref{Theorem2}, in Appendix E, gives a formula where the
above action is written in a matrix form; using this formula, in Section
\ref{SECT_8}, we derive a system of differential equations for the covariance
functions that controlled the MFT. The functional $\overline{{\LARGE Z}}%
_{M}\left(  \boldsymbol{j},\widetilde{\boldsymbol{j}}\right)  $\ corresponds
to an infinite Gaussian NN. This type of network has been studied intensively
in the literature; see e.g. \cite{Halverson et al}, \cite{Lee et al},
\cite{Neal}, \cite{Pleiss et al}, \cite{Poole et al}. These works deal with
neural network architectures with a finite number of neurons $N$, admitting a
limit as $N\rightarrow\infty$ in which the networks are drawn from a Gaussian
process. Here, the infinite Gaussian NN corresponding to $\overline
{{\LARGE Z}}_{M}$ is \textquotedblleft just an approximation\textquotedblright%
\ of the original network. It is worth quoting here that Neal [69, pp. 161]
describes the Gaussian process limit as \textquotedblleft
disappointing,\textquotedblright\ noting that \textquotedblleft infinite
networks do not have hidden units that represent `hidden features' ... often
seen [as the] interesting aspect of neural network learning.\textquotedblright%
\ This phenomenon does not occur if we use $p$-adic numbers to codify the
hierarchical structure of a NN. On the other hand, the mentioned phenomenon
occurs if we use real numbers.

There are many open problems related to our work on the MFT approximation. We
just mention two.

\begin{problem}
Is it possible to recover the usual (discrete) MFT of an NN as the
discretization of a continuous MFT? The partition function $\overline
{{\LARGE Z}}_{M}$ is determined by an action, which is a functional on a
function space. By restricting this functional to a finite-dimensional vector
space a discretization of it can be obtained. This technique works very well
in the case of $p$-adic NNs; see \cite{Zuniga-ATMP}, and the references therein.
\end{problem}

In Section \ref{Eq_9}, we give a rigorous formulation of the partition
function of the double-copy system. By taking the formal limit of the discrete
moment-generating function of the double-copy system given in \cite[Formula
(10.32)]{Helias et al}, we introduce a continuous functional ${\LARGE Z}%
_{M}^{\left(  2\right)  }\left(  \left\{  \boldsymbol{j}^{\alpha}%
,\widetilde{\boldsymbol{j}}^{\alpha}\right\}  _{\alpha\in\left\{  1,2\right\}
};\boldsymbol{J}\right)  $ for the continuous version of the double-copy
system. As before, we work with the partition function ${\LARGE Z}%
_{M}^{\left(  2\right)  }\left(  \left\{  \boldsymbol{0},\boldsymbol{0}%
\right\}  _{\alpha\in\left\{  1,2\right\}  };\mathbf{J}\right)  ={\LARGE Z}%
_{M}^{\left(  2\right)  }\left(  \mathbf{J}\right)  $. We compute rigorously
the expected value $\left\langle {\LARGE Z}_{M}^{\left(  2\right)  }\left(
\mathbf{J}\right)  \right\rangle _{\mathbf{J}}=\overline{{\LARGE Z}%
_{M}^{\left(  2\right)  }}$, which is the partition function for the
continuous double-copy system, see Theorem \ref{Lemma_15}, Theorem
\ref{Theorem3} in Appendix F. Then we compute the matrix propagator for the
double-copy system, which is system of differential equations on the
covariance matrices controlling the dynamics of the model. An important open
problem is the following:

\begin{problem}
Is it possible to recover the usual (discrete) double-copy system of an NN as
the discretization of a continuous version of it?
\end{problem}

In Section \ref{SECT_10}, using\ \ the system of differential equations that
control the covariance functions of the double-copy system, we formally adapt
the technique of Sompolinsky et al., \cite{Sompolinsky et al}, that allows us
to derive a condition for the criticality of networks of type
(\ref{Continuous_Network}) via the largest Lyapunov exponent. The problem of
determining rigorously largest Lyapunov exponent of a random NN of type
(\ref{Eq_Network_gneral}) is open. In our view, from a mathematical
perspective, the technique introduced in \cite{Sompolinsky et al} for
determining largest Lyapunov is a physical reasoning, not a mathematical theorem.

The Section \ref{SECT_11} is dedicated to the self-averaging phenomena. The
dynamics of the double-copy system is controlled by a system of differential
equations, whose solutions (the covariance functions) are distributions
depending on the spatial variables $x,y\in\Omega$. Then, several spatial
averages of the covariance functions satisfy the above-mentioned system of
differential equations. Appendix G collects some results about this
phenomenon; see Theorems \ref{Theorem4} and \ref{Theorem5}.

In Section \ref{SECT_12}, we introduce a toy model of a random NN, and compute
the covariance functions of the MFT of this network.

We set $\Omega=\mathbb{Z}_{p}$, so the fields $\boldsymbol{h}\left(
x,t\right)  \in L^{2}\left(  \mathbb{Z}_{p}\times\mathbb{R}\right)  $, and
pick an integral trace class operator on $L^{2}\left(  \mathbb{Z}_{p}%
\times\mathbb{Z}_{p}\right)  $, whose kernel depends on a real parameter
$\rho\in\left(  1,\infty\right)  $; we interpret $\rho$ as the control
parameter of our model. In this framework, $\overline{{\LARGE Z}}_{M}\left(
\rho\right)  =\overline{{\LARGE Z}}_{M}\left(  \frac{1}{1-p^{\rho-2}}\right)
$ has a pole at $\rho=2$ ,so the network has a continuous phase transition at
$\rho=2$. We introduce a technique that allows the computation of
$G_{\boldsymbol{hh}}^{\rho}(x,y,\tau)=\left\langle \boldsymbol{h}%
(x,t_{1}),\boldsymbol{h}(y,t_{2})\right\rangle _{\boldsymbol{h}}$, $\tau
=t_{1}-t_{2}$, for $\boldsymbol{h}(x,t)$ in a dense subset of $L^{2}\left(
\mathbb{Z}_{p}\times\mathbb{R}\right)  $. Interpreting $G_{\boldsymbol{hh}%
}^{\rho}(x,y,\tau),\rho\in\left(  1,\infty\right)  ,$ as an order parameter,
we argue that $G_{\boldsymbol{hh}}^{\rho}(x,y,\tau)$, $\rho\in\left(
1,\infty\right)  \smallsetminus\left\{  2\right\}  $ describes a disordered
phase, while $G_{\boldsymbol{hh}}^{2}(x,y,\tau)$ describes an ordered one.
This behavior shows that the NN works near phase transition. We believe that
the computation of the double-copy model for this toy NN and its application
to derive a condition for the criticality of NN via the largest Lyapunov
exponent is relevant.

\section{\label{SECT_6}A rigorous definition of the generating functional
${\protect\LARGE Z}_{M}\left(  \boldsymbol{j},\widetilde{\boldsymbol{j}%
};\boldsymbol{J}\right)  $}

By the considerations made at the end of the previous section, it is natural
to consider that the fields $\boldsymbol{h},\widetilde{\boldsymbol{h}%
},\boldsymbol{j},\widetilde{\boldsymbol{j}}$ are elements of the Hilbert space%
\[
L^{2}(\Omega\times\mathbb{R)}=L^{2}(\Omega\times\mathbb{R},d\mu dt)\mathbb{)=}%
\left\{  \boldsymbol{f}:\Omega\times\mathbb{R\rightarrow R};\left\Vert
\boldsymbol{f}\right\Vert _{2}<\infty\right\}  ,
\]
where $d\mu$ is fixed measure in $\Omega$, and $dt$ is the Lebesgue measure in
$\mathbb{R}$, $\left\Vert \boldsymbol{f}\right\Vert _{2}^{2}=\left\langle
\boldsymbol{f},\boldsymbol{f}\right\rangle $, and
\[
\left\langle \boldsymbol{f},\boldsymbol{g}\right\rangle =%
{\displaystyle\int\limits_{\Omega}}
{\displaystyle\int\limits_{\mathbb{R}}}
\boldsymbol{f}\left(  x,t\right)  \boldsymbol{g}\left(  x,t\right)
dtd\mu\left(  x\right)  .
\]

In order to provide a rigorous mathematical version of ${\LARGE Z}\left(
\boldsymbol{j},\widetilde{\boldsymbol{j}};\boldsymbol{J}\right)  $, it is
necessary to solve the following problems:

\begin{problem}
\label{P1} Find a Hilbert subspace $\mathcal{H}\subset L^{2}(\Omega
\times\mathbb{R)}$ such that $\partial_{t}\boldsymbol{f}\left(  x,t\right)
\in\mathcal{H}$ for any $\boldsymbol{f}\left(  x,t\right)  \in\mathcal{H}$;
\end{problem}

\begin{problem}
\label{P2} Show that the bilinear form $\left\langle \widetilde{\boldsymbol{h}%
},%
{\displaystyle\int\limits_{\Omega}}
\boldsymbol{J}\left(  x,y\right)  \phi\left(  \boldsymbol{h}\left(
y,t\right)  \right)  d\mu\left(  y\right)  \right\rangle $ is well-defined for
any $\boldsymbol{h},\widetilde{\boldsymbol{h}}\in\mathcal{H}$, and any
$\boldsymbol{J}\in L^{2}(\Omega\times\mathbb{R)}$;
\end{problem}

\begin{problem}
\label{P3}Find probability measures $d\boldsymbol{P}\left(  \boldsymbol{h}%
\right)  d\widetilde{\boldsymbol{P}}\left(  \widetilde{\boldsymbol{h}}\right)
$ in $\mathcal{H\times H}$;
\end{problem}

\begin{problem}
\label{P4}Verify that
\begin{equation}
\exp\left(  S_{0}\left[  \boldsymbol{h},\widetilde{\boldsymbol{h}}\right]
-S_{int}\left[  \boldsymbol{h},\widetilde{\boldsymbol{h}};\boldsymbol{J}%
\right]  +\left\langle \boldsymbol{j},\boldsymbol{h}\right\rangle
+\left\langle \widetilde{\boldsymbol{j}},\widetilde{\boldsymbol{h}%
}\right\rangle \right)  \label{Function_exponential}%
\end{equation}
is in $L^{1}(\mathcal{H}\times\mathcal{H},d\boldsymbol{P}\left(
\boldsymbol{h}\right)  d\widetilde{\boldsymbol{P}}\left(  \widetilde
{\boldsymbol{h}}\right)  )$.
\end{problem}

The mathematical constructions and demonstrations related to these problems
are given in Appendices A, B, and C. Here, we give a general discussion
emphasizing the connection with the classical case presented in \cite[Chapter
10]{Helias et al}.

\subsection{The construction of the space of fields}

We first note that if $\boldsymbol{h}\left(  x,t\right)  \in L^{2}%
(\Omega\times\mathbb{R)}$, then $\boldsymbol{h}\left(  \cdot,t\right)  \in
L^{2}(\Omega\mathbb{)}$ for almost every $t\in\mathbb{R}$, and $\boldsymbol{h}%
\left(  x,\cdot\right)  \in L^{2}(\mathbb{R)}$ for almost every $x\in\Omega$.
This fact is a consequence of the Fubini theorem. Now, it is necessary to
define $\partial_{t}\boldsymbol{h}\left(  x,\cdot\right)  $ for
$\boldsymbol{h}\left(  x,\cdot\right)  \in L^{2}(\mathbb{R)}$. This drives
naturally to the Sobolev space $W_{1}(\mathbb{R})$:%
\[
W_{1}\left(  \mathbb{R}\right)  =\left\{  G\left(  t\right)
:\mathbb{R\rightarrow C};G\in L^{2}(\mathbb{R})\text{, }\left(  \sqrt
{1+\left\vert \xi\right\vert ^{2}}\right)  \mathcal{F}_{t\rightarrow\xi
}\left(  G\right)  \in L^{2}(\mathbb{R})\right\}  ,
\]
where $\mathcal{F}_{t\rightarrow\xi}\left(  G\right)  $ denotes the Fourier
transform. $W_{1}\left(  \mathbb{R}\right)  $ is a real separable Hilbert
space. In this space, the definition%
\begin{equation}
\partial_{t}G\left(  t\right)  =\mathcal{F}_{\xi\rightarrow t}^{-1}\left(
\sqrt{-1}\xi\mathcal{F}_{t\rightarrow\xi}\left(  G\right)  \right)
\label{Derivative}%
\end{equation}
makes sense, cf. Lemma \ref{Lemma_1} in Appendix A. We warn the reader that in
this article, the fields are real-valued functions, but most of the
mathematical literature about functional analysis required here considers
complex-valued functions. This fact does not create any problem in our
presentation; however, in some cases, like in the formula (\ref{Derivative}),
it is necessary to ensure that certain functions are real-valued, cf. Remark
\ref{Nota1A} in Appendix A.

We now note that if $f_{i}(t)\in W_{1}(\mathbb{R})$, $g_{i}(x)\in L^{2}%
(\Omega\mathbb{)}$, $c_{i}\in\mathbb{R}$, $i=1,\ldots,m$, then
\begin{equation}%
{\displaystyle\sum\limits_{i=1}^{m}}
c_{i}f_{i}(t)%
{\textstyle\bigotimes\limits_{\text{alg}}}
g_{i}(x):=%
{\displaystyle\sum\limits_{i=1}^{m}}
c_{i}f_{i}(t)g_{i}(x)\in L^{2}(\Omega\times\mathbb{R)},
\label{LInear combinations}%
\end{equation}
and%
\[
\partial_{t}\left(
{\displaystyle\sum\limits_{i=1}^{m}}
c_{i}f_{i}(t)%
{\textstyle\bigotimes\limits_{\text{alg}}}
g_{i}(x)\right)  =%
{\displaystyle\sum\limits_{i=1}^{m}}
c_{i}\partial_{t}\left(  f_{i}(t)\right)
{\textstyle\bigotimes\limits_{\text{alg}}}
g_{i}(x)\in L^{2}(\Omega\times\mathbb{R)}.
\]
So, the natural candidate for the space of fields is the closure in
$L^{2}(\Omega\times\mathbb{R)}$ of the set of linear combinations of type
(\ref{LInear combinations}). This insight is correct, and its mathematical
formulation requires some results about the tensor product of separable
Hilbert spaces. Assuming that $L^{2}(\Omega\mathbb{)}$ is separable Hilbert
space, the Hilbert tensor product
\[
\mathcal{H}:=L^{2}(\Omega\mathbb{)}%
{\textstyle\bigotimes}
W_{1}\left(  \mathbb{R}\right)
\]
of $L^{2}(\Omega\mathbb{)}$ and $W_{1}\left(  \mathbb{R}\right)  $, which is a
separable Hilbert space, is continuously embedded in $L^{2}(\Omega
\times\mathbb{R)}$, cf. Lemmas \ref{Lemma_2} and \ref{Lemma_3} in Appendix A.
In this space, the functions
\[
S_{0}\left[  \boldsymbol{h},\widetilde{\boldsymbol{h}}\right]  ,\left\langle
\boldsymbol{j},\boldsymbol{h}\right\rangle ,\left\langle \widetilde
{\boldsymbol{j}},\widetilde{\boldsymbol{h}}\right\rangle ,\text{ for
}\boldsymbol{h},\widetilde{\boldsymbol{h}},\boldsymbol{j},\widetilde
{\boldsymbol{j}}\in\mathcal{H}%
\]
are well-defined. The term $S_{int}\left[  \boldsymbol{h},\widetilde
{\boldsymbol{h}};\boldsymbol{J}\right]  $ is also a well-defined real-valued
function due to the fact that for $\boldsymbol{J}\left(  x,y\right)  \in
L^{2}\left(  \Omega\times\Omega,d\mu d\mu\right)  =L^{2}\left(  \Omega
^{2}\right)  $, and $g\left(  y,t\right)  \in L^{2}\left(  \Omega
\times\mathbb{R},d\mu dt\right)  $, the function $\int_{\Omega}\boldsymbol{J}%
\left(  x,y\right)  \phi\left(  g\left(  y,t\right)  \right)  d\mu\left(
x\right)  $ belongs to $L^{2}\left(  \Omega\times\mathbb{R},d\mu dt\right)  $,
cf. Lemma \ref{Lemma_4} in Appendix A. Our results are valid for a general
class of activation functions, which contains $\phi\left(  z\right)
=\tanh(z)$ as a particular case, see Remark \ref{Note_act_fun} in Appendix A.

An explicit description of the fields, i.e., the elements of $\mathcal{H}$, is
the following:%
\[
\boldsymbol{h}\left(  x,t\right)  =%
{\displaystyle\sum\limits_{n=0}^{\infty}}
\text{ }%
{\displaystyle\sum\limits_{n=0}^{\infty}}
c_{n,m}e_{n}\left(  x\right)  f_{m}\left(  t\right)  ,
\]
where $\left\{  e_{n}\left(  x\right)  \right\}  _{n\in\mathbb{N}}$ is an
orthonormal basis of $L^{2}\left(  \Omega,d\mu\right)  $ and $\left\{
f_{m}\left(  t\right)  \right\}  _{m\in\mathbb{N}}$ is an orthonormal basis of
$W_{1}$, and the $c_{n,m}$\ are real numbers, cf. Lemma \ref{Lemma_2}-(ii) in
Appendix A.

In conclusion, in the space $\mathcal{H\times H}$, the expression
\[
\exp\left(  S_{0}\left[  \boldsymbol{h},\widetilde{\boldsymbol{h}}\right]
-S_{int}\left[  \boldsymbol{h},\widetilde{\boldsymbol{h}};\boldsymbol{J}%
\right]  +\left\langle \boldsymbol{j},\boldsymbol{h}\right\rangle
+\left\langle \widetilde{\boldsymbol{j}},\widetilde{\boldsymbol{h}%
}\right\rangle \right)
\]
is a well-defined real-valued function, cf. Proposition \ref{Porp1} in
Appendix A. For further details, the reader may consult Appendix A and the
references therein.

\begin{remark}
We use the symbols $\left\langle \cdot,\cdot\right\rangle _{\mathcal{H}}$,
respectively $\left\Vert \cdot\right\Vert _{\mathcal{H}}$, for the inner
product, the norm, in $\mathcal{H}$. The symbols $\left\langle \cdot
,\cdot\right\rangle $, $\left\Vert \cdot\right\Vert $ are used only in the
case $L^{2}(\Omega\times\mathbb{R)}$. For other $L^{2}$ spaces, we use
specific notation that allows us to identify the space uniquely.
\end{remark}

\subsection{Probability measures on the space of fields}

A central difference between the study of continuous random NNs versus the
discrete counterparts is that the continuous case requires random variables
taking values in infinite dimensional vector spaces like $L^{2}\left(
\Omega^{2}\right)  $, $\mathcal{H}$; in contrast, the discrete case uses only
$\mathbb{R}^{D}$. To construct probability measures, we use white noise
calculus techniques; in turn, the construction of such measures is possible in
certain topological spaces (nuclear spaces). In the Appendix B, we review the
basic ideas about white noise calculus. We assume the existence of a Gel'fand
triplet of the form%
\begin{equation}
\Phi\hookrightarrow\mathcal{H}\hookrightarrow\Phi^{\prime}, \label{Triplet1}%
\end{equation}
which means that $\Phi$ is a countably Hilbert nuclear space continuously and
densely embedded into $\mathcal{H=}L^{2}(\Omega\mathbb{)}%
{\textstyle\bigotimes}
W_{1}$. By identifying $\mathcal{H}$ with is dual (by the Riesz representation
theorem), $\mathcal{H}$ is continuously and densely embedded into
$\Phi^{\prime}$ (the space of continuous linear functionals on $\Phi$, which
reflexive $\Phi\mathcal{=}\Phi^{\prime\prime}$). This implies that the pairing
$\left(  T,f\right)  $ on $\Phi^{\prime}\times\Phi$ satisfies $\left(
T,\varphi\right)  =\left\langle T,\varphi\right\rangle $ for $T\in
\mathcal{H}^{\prime}=\mathcal{H}$ and $\varphi\in\Phi\subset\mathcal{H}$. For
further details are given in Section \ref{Section_Gelfand_triplet} of the
Appendix C.

The Bochner-Minlos theorem, see Appendix B for details, establishes a
bijection between probability measures $\boldsymbol{P}$ on $\left(
\Phi^{\prime},\mathcal{B}\right)  $, where $\mathcal{B}$ is the $\sigma
$-algebra of cylindric subsets of $\Phi^{\prime}$, and and characteristic
functionals (or moment-generating functionals), $\mathcal{C}:\Phi
\rightarrow\mathbb{C}$ as follows:%
\[
\mathcal{C}\left(  \varphi\right)  =%
{\displaystyle\int\limits_{\Phi^{\prime}}}
e^{\sqrt{-1}\left(  T,\varphi\right)  }d\boldsymbol{P}(T)\text{, for }%
\varphi\in\Phi\text{.}%
\]
We now select two characteristic functions $\mathcal{C}$, $\widetilde
{\mathcal{C}}:\Phi\rightarrow\mathbb{C}$ to get two probability measures
$d\boldsymbol{P}\left(  \boldsymbol{h}\right)  $, $d\widetilde{\boldsymbol{P}%
}\left(  \widetilde{\boldsymbol{h}}\right)  $ on the $\sigma$-algebra
$\mathcal{B}$ and the probability product measure $d\boldsymbol{P}\left(
\boldsymbol{h}\right)  d\widetilde{\boldsymbol{P}}\left(  \widetilde
{\boldsymbol{h}}\right)  $ on the product $\sigma$-algebra $\mathcal{B\times
B}$. The restriction of $d\boldsymbol{P}\left(  \boldsymbol{h}\right)
d\widetilde{\boldsymbol{P}}\left(  \widetilde{\boldsymbol{h}}\right)  $ to
$\mathcal{H\times H}$ is a probability measure (cf. Lemma \ref{Lemma_5A} in
the Appendix C), we use this measure as a replacement of the ill-defined
measure $D\boldsymbol{h}D\widetilde{\boldsymbol{h}}$ in the generating
functional (\ref{Generating_Funt_J_2}). For some examples of measures of the
type $d\boldsymbol{P}\left(  \boldsymbol{h}\right)  d\widetilde{\boldsymbol{P}%
}\left(  \widetilde{\boldsymbol{h}}\right)  $, the reader may consult Section
\ref{Sect_examples} in the Appendix C. The explicit construction of the
measures $d\boldsymbol{P}\left(  \boldsymbol{h}\right)  d\widetilde
{\boldsymbol{P}}\left(  \widetilde{\boldsymbol{h}}\right)  $ for the
Archimedean/non-Archimedean NNs is given at the end of the Appendix C.

\begin{remark}
The technical details about constructing the specific Gel'fand triplets
required for studying Archimedean/non-Archimedean random NNs are given in
Section \ref{Section_Gelfand_triplets} in Appendix C. More precisely, in
Theorem \ref{Theorem6}.
\end{remark}

\subsection{SFTs with cutoff}

We now consider the Problem \ref{P4}. General results provide explicit
families of functions in $L^{1}(\mathcal{H}\times\mathcal{H},d\boldsymbol{P}%
\left(  \boldsymbol{h}\right)  d\widetilde{\boldsymbol{P}}\left(
\widetilde{\boldsymbol{h}}\right)  )$; see the last section in Appendix B and
the references therein; the function (\ref{Function_exponential}) does not
belong to any of these families, and thus, currently the Problem \ref{P4} is
open. However, we do not need the solution of this problem in the rigorous
construction of the moment-generating functionals for NNs.

The fields $\boldsymbol{h}$ and $\widetilde{\boldsymbol{h}}$ represent the
physical signals produced by a network. The physical signals have finite
energy, i.e., $\left\Vert \boldsymbol{h}\right\Vert _{\mathcal{H}},\left\Vert
\widetilde{\boldsymbol{h}}\right\Vert _{\mathcal{H}}<M$, where the constant
$M$ depends on each particular $NN$. Motivated by this observation, we set%
\[
\mathcal{P}_{M}:=\left\{  \left(  \boldsymbol{h},\widetilde{\boldsymbol{h}%
}\right)  \in\mathcal{H}\times\mathcal{H};\left\Vert \boldsymbol{h}\right\Vert
_{\mathcal{H}},\left\Vert \widetilde{\boldsymbol{h}}\right\Vert _{\mathcal{H}%
}<M\right\}  ,
\]
where $M$ is a positive constant. The polydics $\mathcal{P}_{M}$ contains all
the physical signals produced by a NN, constituting a parameter of the STF of
the network. We denote by ${\large 1}_{\mathcal{P}_{M}}\left(  \boldsymbol{h}%
,\widetilde{\boldsymbol{h}}\right)  $ the characteristic function of
$\mathcal{P}_{M}$. Now%
\[
{\large 1}_{\mathcal{P}_{M}}\left(  \boldsymbol{h},\widetilde{\boldsymbol{h}%
}\right)  \exp\left(  S_{0}\left[  \boldsymbol{h},\widetilde{\boldsymbol{h}%
}\right]  -S_{int}\left[  \boldsymbol{h},\widetilde{\boldsymbol{h}%
};\boldsymbol{J}\right]  +\left\langle \boldsymbol{j},\boldsymbol{h}%
\right\rangle +\left\langle \widetilde{\boldsymbol{j}},\widetilde
{\boldsymbol{h}}\right\rangle \right)  ,
\]
belongs to $L^{1}\left(  \mathcal{H}\times\mathcal{H},\mathcal{B}%
\times\mathcal{B},d\boldsymbol{h}d\widetilde{\boldsymbol{h}}\right)  $, see
Lemma \ref{Lemma_5} in Appendix C. Now, for $\boldsymbol{J}\left(  x,y\right)
\in L^{2}\left(  \Omega\times\Omega\right)  $, $\boldsymbol{j}\left(
x,t\right)  $, $\widetilde{\boldsymbol{j}}\left(  x,t\right)  \in\mathcal{H}$,
the functional%
\begin{align}
{\LARGE Z}_{M}\left(  \boldsymbol{j},\widetilde{\boldsymbol{j}};\boldsymbol{J}%
\right)   &  =%
{\displaystyle\iint\limits_{\mathcal{H}\times\mathcal{H}}}
{\large 1}_{\mathcal{P}_{M}}\left(  \boldsymbol{h},\widetilde{\boldsymbol{h}%
}\right)  \exp\left(  S_{0}\left[  \boldsymbol{h},\widetilde{\boldsymbol{h}%
}\right]  -S_{int}\left[  \boldsymbol{h},\widetilde{\boldsymbol{h}%
};\boldsymbol{J}\right]  \right)  \times\nonumber\\
&  \exp\left(  \left\langle \boldsymbol{j},\boldsymbol{h}\right\rangle
+\left\langle \widetilde{\boldsymbol{j}},\widetilde{\boldsymbol{h}%
}\right\rangle \right)  d\boldsymbol{P}\left(  \boldsymbol{h}\right)
d\widetilde{\boldsymbol{P}}\left(  \widetilde{\boldsymbol{h}}\right)
\label{Generating_Funt_J_2_M}%
\end{align}
is well-defined, see Theorem \ref{Porp2} in Appendix C.

\begin{remark}
We warn the reader that the construction of the generating functionals
${\LARGE Z}_{M}\left(  \boldsymbol{j},\widetilde{\boldsymbol{j}}%
;\boldsymbol{J}\right)  $ involves several technical subtleties, which are
discussed in Appendix A. First, the fields $\boldsymbol{h},\widetilde
{\boldsymbol{h}}$ are considered at the same foot; this means that \ in
(\ref{Generating_Funt_J_2_M}) $\boldsymbol{h}$ and $\widetilde{\boldsymbol{h}%
}$ are interchangeable. There is an inclusion $\mathfrak{i}:\mathcal{H}%
\rightarrow L^{2}(\Omega\times\mathbb{R})$ such that $\mathfrak{i}\left(
\mathcal{H}\right)  $ is a dense pre-Hilbert subspace of $L^{2}(\Omega
\times\mathbb{R})$, cf. Lemma \ref{Lemma_3} in Appendix A. In
(\ref{Generating_Funt_J_2_M}), we are obligated to integrate over two copies
of $\mathcal{H}$\ so $\partial_{t}\boldsymbol{h}$ is well-defined. However,
the integrand in (\ref{Generating_Funt_J_2_M}) involves the inner product of
$L^{2}(\Omega\times\mathbb{R})$; this is not a problem since we identify
$\mathcal{H}$ with a pre-Hilbert subspace of $L^{2}(\Omega\times\mathbb{R})$.
On the other hand, we can also use a triplet like $\Psi\hookrightarrow
L^{2}\left(  \Omega\right)
{\textstyle\bigotimes}
L^{2}\left(  \mathbb{R}\right)  \hookrightarrow\Psi^{\prime}$ to construct the
measure $d\boldsymbol{P}\left(  \boldsymbol{h}\right)  d\widetilde
{\boldsymbol{P}}\left(  \widetilde{\boldsymbol{h}}\right)  $, see Lemma
\ref{Lemma_5A} in Appendix C. This implies that in
(\ref{Generating_Funt_J_2_M}), we can integrate on $L^{2}(\Omega
\times\mathbb{R})\times L^{2}(\Omega\times\mathbb{R})$, and the integrand is
well-defined up to a set of $d\boldsymbol{P}\left(  \boldsymbol{h}\right)
d\widetilde{\boldsymbol{P}}\left(  \widetilde{\boldsymbol{h}}\right)
$-measure zero.
\end{remark}

As an application of the rigorous construction of ${\LARGE Z}_{M}\left(
\boldsymbol{j},\widetilde{\boldsymbol{j}};\boldsymbol{J}\right)  $, the
mapping
\begin{equation}
\boldsymbol{J}\left(  \in L^{2}(\Omega^{2})\right)  \text{ }%
\boldsymbol{\rightarrow}\text{ }{\LARGE Z}_{M}\left(  \boldsymbol{j}%
,\widetilde{\boldsymbol{j}};\boldsymbol{J}\right)  \label{Random_variable}%
\end{equation}
is a well-defined (generalized) random variable. The mean-field theory
corresponding to ${\LARGE Z}_{M}\left(  \boldsymbol{j},\widetilde
{\boldsymbol{j}};\boldsymbol{J}\right)  $ is given by%
\[
\overline{{\LARGE Z}}_{M}\left(  \boldsymbol{j},\widetilde{\boldsymbol{j}%
}\right)  :=\left\langle {\LARGE Z}_{M}\left(  \boldsymbol{j},\widetilde
{\boldsymbol{j}};\boldsymbol{J}\right)  \right\rangle _{\boldsymbol{J}},
\]
where $\left\langle \cdot\right\rangle _{\boldsymbol{J}}$ denotes the expected
value with respect to the noise $\boldsymbol{J}$. In the next section, we
compute this expectation in a rigorous mathematical way. Here, it is relevant
to mention that in the discrete case, it is widely accepted that \ in the
limit $N\rightarrow\infty$ the moment-generating functional ${\LARGE Z}\left(
\boldsymbol{j};\left[  \boldsymbol{J}_{ij}\right]  \right)  $, which is a
random object due to $\left[  \boldsymbol{J}_{ij}\right]  $, has a
concentration measure, cf. \cite[pg. 98]{Helias et al}. This is precisely the
meaning of (\ref{Random_variable}), but for a moment-generating functional
with cutoff. In the next sections, we show that the mean-field approximation
and the double-copy techniques can be rigorously formulated, under suitable
hypotheses, using $\overline{{\LARGE Z}}_{M}\left(  \boldsymbol{j}%
,\widetilde{\boldsymbol{j}}\right)  $.

\section{\label{SECT_7}Mean-field theory approximation}

\subsection{Averaging random NNs}

In $S_{int}\left[  \boldsymbol{h},\widetilde{\boldsymbol{h}};\boldsymbol{J}%
\right]  $, see (\ref{S_interaction}), $\boldsymbol{J}$ is a realization of a
generalized Gaussian random variable on the measurable space $\left(
L^{2}\left(  \Omega^{2}\right)  ,\mathcal{B}(\Omega^{2})\right)  $, with mean
zero. This means\ that there exists a probability measure $\boldsymbol{P}%
_{\boldsymbol{J}}$, on the Borel $\mathcal{\sigma}$-algebra $\mathcal{B}%
(\Omega^{2})$ of the space $\Omega^{2}$, such that the probability of the
event $\boldsymbol{J}\in B$, for $B\in\mathcal{B}(\Omega^{2})$, is $\int
_{B}d\boldsymbol{P}_{\boldsymbol{J}}$. For the sake of simplicity, we use
$\boldsymbol{J}$ for the random variable and its value. In the
infinite-dimensional case, there are no explicit formulas for the measure
$\boldsymbol{P}_{\boldsymbol{J}}$; only the Fourier transform of
$d\boldsymbol{P}_{\boldsymbol{J}}$\ is available:%
\begin{equation}%
{\displaystyle\int\limits_{L^{2}\left(  \Omega^{2}\right)  }}
e^{\sqrt{-1}\lambda\left\langle \boldsymbol{J},f\right\rangle _{L^{2}\left(
\Omega^{2}\right)  }}d\boldsymbol{P}_{\boldsymbol{J}}=e^{\frac{-\lambda^{2}%
}{2}\left\langle \square_{\boldsymbol{J}}f,f\right\rangle _{L^{2}\left(
\Omega^{2}\right)  }}\text{, } \label{Fourier_Trans_Prob}%
\end{equation}
for $f\in L^{2}\left(  \Omega^{2}\right)  $ and $\lambda\in\mathbb{R}$, where
$\square_{\boldsymbol{J}}:L^{2}\left(  \Omega^{2}\right)  \rightarrow
L^{2}\left(  \Omega^{2}\right)  $\ is a symmetric, positive definite, trace
class operator. Appendix D reviews the results of the Gaussian measures on
Hilbert spaces required here.

We assume that $\square_{\boldsymbol{J}}$ is an integral operator, which means
that there exists a kernel $K_{\boldsymbol{J}}\left(  u_{1},u_{2},y_{1}%
,y_{2}\right)  \in L^{2}\left(  \Omega^{4}\right)  $ such that%
\[
\square_{\boldsymbol{J}}f\left(  y_{1},y_{2}\right)  =%
{\displaystyle\int\limits_{\Omega^{2}}}
K_{\boldsymbol{J}}\left(  u_{1},u_{2},y_{1},y_{2}\right)  f\left(  u_{1}%
,u_{2}\right)  d\mu\left(  u_{1}\right)  d\mu\left(  u_{2}\right)  .
\]
This generic class of operators is widely used in QFT and amenable for
calculations, \cite[Section 5.1]{Zinn-Justin}; see also Remark
\ref{Nota_Mercer} in Appendix D.

In the classical calculations using the technique of mean-field approximation,
the auxiliary fields $\boldsymbol{j}\left(  x,t\right)  ,\widetilde
{\boldsymbol{j}}\left(  x,t\right)  $ are taken to be zero because they are
not physical fields; this fact also occurs in the continuous counterpart, for
this reason we work only with ${\LARGE Z}_{M}\left(  \boldsymbol{0}%
,\boldsymbol{0};\boldsymbol{J}\right)  =:{\LARGE Z}_{M}\left(  \boldsymbol{J}%
\right)  $, which is an analog of a partition function with cutoff. In this
section, we compute the expectation
\begin{multline*}
\overline{{\LARGE Z}}_{M}:=\left\langle {\LARGE Z}_{M}\left(  \boldsymbol{J}%
\right)  \right\rangle _{\boldsymbol{J}}=%
{\displaystyle\iint\limits_{\mathcal{H}\times\mathcal{H}}}
{\large 1}_{\mathcal{P}_{M}}\left(  \boldsymbol{h},\widetilde{\boldsymbol{h}%
}\right)  \exp\left(  S_{0}\left[  \boldsymbol{h},\widetilde{\boldsymbol{h}%
}\right]  \right)  \times\\
\left\{  \text{ \ }%
{\displaystyle\int\limits_{L^{2}\left(  \Omega^{2}\right)  }}
\exp\left(  -S_{int}\left[  \boldsymbol{h},\widetilde{\boldsymbol{h}%
};\boldsymbol{J}\right]  \right)  d\boldsymbol{P}_{\boldsymbol{J}}\right\}
d\boldsymbol{P}\left(  \boldsymbol{h}\right)  d\widetilde{\boldsymbol{P}%
}\left(  \widetilde{\boldsymbol{h}}\right)  .
\end{multline*}
Set
\begin{align*}
\mathcal{I}(\boldsymbol{h},\widetilde{\boldsymbol{h}})  &  :=\left\langle
\exp\left(  -S_{int}\left[  \boldsymbol{h},\widetilde{\boldsymbol{h}%
};\boldsymbol{J}\right]  \right)  \right\rangle _{_{\boldsymbol{J}}}\\
&  =%
{\displaystyle\int\limits_{L^{2}\left(  \Omega^{2}\right)  }}
\exp\left(  \left\langle \widetilde{\boldsymbol{h}}\left(  x,t\right)  ,%
{\displaystyle\int\limits_{\Omega}}
\boldsymbol{J}\left(  x,y\right)  \phi\left(  \boldsymbol{h}\left(
y,t\right)  \right)  d\mu\left(  y\right)  \right\rangle \right)
d\boldsymbol{P}_{\boldsymbol{J}}.
\end{align*}
By using the identity,%
\begin{gather*}
\left\langle \widetilde{\boldsymbol{h}}\left(  x,t\right)  ,%
{\displaystyle\int\limits_{\Omega}}
\boldsymbol{J}\left(  x,y\right)  \phi\left(  \boldsymbol{h}\left(
y,t\right)  \right)  d\mu\left(  y\right)  \right\rangle =\\
\left\langle \boldsymbol{J}\left(  x,y\right)  ,%
{\displaystyle\int\limits_{\mathbb{R}}}
\widetilde{\boldsymbol{h}}\left(  x,t\right)  \phi\left(  \boldsymbol{h}%
\left(  y,t\right)  \right)  dt\right\rangle _{L^{2}\left(  \Omega\times
\Omega\right)  },
\end{gather*}
see Remark \ref{Nota7} in Appendix D,
\begin{equation}
\mathcal{I}(\boldsymbol{h},\widetilde{\boldsymbol{h}})=%
{\displaystyle\int\limits_{L^{2}\left(  \Omega^{2}\right)  }}
\exp\left(  \left\langle \boldsymbol{J}\left(  x,y\right)  ,%
{\displaystyle\int\limits_{\mathbb{R}}}
\widetilde{\boldsymbol{h}}\left(  x,t\right)  \phi\left(  \boldsymbol{h}%
\left(  y,t\right)  \right)  dt\right\rangle _{L^{2}\left(  \Omega\times
\Omega\right)  }\right)  d\boldsymbol{P}_{\boldsymbol{J}}.
\label{Integral_exp}%
\end{equation}
The integral in (\ref{Integral_exp}) is similar to the one in
(\ref{Fourier_Trans_Prob}), if we take
\[
f=%
{\displaystyle\int\limits_{\mathbb{R}}}
\widetilde{\boldsymbol{h}}\left(  x,t\right)  \phi\left(  \boldsymbol{h}%
\left(  y,t\right)  \right)  dt.
\]
Indeed, formally, $\mathcal{I}(\boldsymbol{h},\widetilde{\boldsymbol{h}}%
)$\ can be calculated by performing a Wick rotation $\lambda\rightarrow
-\sqrt{-1}\lambda$ in (\ref{Fourier_Trans_Prob}),%
\begin{equation}%
{\displaystyle\int\limits_{L^{2}\left(  \Omega^{2}\right)  }}
e^{\lambda\left\langle \boldsymbol{J},f\right\rangle _{L^{2}\left(  \Omega
^{2}\right)  }}d\boldsymbol{P}_{\boldsymbol{J}}=e^{\frac{\lambda^{2}}%
{2}\left\langle \square_{\boldsymbol{J}}f,f\right\rangle _{L^{2}\left(
\Omega^{2}\right)  }}\text{, for }f\in L^{2}\left(  \Omega^{2}\right)  .
\label{Key_formula_A}%
\end{equation}
This heuristics is a mathematical theorem; see the formula (\ref{Key_formula})
in Appendix D. Then, we can compute the integral $\mathcal{I}(\boldsymbol{h}%
,\widetilde{\boldsymbol{h}})$ using (\ref{Key_formula_A}), with $\lambda=1$,%
\begin{multline*}
\mathcal{I}(\boldsymbol{h},\widetilde{\boldsymbol{h}})=\\
\exp\left(  \frac{1}{2}%
{\displaystyle\int\limits_{\mathbb{R}^{2}}}
\text{ }%
{\displaystyle\int\limits_{\Omega^{2}}}
\widetilde{\boldsymbol{h}}\left(  u_{1},t_{1}\right)  C_{\phi\phi}\left(
u_{1},y_{1},t_{1},t_{2}\right)  \widetilde{\boldsymbol{h}}\left(  y_{1}%
,t_{2}\right)  d\mu\left(  u_{1}\right)  d\mu\left(  y_{1}\right)
dt_{1}dt_{2}\right)  ,
\end{multline*}
where%
\begin{gather*}
C_{\phi\left(  \boldsymbol{h}\right)  \phi\left(  \boldsymbol{h}\right)
}\left(  u_{1},y_{1},t_{1},t_{2}\right)  :=\\%
{\displaystyle\int\limits_{\Omega^{2}}}
\phi\left(  \boldsymbol{h}\left(  y_{2},t_{2}\right)  \right)
K_{\boldsymbol{J}}\left(  u_{1},u_{2},y_{1},y_{2}\right)  \phi\left(
\boldsymbol{h}\left(  u_{2},t_{1}\right)  \right)  d\mu\left(  y_{2}\right)
d\mu\left(  u_{2}\right)  ,
\end{gather*}
cf. Lemma \ref{Lemma_7} in Appendix D.

In conclusion,
\begin{gather}
\overline{{\LARGE Z}}_{M}=%
{\displaystyle\iint\limits_{\mathcal{H}\times\mathcal{H}}}
{\large 1}_{\mathcal{P}_{M}}\left(  \boldsymbol{h},\widetilde{\boldsymbol{h}%
}\right)  \exp\left(  S_{0}\left[  \widetilde{\boldsymbol{h}},\boldsymbol{h}%
\right]  +\left\langle \widetilde{\boldsymbol{h}},C_{\phi\left(
\boldsymbol{h}\right)  \phi\left(  \boldsymbol{h}\right)  }\widetilde
{\boldsymbol{h}}\right\rangle _{\left(  L^{2}(\Omega\times\mathbb{R})\right)
^{2}}\right)  \times\label{Averaged_part_funct}\\
d\boldsymbol{P}\left(  \boldsymbol{h}\right)  d\widetilde{\boldsymbol{P}%
}\left(  \widetilde{\boldsymbol{h}}\right)  ,\nonumber
\end{gather}
where $S_{0}\left[  \widetilde{\boldsymbol{h}},\boldsymbol{h}\right]
=-\left\langle \widetilde{\boldsymbol{h}},\left(  \partial_{t}+\gamma\right)
\boldsymbol{h}\right\rangle +\frac{1}{2}\sigma^{2}\left\langle \widetilde
{\boldsymbol{h}},\widetilde{\boldsymbol{h}}\right\rangle $, and%
\begin{gather*}
\left\langle \widetilde{\boldsymbol{h}},C_{\phi\left(  \boldsymbol{h}\right)
\phi\left(  \boldsymbol{h}\right)  }\widetilde{\boldsymbol{h}}\right\rangle
_{\left(  L^{2}(\Omega\times\mathbb{R})\right)  ^{2}}:=\\%
{\displaystyle\int\limits_{\Omega^{2}}}
\text{ }%
{\displaystyle\int\limits_{\mathbb{R}^{2}}}
\widetilde{\boldsymbol{h}}\left(  x,t_{1}\right)  C_{\phi\left(
\boldsymbol{h}\right)  \phi\left(  \boldsymbol{h}\right)  }\left(
x,y,t_{1},t_{2}\right)  \widetilde{\boldsymbol{h}}\left(  y,t_{2}\right)
d\mu\left(  x\right)  d\mu\left(  y\right)  dt_{1}dt_{2},
\end{gather*}
cf. Theorem \ref{Porp3} in Appendix D. The bilinear from $\left\langle
\widetilde{\boldsymbol{h}},C_{\phi\left(  \boldsymbol{h}\right)  \phi\left(
\boldsymbol{h}\right)  }\widetilde{\boldsymbol{h}}\right\rangle _{\left(
L^{2}(\Omega\times\mathbb{R})\right)  ^{2}}$ is the correlation functional of
a Gaussian noise, cf. Remark \ref{Nota2} in Appendix D.

\begin{remark}
The averaged partition function is a continuous analog of the discrete
partition function obtained by applying the classical technique based on the
saddle point approximation, see \cite[Formula 10.18]{Helias et al}. Following
\cite{Helias et al}, we say that the action in (\ref{Generating_Funt_J_2_M})
decomposes into a sum of actions for individual, non-interacting units
(neurons) that feel a common field with self-consistently determined
statistics, characterized by its second cumulant $C_{\phi\left(
\boldsymbol{h}\right)  \phi\left(  \boldsymbol{h}\right)  }$. Hence, the
averaging approximation reduces the network to infinitely many non-interacting
units (neurons).
\end{remark}

\subsection{The dependency on the cutoff}

The function $S_{0}\left[  \widetilde{\boldsymbol{h}},\boldsymbol{h}\right]  $
is well-defined for any $\widetilde{\boldsymbol{h}}$, $\boldsymbol{h}\in
L^{2}(\Omega\times\mathbb{R)}$, so it is not affected by the support of the
cutoff function ${\large 1}_{\mathcal{P}_{M}}\left(  \boldsymbol{h}%
,\widetilde{\boldsymbol{h}}\right)  $, see (\ref{Averaged_part_funct}). The
same assertion is valid for $\left\langle \widetilde{\boldsymbol{h}}%
,C_{\phi\left(  \boldsymbol{h}\right)  \phi\left(  \boldsymbol{h}\right)
}\widetilde{\boldsymbol{h}}\right\rangle _{\left(  L^{2}(\Omega\times
\mathbb{R})\right)  ^{2}}$, cf. Lemma \ref{Lemma_9} in Appendix D.
Consequently, in the formula (\ref{Averaged_part_funct}), the term
\[
\exp\left(  S_{0}\left[  \widetilde{\boldsymbol{h}},\boldsymbol{h}\right]
+\left\langle \widetilde{\boldsymbol{h}},C_{\phi\left(  \boldsymbol{h}\right)
\phi\left(  \boldsymbol{h}\right)  }\widetilde{\boldsymbol{h}}\right\rangle
_{\left(  L^{2}(\Omega\times\mathbb{R})\right)  ^{2}}\right)
\]
is defined for any $\widetilde{\boldsymbol{h}}$, $\boldsymbol{h}\in
L^{2}(\Omega\times\mathbb{R)}$, i.e., independent of the cutoff, cf. Theorem
\ref{Porp3} in Appendix D.

\section{\label{SECT_8}The computation of the matrix propagator}

The Mean-field theory (MFT) of a random NN provides a simpler model of it;
this approximation is obtained by averaging over the states of the network
realizations. The approximated model is encoded in the averaged partition
function $\overline{{\LARGE Z}}_{M}$. The entries of the matrix propagator
associated with $\overline{{\LARGE Z}}_{M}$ satisfy a system of differential
equations that control the dynamic of the mean-field approximation. The
explicit determination of this system requires a convenient formula for
$\overline{{\LARGE Z}}_{M}$, which is our next step.

\subsection{A formula for $\overline{{\protect\LARGE Z}}_{M}$}

We now find a matrix representation for $S_{0}\left[  \widetilde
{\boldsymbol{h}},\boldsymbol{h}\right]  +\left\langle \widetilde
{\boldsymbol{h}},C_{\phi\left(  \boldsymbol{h}\right)  \phi\left(
\boldsymbol{h}\right)  }\widetilde{\boldsymbol{h}}\right\rangle _{\left(
L^{2}(\Omega\times\mathbb{R})\right)  ^{2}}$ in (\ref{Averaged_part_funct}).
This calculation is completely analogue to the one given in
\cite{Grosvenor-Jefferson}. Initially, we proceed formally
following\ \cite[pp. 15-16]{Grosvenor-Jefferson}. We set%
\[
\boldsymbol{z}\left(  y,t\right)  =\left[
\begin{array}
[c]{c}%
\boldsymbol{h}\left(  y,t\right) \\
\\
\widetilde{\boldsymbol{h}}\left(  y,t\right)
\end{array}
\right]  \text{, \ }\Xi\left(  x,y,t_{1},t_{2}\right)  =\left[
\begin{array}
[c]{ccc}%
\Xi_{\boldsymbol{hh}} &  & \Xi_{\boldsymbol{h}\widetilde{\boldsymbol{h}}}\\
&  & \\
\Xi_{\widetilde{\boldsymbol{h}}\boldsymbol{h}} &  & \Xi_{\widetilde
{\boldsymbol{h}}\widetilde{\boldsymbol{h}}}%
\end{array}
\right]  \text{,}%
\]
with
\[
\Xi_{\boldsymbol{hh}}=0\text{, }\Xi_{\boldsymbol{h}\widetilde{\boldsymbol{h}}%
}=\delta\left(  y-x\right)  \delta\left(  t_{2}-t_{1}\right)  \left(
\partial_{t_{2}}-\gamma\right)  \text{,}%
\]%
\[
\Xi_{\widetilde{\boldsymbol{h}}\boldsymbol{h}}=-\delta\left(  y-x\right)
\delta\left(  t_{2}-t_{1}\right)  \left(  \partial_{t_{2}}+\gamma\right)
\text{,}%
\]%
\[
\Xi_{\widetilde{\boldsymbol{h}}\widetilde{\boldsymbol{h}}}=C_{\phi\left(
\boldsymbol{h}\right)  \phi\left(  \boldsymbol{h}\right)  }\left(
x,y,t_{1},t_{2}\right)  +\sigma^{2}\delta\left(  y-x\right)  \delta\left(
t_{2}-t_{1}\right)  ,
\]
and%
\[
S\left[  \boldsymbol{h},\widetilde{\boldsymbol{h}};\Xi\right]  =\frac{1}{2}%
{\displaystyle\int\limits_{\Omega^{2}}}
\text{ }%
{\displaystyle\int\limits_{\mathbb{R}^{2}}}
\text{\ }\boldsymbol{z}\left(  x,t_{1}\right)  ^{T}\Xi\left(  x,y,t_{1}%
,t_{2}\right)  \boldsymbol{z}\left(  y,t_{2}\right)  d\mu\left(  x\right)
d\mu\left(  y\right)  dt_{1}dt_{2},
\]
where $\delta\left(  y-x\right)  $ $\ $is the Dirac delta function, with $x$,
$y$, $y-x\in\Omega$, and $\delta\left(  t_{2}-t_{1}\right)  $ is the Dirac
delta function on $\mathbb{R}$. Here we need a space of continuous functions
$\mathcal{D}(\Omega)$, where the Dirac function can be defined:%
\[
\left(  \delta\left(  y-x\right)  ,\varphi\left(  y\right)  \right)  :=%
{\displaystyle\int\limits_{\Omega}}
\delta\left(  y-x\right)  \varphi\left(  y\right)  d\mu\left(  y\right)
=\varphi\left(  x\right)  ,\text{ for }\varphi\in\mathcal{D}(\Omega).
\]
For this reason, we assume (i) that the space of continuous functions
$\mathcal{D}(\Omega)$ is a nuclear space, and
\[
\mathcal{D}(\Omega)\subset L^{2}(\Omega)\subset\mathcal{D}^{\prime}(\Omega)
\]
is a Gel'fand triplet, and (ii) $\Omega$ is contained in additive group. The
nuclearity of $\mathcal{D}(\Omega)$ is required in the specific construction
of the measures $d\boldsymbol{P}\left(  \boldsymbol{h}\right)  d\widetilde
{\boldsymbol{P}}\left(  \widetilde{\boldsymbol{h}}\right)  $, see Section
\ref{Section_Gelfand_triplets} in Appendix C.

With these hypotheses one shows that
\[
\overline{{\LARGE Z}}_{M}=%
{\displaystyle\iint\limits_{\mathcal{H}\times\mathcal{H}}}
{\large 1}_{\mathcal{P}_{M}}\left(  \boldsymbol{h},\widetilde{\boldsymbol{h}%
}\right)  \exp\left(  S\left[  \boldsymbol{h},\widetilde{\boldsymbol{h}}%
;\Xi\right]  \right)  d\boldsymbol{P}\left(  \boldsymbol{h}\right)
d\widetilde{\boldsymbol{P}}\left(  \widetilde{\boldsymbol{h}}\right)  ,
\]
and that the functions $\Xi$, $S\left[  \boldsymbol{h},\widetilde
{\boldsymbol{h}};\Xi\right]  $ do not depend on the cutoff ${\large 1}%
_{\mathcal{P}_{M}}\left(  \boldsymbol{h},\widetilde{\boldsymbol{h}}\right)  $.
The calculations are given in the proof of Theorem \ref{Theorem2} in Appendix E.

\subsection{The matrix propagator}

We now consider the solutions of the following matrix integral equation:%
\[%
{\displaystyle\int\limits_{\Omega}}
{\displaystyle\int\limits_{\mathbb{R}}}
\Xi\left(  x,z,t_{1},s\right)  G\left(  z,y,s,t_{2}\right)  dsd\mu\left(
z\right)  =\delta\left(  x-y\right)  \delta\left(  t_{1}-t_{2}\right)  \left[
\begin{array}
[c]{cc}%
1 & 0\\
0 & 1
\end{array}
\right]  \text{,}%
\]
where%
\[
G\left(  x,y,t_{1},t_{2}\right)  =\left[
\begin{array}
[c]{cc}%
G_{\boldsymbol{hh}} & G_{\boldsymbol{h}\widetilde{\boldsymbol{h}}}\\
& \\
G_{\widetilde{\boldsymbol{h}}\boldsymbol{h}} & G_{\widetilde{\boldsymbol{h}%
}\widetilde{\boldsymbol{h}}}%
\end{array}
\right]  ,\text{ with }G_{\widetilde{\boldsymbol{h}}\widetilde{\boldsymbol{h}%
}}=0.
\]
The condition $G_{\widetilde{\boldsymbol{h}}\widetilde{\boldsymbol{h}}}=0$
occurs naturally in discrete versions of models considered here, see \cite[the
top of pg. 16 and the formula 14]{Grosvenor-Jefferson}. For this reason, we
propose $G_{\widetilde{\boldsymbol{h}}\widetilde{\boldsymbol{h}}}=0$ as an
ansatz. Now, from Theorem \ref{Theorem2}, by a direct calculation, we obtain
the following result:

Under the hypothesis of Theorem \ref{Theorem2} in Appendix E, the entries of
the matrix propagator $G$ are determined by the following system of equations:%
\begin{equation}
\left(  \partial_{t_{1}}-\gamma\right)  G_{\widetilde{\boldsymbol{h}%
}\boldsymbol{h}}\left(  x,y,t_{1},t_{2}\right)  =\delta\left(  x-y\right)
\delta\left(  t_{1}-t_{2}\right)  ; \label{Eq_1DD}%
\end{equation}

\begin{gather}
-\left(  \partial_{t_{1}}+\gamma\right)  G_{\boldsymbol{hh}}\left(
x,y,t_{1},t_{2}\right)  +%
{\displaystyle\int\limits_{\Omega}}
{\displaystyle\int\limits_{\mathbb{R}}}
C_{\phi\left(  \boldsymbol{h}\right)  \phi\left(  \boldsymbol{h}\right)
}\left(  x,z,t_{1},s\right)  G_{\widetilde{\boldsymbol{h}}\boldsymbol{h}%
}\left(  z,y,s,t_{2}\right)  dsd\mu\left(  z\right) \label{Eq_2DD}\\
\nonumber\\
+\sigma^{2}G_{\widetilde{\boldsymbol{h}}\boldsymbol{h}}\left(  x,y,t_{1}%
,t_{2}\right)  =0;\nonumber
\end{gather}

\begin{equation}
-\left(  \partial_{t_{1}}+\gamma\right)  G_{\boldsymbol{h}\widetilde
{\boldsymbol{h}}}\left(  x,y,t_{1},t_{2}\right)  =\delta\left(  x-y\right)
\delta\left(  t_{1}-t_{2}\right)  . \label{Eq_3DD}%
\end{equation}
which is independent of the averaged partition function $\overline{{\LARGE Z}%
}_{M}$ used.

\subsection{Time translational invariance}

We now adapt the techniques developed in \cite{Sompolinsky et al} and
\cite{Schuecker et al}, see also \cite[Chapter 10]{Helias et al},
\cite{Grosvenor-Jefferson}, for\ studying discrete\ random NNs using the
system of differential equations (\ref{Eq_1DD})-(\ref{Eq_3DD}), under the
assumption of time translational symmetry. We proceed formally, i.e., we do
not determine the function spaces where the solutions of mentioned
differential equations exist. However, these calculations can be rigorously
justified with the techniques and results given in Appendix G.

In order to obtain additional information from the system (\ref{Eq_1DD}%
)-(\ref{Eq_3DD}), we assume time translation symmetry, so $G\left(
x,y,t_{1},t_{2}\right)  =G\left(  x,y,t_{1}-t_{2}\right)  $; this implies that
$\partial_{t_{1}}G\left(  x,y,t_{1}-t_{2}\right)  =-\partial_{t_{2}}G\left(
x,y,t_{1}-t_{2}\right)  $. Since we have only time derivatives, we consider
$x$ and $y$ as parameters and use the notation $\partial_{t_{1}}=\partial_{1}%
$, \ $\partial_{2}=\partial_{t_{2}}$. Then, the equation (\ref{Eq_1DD}) can be
rewritten as%
\begin{equation}
-\left(  \partial_{2}+\gamma\right)  G_{\widetilde{\boldsymbol{h}%
}\boldsymbol{h}}\left(  x,y,t_{1}-t_{2}\right)  =\delta\left(  x-y\right)
\delta\left(  t_{1}-t_{2}\right)  . \label{Eq_1AA}%
\end{equation}
The time translation symmetry requires that
\begin{align*}
G_{\boldsymbol{hh}}\left(  x,y,t_{1},t_{2}\right)   &  =G_{\boldsymbol{hh}%
}\left(  x,y,t_{1}-t_{2}\right)  \text{, }\\
C_{\phi\left(  \boldsymbol{h}\right)  \phi\left(  \boldsymbol{h}\right)
}\left(  x,y,t_{1},t_{2}\right)   &  =C_{\phi\left(  \boldsymbol{h}\right)
\phi\left(  \boldsymbol{h}\right)  }\left(  x,y,t_{1}-t_{2}\right)  .
\end{align*}
Now, the equation (\ref{Eq_2DD}) becomes%
\begin{gather}
-\left(  \partial_{1}+\gamma\right)  G_{\boldsymbol{hh}}\left(  x,y,t_{1}%
-t_{2}\right) \nonumber\\
+%
{\displaystyle\int\limits_{\Omega}}
{\displaystyle\int\limits_{\mathbb{R}}}
C_{\phi\left(  \boldsymbol{h}\right)  \phi\left(  \boldsymbol{h}\right)
}\left(  x,z,t_{1}-s\right)  G_{\widetilde{\boldsymbol{h}}\boldsymbol{h}%
}\left(  z,y,s-t_{2}\right)  dsd\mu\left(  z\right) \label{Eq_1BB}\\
+\sigma^{2}G_{\widetilde{\boldsymbol{h}}\boldsymbol{h}}\left(  x,y,t_{1}%
-t_{2}\right)  =0.\nonumber
\end{gather}
By multiplying both sides in (\ref{Eq_1BB}) by $\left(  \partial_{2}%
+\gamma\right)  $, and using (\ref{Eq_1AA}), we get%
\begin{gather*}
\left(  \partial_{2}+\gamma\right)  \left(  \partial_{1}+\gamma\right)
G_{\boldsymbol{hh}}\left(  x,y,t_{1}-t_{2}\right)  =\\%
{\displaystyle\int\limits_{\Omega}}
{\displaystyle\int\limits_{\mathbb{R}}}
C_{\phi\left(  \boldsymbol{h}\right)  \phi\left(  \boldsymbol{h}\right)
}\left(  x,z,t_{1}-s\right)  \left(  \partial_{2}+\gamma\right)
G_{\widetilde{\boldsymbol{h}}\boldsymbol{h}}\left(  z,y,s-t_{2}\right)
dsd\mu\left(  z\right) \\
+\sigma^{2}\left(  \partial_{2}+\gamma\right)  G_{\widetilde{\boldsymbol{h}%
}\boldsymbol{h}}\left(  x,y,t_{1}-t_{2}\right)
\end{gather*}

i.e.,
\begin{align}
\left(  \partial_{1}+\gamma\right)  \left(  \partial_{2}+\gamma\right)
G_{\boldsymbol{hh}}\left(  x,y,t_{1}-t_{2}\right)   &  =C_{\phi\left(
\boldsymbol{h}\right)  \phi\left(  \boldsymbol{h}\right)  }\left(
x,y,t_{1}-t_{2}\right) \nonumber\\
&  +\sigma^{2}\delta\left(  x-y\right)  \delta\left(  t_{1}-t_{2}\right)  .
\label{Eq_3BB}%
\end{align}

We now take $\tau=t_{1}-t_{2}$, then, the equations (\ref{Eq_3DD}) and
(\ref{Eq_3BB}) can be re\-written as%
\begin{equation}
-\left(  \partial_{\tau}+\gamma\right)  G_{\widetilde{\boldsymbol{h}%
}\boldsymbol{h}}\left(  x,y,\tau\right)  =\left(  \partial_{\tau}%
-\gamma\right)  G_{\boldsymbol{h}\widetilde{\boldsymbol{h}}}\left(
x,y,\tau\right)  =\delta\left(  x-y\right)  \delta\left(  \tau\right)  ,
\label{Eq_11AA}%
\end{equation}%
\begin{equation}
\left(  -\partial_{\tau}^{2}+\gamma^{2}\right)  G_{\boldsymbol{hh}}\left(
x,y,\tau\right)  =C_{\phi\left(  \boldsymbol{h}\right)  \phi\left(
\boldsymbol{h}\right)  }\left(  x,y,\tau\right)  +\sigma^{2}\delta\left(
x-y\right)  \delta\left(  \tau\right)  . \label{Eq_11}%
\end{equation}

From the right-side of the differential equations (\ref{Eq_11AA}) and
(\ref{Eq_11}) follow that $G_{\boldsymbol{hh}}$, $G_{\boldsymbol{h}%
\widetilde{\boldsymbol{h}}}$, $G_{\widetilde{\boldsymbol{h}}\boldsymbol{h}}$,
$G_{\widetilde{\boldsymbol{h}}\widetilde{\boldsymbol{h}}}$ are distributions
from $\mathcal{D}^{\prime}(\Omega)%
{\textstyle\bigotimes_{\text{alg}}}
\mathcal{S}^{\prime}(\mathbb{R})$, and consequently, all the calculations done
to derive the mentioned differential equations should interpreted in
distributional sense; see Appendix G.

\begin{remark}
Following \cite[Section 10.4]{Helias et al}, we interpret $\overline
{{\LARGE Z}}_{M}\left(  \boldsymbol{j},\widetilde{\boldsymbol{j}}\right)  $ as
the generating functional associated with the network%
\begin{equation}
\left(  \partial_{t}+\gamma\right)  \boldsymbol{h}\left(  x,t\right)
=\theta\left(  x,t\right)  , \label{Eq_12A}%
\end{equation}
driving by the Gaussian noise $\theta\left(  x,t\right)  =\varsigma\left(
x,t\right)  +\eta\left(  x,t\right)  $, with autocorrelation
\[
\left\langle \theta\left(  x,t\right)  \theta\left(  y,s\right)  \right\rangle
=C_{\phi\left(  \boldsymbol{h}\right)  \phi\left(  \boldsymbol{h}\right)
}\left(  x,y,t,s\right)  +\sigma^{2}\delta\left(  y-x\right)  \delta\left(
t-s\right)  ,
\]
see Remark \ref{Nota2} in Appendix E. We now set
\[
C_{\boldsymbol{hh}}\left(  x,y,t,s\right)  =\left\langle \boldsymbol{h}\left(
x,t\right)  \boldsymbol{h}\left(  y,s\right)  \right\rangle _{\boldsymbol{h}}%
\]
for the autocorrelation of $\boldsymbol{h}$. By taking in
\[
\left(  \partial_{t}+\gamma\right)  \boldsymbol{h}\left(  x,t\right)  \left(
\partial_{s}+\gamma\right)  \boldsymbol{h}\left(  y,s\right)  =\theta\left(
x,t\right)  \theta\left(  y,s\right)
\]
\ the expectation with respect to the noise $\theta$, we obtain formally,
\[
\left(  \partial_{t}+\gamma\right)  \left(  \partial_{s}+\gamma\right)
C_{\boldsymbol{hh}}\left(  x,y,t,s\right)  =C_{\phi\left(  \boldsymbol{h}%
\right)  \phi\left(  \boldsymbol{h}\right)  }\left(  x,y,t,s\right)
+\sigma^{2}\delta\left(  y-x\right)  \delta\left(  t-s\right)  .
\]
Now by assuming time translation invariance and taking $\tau=t-s$, we have%
\begin{equation}
\left(  -\partial_{\tau}^{2}+\gamma^{2}\right)  C_{\boldsymbol{hh}}\left(
x,y,\tau\right)  =C_{\phi\left(  \boldsymbol{h}\right)  \phi\left(
\boldsymbol{h}\right)  }\left(  x,y,\tau\right)  +\sigma^{2}\delta\left(
y-x\right)  \delta\left(  \tau\right)  . \label{Eq_12B}%
\end{equation}
Then, $G_{\boldsymbol{hh}}\left(  x,y,\tau\right)  =C_{\boldsymbol{hh}}\left(
x,y,\tau\right)  $, so we can interpret $G_{\boldsymbol{hh}}\left(
x,y,\tau\right)  $ as the autocorrelation of $\boldsymbol{h}$.
\end{remark}

\subsection{\label{Section_Propagator}The propagator $G_{\boldsymbol{hh}%
}\left(  x,y,\tau\right)  $}

There exists\ a subset $\mathcal{M}\subset\Omega\times\mathbb{R}$ with $d\mu
dt$-measure zero, such that for any $\left(  x,t\right)  \in\Omega
\times\mathbb{R}$ $\mathbb{\smallsetminus}$ $\mathcal{M}$, the mapping%
\[%
\begin{array}
[c]{cccc}%
\operatorname{eval}_{\left(  x,t\right)  } & \mathcal{H} & \rightarrow &
\mathbb{R}\\
&  &  & \\
& \boldsymbol{h} & \rightarrow & \boldsymbol{h}\left(  x,t\right)
\end{array}
\]
is an ordinary Gaussian random variable with mean zero, cf. Lemma
\ref{Lemma_16} in Appendix E. We now fixed to two points $\left(
x,t_{1}\right)  $, $\left(  y,t_{2}\right)  \in\Omega\times\mathbb{R}$
$\mathbb{\smallsetminus}$ $\mathcal{M}$, and take $\tau=t_{1}-t_{2}$, then
$\boldsymbol{h}\left(  x,t_{1}\right)  $, $\boldsymbol{h}\left(
y,t_{2}\right)  $ are Gaussian random variables with covariance matrix
\[
G_{\boldsymbol{hh}}\left(  x,y,\tau\right)  =\left\langle \boldsymbol{h}%
\left(  x,t_{1}\right)  \boldsymbol{h}\left(  y,t_{2}\right)  \right\rangle
_{\boldsymbol{h}}.
\]
Here, it is important to mention that under the time translation symmetry
hypothesis,
\[
G_{\boldsymbol{hh}}\left(  x,y,0\right)  =G_{\boldsymbol{hh}}\left(
x,y,t_{1},t_{2}\right)  \mid_{t_{1}=t_{2}}\text{, for }\left(  x,t_{1}\right)
,\left(  x,t_{2}\right)  \in\Omega\times\mathbb{R\smallsetminus}%
\mathcal{M}\text{,}%
\]
is well-defined.

We adapt the techniques presented in \cite{Sompolinsky et al}, \cite{Helias et
al}, \cite{Grosvenor-Jefferson}, among many references, to derive a
differential equation for $G_{\boldsymbol{hh}}\left(  x,y,\tau\right)  $.

For the rest of this section, we fix $x,y\in\Omega$, and adopt the notation,%
\[
\boldsymbol{h}\left(  \cdot,t_{i}\right)  =\boldsymbol{h}\left(  t_{i}\right)
=\boldsymbol{h}_{i}\text{, for }i\in\left\{  1,2\right\}  .
\]
These values are realizations of the random vector $\left(  \boldsymbol{h}%
_{1}\text{, }\boldsymbol{h}_{2}\right)  $ with bivariate normal distribution%
\begin{equation}
\left(  \boldsymbol{h}_{1}\text{, }\boldsymbol{h}_{2}\right)  \sim
\mathcal{N}\left(  0,\left[
\begin{array}
[c]{cc}%
G_{\boldsymbol{hh}}\left(  x,y,0\right)  & G_{\boldsymbol{hh}}\left(
x,y,\tau\right) \\
& \\
G_{\boldsymbol{hh}}\left(  x,y,\tau\right)  & G_{\boldsymbol{hh}}\left(
x,y,0\right)
\end{array}
\right]  \right)  =\mathcal{N}\left(  0,\left[
\begin{array}
[c]{cc}%
c_{0} & c_{\tau}\\
& \\
c_{\tau} & c_{0}%
\end{array}
\right]  \right)  , \label{covariance_single_copy}%
\end{equation}
where $c_{0}:=c_{0}\left(  x,y\right)  =\left\langle \boldsymbol{h}\left(
x,t\right)  \boldsymbol{h}\left(  x,t\right)  \right\rangle _{\boldsymbol{h}}%
$, and
\begin{equation}
c_{\tau}:=c_{\tau}\left(  x,y\right)  =\left\langle \boldsymbol{h}\left(
x,t_{1}\right)  \boldsymbol{h}\left(  y,t_{2}\right)  \right\rangle
_{\boldsymbol{h}}\text{ \ for }\tau=t_{1}-t_{2}\neq0. \label{correlator_tau}%
\end{equation}
On the other hand, we set%
\[
C_{\phi\phi}:=C_{\phi\left(  \boldsymbol{h}_{1}\right)  \phi\left(
\boldsymbol{h}_{2}\right)  }=C_{\phi\left(  \boldsymbol{h}\right)  \phi\left(
\boldsymbol{h}\right)  }\left(  x,y,t_{1},t_{2}\right)  ,
\]
then%
\[
C_{\phi\phi}=\left\langle \phi\left(  \boldsymbol{h}_{1}\right)  \phi\left(
\boldsymbol{h}_{2}\right)  \right\rangle _{\boldsymbol{h}}=%
{\displaystyle\int\limits_{\mathbb{R}^{2}}}
\phi\left(  \boldsymbol{h}_{1}\right)  \phi\left(  \boldsymbol{h}_{2}\right)
D\boldsymbol{h}_{1}D\boldsymbol{h}_{2},
\]
where%
\[
D\boldsymbol{h}_{1}D\boldsymbol{h}_{2}=\frac{\exp\left(  -\frac{1}%
{2c_{0}\left(  1-\rho^{2}\right)  }\left\{  \boldsymbol{h}_{1}^{2}%
+\boldsymbol{h}_{2}^{2}-2\rho\boldsymbol{h}_{1}\boldsymbol{h}_{2}\right\}
\right)  }{2\pi\sqrt{c_{0}^{2}\left(  1-\rho^{2}\right)  }}d\boldsymbol{h}%
_{1}d\boldsymbol{h}_{2},
\]
here $d\boldsymbol{h}_{1}d\boldsymbol{h}_{2}$ is the Lebesgue measure of
\ $\mathbb{R\times R}$, and $\rho=\frac{c_{\tau}}{c_{0}}$. By changing
variables as%
\begin{align*}
\boldsymbol{h}_{1}  &  =\sqrt{c_{0}}\boldsymbol{h}_{a}\text{, \ }%
\boldsymbol{h}_{2}=\sqrt{c_{0}}\left(  \rho\boldsymbol{h}_{b}+\sqrt{\left(
1-\rho^{2}\right)  }\boldsymbol{h}_{b}\right)  ,\\
D\boldsymbol{h}_{a}D\boldsymbol{h}_{b}  &  =\frac{\exp\left(  -\frac{1}%
{2}\left\{  \boldsymbol{h}_{a}^{2}+\boldsymbol{h}_{b}^{2}\right\}  \right)
}{2\pi}d\boldsymbol{h}_{a}d\boldsymbol{h}_{b},
\end{align*}
we have%
\[
C_{\phi\phi}=%
{\displaystyle\int\limits_{\mathbb{R}^{2}}}
\phi\left(  \sqrt{c_{0}}\boldsymbol{h}_{a}\right)  \phi\left(  \sqrt{c_{0}%
}\left(  \rho\boldsymbol{h}_{b}+\sqrt{\left(  1-\rho^{2}\right)
}\boldsymbol{h}_{b}\right)  \right)  D\boldsymbol{h}_{a}D\boldsymbol{h}_{b}.
\]
By Price's theorem, \cite{Price}-\cite{Papoulis et al}, for a Gaussian
process, it verifies that%
\begin{equation}
C_{\phi\phi}=\frac{\partial}{\partial c_{\tau}}C_{\Phi\Phi},\text{ with }%
\Phi\left(  u\right)  =%
{\displaystyle\int\nolimits_{0}^{u}}
\phi\left(  s\right)  ds\text{, }s\in\mathbb{R}_{+}. \label{Price_theorem}%
\end{equation}
Now, using equation (\ref{Eq_11}), with%
\[
G_{\boldsymbol{hh}}\left(  x,y,\tau\right)  =c_{\tau}\left(  x,y\right)
=c_{\tau}\text{, }C_{\phi\phi}\left(  \tau\right)  :=C_{\phi\phi}\left(
x,y,\tau\right)  \text{, }%
\]
we have%
\begin{equation}
\partial_{\tau}^{2}c_{\tau}=\gamma^{2}c_{\tau}-C_{\phi\phi}\left(
\tau\right)  -\sigma^{2}\delta\left(  y-x\right)  \delta\left(  \tau\right)  .
\label{Eq_12}%
\end{equation}
We now introduce the potential%
\begin{equation}
V\left(  c_{\tau},c_{0}\right)  =-\frac{1}{2}\gamma^{2}c_{\tau}^{2}%
+C_{\Phi\Phi}\left(  c_{\tau},c_{0}\right)  . \label{Potential}%
\end{equation}
Then, equation (\ref{Eq_12}) takes the form%
\begin{equation}
\partial_{\tau}^{2}c_{\tau}=-V^{\prime}\left(  c_{\tau},c_{0}\right)
-\sigma^{2}\delta\left(  y-x\right)  \delta\left(  \tau\right)  ,
\label{Key_Eq_1}%
\end{equation}
where the prime denotes the derivative with respect to $c_{\tau}$. The
equation (\ref{Key_Eq_1}) is a generalization of the differential equations
(6-7) in \cite{Sompolinsky et al}, and (10.26) in \cite{Helias et al}. In our
generalized equation the noise and the covariance function $c_{\tau}%
\in\mathcal{D}^{\prime}(\Omega)%
{\textstyle\bigotimes_{\text{alg}}}
\mathcal{S}^{\prime}(\mathbb{R})$ depend on the position of two neurons $x,y$;
see Appendix G.

\section{\label{SECT_9}The double-copy system}

In this section, we present a mathematical formulation of the continuous
version of the double-copy system. Our generating functional for the
double-copy system is the \ formal thermodynamic limit of the one given in
\cite[Formula (10.32)]{Helias et al}:%
\begin{gather*}
{\LARGE Z}^{\left(  2\right)  }\left(  \left\{  \boldsymbol{j}^{\alpha
},\widetilde{\boldsymbol{j}}^{\alpha}\right\}  _{\alpha\in\left\{
1,2\right\}  };\mathbf{J}\right)  =\\%
{\displaystyle\prod\limits_{\alpha=1}^{2}}
\left\{
{\displaystyle\int}
D\boldsymbol{h}^{\alpha}%
{\displaystyle\int}
D\widetilde{\boldsymbol{h}}^{\alpha}\exp\left(  S_{0}\left[  \boldsymbol{h}%
^{\alpha},\widetilde{\boldsymbol{h}}^{\alpha}\right]  -\left\langle
\widetilde{\boldsymbol{h}}^{\alpha},%
{\displaystyle\int\limits_{\Omega}}
\boldsymbol{J}\left(  x,y\right)  \phi\left(  \boldsymbol{h}^{\alpha}\left(
y,t\right)  \right)  d\mu\left(  y\right)  \right\rangle \right)
\times\right. \\
\left.  \exp\left(  \left\langle \boldsymbol{j}^{\alpha},\boldsymbol{h}%
^{\alpha}\right\rangle +\left\langle \widetilde{\boldsymbol{j}}^{\alpha
},\widetilde{\boldsymbol{h}}^{\alpha}\right\rangle \right)  \right\}
\exp\left(  \frac{\sigma^{2}}{2}\left\langle \widetilde{\boldsymbol{h}}%
^{1},\widetilde{\boldsymbol{h}}^{2}\right\rangle \right)
\end{gather*}

\begin{gather*}
=%
{\displaystyle\iint}
{\displaystyle\iint}
\left(
{\displaystyle\prod\limits_{\alpha=1}^{2}}
D\widetilde{\boldsymbol{h}}^{\alpha}D\boldsymbol{h}^{\alpha}\right)
\exp\left(  -S_{0}\left(  \widetilde{\boldsymbol{h}}^{\alpha},\boldsymbol{h}%
^{\alpha}\right)  \right)  \times\\
\exp\left(  S_{int}\left(  \widetilde{\boldsymbol{h}}^{\alpha},\boldsymbol{h}%
^{\alpha},\boldsymbol{J}\right)  +\sum_{\alpha=1}^{2}\left(  \left\langle
\boldsymbol{j}^{\alpha},\boldsymbol{h}^{\alpha}\right\rangle +\left\langle
\widetilde{\boldsymbol{j}}^{\alpha},\widetilde{\boldsymbol{h}}^{\alpha
}\right\rangle \right)  \right)  \exp\left(  \frac{\sigma^{2}}{2}\left\langle
\widetilde{\boldsymbol{h}}^{1},\widetilde{\boldsymbol{h}}^{2}\right\rangle
\right)  ,
\end{gather*}
where%
\[
S_{0}\left(  \widetilde{\boldsymbol{h}}^{\alpha},\boldsymbol{h}^{\alpha
}\right)  :=\sum_{\alpha=1}^{2}\left\langle \widetilde{\boldsymbol{h}}%
^{\alpha},\left(  \partial_{t}+\gamma\right)  \boldsymbol{h}^{\alpha
}\right\rangle \text{,}%
\]
and
\[
S_{int}\left(  \widetilde{\boldsymbol{h}}^{\alpha},\boldsymbol{h}^{\alpha
},\boldsymbol{J}\right)  :=\sum_{\alpha=1}^{2}\left\langle \widetilde
{\boldsymbol{h}}^{\alpha},%
{\displaystyle\int\limits_{\Omega}}
\boldsymbol{J}\left(  x,y\right)  \phi\left(  \boldsymbol{h}^{\alpha}\left(
y,t\right)  \right)  d\mu\left(  y\right)  \right\rangle .
\]
We use the following mathematical version of the ${\LARGE Z}^{\left(
2\right)  }(\left\{  \boldsymbol{j}^{\alpha},\widetilde{\boldsymbol{j}%
}^{\alpha}\right\}  _{\alpha\in\left\{  1,2\right\}  })\left(  \mathbf{J}%
\right)  $:%
\begin{multline*}
{\LARGE Z}_{M}^{\left(  2\right)  }\left(  \left\{  \boldsymbol{j}^{\alpha
},\widetilde{\boldsymbol{j}}^{\alpha}\right\}  _{\alpha\in\left\{
1,2\right\}  };\boldsymbol{J}\right)  :=%
{\displaystyle\iint\limits_{\mathcal{H}\times\mathcal{H}}}
\text{ \ }%
{\displaystyle\iint\limits_{\mathcal{H}\times\mathcal{H}}}
\exp\left(  S_{0}\left(  \widetilde{\boldsymbol{h}}^{\alpha},\boldsymbol{h}%
^{\alpha}\right)  -S_{int}\left(  \widetilde{\boldsymbol{h}}^{\alpha
},\boldsymbol{h}^{\alpha}\right)  \right)  \times\\
\exp\left(  \sum_{\alpha=1}^{2}\left(  \left\langle \boldsymbol{j}^{\alpha
},\boldsymbol{h}^{\alpha}\right\rangle +\left\langle \widetilde{\boldsymbol{j}%
}^{\alpha},\widetilde{\boldsymbol{h}}^{\alpha}\right\rangle \right)
+\frac{\sigma^{2}}{2}\left\langle \widetilde{\boldsymbol{h}}^{1}%
,\widetilde{\boldsymbol{h}}^{2}\right\rangle \right)  \times\\%
{\displaystyle\prod\limits_{\alpha=1}^{2}}
{\large 1}_{\mathcal{P}_{M}}\left(  \boldsymbol{h}^{\alpha},\widetilde
{\boldsymbol{h}}^{\alpha}\right)  d\boldsymbol{P}\left(  \boldsymbol{h}%
^{\alpha}\right)  d\widetilde{\boldsymbol{P}}\left(  \widetilde{\boldsymbol{h}%
}^{\alpha}\right)  .
\end{multline*}
We only require the partition function ${\LARGE Z}_{M}^{\left(  2\right)
}\left(  \left\{  \boldsymbol{0},\boldsymbol{0}\right\}  _{\alpha\in\left\{
1,2\right\}  };\mathbf{J}\right)  ={\LARGE Z}_{M}^{\left(  2\right)  }\left(
\mathbf{J}\right)  $. The calculation of the expected value$\left\langle
{\LARGE Z}_{M}^{\left(  2\right)  }\left(  \mathbf{J}\right)  \right\rangle
_{\mathbf{J}}=\overline{{\LARGE Z}_{M}^{\left(  2\right)  }}$ is similar the
single-copy case, the details are given in the Appendix F, see Lemma
\ref{Lemma_14} and Theorem \ref{Lemma_15}. The average partition function for
the double-copy system is given by%
\begin{gather*}
\overline{{\LARGE Z}_{M}^{\left(  2\right)  }}=%
{\displaystyle\iint\limits_{\mathcal{H}\times\mathcal{H}}}
\text{ \ }%
{\displaystyle\iint\limits_{\mathcal{H}\times\mathcal{H}}}
\exp\left(  S_{0}\left(  \widetilde{\boldsymbol{h}}^{\alpha},\boldsymbol{h}%
^{\alpha}\right)  +\frac{\sigma^{2}}{2}\left\langle \widetilde{\boldsymbol{h}%
}^{1},\widetilde{\boldsymbol{h}}^{2}\right\rangle \right)  \times\\
\exp\left(  \frac{1}{2}%
{\displaystyle\sum\limits_{\alpha=1}^{2}}
\text{ }%
{\displaystyle\sum\limits_{\beta=1}^{2}}
\text{ \ }\left\langle \widetilde{\boldsymbol{h}}^{\alpha},C_{\phi\left(
\boldsymbol{h}^{\alpha}\right)  \phi\left(  \boldsymbol{h}^{\beta}\right)
}^{\alpha\beta}\widetilde{\boldsymbol{h}}^{\beta}\right\rangle _{\left(
L^{2}(\Omega\times\mathbb{R})\right)  ^{2}}\right)  \times\\%
{\displaystyle\prod\limits_{\alpha=1}^{2}}
{\large 1}_{\mathcal{P}_{M}}\left(  \boldsymbol{h}^{\alpha},\widetilde
{\boldsymbol{h}}^{\alpha}\right)  d\boldsymbol{P}\left(  \boldsymbol{h}%
^{\alpha}\right)  d\widetilde{\boldsymbol{P}}\left(  \widetilde{\boldsymbol{h}%
}^{\alpha}\right)  ,
\end{gather*}
where the bilinear forms $\left\langle \widetilde{\boldsymbol{h}}^{\alpha
},C_{\phi\left(  \boldsymbol{h}^{\alpha}\right)  \phi\left(  \boldsymbol{h}%
^{\beta}\right)  }^{\alpha\beta}\widetilde{\boldsymbol{h}}^{\beta
}\right\rangle _{\left(  L^{2}(\Omega\times\mathbb{R})\right)  ^{2}}$ are
correlation functionals of a Gaussian noises with mean zero, where the
$C_{\phi\left(  \boldsymbol{h}^{\alpha}\right)  \phi\left(  \boldsymbol{h}%
^{\beta}\right)  }^{\alpha\beta}$ are defined in the formula (\ref{Eq_8}), in
Appendix F.

In the next step, we show that%
\begin{multline*}
\overline{{\LARGE Z}_{M}^{\left(  2\right)  }}=%
{\displaystyle\iint\limits_{\mathcal{H}\times\mathcal{H}}}
\text{ \ }%
{\displaystyle\iint\limits_{\mathcal{H}\times\mathcal{H}}}
\exp\left(  S\left[  \boldsymbol{h},\widetilde{\boldsymbol{h}};\left\{
\Xi^{\alpha\beta}\right\}  _{\alpha,\beta\in\left\{  1,2\right\}  }\right]
\right)  \times\\%
{\displaystyle\prod\limits_{\alpha=1}^{2}}
{\large 1}_{\mathcal{P}_{M}}\left(  \boldsymbol{h}^{\alpha},\widetilde
{\boldsymbol{h}}^{\alpha}\right)  d\boldsymbol{P}\left(  \boldsymbol{h}%
^{\alpha}\right)  d\widetilde{\boldsymbol{P}}\left(  \widetilde{\boldsymbol{h}%
}^{\alpha}\right)  ,
\end{multline*}
where
\begin{gather*}
S\left[  \boldsymbol{h},\widetilde{\boldsymbol{h}};\left\{  \Xi^{\alpha\beta
}\right\}  _{\alpha,\beta\in\left\{  1,2\right\}  }\right]  :=\\
\frac{1}{2}%
{\displaystyle\int\limits_{\Omega^{2}}}
\text{ }%
{\displaystyle\int\limits_{\mathbb{R}^{2}}}
\boldsymbol{z}^{\alpha}\left(  x,t_{1}\right)  ^{T}\Xi^{\alpha\beta}\left(
x,y,t_{1},t_{2}\right)  \boldsymbol{z}^{\beta}\left(  y,t_{2}\right)
d\mu\left(  x\right)  d\mu\left(  y\right)  dt_{1}dt_{2}.
\end{gather*}
and the functions $\left\{  \Xi^{\alpha\beta}\right\}  _{\alpha,\beta
\in\left\{  1,2\right\}  }$, $S\left[  \boldsymbol{h},\widetilde
{\boldsymbol{h}};\left\{  \Xi^{\alpha\beta}\right\}  _{\alpha,\beta\in\left\{
1,2\right\}  }\right]  $ are well-defined in $L^{2}\left(  \Omega
\times\mathbb{R}\right)  \times L^{2}\left(  \Omega\times\mathbb{R}\right)  $,
so they do not depend on the cutoff function. See Theorem \ref{Theorem3} in
Appendix F for all the details.

\subsection{\label{Section matrix prop}Matrix propagators of the double-copy}

We set%
\[
G^{\alpha\beta}\left(  x,y,t_{1},t_{2}\right)  =\left[
\begin{array}
[c]{cc}%
G_{\boldsymbol{hh}}^{\alpha\beta} & G_{\boldsymbol{h}\widetilde{\boldsymbol{h}%
}}^{\alpha\beta}\\
& \\
G_{\widetilde{\boldsymbol{h}}\boldsymbol{h}}^{\alpha\beta} & G_{\widetilde
{\boldsymbol{h}}\widetilde{\boldsymbol{h}}}^{\alpha\beta}%
\end{array}
\right]  \text{, with \ }G_{\widetilde{\boldsymbol{h}}\widetilde
{\boldsymbol{h}}}^{\alpha\beta}=0\text{,}%
\]
and
\[
\Xi^{\alpha\beta}\left(  x,y,t_{1},t_{2}\right)  =\left[
\begin{array}
[c]{ccc}%
\Xi_{\boldsymbol{hh}}^{\alpha\beta} &  & \Xi_{\boldsymbol{h}\widetilde
{\boldsymbol{h}}}^{\alpha\beta}\\
&  & \\
\Xi_{\widetilde{\boldsymbol{h}}\boldsymbol{h}}^{\alpha\beta} &  &
\Xi_{\widetilde{\boldsymbol{h}}\widetilde{\boldsymbol{h}}}^{\alpha\beta}%
\end{array}
\right]  ,
\]
for $\alpha,\beta\in\left\{  1,2\right\}  $.

We now look for the propagators $G^{\alpha\beta}\left(  x,y,t_{1}%
,t_{2}\right)  $ of the double-copy MFT defined as the matrix elements of the
right-inverse of $\Xi^{\alpha\beta}\left(  x,y,t_{1},t_{2}\right)  $:%
\begin{gather}%
{\displaystyle\sum\limits_{\mu=1}^{2}}
\text{ }%
{\displaystyle\int\limits_{\Omega}}
{\displaystyle\int\limits_{\mathbb{R}}}
\Xi^{\alpha\mu}\left(  x,z,t_{1},s\right)  G^{\mu\beta}\left(  z,y,s,t_{2}%
\right)  dsd\mu\left(  z\right)  =\label{Motion_equation}\\
\delta_{\alpha\beta}\delta\left(  x-y\right)  \delta\left(  t_{1}%
-t_{2}\right)  \left[
\begin{array}
[c]{cc}%
1 & 0\\
0 & 1
\end{array}
\right]  .\nonumber
\end{gather}
This yields to the following system of equations for the $G^{\alpha\beta
}\left(  x,y,t_{1},t_{2}\right)  $:%
\begin{equation}
\left(  \partial_{t_{1}}-\gamma\right)  G_{\widetilde{\boldsymbol{h}%
}\boldsymbol{h}}^{\alpha\beta}\left(  x,y,t_{1},t_{2}\right)  =\delta
_{\alpha\beta}\delta\left(  x-y\right)  \delta\left(  t_{1}-t_{2}\right)  ,
\label{Eq_9}%
\end{equation}%
\begin{gather}
-\left(  \partial_{t_{1}}+\gamma\right)  G_{\boldsymbol{hh}}^{\alpha\beta
}\left(  x,y,t_{1},t_{2}\right)  +\label{Eq_10}\\%
{\displaystyle\sum\limits_{\mu=1}^{2}}
{\displaystyle\int\limits_{\Omega}}
{\displaystyle\int\limits_{\mathbb{R}}}
\left\{  C_{\phi\left(  \boldsymbol{h}^{\alpha}\right)  \phi\left(
\boldsymbol{h}^{\beta}\right)  }^{\alpha\mu}\left(  x,z,t_{1},s\right)
G_{\widetilde{\boldsymbol{h}}\boldsymbol{h}}^{\mu\beta}\left(  z,y,s,t_{2}%
\right)  dsd\mu\left(  z\right)  \right\} \nonumber\\
+\sigma^{2}%
{\displaystyle\sum\limits_{\mu=1}^{2}}
G_{\widetilde{\boldsymbol{h}}\boldsymbol{h}}^{\mu\beta}\left(  x,y,t_{1}%
,t_{2}\right)  ds=0\nonumber
\end{gather}%
\begin{equation}
-\left(  \partial_{t_{1}}+\gamma\right)  G_{\boldsymbol{h}\widetilde
{\boldsymbol{h}}}^{\alpha\beta}\left(  x,y,t_{1},t_{2}\right)  =\delta
_{\alpha\beta}\delta\left(  x-y\right)  \left(  t_{1}-t_{2}\right)  .
\label{Eq_11A}%
\end{equation}
Besides that our equation (\ref{Motion_equation}) is different from the one in
\cite[equation 65]{Grosvenor-Jefferson}, the derivation of the equations
(\ref{Eq_9})-(\ref{Eq_11A}) is completely similar to the one presented in
\cite[pp. 21-22]{Grosvenor-Jefferson}. The formal calculations presented in
this section can be rigorously justified with the techniques presented in
Appendix G. Again, we mention that the system (\ref{Eq_9})-(\ref{Eq_11A}) is
independent of the cutoff function used in the definition of the partition function.

Assuming time-translation invariance, and replacing $\partial_{1}$ by
$-\partial_{2}$, $\tau:=t_{1}-t_{2}$, equations (\ref{Eq_9})and (\ref{Eq_11A})
become%
\[
-\left(  \partial_{\tau}+\gamma\right)  G_{\widetilde{\boldsymbol{h}%
}\boldsymbol{h}}^{\alpha\beta}\left(  x,y,\tau\right)  =\left(  \partial
_{\tau}-\gamma\right)  G_{\boldsymbol{h}\widetilde{\boldsymbol{h}}}%
^{\alpha\beta}\left(  x,y,\tau\right)  =\delta_{\alpha\beta}\delta\left(
x-y\right)  \delta\left(  \tau\right)  .
\]
By multiplying the equation (\ref{Eq_10}) by $\left(  \partial_{2}%
+\gamma\right)  $, and using equation (\ref{Eq_9}) with $\partial
_{1}=-\partial_{2}$, we obtain%
\begin{gather}
\left(  \partial_{2}+\gamma\right)  \left(  \partial_{1}+\gamma\right)
G_{\boldsymbol{hh}}^{\alpha\beta}\left(  x,y,t_{1},t_{2}\right)
=\label{Eq_15A}\\
C_{\phi\left(  \boldsymbol{h}^{\alpha}\right)  \phi\left(  \boldsymbol{h}%
^{\beta}\right)  }^{\alpha\beta}\left(  x,y,t_{1},t_{2}\right)  +\sigma
^{2}\delta\left(  x-y\right)  \delta\left(  t_{1}-t_{2}\right)  .\nonumber
\end{gather}
Finally taking $\partial_{1}=-\partial_{2}$, and $\tau=t_{1}-t_{2}$, we have%
\begin{equation}
\left(  -\partial_{\tau}^{2}+\gamma^{2}\right)  G_{\boldsymbol{hh}}%
^{\alpha\beta}\left(  x,y,\tau\right)  =C_{\phi\left(  \boldsymbol{h}^{\alpha
}\right)  \phi\left(  \boldsymbol{h}^{\beta}\right)  }^{\alpha\beta}\left(
x,y,\tau\right)  +\sigma^{2}\delta\left(  x-y\right)  \delta\left(
\tau\right)  . \label{Eq_15}%
\end{equation}

We note that the $\alpha=\beta$ components in the equation (\ref{Eq_15}) are
precisely (\ref{Eq_12B}), as expected, and thus the autocorrelation
$G_{\boldsymbol{hh}}^{\alpha\alpha}\left(  x,y,\tau\right)  :=c^{\alpha\alpha
}\left(  x,y,\tau\right)  $ agrees with the autocorrelation of the single-copy
case $c\left(  x,y,\tau\right)  $, i.e. $c^{11}\left(  x,y,\tau\right)
=c^{22}\left(  x,y,\tau\right)  =c\left(  x,y,\tau\right)  $. Finally, when
interpreting the propagators as correlations between ordinary Gaussian
variables, it is necessary to recall that this interpretation is valid when
the parameters $x,y$, and $\tau$ are taken outside of a set of $d\mu
dt$-measure zero in $\Omega\times\mathbb{R}$.

\section{\label{SECT_10}The edge of chaos}

In this section, we derive the largest Lyapunov exponent of a random NN of
type (\ref{Eq_Network_gneral}), which allows us to assess the conditions under
which the network transitions into the chaotic regime. We adapt formally the
techniques developed in \cite{Sompolinsky et al}, \cite[Chapter 10]{Helias et
al}, \cite{Grosvenor-Jefferson}. However, this task is not straightforward
because new features appear in the continuous case.

We consider a pair of identically prepared networks with identical parameter
$\boldsymbol{J}$, and also with the same realization of the Gaussian noise
$\eta\left(  x,t\right)  $. Let $\boldsymbol{h}^{1}\left(  x,t_{1}\right)  $,
$\boldsymbol{h}^{2}\left(  y,t_{2}\right)  $ be the trajectories of the
prepared networks.\ Given $\left(  x,t_{1}\right)  $, $\left(  y,t_{2}\right)
\in\Omega\times\mathbb{R}$, the cross-correlator of $\boldsymbol{h}^{\alpha
}\left(  x,t_{1}\right)  $, $\boldsymbol{h}^{\beta}\left(  y,t_{2}\right)  $
is defined \ as%
\begin{align*}
G_{\boldsymbol{hh}}^{\alpha\beta}\left(  x,y,t_{1},t_{2}\right)   &
=c^{\alpha\beta}\left(  x,y,t_{1},t_{2}\right)  =\left\langle \boldsymbol{h}%
^{\alpha}\left(  x,t_{1}\right)  \boldsymbol{h}^{\beta}\left(  y,t_{2}\right)
\right\rangle _{\left(  \boldsymbol{h}^{\alpha},\boldsymbol{h}^{\beta}\right)
}=\\
&
{\displaystyle\iint\limits_{\left(  L^{2}(\Omega\times\mathbb{R})\right)
^{2}}}
\text{ }\boldsymbol{h}^{\alpha}\left(  x,t_{1}\right)  \boldsymbol{h}^{\beta
}\left(  y,t_{2}\right)  d\boldsymbol{Q}\left(  \boldsymbol{h}^{\alpha
}\right)  d\boldsymbol{Q}\left(  \boldsymbol{h}^{\beta}\right)  ,
\end{align*}
where $d\boldsymbol{Q}\left(  \boldsymbol{h}^{\alpha}\right)  d\boldsymbol{Q}%
\left(  \boldsymbol{h}^{\beta}\right)  $\ is the product of two probability
measures $d\boldsymbol{Q}\left(  \boldsymbol{h}^{\alpha}\right)  $,
$d\boldsymbol{Q}\left(  \boldsymbol{h}^{\beta}\right)  $ defined on
$L^{2}(\Omega\times\mathbb{R})$.

Our next goal is to define a notion of distance between $\boldsymbol{h}%
^{1}\left(  \cdot,\cdot\right)  $, $\boldsymbol{h}^{2}\left(  \cdot
,\cdot\right)  \in L^{2}(\Omega\times\mathbb{R})$. We fix two points $\left(
x,t\right)  $, $\left(  x,s\right)  \in\Omega\times\mathbb{R}$, and use the
cross-correlator $G_{\boldsymbol{hh}}^{\alpha\beta}$ to define the following
pseudo-distance:%
\begin{align}
d(\boldsymbol{h}^{1},\boldsymbol{h}^{2};x,x,t,s)  &  =d(x,x,t,s)\nonumber\\
&  :=%
{\displaystyle\iint\limits_{\left(  L^{2}(\Omega\times\mathbb{R})\right)
^{2}}}
\text{ \ }\left\{  \boldsymbol{h}^{1}\left(  x,t\right)  -\boldsymbol{h}%
^{2}\left(  x,s\right)  \right\}  ^{2}d\boldsymbol{Q}\left(  \boldsymbol{h}%
^{1}\right)  d\boldsymbol{Q}\left(  \boldsymbol{h}^{2}\right) \nonumber\\
&  =%
{\displaystyle\int\limits_{L^{2}\left(  \Omega\times\mathbb{R}\right)  }}
\text{ }\boldsymbol{h}^{1}\left(  x,t\right)  ^{2}d\boldsymbol{Q}\left(
\boldsymbol{h}^{1}\right)  +%
{\displaystyle\int\limits_{L^{2}\left(  \Omega\times\mathbb{R}\right)  }}
\text{ }\boldsymbol{h}^{2}\left(  x,s\right)  ^{2}d\boldsymbol{Q}\left(
\boldsymbol{h}^{2}\right) \nonumber\\
&  -2%
{\displaystyle\iint\limits_{\left(  L^{2}(\Omega\times\mathbb{R})\right)
^{2}}}
\text{ \ }\boldsymbol{h}^{1}\left(  x,t\right)  \boldsymbol{h}^{2}\left(
x,s\right)  d\boldsymbol{Q}\left(  \boldsymbol{h}^{1}\right)  d\boldsymbol{Q}%
\left(  \boldsymbol{h}^{2}\right) \nonumber\\
&  =c^{11}\left(  x,x,t,t\right)  +c^{22}\left(  x,x,s,s\right)
-2c^{12}\left(  x,x,t,s\right)  . \label{distance_formula}%
\end{align}
We note that $\sqrt{d(x,x,t,s)}$ is a distance. Here, we should mention a
subtle point: the autocorrelation $c^{11}\left(  x,x,t,t\right)  $ is the
restriction of the distribution $c^{11}\left(  x,y,t,s\right)  $ to the plane
$\left\{  y=x,s=t\right\}  $, which is a non-trivial operation. However, by
Lemma \ref{Lemma_16} in Appendix E, that for generic points $\left(
x,t\right)  ,\left(  y,s\right)  $, $c^{11}\left(  x,y,t,s\right)  $ is the
correlation function of two ordinary Gaussian variables, with mean zero, and
consequently, $c^{11}\left(  x,x,t,t\right)  $ is a autocorrelation of an
ordinary Gaussian random variable for $\left(  x,t\right)  $ in a generic set.

\subsection{Lyapunov exponents}

We now adapt formally the technique introduced in \cite{Sompolinsky et al},
see also \cite[Chapter 10]{Helias et al} and \cite{Grosvenor-Jefferson}, to
compute the Lyapunov exponents associate with continuous random networks.

We fix $x\in\Omega$, and use the notation $c^{\alpha\beta}\left(
x,x,t,s\right)  :=c^{\alpha\beta}\left(  t,s\right)  =c^{\alpha\beta}\left(
\tau\right)  $, with $\tau=t-s$. We note that
\[
c\left(  \tau\right)  =c^{11}\left(  \tau\right)  =c^{22}\left(  \tau\right)
,
\]
where $c\left(  \tau\right)  =c_{\tau}$ was already introduced in
(\ref{correlator_tau}). Here, we used the results of Section
\ref{Section_Propagator}, however, in this section, we consider correlations
$c_{\tau}\left(  x,y\right)  ,$ $\tau=t_{1}-t_{2}$ depending on two points $x$
and $y$, here $x=y$, and thus $c_{0}=c_{0}\left(  x,x\right)  $, and $c_{\tau
}=c_{\tau}\left(  x,x\right)  $.

We start with the following linear approximation of the cross-correlator
$c^{12}\left(  t,s\right)  =G_{\boldsymbol{hh}}^{12}$:%
\begin{equation}
c^{12}\left(  t,s\right)  =c\left(  \tau\right)  +\eta k\left(  t,s\right)
+O\left(  \eta^{2}\right)  , \label{EQ_C_12}%
\end{equation}
where $\eta\ll1$ is some small expansion parameter.

\begin{remark}
We set $C_{\phi\phi}^{\alpha\beta}\left(  x,x,t,s\right)  :=C_{\phi\left(
\boldsymbol{h}^{\alpha}\right)  \phi\left(  \boldsymbol{h}^{\beta}\right)
}^{\alpha\beta}\left(  x,x,t,s\right)  $, and propose that%
\[
C_{\phi\phi}^{\alpha\beta}\left(  x,x,t,s\right)  =C_{\phi\phi}^{\alpha\beta
}\left(  t,s\right)  =\left\langle \phi\left(  \boldsymbol{h}^{\alpha}\left(
x,t\right)  \right)  \phi\left(  \boldsymbol{h}^{\beta}\left(  x,s\right)
\right)  \right\rangle _{\boldsymbol{h}^{\alpha}\boldsymbol{h}^{\beta}}.
\]
This interpretation is in line with the fact that $\left\langle \widetilde
{\boldsymbol{h}}^{\alpha},C_{\phi\left(  \boldsymbol{h}^{\alpha}\right)
\phi\left(  \boldsymbol{h}^{\beta}\right)  }^{\alpha\beta}\widetilde
{\boldsymbol{h}}^{\beta}\right\rangle _{\left(  L^{2}(\Omega\times
\mathbb{R})\right)  ^{2}}$ are correlation functionals of a Gaussian noises
with mean zero, and then, the $C_{\phi\left(  \boldsymbol{h}^{\alpha}\right)
\phi\left(  \boldsymbol{h}^{\beta}\right)  }^{\alpha\beta}$ are the kernels of
the covariance operators of these noises. From a mathematical perspective, the
existence of a restriction of the distribution $C_{\phi\phi}^{\alpha\beta
}\left(  x,y,t,s\right)  $ to the hyperplane $y=x$ is a non-trivial result. In
the Remark \ref{Nota3}, we \ discuss these matters under suitable conditions
on the space of neurons. The problem of determining rigorously largest
Lyapunov exponent of a random NN of type (\ref{Eq_Network_gneral}) is open. In
our view, the technique introduced in \cite{Sompolinsky et al} for determining
largest Lyapunov is a physical reasoning not a mathematical theorem.
\end{remark}

Following \cite{Grosvenor-Jefferson}, we introduce the following notation
\ for the two-point correlation of $\phi\left(  \boldsymbol{h}\right)  $:%
\begin{equation}
f_{\phi}\left(  c^{12},c_{0}\right)  :=C_{\phi\phi}^{12}\text{, \ }f_{\phi
}\left(  c,c_{0}\right)  :=C_{\phi\phi}^{11}=C_{\phi\phi}^{22}, \label{eq_F}%
\end{equation}
where $c_{0}$, $c:=c_{\tau}$ are the components of the covariance matrix of
the single copy in (\ref{covariance_single_copy}), and $c_{0}=c_{0}^{11}%
=c_{0}^{22}$, $c^{12}$ are defined as follows. Then, the random vector
$\left(  \boldsymbol{h}^{1}\left(  t\right)  \text{, }\boldsymbol{h}%
^{2}\left(  s\right)  \right)  $ is bivariate\ with normal distribution%
\begin{multline*}
\left(  \boldsymbol{h}^{1}\left(  t\right)  \text{, }\boldsymbol{h}^{2}\left(
s\right)  \right)  \sim\mathcal{N}\left(  0,\left[
\begin{array}
[c]{cc}%
G_{\boldsymbol{hh}}^{11}\left(  x,x,t,t\right)  & G_{\boldsymbol{hh}}%
^{12}\left(  x,x,t,s\right) \\
& \\
G_{\boldsymbol{hh}}^{12}\left(  x,x,t,s\right)  & G_{\boldsymbol{hh}}%
^{22}\left(  x,x,s,s\right)
\end{array}
\right]  \right) \\
\\
=:\mathcal{N}\left(  0,\left[
\begin{array}
[c]{cc}%
c_{0} & c^{12}\\
& \\
c^{12} & c_{0}%
\end{array}
\right]  \right)  .
\end{multline*}
We now Taylor expand\ (\ref{Eq_15A}) around $c$, on the left-side we
substitute $G_{\boldsymbol{hh}}^{12}=c^{12}$ and use (\ref{EQ_C_12}). On the
right-side we expand $C_{\phi\phi}^{12}$ as%
\begin{align*}
f_{\phi}\left(  c^{12},c_{0}\right)   &  =f_{\phi}\left(  c,c_{0}\right)
+\left(  c^{12}-c\right)  \frac{\partial}{\partial c^{12}}f_{\phi}\left(
c^{12},c_{0}\right)  \mid_{c^{12}=c}+\cdots\\
&  =f_{\phi}\left(  c,c_{0}\right)  +\eta k\left(  t,s\right)  f_{\phi
^{\prime}}\left(  c,c_{0}\right)  +O\left(  \eta^{2}\right)  ,
\end{align*}
where we used Price's theorem, see (\ref{Price_theorem}). By substituting
these expression in (\ref{Eq_15A}), and using \ref{Eq_3BB},%

\begin{equation}
\left(  \partial_{1}+\gamma\right)  \left(  \partial_{2}+\gamma\right)
k\left(  t,s\right)  =\eta k\left(  t,s\right)  f_{\phi^{\prime}}\left(
c,c_{0}\right)  \label{Eq_eta}%
\end{equation}
at order $O\left(  \eta\right)  $.

Now, we use the formula for the pseudo-distance between $\boldsymbol{h}%
^{1},\boldsymbol{h}^{2}$, with $t=s$, $c\left(  \tau\right)  =c\left(
x,x,\tau\right)  ,$ and $c\left(  0\right)  =c\left(  x,x,0\right)  $,%

\begin{align*}
d(\boldsymbol{h}^{1},\boldsymbol{h}^{2};x,x,t,t)  &  =d(\boldsymbol{h}%
^{1},\boldsymbol{h}^{2};x,t)\\
&  =c^{11}\left(  x,x,t,t\right)  +c^{22}\left(  x,x,t,t\right)
-2c^{12}\left(  x,x,t,t\right) \\
&  =2c\left(  0\right)  -2\left(  c\left(  0\right)  +\eta k\left(
t,t\right)  +O\left(  \eta^{2}\right)  \right) \\
&  =-2\eta k\left(  t,t\right)  +O\left(  \eta^{2}\right)
\end{align*}
Following, \cite{Helias et al}-\cite{Grosvenor-Jefferson}, we\ rewrite
(\ref{Eq_eta}) changing variables $\tau=t-s$, $u=t+s$, and using that%
\[
\partial_{t}=\partial_{\tau}+\partial_{u}\text{, \ }\partial_{s}%
=-\partial_{\tau}+\partial_{u},
\]
to get
\[
\left[  \left(  \partial_{u}+\gamma\right)  ^{2}-\partial_{\tau}^{2}\right]
k\left(  t,s\right)  =f_{\phi^{\prime}}\left(  c,c_{0}\right)  k\left(
\tau,u\right)  .
\]
To solve this last equation, we assume that
\begin{equation}
k\left(  \tau,u\right)  =e^{\lambda u}\psi\left(  \tau\right)  ,
\label{K-Function}%
\end{equation}
then%
\[
\left[  \left(  \lambda+\gamma\right)  ^{2}-\partial_{\tau}^{2}\right]
\psi\left(  \tau\right)  =f_{\phi^{\prime}}\left(  c,c_{0}\right)  \psi\left(
\tau\right)  ;
\]
we now use that $V^{\prime\prime}\left(  c_{\tau},c_{0}\right)  =-\gamma
^{2}+f_{\phi^{\prime}}\left(  c_{\tau},c_{0}\right)  $, see (\ref{Potential}),
to rewrite this equation as%
\begin{equation}
\mathbf{H}\psi\left(  \tau\right)  =E\psi\left(  \tau\right)  ,
\label{Schodinger_equation}%
\end{equation}
where $\boldsymbol{H}=-\partial_{\tau}^{2}-V^{\prime\prime}\left(  c_{\tau
},c_{0}\right)  $, and $E:=\gamma^{2}-\left(  \lambda+\gamma\right)  ^{2}%
$.\ We now consider (\ref{Schodinger_equation}) as \ time-independent
Schr\"{o}dinger equation with Hamiltonian $\boldsymbol{H}$. We assume that
existence \ of a countable set of eigenvalues
\[
E_{n}=\gamma^{2}-\left(  \lambda_{n}+\gamma\right)  ^{2}\text{, }%
n\in\mathbb{N}=\left\{  0,1,\ldots\right\}  ,
\]
which implies the existence a countable set of Lyapunov exponents $\lambda
_{n}=\lambda_{n}\left(  x,y\right)  $ which controls the growth of $k\left(
\tau,u\right)  $, see (\ref{K-Function}). In particular, the largest exponent
$\lambda_{0}$ corresponds to ground state energy $E_{0}$:%
\[
\lambda_{0}=-\gamma+\sqrt{\gamma^{2}-E_{0}}.
\]
In the unstable regime $\lambda_{0}>0$ the pseudo-distance $d(\boldsymbol{h}%
^{1},\boldsymbol{h}^{2};x,t)$ grows exponentially with time. Note that if
$\gamma<0$, then $\lambda_{0}>0$ for all possible values of the ground state
energy $E_{0}$. For this reason, we consider only the case $\gamma>0$. Our
goal is to determine \ the point at which ground state energy becomes
negative, $E_{0}<0$. Following, \cite[Chapter 10]{Helias et al},\ we construct
a normalizable solution\ with $E=0$, and derive the necessary conditions on
the existence\ of\ lower-energy states.

By differentiating equation (\ref{Key_Eq_1}) with respect to $\tau$ for
$\tau\neq0$, and\ denoting $\partial_{\tau}c_{\tau}=\overset{\cdot}{c_{\tau}}%
$, we have
\begin{equation}
\partial_{\tau}^{2}\overset{\cdot}{c_{\tau}}+\partial_{\tau}V^{\prime}\left(
c_{\tau},c_{0}\right)  =\left[  \partial_{\tau}^{2}+V^{\prime\prime}\left(
c_{\tau},c_{0}\right)  \right]  \overset{\cdot}{c_{\tau}}=0.
\label{Ptential_2}%
\end{equation}
now, comparing with (\ref{Schodinger_equation}), we conclude that $\psi\left(
\tau\right)  =\overset{\cdot}{c_{\tau}}$ is an eigenfunction for $E=0$, for
$\tau\neq0$. The discontinuity at $\tau=0$ comes from the delta function
$\delta\left(  \tau\right)  $ in (\ref{Key_Eq_1}). To fix this problem, we
define%
\[
y\left(  \tau\right)  =\left\{
\begin{array}
[c]{ll}%
\overset{\cdot}{c_{\tau}} & \text{for }\tau>0\\
& \\
-\overset{\cdot}{c_{\tau}} & \text{for }\tau<0,
\end{array}
\right.
\]
and impose the condition that $y\left(  \tau\right)  $ be differentiable for
any $\tau\in\mathbb{R}$. Then, $\overset{\cdot}{y}\left(  \tau\right)  $
should be continuous at $\tau=0$, which implies that%
\[
\lim_{\tau\rightarrow0^{+}}\overset{\cdot}{y}\left(  \tau\right)
=\overset{\cdot}{y}\left(  0^{+}\right)  =\overset{\cdot}{y}\left(
0^{-}\right)  =\lim_{\tau\rightarrow0^{-}}\overset{\cdot}{y}\left(
\tau\right)  ;
\]
then, using (\ref{Key_Eq_1}), (\ref{Potential}), (\ref{Price_theorem}), and
\ref{eq_F},%
\begin{align}
\overset{\cdot}{y}\left(  0^{+}\right)  -\overset{\cdot}{y}\left(
0^{-}\right)   &  =\overset{\cdot\cdot}{c_{\tau}}\left(  0^{+}\right)
+\overset{\cdot\cdot}{c_{\tau}}\left(  0^{-}\right)  =-2V^{\prime}\left(
c_{0},c_{0}\right) \label{extremum}\\
&  =-2\left(  -\gamma^{2}c_{0}+\text{\ }f_{\phi}\left(  c_{0},c_{0}\right)
\right)  =0,\nonumber
\end{align}
i.e.,%
\[
V^{\prime}\left(  c_{0},c_{0}\right)  =-\gamma^{2}c_{0}+\text{\ }f_{\phi
}\left(  c_{0},c_{0}\right)  =0.
\]
On the other hand, $y\left(  \tau\right)  $ has zero total energy by
construction, then, it is necessary to impose \ that the potential be less
than or equal zero\ at the extremum given by (\ref{extremum}), see
\cite{Helias et al}, and \cite{Schuecker et al}, for a detailed discussion.
Now, by (\ref{Ptential_2}), the potential for $y\left(  \tau\right)  $ is
$-V^{\prime\prime}\left(  c_{\tau},c_{0}\right)  $, then%
\begin{equation}
V^{\prime\prime}\left(  c_{0},c_{0}\right)  \geq0\text{, i.e. }-\gamma
^{2}+\text{\ }f_{\phi^{\prime}}\left(  c_{0},c_{0}\right)  \geq0.
\label{inequality}%
\end{equation}
When this equality is attained in (\ref{inequality}), $y\left(  \tau\right)  $
must be the ground state with $E_{0}=0$; when the strict inequality occurs,
the energy of the ground state is $E_{0}<0$.

\section{\label{SECT_11}Self-averaging}

The dynamics of the mean-field of the double-copy system is controlled by the
equations (\ref{Eq_15A}). The solutions (the covariance functions) are
distributions depending on the spatial variables $x,y\in\Omega$. Then, several
spatial averages of the covariance functions satisfy the above-mentioned
system of differential equations.

Let us multiply both sides of (\ref{Eq_15A}) by a test function $\theta\left(
x,y\right)  $, and then integrate with respect $d\mu\left(  x\right)
d\mu\left(  y\right)  $, formally, we get%
\begin{gather*}
\left(  \partial_{1}+\gamma\right)  \left(  \partial_{2}+\gamma\right)
\left\{  \text{ }%
{\displaystyle\int\limits_{\Omega^{2}}}
G_{\boldsymbol{hh}}^{\alpha\beta}\left(  x,y,t_{1},t_{2}\right)  \theta\left(
x,y\right)  d\mu\left(  x\right)  d\mu\left(  y\right)  \right\}  =\\%
{\displaystyle\int\limits_{\Omega^{2}}}
C_{\phi\phi}^{\alpha\beta}\left(  x,y,t_{1},t_{2}\right)  \theta\left(
x,y\right)  d\mu\left(  x\right)  d\mu\left(  y\right)  +\\
\sigma^{2}\left\{  \text{ }%
{\displaystyle\int\limits_{\Omega^{2}}}
\delta\left(  x-y\right)  \theta\left(  x,y\right)  d\mu\left(  x\right)
d\mu\left(  y\right)  \right\}  \delta\left(  t_{1}-t_{2}\right)  ,
\end{gather*}
i.e.,%
\begin{equation}
\left(  \partial_{1}+\gamma\right)  \left(  \partial_{2}+\gamma\right)
\overline{G}_{\boldsymbol{hh}}^{\alpha\beta}\left(  t_{1},t_{2}\right)
=\overline{C}_{\phi\phi}^{\alpha\beta}\left(  t_{1},t_{2}\right)
+\sigma_{\theta}^{2}\delta\left(  t_{1}-t_{2}\right)  ,
\label{Average-dynamics}%
\end{equation}
as distributions, where%
\begin{equation}
\overline{G}_{\boldsymbol{hh}}^{\alpha\beta}\left(  t_{1},t_{2}\right)  =%
{\displaystyle\int\limits_{\Omega^{2}}}
G_{\boldsymbol{hh}}^{\alpha\beta}\left(  x,y,t_{1},t_{2}\right)  \theta\left(
x,y\right)  d\mu\left(  x\right)  d\mu\left(  y\right)  , \label{average_0}%
\end{equation}%
\begin{equation}
\overline{C}_{\phi\phi}^{\alpha\beta}\left(  t_{1},t_{2}\right)  =%
{\displaystyle\int\limits_{\Omega^{2}}}
C_{\phi\phi}^{\alpha\beta}\left(  x,y,t_{1},t_{2}\right)  \theta\left(
x,y\right)  d\mu\left(  x\right)  d\mu\left(  y\right)  \text{, }
\label{average_1A}%
\end{equation}
and
\begin{equation}
\sigma_{\theta}^{2}=\sigma^{2}%
{\displaystyle\int\limits_{\Omega}}
\theta\left(  y,y\right)  d\mu\left(  y\right)  . \label{average_2}%
\end{equation}
By interpreting $G_{\boldsymbol{hh}}^{\alpha\beta}\left(  x,y,t_{1}%
,t_{2}\right)  $ as the cross-correlations of the double-copy system, and
\ (\ref{average_0})-(\ref{average_2}) as spatial averages, we can expect that
$G_{\boldsymbol{hh}}^{\alpha\beta}\left(  x,y,t_{1},t_{2}\right)  $ and
$\overline{G}_{\boldsymbol{hh}}^{\alpha\beta}\left(  t_{1},t_{2}\right)  $
provide\ a similar\ descriptions of the double-copy system. A precise
mathematical formulation of this result is given in the Appendix G, Theorems
\ref{Theorem4} and \ref{Theorem5}.

\section{\label{SECT_12}Continuous phase transitions in a toy model}

In this section, we introduce a family of NNs corresponding to certain
partition functions $\overline{{\LARGE Z}}_{M}$, for which the function
$C_{\phi\left(  \boldsymbol{h}\right)  \phi\left(  \boldsymbol{h}\right)  }$
can be studied explicitly. We take $\Omega=\mathbb{Z}_{p}$, and denote by $dx$
the Haar measure of $\mathbb{Q}_{p}$ normalized by the condition
$\int_{\mathbb{Z}_{p}}dx=1$. We show the existence of a continuous phase
transition for this type of NNs.

\subsection{The generating functional}

For our convenience, we recall some results established in the previous
sections. We consider the functional%
\[
\overline{{\LARGE Z}}_{M}=%
{\displaystyle\iint\limits_{\mathcal{H}\times\mathcal{H}}}
{\large 1}_{\mathcal{P}_{M}}\left(  \boldsymbol{h},\widetilde{\boldsymbol{h}%
}\right)  \exp\left(  S\left[  \widetilde{\boldsymbol{h}},\boldsymbol{h}%
\right]  \right)  d\boldsymbol{P}\left(  \boldsymbol{h}\right)  d\widetilde
{\boldsymbol{P}}\left(  \widetilde{\boldsymbol{h}}\right)  ,
\]
where $S\left[  \widetilde{\boldsymbol{h}},\boldsymbol{h}\right]
=S_{0}\left[  \widetilde{\boldsymbol{h}},\boldsymbol{h}\right]  +\left\langle
\widetilde{\boldsymbol{h}},C_{\phi\left(  \boldsymbol{h}\right)  \phi\left(
\boldsymbol{h}\right)  }\widetilde{\boldsymbol{h}}\right\rangle _{\left(
L^{2}(\mathbb{Z}_{p}\times\mathbb{R})\right)  ^{2}}$, $\ $with $S_{0}\left[
\widetilde{\boldsymbol{h}},\boldsymbol{h}\right]  =-\left\langle
\widetilde{\boldsymbol{h}},\left(  \partial_{t}+\gamma\right)  \boldsymbol{h}%
\right\rangle +\frac{1}{2}\sigma^{2}\left\langle \widetilde{\boldsymbol{h}%
},\widetilde{\boldsymbol{h}}\right\rangle $, and%
\begin{gather*}
\left\langle \widetilde{\boldsymbol{h}},C_{\phi\left(  \boldsymbol{h}\right)
\phi\left(  \boldsymbol{h}\right)  }\widetilde{\boldsymbol{h}}\right\rangle
_{\left(  L^{2}(\mathbb{Z}_{p}\times\mathbb{R})\right)  ^{2}}=\\%
{\displaystyle\int\limits_{\mathbb{Z}_{p}^{2}}}
\text{ }%
{\displaystyle\int\limits_{\mathbb{R}^{2}}}
\widetilde{\boldsymbol{h}}\left(  x,t_{1}\right)  C_{\phi\left(
\boldsymbol{h}\right)  \phi\left(  \boldsymbol{h}\right)  }\left(
x,y,t_{1},t_{2}\right)  \widetilde{\boldsymbol{h}}\left(  y,t_{2}\right)
dxdydt_{1}dt_{2}.
\end{gather*}
The partition function is essentially determined by functional $C_{\phi\left(
\boldsymbol{h}\right)  \phi\left(  \boldsymbol{h}\right)  }$, which in turn is
determined by an integral operator. In this section, we introduce a family of
integral operators for which, we can describe the functionals $C_{\phi\left(
\boldsymbol{h}\right)  \phi\left(  \boldsymbol{h}\right)  }$ for
$\boldsymbol{h}$ in a dense space of $\mathcal{H}$.

We also recall that
\begin{gather*}
S\left[  \boldsymbol{h},\widetilde{\boldsymbol{h}};\Xi\right]  =\frac{1}{2}%
{\displaystyle\int\limits_{\mathbb{Z}_{p}^{2}}}
\text{ }%
{\displaystyle\int\limits_{\mathbb{R}^{2}}}
\text{\ }\boldsymbol{z}\left(  x,t_{1}\right)  ^{T}\Xi\left(  x,y,t_{1}%
,t_{2}\right)  \boldsymbol{z}\left(  y,t_{2}\right)  dxdydt_{1}dt_{2}=\\
\frac{1}{2}%
{\displaystyle\int\limits_{\mathbb{Z}_{p}^{2}}}
\text{ }%
{\displaystyle\int\limits_{\mathbb{R}^{2}}}
\text{\ }\left[  \boldsymbol{h}\left(  x,t_{1}\right)  ,\widetilde
{\boldsymbol{h}}\left(  x,t_{1}\right)  \right]  \left[
\begin{array}
[c]{cc}%
\Xi_{\boldsymbol{hh}} & \Xi_{\boldsymbol{h}\widetilde{\boldsymbol{h}}}\\
& \\
\Xi_{\widetilde{\boldsymbol{h}}\boldsymbol{h}} & \Xi_{\widetilde
{\boldsymbol{h}}\widetilde{\boldsymbol{h}}}%
\end{array}
\right]  \left(  x,y,t_{1},t_{2}\right)  \left[
\begin{array}
[c]{c}%
\boldsymbol{h}\left(  y,t_{2}\right) \\
\widetilde{\boldsymbol{h}}\left(  y,t_{2}\right)
\end{array}
\right]  dxdydt_{1}dt_{2},
\end{gather*}
\ and that the functional inverse of $\Xi\left(  x,y,t_{1},t_{2}\right)  $ is
given by%
\[%
{\displaystyle\int\limits_{\mathbb{Z}_{p}}}
{\displaystyle\int\limits_{\mathbb{R}}}
\Xi\left(  x,z,t_{1},s\right)  G\left(  z,y,s,t_{2}\right)  dsdz=\delta\left(
x-y\right)  \delta\left(  t_{1}-t_{2}\right)  \left[
\begin{array}
[c]{cc}%
1 & 0\\
0 & 1
\end{array}
\right]  ,
\]
where%
\[
G\left(  x,y,t_{1},t_{2}\right)  =\left[
\begin{array}
[c]{cc}%
G_{\boldsymbol{hh}} & G_{\boldsymbol{h}\widetilde{\boldsymbol{h}}}\\
& \\
G_{\widetilde{\boldsymbol{h}}\boldsymbol{h}} & G_{\widetilde{\boldsymbol{h}%
}\widetilde{\boldsymbol{h}}}%
\end{array}
\right]  ,\text{ with }G_{\widetilde{\boldsymbol{h}}\widetilde{\boldsymbol{h}%
}}=0.
\]
Furthermore, the entries of the matrix propagator $G$ are determined by the
following system of equations:%
\begin{equation}
\left\{
\begin{array}
[c]{l}%
\left(  \partial_{t_{1}}-\gamma\right)  G_{\widetilde{\boldsymbol{h}%
}\boldsymbol{h}}\left(  x,y,t_{1},t_{2}\right)  =\delta\left(  x-y\right)
\delta\left(  t_{1}-t_{2}\right)  ;\\
\\%
\begin{array}
[c]{c}%
-\left(  \partial_{t_{1}}+\gamma\right)  G_{\boldsymbol{hh}}\left(
x,y,t_{1},t_{2}\right)  +%
{\displaystyle\int\limits_{\mathbb{Z}_{p}}}
{\displaystyle\int\limits_{\mathbb{R}}}
C_{\phi\left(  \boldsymbol{h}\right)  \phi\left(  \boldsymbol{h}\right)
}\left(  x,z,t_{1},s\right)  G_{\widetilde{\boldsymbol{h}}\boldsymbol{h}%
}\left(  z,y,s,t_{2}\right)  dsdz+\\
\multicolumn{1}{l}{}\\
\multicolumn{1}{l}{\sigma^{2}G_{\widetilde{\boldsymbol{h}}\boldsymbol{h}%
}\left(  x,y,t_{1},t_{2}\right)  =0;}%
\end{array}
\\
\\
-\left(  \partial_{t_{1}}+\gamma\right)  G_{\boldsymbol{h}\widetilde
{\boldsymbol{h}}}\left(  x,y,t_{1},t_{2}\right)  =\delta\left(  x-y\right)
\delta\left(  t_{1}-t_{2}\right)  .
\end{array}
\right.  \label{Equation_G}%
\end{equation}

\subsection{A family of trace class operators}

For $\xi=\left(  \xi_{1},\xi_{2}\right)  \in\mathbb{Q}_{p}^{2}$, we set
$\left\Vert \xi\right\Vert _{p}=\max\left\{  \left\vert \xi_{1}\right\vert
_{p},\left\vert \xi_{2}\right\vert _{p}\right\}  $. We use the results from
(\cite[Section 4.1]{Zuniga-ATMP}), so we fix a function $\widehat
{K}:\mathbb{Q}_{p}^{2}\rightarrow\mathbb{R}$ \ satisfying the following:

\begin{enumerate}
\item[(H1)] $\widehat{K}\in$ $L^{1}\left(  \mathbb{Q}_{p}^{2}\right)  \cap
C(\mathbb{Q}_{p}^{2},\mathbb{R})$,

\item[(H2)] $\widehat{K}\left(  \xi\right)  =\widehat{K}(\max\left\{
1,\left\Vert \xi\right\Vert _{p}\right\}  )$, for $\xi\in\mathbb{Q}_{p}^{2}$,

\item[(H3)] $\widehat{K}\left(  \xi\right)  >0$, for $\xi\in\mathbb{Q}_{p}%
^{2}$.
\end{enumerate}

Here $\widehat{K}$ denotes the Fourier transform of an integral function $K$.
For further details about the Fourier transform the reader may consul Appendix
H. Now, since $\widehat{K}(\max\left\{  1,\left\Vert \xi\right\Vert
_{p}\right\}  )$ is integrable, $K$ is a continuous real-valued function.

We now define the operator%
\[%
\begin{array}
[c]{ccc}%
L^{2}\left(  \mathbb{Z}_{p}^{2}\right)  & \rightarrow & L^{2}\left(
\mathbb{Z}_{p}^{2}\right) \\
f & \rightarrow & \square_{K}f,
\end{array}
\]
where $\left(  \square_{K}f\right)  \left(  x\right)  :=K(x)\ast f(x)$. Since
$\left(  \mathbb{Z}_{p}^{2},+\right)  $ is an Abelian group, the function
$K(x)\ast f(x)$\ is supported in the two-dimensional unit ball.

Under the hypothesis
\begin{equation}
\widehat{K}(1)+%
{\displaystyle\sum\limits_{\lambda=1}^{\infty}}
\left(  2\lambda p^{\lambda+1}\right)  ^{N}\widehat{K}(p^{\lambda})<\infty,
\label{hypo1}%
\end{equation}
operator $\square_{K}$ is a symmetric, positive, trace class operator on
$L^{2}\left(  \mathbb{Z}_{p}^{2}\right)  $. cf. \cite[Lemmas 1, 2, Theorem
1]{Zuniga-ATMP}.

We set $\Omega\left(  z\right)  $ for the characteristic function of the
interval $\left[  0,1\right]  $. Then, $\Omega\left(  p^{l}\left\vert
x-I\right\vert _{p}\right)  $ is the characteristic function of the ball
\[
B_{-l}(I)=\left\{  x\in\mathbb{Q}_{p};\left\vert x-I\right\vert _{p}\leq
p^{l}\right\}  .
\]

The next step is to pick a convenient family of kernels $K$. For $\rho>0$, for
$\xi=\left(  \xi_{1},\xi_{2}\right)  \in\mathbb{Q}_{p}^{2}$, we set
\[
\widehat{K}_{\rho,2}\left(  \xi\right)  :=\max\left(  1,\left\Vert
\xi\right\Vert _{p}\right)  ^{-\rho}.
\]
Then $K_{\rho,2}\in L^{1}\left(  \mathbb{Q}_{p}^{2}\right)  $, and
\begin{equation}
K_{\rho,2}\left(  x\right)  =\left\{
\begin{array}
[c]{lll}%
\frac{1}{\Gamma_{2}\left(  \rho\right)  }\left(  \left\Vert x\right\Vert
_{p}^{\rho-2}-p^{\rho-2}\right)  \Omega\left(  \left\Vert x\right\Vert
_{p}\right)  & \text{if} & \rho\neq2\\
&  & \\
\left(  1-p^{-2}\right)  \log_{p}\left(  \frac{p}{\left\Vert x\right\Vert
_{p}}\right)  \Omega\left(  \left\Vert x\right\Vert _{p}\right)  & \text{if} &
\rho=2,
\end{array}
\right.  \label{Kernel}%
\end{equation}
where $\Gamma_{2}\left(  \rho\right)  =\frac{1-p^{\rho-2}}{1-p^{-\rho}}$, for
$\rho\neq0$, cf. \cite[ Chapter III, Lema 5.2]{Taibleson}. Since $\widehat
{K}_{\rho}\left(  \xi\right)  $ verifies (\ref{hypo1}),
\[%
\begin{array}
[c]{ccc}%
L^{2}\left(  \mathbb{Z}_{p}^{2}\right)  & \rightarrow & L^{2}\left(
\mathbb{Z}_{p}^{2}\right) \\
&  & \\
f\left(  x_{1},x_{2}\right)  & \rightarrow & K_{\rho,2}\left(  x_{1}%
,x_{2}\right)  \ast f\left(  x_{1},x_{2}\right)  .
\end{array}
\]
is a symmetric, positive, trace class operator on $L^{2}\left(  \mathbb{Z}%
_{p}^{2}\right)  $.

\subsection{Computation of $C_{\phi\left(  \boldsymbol{h}\right)  \phi\left(
\boldsymbol{h}\right)  }^{\rho}$}

The $\mathbb{R}$-vector space $\mathcal{D}_{l}\left(  \mathbb{Z}_{p}\right)
$, $l\in\mathbb{N}$, consists of the test functions supported in the ball
$\mathbb{Z}_{p}$\ of the form%
\[
\varphi\left(  x\right)  =%
{\displaystyle\sum\limits_{I\in G_{l}}}
\varphi_{I}\Omega\left(  p^{l}\left\vert x-I\right\vert _{p}\right)  \text{,
}\varphi_{I}\in\mathbb{R}\text{,}%
\]
where $G_{l}=\mathbb{Z}_{p}/p^{l}\mathbb{Z}_{p}\simeq\mathbb{Z}/p^{l}%
\mathbb{Z}$ is the set of $p$-adic numbers of the form $I=\sum_{k=0}%
^{l-1}I_{k}p^{k}$ , with $I_{k}\in\left\{  0,\ldots,p-1\right\}  $.

The space of test functions $\mathcal{D}\left(  \mathbb{Z}_{p}\right)
=\cup_{l\in\mathbb{N}}\mathcal{D}_{l}\left(  \mathbb{Z}_{p}\right)  $,
$\mathcal{D}_{l}\left(  \mathbb{Z}_{p}\right)  \subset\mathcal{D}_{l+1}\left(
\mathbb{Z}_{p}\right)  $, is dense in $L^{1}(\mathbb{Z}_{p})\supset
L^{2}(\mathbb{Z}_{p})$.

We set%
\begin{multline*}
C_{\phi\left(  \boldsymbol{h}\right)  \phi\left(  \boldsymbol{h}\right)
}^{\rho}\left(  u_{1},y_{1},t_{1},t_{2}\right)  :=\\%
{\displaystyle\int\limits_{\mathbb{Z}_{p}^{2}}}
K_{\rho,2}\left(  y_{1}-u_{1},y_{2}-u_{2}\right)  \phi\left(  \boldsymbol{h}%
\left(  y_{2},t_{2}\right)  \right)  \phi\left(  \boldsymbol{h}\left(
u_{2},t_{1}\right)  \right)  dy_{2}du_{2}.
\end{multline*}

The functions of type%
\[
\boldsymbol{h}\left(  x,t\right)  =%
{\displaystyle\sum\limits_{I\in G_{l}}}
\boldsymbol{h}_{I}\left(  t\right)  \Omega\left(  p^{l}\left\vert
x-I\right\vert _{p}\right)  \in C^{\infty}(\mathbb{R})%
{\textstyle\bigotimes\nolimits_{\text{alg}}}
\mathcal{D}_{l}\left(  \mathbb{Z}_{p}\right)  ,
\]
where each $\boldsymbol{h}_{I}\left(  t\right)  \in C^{\infty}(\mathbb{R})$
are dense in $\mathcal{H}$. We now compute the functional $C_{\phi\left(
\boldsymbol{h}\right)  \phi\left(  \boldsymbol{h}\right)  }^{\rho}$ on
$C^{\infty}(\mathbb{R})%
{\textstyle\bigotimes\nolimits_{\text{alg}}}
\mathcal{D}_{l}\left(  \mathbb{Z}_{p}\right)  $, for $l$ fixed.

\paragraph{Computation of $C_{\phi\left(  \boldsymbol{h}\right)  \phi\left(
\boldsymbol{h}\right)  }^{\rho}$ for $\rho\neq2$}

We take $\rho\neq2$, and%
\[
\boldsymbol{h}\left(  y_{2},t_{2}\right)  =%
{\displaystyle\sum\limits_{I\in G_{l}}}
\boldsymbol{h}_{I}\left(  t_{2}\right)  \Omega\left(  p^{l}\left\vert
y_{2}-I\right\vert _{p}\right)  \text{, }\boldsymbol{h}\left(  u_{2}%
,t_{1}\right)  =%
{\displaystyle\sum\limits_{I\in G_{l}}}
\boldsymbol{h}_{I}\left(  t_{1}\right)  \Omega\left(  p^{l}\left\vert
u_{2}-I\right\vert _{p}\right)  ,
\]

then%
\begin{gather*}
\phi\left(  \boldsymbol{h}\left(  y_{2},t_{2}\right)  \right)  \phi\left(
\boldsymbol{h}\left(  u_{2},t_{1}\right)  \right)  =\\
\left\{
{\displaystyle\sum\limits_{I\in G_{l}}}
\phi\left(  \boldsymbol{h}_{I}\left(  t_{2}\right)  \right)  \Omega\left(
p^{l}\left\vert y_{2}-I\right\vert _{p}\right)  \right\}  \left\{
{\displaystyle\sum\limits_{J\in G_{l}}}
\phi\left(  \boldsymbol{h}_{J}\left(  t_{1}\right)  \right)  \Omega\left(
p^{l}\left\vert u_{2}-J\right\vert _{p}\right)  \right\} \\
=%
{\displaystyle\sum\limits_{I\in G_{l}}}
{\displaystyle\sum\limits_{J\in G_{l}}}
\phi\left(  \boldsymbol{h}_{I}\left(  t_{2}\right)  \right)  \phi\left(
\boldsymbol{h}_{J}\left(  t_{1}\right)  \right)  \Omega\left(  p^{l}\left\vert
y_{2}-I\right\vert _{p}\right)  \Omega\left(  p^{l}\left\vert u_{2}%
-J\right\vert _{p}\right)  .
\end{gather*}
Now by using that $\mathbb{Z}_{p}=\left\{  x\in\mathbb{Q}_{p};\left\vert
x\right\vert _{p}\leq1\right\}  $ is an additive group,
\[
\Omega\left(  \max\left\{  \left\vert y_{1}-u_{1}\right\vert _{p},\left\vert
y_{2}-u_{2}\right\vert _{p}\right\}  \right)  =1,
\]
for $y_{1},u_{1},y_{2},u_{2}\in\mathbb{Z}_{p}$, we have
\begin{gather}
C_{\phi\left(  \boldsymbol{h}\right)  \phi\left(  \boldsymbol{h}\right)
}^{\rho}\left(  u_{1},y_{1},t_{1},t_{2}\right)  =\frac{1}{\Gamma_{2}\left(
\rho\right)  }%
{\displaystyle\sum\limits_{I\in G_{l}}}
{\displaystyle\sum\limits_{J\in G_{l}}}
\phi\left(  \boldsymbol{h}_{I}\left(  t_{2}\right)  \right)  \phi\left(
\boldsymbol{h}_{J}\left(  t_{1}\right)  \right)  \times\label{Fomula_C_phi}\\%
{\displaystyle\int\limits_{J+p^{l}\mathbb{Z}_{p}}}
\text{ }%
{\displaystyle\int\limits_{I+p^{l}\mathbb{Z}_{p}}}
\text{ }\left(  \left[  \max\left\{  \left\vert y_{1}-u_{1}\right\vert
_{p},\left\vert y_{2}-u_{2}\right\vert _{p}\right\}  \right]  ^{\rho
-2}-p^{\rho-2}\right)  dy_{2}du_{2}.\nonumber
\end{gather}
We now compute this last integral.

\textbf{Claim 1}

For $\rho\in\left(  1,\infty\right)  \smallsetminus\left\{  2\right\}  $,
\begin{align}
{\LARGE L}(y_{1}-u_{1};J,I)  &  :=%
{\displaystyle\int\limits_{J+p^{l}\mathbb{Z}_{p}}}
{\displaystyle\int\limits_{I+p^{l}\mathbb{Z}_{p}}}
\left[  \max\left\{  \left\vert y_{1}-u_{1}\right\vert _{p},\left\vert
y_{2}-u_{2}\right\vert _{p}\right\}  \right]  ^{\rho-2}dy_{2}du_{2}%
\label{Formula_L_0}\\
&  =\left\{
\begin{array}
[c]{lll}%
p^{-l\rho}%
{\displaystyle\int\limits_{\mathbb{Z}_{p}}}
\left[  \max\left\{  p^{l}\left\vert y_{1}-u_{1}\right\vert _{p},\left\vert
w\right\vert _{p}\right\}  \right]  ^{\rho-2}dw & \text{if } & I=J\\
&  & \\
p^{-2l}\left[  \max\left\{  \left\vert y_{1}-u_{1}\right\vert _{p},\left\vert
I-J\right\vert _{p}\right\}  \right]  ^{\rho-2} & \text{if} & I\neq J.
\end{array}
\right. \nonumber
\end{align}

By changing variables in ${\LARGE L}(y_{1}-u_{1};J,I)$ as $y_{2}%
=I+p^{l}\widetilde{y}_{2}$, $u_{2}=J+p^{l}\widetilde{u}_{2}$, $dy_{2}%
du_{2}=p^{-2l}d\widetilde{y}_{2}d\widetilde{u}_{2}$,%
\[
{\LARGE L}(y_{1}-u_{1};J,I)=p^{-2l}%
{\displaystyle\int\limits_{\mathbb{Z}_{p}}}
{\displaystyle\int\limits_{\mathbb{Z}_{p}}}
\left[  \max\left\{  \left\vert y_{1}-u_{1}\right\vert _{p},\left\vert
I-J+p^{l}\left(  \widetilde{y}_{2}-\widetilde{u}_{2}\right)  \right\vert
_{p}\right\}  \right]  ^{\rho-2}d\widetilde{y}_{2}d\widetilde{u}_{2}.
\]
If $I\neq J$, the ultrametric property of $\left\vert \cdot\right\vert _{p}$
implies that
\[
\left\vert I-J+p^{l}\left(  \widetilde{y}_{2}-\widetilde{u}_{2}\right)
\right\vert _{p}=\left\vert I-J\right\vert _{p},
\]
then%
\begin{align*}
{\LARGE L}(y_{1}-u_{1};J,I)  &  =p^{-2l}%
{\displaystyle\int\limits_{\mathbb{Z}_{p}}}
{\displaystyle\int\limits_{\mathbb{Z}_{p}}}
\left[  \max\left\{  \left\vert y_{1}-u_{1}\right\vert _{p},\left\vert
I-J\right\vert _{p}\right\}  \right]  ^{\rho-2}d\widetilde{y}_{2}%
d\widetilde{u}_{2}\\
&  =p^{-2l}\left[  \max\left\{  \left\vert y_{1}-u_{1}\right\vert
_{p},\left\vert I-J\right\vert _{p}\right\}  \right]  ^{\rho-2}.
\end{align*}
If $I=J$,%
\[
{\LARGE L}(y_{1}-u_{1};J,I)=p^{-l\rho}%
{\displaystyle\int\limits_{\mathbb{Z}_{p}}}
{\displaystyle\int\limits_{\mathbb{Z}_{p}}}
\left[  \max\left\{  p^{l}\left\vert y_{1}-u_{1}\right\vert _{p},\left\vert
\left(  \widetilde{y}_{2}-\widetilde{u}_{2}\right)  \right\vert _{p}\right\}
\right]  ^{\rho-2}d\widetilde{y}_{2}d\widetilde{u}_{2}.
\]
By changing variables as $w=\widetilde{y}_{2}-\widetilde{u}_{2}$,
$\widetilde{w}=\widetilde{y}_{2}$, $d\widetilde{y}_{2}d\widetilde{u}%
_{2}=dwd\widetilde{w}$,%
\[
{\LARGE L}(y_{1}-u_{1};J,I)=p^{-l\rho}%
{\displaystyle\int\limits_{\mathbb{Z}_{p}}}
{\displaystyle\int\limits_{\mathbb{Z}_{p}}}
\left[  \max\left\{  p^{l}\left\vert y_{1}-u_{1}\right\vert _{p},\left\vert
w\right\vert _{p}\right\}  \right]  ^{\rho-2}dw.
\]
In the case $y_{1}=u_{1}$,%
\[
{\LARGE L}(y_{1}-u_{1};J,I)=p^{-l\rho}%
{\displaystyle\int\limits_{\mathbb{Z}_{p}}}
\left\vert w\right\vert _{p}^{\rho-2}dw=\frac{p^{-l\rho}\left(  1-p^{-1}%
\right)  }{1-p^{1-\rho}}\text{, for }\rho>1\text{.}%
\]
In this way, we obtain the condition $\rho\in\left(  1,\infty\right)
\smallsetminus\left\{  2\right\}  $.

\textbf{Claim 2}

If $\rho\in\left(  1,\infty\right)  \smallsetminus\left\{  2\right\}  $, then
\begin{gather*}
C_{\phi\left(  \boldsymbol{h}\right)  \phi\left(  \boldsymbol{h}\right)
}^{\rho}\left(  u_{1},y_{1},t_{1},t_{2}\right)  =\frac{p^{-l\rho}}{\Gamma
_{2}\left(  \rho\right)  }%
{\displaystyle\sum\limits_{I\in G_{l}}}
\phi\left(  \boldsymbol{h}_{I}\left(  t_{2}\right)  \right)  \phi\left(
\boldsymbol{h}_{I}\left(  t_{1}\right)  \right)  \times\\
\left\{
{\displaystyle\int\limits_{\mathbb{Z}_{p}}}
\left[  \max\left\{  p^{l}\left\vert y_{1}-u_{1}\right\vert _{p},\left\vert
w\right\vert _{p}\right\}  \right]  ^{\rho-2}dw-p^{-2l+\rho-2}\right\}  +\\
\frac{p^{-2l}}{\Gamma_{2}\left(  \rho\right)  }%
{\displaystyle\sum\limits_{\substack{I,J\in G_{l}\\I\neq J}}}
\phi\left(  \boldsymbol{h}_{I}\left(  t_{2}\right)  \right)  \phi\left(
\boldsymbol{h}_{J}\left(  t_{1}\right)  \right)  \left[  \max\left\{
\left\vert y_{1}-u_{1}\right\vert _{p},\left\vert I-J\right\vert _{p}\right\}
\right]  ^{\rho-2}.
\end{gather*}
This formula follows from formula (\ref{Fomula_C_phi}) by Claim 1.

\begin{remark}
\label{Nota1AAA}For $A>0$, by L'Hospital's rule,%
\begin{gather*}
\lim_{\rho\rightarrow2}\frac{p^{-2l}A^{\rho-2}-p^{-2l+\rho-2}}{\Gamma
_{2}\left(  \rho\right)  }=\lim_{\rho\rightarrow2}\frac{\left(  p^{-2l}%
A^{\rho-2}-p^{-2l+\rho-2}\right)  \left(  1-p^{-\rho}\right)  }{1-p^{\rho-2}%
}\\
=\left(  1-p^{-2}\right)  \lim_{\rho\rightarrow2}\frac{\left(  p^{-2l}%
A^{\rho-2}-p^{-2l+\rho-2}\right)  }{1-p^{\rho-2}}=\frac{p^{-2l}\left(
1-p^{-2}\right)  \ln\frac{p}{A}}{\ln p}.
\end{gather*}

\end{remark}

By applying this formula, we obtain%
\begin{gather*}
\lim_{\rho\rightarrow2}\frac{p^{-2l}}{\Gamma_{2}\left(  \rho\right)  }\left[
\max\left\{  \left\vert y_{1}-u_{1}\right\vert _{p},\left\vert I-J\right\vert
_{p}\right\}  \right]  ^{\rho-2}=\\
\frac{p^{-2l}\left(  1-p^{-2}\right)  }{\ln p}\ln\left(  \frac{p}{\max\left\{
\left\vert y_{1}-u_{1}\right\vert _{p},\left\vert I-J\right\vert _{p}\right\}
}\right)  .
\end{gather*}

\begin{remark}
\label{Nota2AAA}By L'Hospital's rule,%
\begin{gather*}
\lim_{\rho\rightarrow2}\frac{p^{-l\rho}%
{\displaystyle\int\limits_{\mathbb{Z}_{p}}}
\left[  \max\left\{  p^{l}\left\vert y_{1}-u_{1}\right\vert _{p},\left\vert
w\right\vert _{p}\right\}  \right]  ^{\rho-2}dw-p^{-2l+\rho-2}}{\Gamma
_{2}\left(  \rho\right)  }\\
=\lim_{\rho\rightarrow2}\frac{p^{-l\rho}\left(  1-p^{-\rho}\right)
{\displaystyle\int\limits_{\mathbb{Z}_{p}}}
\left[  \max\left\{  p^{l}\left\vert y_{1}-u_{1}\right\vert _{p},\left\vert
w\right\vert _{p}\right\}  \right]  ^{\rho-2}dw-p^{-2l+\rho-2}}{1-p^{\rho-2}%
}\\
=\left(  1-p^{-2}\right)  \lim_{\rho\rightarrow2}\frac{p^{-l\rho}%
{\displaystyle\int\limits_{\mathbb{Z}_{p}}}
\left[  \max\left\{  p^{l}\left\vert y_{1}-u_{1}\right\vert _{p},\left\vert
w\right\vert _{p}\right\}  \right]  ^{\rho-2}dw-p^{-2l+\rho-2}}{1-p^{\rho-2}%
}\\
=\left(  1-p^{-2}\right)  \frac{-lp^{-2l}\ln p+%
{\displaystyle\int\limits_{\mathbb{Z}_{p}}}
\ln\left[  \max\left\{  p^{l}\left\vert y_{1}-u_{1}\right\vert _{p},\left\vert
w\right\vert _{p}\right\}  \right]  dw-p^{-2l}\ln p}{-\ln p}\\
=\left(  1-p^{-2}\right)  \left(  l+1\right)  p^{-2l}+\frac{\left(
1-p^{-2}\right)  }{\ln p}%
{\displaystyle\int\limits_{\mathbb{Z}_{p}}}
\ln\left[  \max\left\{  p^{l}\left\vert y_{1}-u_{1}\right\vert _{p},\left\vert
w\right\vert _{p}\right\}  \right]  dw.
\end{gather*}

\end{remark}

\textbf{Claim 3}

By Remarks (\ref{Nota1AAA})-(\ref{Nota2AAA}),%

\begin{gather*}
C_{\phi\left(  \boldsymbol{h}\right)  \phi\left(  \boldsymbol{h}\right)  }%
^{2}\left(  u_{1},y_{1},t_{1},t_{2}\right)  :=\lim_{\rho\rightarrow2}%
C_{\phi\left(  \boldsymbol{h}\right)  \phi\left(  \boldsymbol{h}\right)
}^{\rho}\left(  u_{1},y_{1},t_{1},t_{2}\right) \\
=%
{\displaystyle\sum\limits_{I\in G_{l}}}
\phi\left(  \boldsymbol{h}_{I}\left(  t_{2}\right)  \right)  \phi\left(
\boldsymbol{h}_{I}\left(  t_{1}\right)  \right)  \times\\
\left\{  \left(  1-p^{-2}\right)  \left(  l+1\right)  p^{-2l}+\frac{\left(
1-p^{-2}\right)  }{\ln p}%
{\displaystyle\int\limits_{\mathbb{Z}_{p}}}
\ln\left[  \max\left\{  p^{l}\left\vert y_{1}-u_{1}\right\vert _{p},\left\vert
w\right\vert _{p}\right\}  \right]  dw\right\}  +\\
\frac{p^{-2l}\left(  1-p^{-2}\right)  }{\ln p}%
{\displaystyle\sum\limits_{\substack{I,J\in G_{l}\\I\neq J}}}
\phi\left(  \boldsymbol{h}_{I}\left(  t_{2}\right)  \right)  \phi\left(
\boldsymbol{h}_{J}\left(  t_{1}\right)  \right)  \ln\left(  \frac{p}%
{\max\left\{  \left\vert y_{1}-u_{1}\right\vert _{p},\left\vert I-J\right\vert
_{p}\right\}  }\right)  .
\end{gather*}

We note by Remark \ref{Nota1AAA},%
\[
\lim_{\rho\rightarrow2}K_{\rho,2}\left(  x\right)  =K_{2,2}(x),
\]
see (\ref{Kernel}). For this reason, instead of repeating the calculation for
$C_{\phi\left(  \boldsymbol{h}\right)  \phi\left(  \boldsymbol{h}\right)
}^{2}$ starting with $K_{2,2}(x)$, we can obtain this functional by
$\lim_{\rho\rightarrow2}C_{\phi\left(  \boldsymbol{h}\right)  \phi\left(
\boldsymbol{h}\right)  }^{\rho}$.

We now solve the equations (\ref{Equation_G}) for the case $C_{\phi\left(
\boldsymbol{h}\right)  \phi\left(  \boldsymbol{h}\right)  }^{\rho}$. To
emphasize the dependency of this particular kernel, we use the notation
$G_{\boldsymbol{hh}}^{\rho}=G_{\boldsymbol{hh}}$. But, we do not use this
notation for $G_{\widetilde{\boldsymbol{h}}\boldsymbol{h}}$, $G_{\widetilde
{\boldsymbol{h}}\boldsymbol{h}}$ because these functions do not depend on
$\rho$

\textbf{Claim 4}

The solution of%

\[
\left(  \partial_{t_{1}}-\gamma\right)  G_{\widetilde{\boldsymbol{h}%
}\boldsymbol{h}}\left(  x,y,t_{1},t_{2}\right)  =\delta\left(  x-y\right)
\delta\left(  t_{1}-t_{2}\right)
\]
is%
\[
G_{\widetilde{\boldsymbol{h}}\boldsymbol{h}}\left(  x,y,t_{1},t_{2}\right)
=e^{\gamma\left(  t_{1}-t_{2}\right)  }\Theta\left(  t_{1}-t_{2}\right)
\delta\left(  x-y\right)  ,
\]
where $\Theta\left(  t\right)  =1$ for $t>0$, and $\Theta\left(  t\right)  =0$
for $t\leq0$.

\textbf{Claim 5}

The solution of%
\[
-\left(  \partial_{t_{1}}+\gamma\right)  G_{\boldsymbol{h}\widetilde
{\boldsymbol{h}}}\left(  x,y,t_{1},t_{2}\right)  =\delta\left(  x-y\right)
\delta\left(  t_{1}-t_{2}\right)  .
\]
is%
\[
G_{\boldsymbol{h}\widetilde{\boldsymbol{h}}}\left(  x,y,t_{1},t_{2}\right)
=-e^{-\gamma\left(  t_{1}-t_{2}\right)  }\Theta\left(  t_{1}-t_{2}\right)
\delta\left(  x-y\right)  .
\]

\textbf{Claim 6}

We take $\rho\in\left(  1,\infty\right)  \smallsetminus\left\{  2\right\}  $.
By Claims 2, 4,
\begin{multline*}
C_{\phi\left(  \boldsymbol{h}\right)  \phi\left(  \boldsymbol{h}\right)
}^{\rho}\left(  x,z,t_{1},s\right)  G_{\widetilde{\boldsymbol{h}%
}\boldsymbol{h}}\left(  z,y,s,t_{2}\right)  =\\
\frac{p^{-l\rho}}{\Gamma_{2}\left(  \rho\right)  }%
{\displaystyle\sum\limits_{I\in G_{l}}}
\phi\left(  \boldsymbol{h}_{I}\left(  s\right)  \right)  \phi\left(
\boldsymbol{h}_{I}\left(  t_{1}\right)  \right)  e^{\gamma\left(
s-t_{2}\right)  }\Theta\left(  s-t_{2}\right)  \delta\left(  z-y\right)
\times\\
\left\{
{\displaystyle\int\limits_{\mathbb{Z}_{p}}}
\left[  \max\left\{  p^{l}\left\vert z-x\right\vert _{p},\left\vert
w\right\vert _{p}\right\}  \right]  ^{\rho-2}dw-p^{-2l+\rho-2}\right\}  +\\
\frac{p^{-2l}}{\Gamma_{2}\left(  \rho\right)  }e^{\gamma\left(  s-t_{2}%
\right)  }\Theta\left(  s-t_{2}\right)  \delta\left(  z-y\right)
{\displaystyle\sum\limits_{\substack{I,J\in G_{l}\\I\neq J}}}
\phi\left(  \boldsymbol{h}_{I}\left(  s\right)  \right)  \phi\left(
\boldsymbol{h}_{J}\left(  t_{1}\right)  \right)  \times\\
\left[  \max\left\{  \left\vert z-x\right\vert _{p},\left\vert I-J\right\vert
_{p}\right\}  \right]  ^{\rho-2}.
\end{multline*}

Then%
\begin{gather*}%
{\displaystyle\int\limits_{\mathbb{Z}_{p}}}
{\displaystyle\int\limits_{\mathbb{R}}}
C_{\phi\left(  \boldsymbol{h}\right)  \phi\left(  \boldsymbol{h}\right)
}^{\rho}\left(  x,z,t_{1},s\right)  G_{\widetilde{\boldsymbol{h}%
}\boldsymbol{h}}\left(  z,y,s,t_{2}\right)  dsdz=\\
\frac{p^{-l\rho}}{\Gamma_{2}\left(  \rho\right)  }\left\{
{\displaystyle\int\limits_{\mathbb{Z}_{p}}}
\left[  \max\left\{  p^{l}\left\vert y-x\right\vert _{p},\left\vert
w\right\vert _{p}\right\}  \right]  ^{\rho-2}dw-p^{-2l+\rho-2}\right\}
\times\\%
{\displaystyle\sum\limits_{I\in G_{l}}}
\phi\left(  \boldsymbol{h}_{I}\left(  t_{1}\right)  \right)
{\displaystyle\int\nolimits_{t_{2}}^{\infty}}
\phi\left(  \boldsymbol{h}_{I}\left(  s\right)  \right)  e^{\gamma\left(
s-t_{2}\right)  }ds+\\
\frac{p^{-2l}}{\Gamma_{2}\left(  \rho\right)  }%
{\displaystyle\sum\limits_{\substack{I,J\in G_{l}\\I\neq J}}}
\phi\left(  \boldsymbol{h}_{J}\left(  t_{1}\right)  \right)  \left(
{\displaystyle\int\nolimits_{t_{2}}^{\infty}}
\phi\left(  \boldsymbol{h}_{I}\left(  s\right)  \right)  e^{\gamma\left(
s-t_{2}\right)  }ds\right)  \left[  \max\left\{  \left\vert y-x\right\vert
_{p},\left\vert I-J\right\vert _{p}\right\}  \right]  ^{\rho-2}.
\end{gather*}

We set
\begin{align*}
A(t_{1},t_{2};\phi,\boldsymbol{h},\gamma)  &  :=%
{\displaystyle\sum\limits_{I\in G_{l}}}
\phi\left(  \boldsymbol{h}_{I}\left(  t_{1}\right)  \right)
{\displaystyle\int\nolimits_{t_{2}}^{\infty}}
\phi\left(  \boldsymbol{h}_{I}\left(  s\right)  \right)  e^{\gamma\left(
s-t_{2}\right)  }ds,\\
B_{I,J}(t_{1},t_{2};\phi,\boldsymbol{h},\gamma)  &  :=\phi\left(
\boldsymbol{h}_{J}\left(  t_{1}\right)  \right)
{\displaystyle\int\nolimits_{t_{2}}^{\infty}}
\phi\left(  \boldsymbol{h}_{I}\left(  s\right)  \right)  e^{\gamma\left(
s-t_{2}\right)  }ds,
\end{align*}
then%
\begin{gather*}%
{\displaystyle\int\limits_{\mathbb{Z}_{p}}}
{\displaystyle\int\limits_{\mathbb{R}}}
C_{\phi\left(  \boldsymbol{h}\right)  \phi\left(  \boldsymbol{h}\right)
}^{\rho}\left(  x,z,t_{1},s\right)  G_{\widetilde{\boldsymbol{h}%
}\boldsymbol{h}}\left(  z,y,s,t_{2}\right)  dsdz=\\
\frac{p^{-l\rho}A(t_{1},t_{2};\phi,\boldsymbol{h},\gamma)}{\Gamma_{2}\left(
\rho\right)  }\left\{
{\displaystyle\int\limits_{\mathbb{Z}_{p}}}
\left[  \max\left\{  p^{l}\left\vert y-x\right\vert _{p},\left\vert
w\right\vert _{p}\right\}  \right]  ^{\rho-2}dw-p^{-2l+\rho-2}\right\}  +\\
\frac{p^{-2l}}{\Gamma_{2}\left(  \rho\right)  }%
{\displaystyle\sum\limits_{\substack{I,J\in G_{l}\\I\neq J}}}
B_{I,J}(t_{1},t_{2};\phi,\boldsymbol{h},\gamma)\left[  \max\left\{  \left\vert
y-x\right\vert _{p},\left\vert I-J\right\vert _{p}\right\}  \right]  ^{\rho-2}%
\end{gather*}

\subsection{Computation of $G_{\boldsymbol{hh}}^{\rho}\left(  x,y,\tau\right)
$, $\rho\neq2$}

To solve the second equation in (\ref{Equation_G}), we assume invariance
translation, so we take $t_{2}=0$, $\tau=t_{1}$, and $G_{\boldsymbol{hh}%
}^{\rho}\left(  x,y,\tau\right)  :=G_{\boldsymbol{hh}}^{\rho}\left(
x,y,-\tau,0\right)  $, and assume that $G_{\boldsymbol{hh}}^{\rho}\left(
x,y,\tau\right)  =G_{\boldsymbol{hh}}^{\rho}\left(  x,y,-\tau\right)  $. We
also set $G_{\widetilde{\boldsymbol{h}}\boldsymbol{h}}\left(  z,y,\tau\right)
:=G_{\widetilde{\boldsymbol{h}}\boldsymbol{h}}\left(  z,y,s,0\right)  $,
$G_{\widetilde{\boldsymbol{h}}\boldsymbol{h}}\left(  x,y,\tau\right)
:=G_{\widetilde{\boldsymbol{h}}\boldsymbol{h}}\left(  x,y,\tau,0\right)  $.
Now, the mentioned equation becomes%
\begin{gather*}
-\left(  \partial_{\tau}+\gamma\right)  G_{\boldsymbol{hh}}^{\rho}\left(
x,y,\tau\right)  +%
{\displaystyle\int\limits_{\mathbb{Z}_{p}}}
{\displaystyle\int\limits_{\mathbb{R}}}
C_{\phi\left(  \boldsymbol{h}\right)  \phi\left(  \boldsymbol{h}\right)
}^{\rho}\left(  x,z,\tau,s\right)  G_{\widetilde{\boldsymbol{h}}%
\boldsymbol{h}}\left(  z,y,s\right)  dsdz\\
\\
+\sigma^{2}G_{\widetilde{\boldsymbol{h}}\boldsymbol{h}}\left(  x,y,\tau
\right)  =0,
\end{gather*}

i.e.,%
\begin{gather*}
\left(  \partial_{\tau}+\gamma\right)  G_{\boldsymbol{hh}}^{\rho}\left(
x,y,\tau\right)  =\\
\frac{p^{-l\rho}A(\tau,0;\phi,\boldsymbol{h},\gamma)}{\Gamma_{2}\left(
\rho\right)  }\left\{
{\displaystyle\int\limits_{\mathbb{Z}_{p}}}
\left[  \max\left\{  p^{l}\left\vert y-x\right\vert _{p},\left\vert
w\right\vert _{p}\right\}  \right]  ^{\rho-2}dw-p^{-2l+\rho-2}\right\}  +\\
\frac{p^{-2l}}{\Gamma_{2}\left(  \rho\right)  }%
{\displaystyle\sum\limits_{\substack{I,J\in G_{l}\\I\neq J}}}
B_{I,J}(\tau,0;\phi,\boldsymbol{h},\gamma)\left[  \max\left\{  \left\vert
y-x\right\vert _{p},\left\vert I-J\right\vert _{p}\right\}  \right]  ^{\rho
-2}\\
-\sigma^{2}e^{\gamma\tau}\Theta\left(  \tau\right)  \delta\left(  x-y\right)
.
\end{gather*}

We set $\Delta\left(  \tau\right)  =e^{-\gamma\tau}\Theta\left(  \tau\right)
$, where $\Theta\left(  \tau\right)  $ is the step function as before. Then,
the solution of the equation%
\[
\left(  \partial_{\tau}+\gamma\right)  U\left(  \tau\right)  =F\left(
\tau\right)
\]
is $U\left(  \tau\right)  =\Delta\left(  \tau\right)  \ast F\left(
\tau\right)  $.

Furthermore,%
\begin{align*}
e^{\gamma\tau}\Theta\left(  \tau\right)  \ast\Delta\left(  \tau\right)   &
=e^{\gamma\tau}\Theta\left(  \tau\right)  \ast e^{-\gamma\tau}\Theta\left(
\tau\right)  =%
{\displaystyle\int\limits_{\mathbb{R}}}
e^{\gamma\left(  \tau-s\right)  }\Theta\left(  \tau-s\right)  e^{-\gamma
s}\Theta\left(  s\right)  ds\\
&  =e^{\gamma\tau}%
{\displaystyle\int\nolimits_{0}^{\tau}}
e^{-2\gamma s}ds=\frac{e^{\gamma\tau}}{2\gamma}\left(  1-e^{-2\gamma\tau
}\right)  .
\end{align*}

\textbf{Claim 7}

If $\rho\in\left(  1,\infty\right)  \smallsetminus\left\{  2\right\}  $,%
\begin{gather}
G_{\boldsymbol{hh}}^{\rho}\left(  x,y,\tau\right)  =\frac{p^{-l\rho}%
A(\tau,0;\phi,\boldsymbol{h},\gamma)\ast\Delta\left(  \tau\right)  }%
{\Gamma_{2}\left(  \rho\right)  }\times\nonumber\\
\left\{
{\displaystyle\int\limits_{\mathbb{Z}_{p}}}
\left[  \max\left\{  p^{l}\left\vert y-x\right\vert _{p},\left\vert
w\right\vert _{p}\right\}  \right]  ^{\rho-2}dw-p^{-2l+\rho-2}\right\}
+\nonumber\\
\frac{p^{-2l}}{\Gamma_{2}\left(  \rho\right)  }%
{\displaystyle\sum\limits_{\substack{I,J\in G_{l}\\I\neq J}}}
\left[  \max\left\{  \left\vert y-x\right\vert _{p},\left\vert I-J\right\vert
_{p}\right\}  \right]  ^{\rho-2}B_{I,J}(\tau,0;\phi,\boldsymbol{h},\gamma
)\ast\Delta\left(  \tau\right) \nonumber\\
-\sigma^{2}\frac{e^{\gamma\tau}}{2\gamma}\left(  1-e^{-2\gamma\tau}\right)
\delta\left(  x-y\right)  . \label{Covar_G_alpha}%
\end{gather}
for $\boldsymbol{h\in}C^{\infty}(\mathbb{R})%
{\textstyle\bigotimes\nolimits_{\text{alg}}}
\mathcal{D}_{l}\left(  \mathbb{Z}_{p}\right)  $.

\subsection{Computation of $G_{\boldsymbol{hh}}^{2}\left(  x,y,\tau\right)  $}%

\begin{gather*}%
{\displaystyle\int\limits_{\mathbb{Z}_{p}}}
{\displaystyle\int\limits_{\mathbb{R}}}
C_{\phi\left(  \boldsymbol{h}\right)  \phi\left(  \boldsymbol{h}\right)  }%
^{2}\left(  x,z,t_{1},s\right)  G_{\widetilde{\boldsymbol{h}}\boldsymbol{h}%
}\left(  z,y,s,t_{2}\right)  dsdz=\\
=%
{\displaystyle\sum\limits_{I\in G_{l}}}
\phi\left(  \boldsymbol{h}_{I}\left(  t_{1}\right)  \right)
{\displaystyle\int\nolimits_{t_{2}}^{\infty}}
\phi\left(  \boldsymbol{h}_{I}\left(  s\right)  \right)  e^{\gamma\left(
s-t_{2}\right)  }ds\times\\
\left\{  \left(  1-p^{-2}\right)  \left(  l+1\right)  p^{-2l}+\frac{\left(
1-p^{-2}\right)  }{\ln p}%
{\displaystyle\int\limits_{\mathbb{Z}_{p}}}
\ln\left[  \max\left\{  p^{l}\left\vert y-x\right\vert _{p},\left\vert
w\right\vert _{p}\right\}  \right]  dw\right\}  +\\
\frac{p^{-2l}\left(  1-p^{-2}\right)  }{\ln p}%
{\displaystyle\sum\limits_{\substack{I,J\in G_{l}\\I\neq J}}}
\phi\left(  \boldsymbol{h}_{J}\left(  t_{1}\right)  \right)  \left(
{\displaystyle\int\nolimits_{t_{2}}^{\infty}}
\phi\left(  \boldsymbol{h}_{I}\left(  s\right)  \right)  e^{\gamma\left(
s-t_{2}\right)  }ds\right)  \times\\
\ln\left(  \frac{p}{\max\left\{  \left\vert y-x\right\vert _{p},\left\vert
I-J\right\vert _{p}\right\}  }\right)  .
\end{gather*}

We now set $A(t_{1},t_{2};\phi,\boldsymbol{h},\gamma)$, $B_{I,J}(t_{1}%
,t_{2};\phi,\boldsymbol{h},\gamma)$\ as before, then%
\begin{gather*}%
{\displaystyle\int\limits_{\mathbb{Z}_{p}}}
{\displaystyle\int\limits_{\mathbb{R}}}
C_{\phi\left(  \boldsymbol{h}\right)  \phi\left(  \boldsymbol{h}\right)  }%
^{2}\left(  x,z,t_{1},s\right)  G_{\widetilde{\boldsymbol{h}}\boldsymbol{h}%
}\left(  z,y,s,t_{2}\right)  dsdz=A(t_{1},t_{2};\phi,\boldsymbol{h}%
,\gamma)\times\\
\left\{  \left(  1-p^{-2}\right)  \left(  l+1\right)  p^{-2l}+\frac{\left(
1-p^{-2}\right)  }{\ln p}%
{\displaystyle\int\limits_{\mathbb{Z}_{p}}}
\ln\left[  \max\left\{  p^{l}\left\vert y-x\right\vert _{p},\left\vert
w\right\vert _{p}\right\}  \right]  dw\right\}  +\\
\frac{p^{-2l}\left(  1-p^{-2}\right)  }{\ln p}%
{\displaystyle\sum\limits_{\substack{I,J\in G_{l}\\I\neq J}}}
B_{I,J}(t_{1},t_{2};\phi,\boldsymbol{h},\gamma)\ln\left(  \frac{p}%
{\max\left\{  \left\vert y-x\right\vert _{p},\left\vert I-J\right\vert
_{p}\right\}  }\right)  .
\end{gather*}

To solve the second equation in (\ref{Equation_G}), for $\rho=2$, we assume
invariance translation, so we take $t_{2}=0$, $\tau=t_{1}$ as before; then,
the mentioned equation becomes%
\begin{gather*}
-\left(  \partial_{\tau}+\gamma\right)  G_{\boldsymbol{hh}}^{2}\left(
x,y,\tau\right)  +%
{\displaystyle\int\limits_{\mathbb{Z}_{p}}}
{\displaystyle\int\limits_{\mathbb{R}}}
C_{\phi\left(  \boldsymbol{h}\right)  \phi\left(  \boldsymbol{h}\right)  }%
^{2}\left(  x,z,\tau,s\right)  G_{\widetilde{\boldsymbol{h}}\boldsymbol{h}%
}\left(  z,y,s\right)  dsdz\\
\\
+\sigma^{2}G_{\widetilde{\boldsymbol{h}}\boldsymbol{h}}\left(  x,y,\tau
\right)  =0,
\end{gather*}

i.e.,%
\begin{gather*}
\left(  \partial_{\tau}+\gamma\right)  G_{\boldsymbol{hh}}^{2}\left(
x,y,\tau\right)  =A(\tau,0;\phi,\boldsymbol{h},\gamma)\times\\
\left\{  \left(  1-p^{-2}\right)  \left(  l+1\right)  p^{-2l}+\frac{\left(
1-p^{-2}\right)  }{\ln p}%
{\displaystyle\int\limits_{\mathbb{Z}_{p}}}
\ln\left[  \max\left\{  p^{l}\left\vert y-x\right\vert _{p},\left\vert
w\right\vert _{p}\right\}  \right]  dw\right\}  +\\
\frac{p^{-2l}\left(  1-p^{-2}\right)  }{\ln p}%
{\displaystyle\sum\limits_{\substack{I,J\in G_{l}\\I\neq J}}}
B_{I,J}(\tau,0;\phi,\boldsymbol{h},\gamma)\ln\left(  \frac{p}{\max\left\{
\left\vert y-x\right\vert _{p},\left\vert I-J\right\vert _{p}\right\}
}\right) \\
-\sigma^{2}e^{\gamma\tau}\Theta\left(  \tau\right)  \delta\left(  x-y\right)
.
\end{gather*}

By proceeding as in the previous section, we obtain the following result.

\textbf{Claim 8 }%

\begin{gather}
G_{\boldsymbol{hh}}^{2}\left(  x,y,\tau\right)  =\left(  A(\tau,0;\phi
,\boldsymbol{h},\gamma)\ast\Delta\left(  \tau\right)  \right)  \times
\nonumber\\
\left\{  \left(  1-p^{-2}\right)  \left(  l+1\right)  p^{-2l}+\frac{\left(
1-p^{-2}\right)  }{\ln p}%
{\displaystyle\int\limits_{\mathbb{Z}_{p}}}
\ln\left[  \max\left\{  p^{l}\left\vert y-x\right\vert _{p},\left\vert
w\right\vert _{p}\right\}  \right]  dw\right\}  +\nonumber\\
\frac{p^{-2l}\left(  1-p^{-2}\right)  }{\ln p}%
{\displaystyle\sum\limits_{\substack{I,J\in G_{l}\\I\neq J}}}
\left(  B_{I,J}(\tau,0;\phi,\boldsymbol{h},\gamma)\ast\Delta\left(
\tau\right)  \right)  \ln\left(  \frac{p}{\max\left\{  \left\vert
y-x\right\vert _{p},\left\vert I-J\right\vert _{p}\right\}  }\right)
\nonumber\\
-\sigma^{2}\frac{e^{\gamma\tau}}{2\gamma}\left(  1-e^{-2\gamma\tau}\right)
\delta\left(  x-y\right)  \label{Covar_G_2}%
\end{gather}
for $\boldsymbol{h\in}C^{\infty}(\mathbb{R})%
{\textstyle\bigotimes\nolimits_{\text{alg}}}
\mathcal{D}_{l}\left(  \mathbb{Z}_{p}\right)  $.

\subsection{Continuous phase transitions}

The partition function $\overline{{\LARGE Z}}_{M}\left(  \rho\right)  $,
$\rho\in\left(  1,\infty\right)  \smallsetminus\left\{  2\right\}  $,
corresponds to a random $p$-adic cellular neural network. The state of the
network $\boldsymbol{h}(x,t)\in L^{2}(\mathbb{Z}_{p}\times\mathbb{R})$ is a
realization of a generalized Gaussian process with covariance
$G_{\boldsymbol{hh}}^{\rho}(x,y,\tau)=\left\langle \boldsymbol{h}%
(x,t_{1}),\boldsymbol{h}(y,t_{2})\right\rangle _{\boldsymbol{h}}$, $\tau
=t_{1}-t_{2}$. We interpret $\boldsymbol{h}(x,t)$ as a spatiotemporal pattern
produced by the network. By construction, $\overline{{\LARGE Z}}_{M}\left(
\rho\right)  $ has a pole at $\rho=2$, i.e. $\overline{{\LARGE Z}}_{M}\left(
\rho\right)  =\overline{{\LARGE Z}}_{M}\left(  \frac{1}{1-p^{\rho-2}}\right)
$, so
\[
\left.  \frac{d^{k}}{d\rho^{k}}\ln\overline{{\LARGE Z}}_{M}\left(
\rho\right)  \right\vert _{\rho=2}=\infty\text{, for }k=1,2,
\]
consequently, the network has a continuous phase transition at $\rho=2$. We
interpret $\rho$ as a control parameter, and $G_{\boldsymbol{hh}}^{\rho
}(x,y,\tau)$, $\rho\in\left(  1,\infty\right)  \smallsetminus\left\{
2\right\}  $ as the order parameter. Then, there are two phases:
$G_{\boldsymbol{hh}}^{\rho}(x,y,\tau)$, $\rho\in\left(  1,\infty\right)  $,
and $G_{\boldsymbol{hh}}^{2}(x,y,\tau)$. In our analysis $\tau$\ is fixed; it
is based on the explicit expressions for $G_{\boldsymbol{hh}}^{\rho}%
(x,y,\tau)$, $\rho\in\left(  1,\infty\right)  \smallsetminus\left\{
2\right\}  $, and $G_{\boldsymbol{hh}}^{2}(x,y,\tau)$ derived in the previous
sections. The functions covariance functions are radial, i.e., they depend on
$\left\vert x-y\right\vert _{p}$ which is the distance between neurons, up to
a term of the form $c\delta\left(  x-y\right)  $. If $A(\tau,0;\phi
,\boldsymbol{h},l)\ast\Delta\left(  \tau\right)  =B_{I,J}(\tau,0;\phi
,\boldsymbol{h},l)\ast\Delta\left(  \tau\right)  =0$, which happens when
$\boldsymbol{h}(x,t)=0$, then $G_{\boldsymbol{hh}}^{\rho}(x,y,\tau
)=-\sigma^{2}\frac{e^{\gamma\tau}}{2\gamma}\left(  1-e^{-2\gamma\tau}\right)
\delta\left(  x-y\right)  $, see (\ref{Covar_G_alpha})-(\ref{Covar_G_2}). A
covariance function of the form $c\delta\left(  x-y\right)  $ corresponds a
white noise, which we interpret as the `baseline state/pattern' (or a
background noise) of the network. We argue, $G_{\boldsymbol{hh}}^{\rho
}(x,y,\tau)$ provides a description of the neuronal connections, and so of the
network organization, for a fixed $\tau$.

\subsection{The behavior of $G_{\boldsymbol{hh}}^{\rho}(x,y,\tau)$, for
$\left\vert x-y\right\vert _{p}\approx0$}

We now study the behavior of $G_{\boldsymbol{hh}}^{\rho}(x,y,\tau)$, $\rho
\in\left(  1,2\right)  $, see (\ref{Covar_G_alpha}), for $\left\vert
x-y\right\vert _{p}\approx0$. Assuming that $A(\tau,0;\phi,\boldsymbol{h}%
,\gamma)\ast\Delta\left(  \tau\right)  \neq0$, we have%
\begin{gather*}
T_{1}\left(  \left\vert y-x\right\vert _{p},\rho\right)  :=\\
\frac{p^{-l\rho}A(\tau,0;\phi,\boldsymbol{h},\gamma)\ast\Delta\left(
\tau\right)  }{\Gamma_{2}\left(  \rho\right)  }\left\{
{\displaystyle\int\limits_{\mathbb{Z}_{p}}}
\left[  \max\left\{  p^{l}\left\vert y-x\right\vert _{p},\left\vert
w\right\vert _{p}\right\}  \right]  ^{\rho-2}dw-p^{-2l+\rho-2}\right\} \\
\approx\frac{p^{-l\rho}A(\tau,0;\phi,\boldsymbol{h},\gamma)\ast\Delta\left(
\tau\right)  }{\Gamma_{2}\left(  \rho\right)  }\left\{
{\displaystyle\int\limits_{\mathbb{Z}_{p}}}
\left\vert w\right\vert _{p}^{\rho-2}dw-p^{-2l+\rho-2}\right\} \\
=\frac{p^{-l\rho}A(\tau,0;\phi,\boldsymbol{h},\gamma)\ast\Delta\left(
\tau\right)  }{\Gamma_{2}\left(  \rho\right)  }\left\{  \frac{1-p^{-1}%
}{1-p^{1-\rho}}-p^{-2l+\rho-2}\right\} \\
=p^{-l\rho}\left\{  A(\tau,0;\phi,\boldsymbol{h},\gamma)\ast\Delta\left(
\tau\right)  \right\}  \frac{1-p^{-\rho}}{1-p^{\rho-2}}\left\{  \frac
{1-p^{-1}}{1-p^{1-\rho}}-p^{-2l+\rho-2}\right\}  .
\end{gather*}
For further details about the calculation of the above integral, the reader
may consult \cite[Chapter 1]{Zuniga-Textbook}.

Now, assuming that not all the $B_{I,J}(\tau,0;\phi,\boldsymbol{h},\gamma
)\ast\Delta\left(  \tau\right)  $ are zero, for $\left\vert x-y\right\vert
_{p}\approx0$,%

\begin{gather*}
T_{2}\left(  \left\vert y-x\right\vert _{p},\rho\right)  :=\\
\frac{p^{-2l}}{\Gamma_{2}\left(  \rho\right)  }%
{\displaystyle\sum\limits_{\substack{I,J\in G_{l}\\I\neq J}}}
\left[  \max\left\{  \left\vert y-x\right\vert _{p},\left\vert I-J\right\vert
_{p}\right\}  \right]  ^{\rho-2}\left(  B_{I,J}(\tau,0;\phi,\boldsymbol{h}%
,\gamma)\ast\Delta\left(  \tau\right)  \right) \\
\approx\frac{p^{-2l}}{\Gamma_{2}\left(  \rho\right)  }%
{\displaystyle\sum\limits_{\substack{I,J\in G_{l}\\I\neq J}}}
B_{I,J}(\tau,0;\phi,\boldsymbol{h},\gamma)\ast\Delta\left(  \tau\right)
\left\vert I-J\right\vert _{p}^{\rho-2}.
\end{gather*}
\ In conclusion, $G_{\boldsymbol{hh}}^{\rho}(x,y,\tau)+\sigma^{2}%
\frac{e^{\gamma\tau}}{2\gamma}\left(  1-e^{-2\gamma\tau}\right)  \delta\left(
x-y\right)  $ does not dependent of $x,y$ for $\left\vert y-x\right\vert
_{p}\approx0$, $\rho\in\left(  1,\infty\right)  \smallsetminus\left\{
2\right\}  $, and for $\boldsymbol{h\in}C^{\infty}(\mathbb{R})%
{\textstyle\bigotimes\nolimits_{\text{alg}}}
\mathcal{D}_{l}\left(  \mathbb{Z}_{p}\right)  $, which means that the
short-range neural interactions do not affect the phase of the network
described by $G_{\boldsymbol{hh}}^{\rho}(x,y,\tau)$.

\subsection{The behavior of $G_{\boldsymbol{hh}}^{\rho}(x,y,\tau)$, for
$\left\vert x-y\right\vert _{p}>p^{-l}$}

The condition $\left\vert x-y\right\vert _{p}>p^{-l}$ is equivalent to
$p^{l}\left\vert y-x\right\vert _{p}>1$, and since $\left\vert w\right\vert
_{p}\leq1$ for $w\in\mathbb{Z}_{p}$, we have%
\[
T_{1}\left(  \left\vert y-x\right\vert _{p},\rho\right)  =\frac{p^{-l\rho
}A(\tau,0;\phi,\boldsymbol{h},\gamma)\ast\Delta\left(  \tau\right)  }%
{\Gamma_{2}\left(  \rho\right)  }\left\{  p^{-2l}\left\vert y-x\right\vert
_{p}^{\rho-2}-p^{-l\left(  \rho-2\right)  +\rho+2}\right\}  .
\]

Now, the set $p^{-l}<\left\vert y-x\right\vert _{p}\leq1$ is a compact subset,
and the function $f(\left\vert y-x\right\vert _{p}):=\frac{p^{-l\rho}}%
{\Gamma_{2}\left(  \rho\right)  }\left(  p^{-2l}\left\vert y-x\right\vert
_{p}^{\rho-2}-p^{-l\left(  \rho-2\right)  +\rho+2}\right)  $, $p^{-l}%
<\left\vert y-x\right\vert _{p}\leq1$, $\rho\in\left(  1,\infty\right)
\smallsetminus\left\{  2\right\}  $, is continuous on it; then, there are
constants $C_{0}(\rho)$, $C_{1}(\rho)$ such that%
\[
C_{0}(\rho)\leq\frac{p^{-l\rho}}{\Gamma_{2}\left(  \rho\right)  }\left(
p^{-2l}\left\vert y-x\right\vert _{p}^{\rho-2}-p^{-l\left(  \rho-2\right)
+\rho+2}\right)  \leq C_{1}(\rho).
\]
In conclusion, for $\left\vert x-y\right\vert _{p}>p^{-l}$ and $\rho\in\left(
1,\infty\right)  \smallsetminus\left\{  2\right\}  $,%
\begin{equation}
C_{0}^{\prime}(\rho,\tau,0;\phi,l)\leq T_{1}\left(  \left\vert y-x\right\vert
_{p},\rho\right)  \leq C_{1}^{\prime}(\rho,\tau,0;\phi,l). \label{Estimation}%
\end{equation}
On the other hand, the function
\[
g(\left\vert y-x\right\vert _{p}):=\frac{1}{\Gamma_{2}\left(  \rho\right)
}\left[  \max\left\{  \left\vert y-x\right\vert _{p},\left\vert I-J\right\vert
_{p}\right\}  \right]  ^{\rho-2},
\]
for $\left(  x,y\right)  \in\mathbb{Z}_{p}\times\mathbb{Z}_{p}$ is continuous,
since $\mathbb{Z}_{p}\times\mathbb{Z}_{p}$ is compact, then, estimation
(\ref{Estimation}) is also valid for $T_{2}\left(  \left\vert y-x\right\vert
_{p},\rho\right)  $. In conclusion,
\begin{multline*}
2C_{0}^{\prime}(\rho,\tau,0;\phi,l)\leq\\
G_{\boldsymbol{hh}}^{\rho}(x,y,\tau)+\sigma^{2}\frac{e^{\gamma\tau}}{2\gamma
}\left(  1-e^{-2\gamma\tau}\right)  \delta\left(  x-y\right)  \leq\\
2C_{1}^{\prime}(\rho,\tau,0;\phi,l),
\end{multline*}
for $\left\vert x-y\right\vert _{p}>p^{-l}$, $\rho\in\left(  1,\infty\right)
\smallsetminus\left\{  2\right\}  $, and $\boldsymbol{h\in}C^{\infty
}(\mathbb{R})%
{\textstyle\bigotimes\nolimits_{\text{alg}}}
\mathcal{D}_{l}\left(  \mathbb{Z}_{p}\right)  $; then, the long-range
connections between neurons do no affect $G^{\rho}(x,y,\tau)$, in addition,
they are masked for the background noise for $\tau$\ large.

\subsection{The behavior of $G_{\boldsymbol{hh}}^{2}(x,y,\tau)$, for
$\left\vert x-y\right\vert _{p}\approx0$}

For $\left\vert x-y\right\vert _{p}\approx0$, we have
\begin{gather*}
T_{1}^{\prime}\left(  \left\vert y-x\right\vert _{p}\right)  :=A(\tau
,0;\phi,\boldsymbol{h},\gamma)\ast\Delta\left(  \tau\right)  \times\\
\left\{  \left(  1-p^{-2}\right)  \left(  l+1\right)  p^{-2l}+\frac{\left(
1-p^{-2}\right)  }{\ln p}%
{\displaystyle\int\limits_{\mathbb{Z}_{p}}}
\ln\left[  \max\left\{  p^{l}\left\vert y-x\right\vert _{p},\left\vert
w\right\vert _{p}\right\}  \right]  dw\right\} \\
\approx\left(  A(\tau,0;\phi,\boldsymbol{h},\gamma)\ast\Delta\left(
\tau\right)  \right)  \left\{  \left(  1-p^{-2}\right)  \left(  l+1\right)
p^{-2l}+\frac{\left(  1-p^{-2}\right)  }{\ln p}%
{\displaystyle\int\limits_{\mathbb{Z}_{p}}}
\ln\left(  \left\vert w\right\vert _{p}\right)  dw\right\}  ,
\end{gather*}
where $\int_{\mathbb{Z}_{p}}\ln\left(  \left\vert w\right\vert _{p}\right)
dw$ exits.

Now, since $\left\vert I-J\right\vert _{p}>p^{-i}$ for $I\neq J$,%
\begin{gather*}
T_{2}^{\prime}\left(  \left\vert y-x\right\vert _{p}\right)  :=\\
\frac{p^{-2l}\left(  1-p^{-2}\right)  }{\ln p}%
{\displaystyle\sum\limits_{\substack{I,J\in G_{l}\\I\neq J}}}
\left(  B_{I,J}(\tau,0;\phi,\boldsymbol{h},\gamma)\ast\Delta\left(
\tau\right)  \right)  \ln\left(  \frac{p}{\max\left\{  \left\vert
y-x\right\vert _{p},\left\vert I-J\right\vert _{p}\right\}  }\right) \\
\approx\frac{p^{-2l}\left(  1-p^{-2}\right)  }{\ln p}%
{\displaystyle\sum\limits_{\substack{I,J\in G_{l}\\I\neq J}}}
\left(  B_{I,J}(\tau,0;\phi,\boldsymbol{h},\gamma)\ast\Delta\left(
\tau\right)  \right)  \ln\left(  \frac{p}{\left\vert I-J\right\vert _{p}%
}\right)  .
\end{gather*}
In conclusion, for $\left\vert x-y\right\vert _{p}\approx0$,
$G_{\boldsymbol{hh}}^{2}\left(  x,y,\tau\right)  +\sigma^{2}\frac
{e^{\gamma\tau}}{2\gamma}\left(  1-e^{-2\gamma\tau}\right)  \delta\left(
x-y\right)  $ is a function independent of $x,y$.

\subsection{The behavior of $G_{\boldsymbol{hh}}^{2}(x,y,\tau)$, for for
$\left\vert x-y\right\vert _{p}>p^{-l}$}

Since $p^{l}\left\vert y-x\right\vert _{p}>1$ and $\left\vert w\right\vert
_{p}\leq1$,%
\begin{align*}
T_{1}^{\prime}\left(  \left\vert y-x\right\vert _{p}\right)   &
=A(\tau,0;\phi,\boldsymbol{h},\gamma)\ast\Delta\left(  \tau\right)  \times\\
&  \left\{  \left(  1-p^{-2}\right)  \left(  l+1\right)  p^{-2l}+\frac{\left(
1-p^{-2}\right)  }{\ln p}\ln\left[  p^{l}\left\vert y-x\right\vert
_{p}\right]  \right\}
\end{align*}
is a non-bounded increasing function of $p^{l}\left\vert y-x\right\vert _{p}$.

Now, since $\left\vert y-x\right\vert _{p},\left\vert I-J\right\vert
_{p}>p^{-l}$, \ the continuity of $\ln\left(  \frac{p}{\max\left\{  \left\vert
y-x\right\vert _{p},\left\vert I-J\right\vert _{p}\right\}  }\right)  $ on
$\mathbb{Z}_{p}\times\mathbb{Z}_{p}$ implies that%
\[
C_{0}^{\prime\prime}(\tau,0;\phi,l)\leq T_{2}^{\prime}\left(  \left\vert
y-x\right\vert _{p}\right)  \leq C_{1}^{\prime\prime}(\tau,0;\phi,l).
\]
In conclusion%
\begin{multline*}
C_{0}^{\prime\prime}(\tau,0;\phi,l)+T_{1}^{\prime}\left(  \left\vert
y-x\right\vert _{p}\right)  \leq\\
G_{\boldsymbol{hh}}^{2}(x,y,\tau)+\sigma^{2}\frac{e^{\gamma\tau}}{2\gamma
}\left(  1-e^{-2\gamma\tau}\right)  \delta\left(  x-y\right)  \leq\\
C_{1}^{\prime\prime}(\tau,0;\phi,l)+T_{1}^{\prime}\left(  \left\vert
y-x\right\vert _{p}\right)  ,
\end{multline*}
for $\boldsymbol{h\in}C^{\infty}(\mathbb{R})%
{\textstyle\bigotimes\nolimits_{\text{alg}}}
\mathcal{D}_{l}\left(  \mathbb{Z}_{p}\right)  $. In conclusion,
$G_{\boldsymbol{hh}}^{2}(x,y,\tau)$ depends only on long-range interactions.
For $\tau$\ fixed these interactions cannot be masked by the background noise.

\subsection{Discussion}

The phase described by $G_{\boldsymbol{hh}}^{\rho}(x,y,\tau)$, $\rho\in\left(
1,\infty\right)  \smallsetminus\left\{  2\right\}  $, is disordered in the
sense that short-range or large-range neuronal connections do not control it.
In addition, for $\tau$ large, any of these interactions is masked by the
background noise. On the other hand, the phase described by
$G_{\boldsymbol{hh}}^{2}(x,y,\tau)$ is ordered in the sense that long-range
interactions control it; in addition, for $\tau$ fixed, these interactions
cannot be masked by the background noise. In our view, the study of the
double-copy model as well as the largest Lyapunov exponent for the toy model
is an interesting problem.

\section{\label{Appen_A}Appendix A}

\subsection{$\partial_{t}$ as pseudo-derivative}

We denote the space of Schwartz functions on $\mathbb{R}$ as $\mathcal{S}%
(\mathbb{R})$, which is the\ set of infinitely differentiable functions
$\varphi\left(  x\right)  $ on $\mathbb{R}$ such that
\[
\left\Vert \varphi\right\Vert _{n,m}:=\sup_{t\in\mathbb{R}}\left\vert
x^{n}\frac{d^{m}\varphi\left(  t\right)  }{dt^{m}}\right\vert <\infty
\]
for all $n,m\in\mathbb{N}=\left\{  0,1,\ldots,k,k+1,\ldots\right\}  $. The
vector space $\mathcal{S}(\mathbb{R})$ with the natural topology given by the
seminorms $\left\Vert \cdot\right\Vert _{n,m}$ is a separable metric space.

Suppose that $\varphi\in\mathcal{S}(\mathbb{R})$. The Fourier transform of
$\varphi$\ is the function $\widehat{\varphi}$ given by%
\[
\widehat{\varphi}\left(  \xi\right)  =\frac{1}{\sqrt{2\pi}}%
{\displaystyle\int\limits_{\mathbb{R}}}
e^{-i\xi t}\varphi\left(  t\right)  dt.
\]
We also use the notation $\mathcal{F}\left(  \varphi\right)  $, $\mathcal{F}%
_{t\rightarrow\xi}$ $\left(  \varphi\right)  $ for the Fourier transform. The
inverse Fourier transform of $\widehat{\varphi}$ is given by
\[
\varphi\left(  t\right)  =\mathcal{F}_{\xi\rightarrow t}^{-1}\left(
\varphi\right)  =\frac{1}{\sqrt{2\pi}}%
{\displaystyle\int\limits_{\mathbb{R}}}
e^{i\xi t}\widehat{\varphi}\left(  \xi\right)  d\xi\text{.}%
\]
The Fourier transform is a linear bicontinuous bijection from $\mathcal{S}%
(\mathbb{R})$\ onto $\mathcal{S}(\mathbb{R})$. Furthermore, for $\varphi
\in\mathcal{S}(\mathbb{R})$,%
\begin{equation}
\frac{d\varphi\left(  t\right)  }{dt}:=\partial_{t}\varphi\left(  t\right)
=\mathcal{F}_{\xi\rightarrow t}^{-1}\left(  \sqrt{-1}\xi\mathcal{F}%
_{t\rightarrow\xi}\left(  \varphi\right)  \right)  . \label{Eq_Derivative}%
\end{equation}
The topological dual of $\mathcal{S}(\mathbb{R})$, denoted as $\mathcal{S}%
^{\prime}(\mathbb{R})$, is the space of tempered distributions. $\mathcal{S}%
^{\prime}(\mathbb{R})$ is a separable space and $\mathcal{S}(\mathbb{R})$ is a
dense subset of $\mathcal{S}^{\prime}(\mathbb{R})$, \cite[Chapter
V]{Reed-Simon-I}. The Fourier transform of a distribution $T\in\mathcal{S}%
^{\prime}(\mathbb{R})$, denoted as $\widehat{T}$ or $\mathcal{F}\left(
T\right)  $, is defined as $\widehat{T}\left(  \varphi\right)  =T(\widehat
{\varphi})$, where $\widehat{\varphi}$ is the Fourier transform of $\varphi
\in\mathcal{S}(\mathbb{R})$.

A distribution $T\in\mathcal{S}^{\prime}(\mathbb{R})$ is said to be in the
Sobolev space $W_{1}(\mathbb{R})$, if $\widehat{T}$ is a measurable function
satisfying
\[
\left\Vert T\right\Vert _{W_{1}(\mathbb{R})}^{2}=%
{\displaystyle\int\limits_{\mathbb{R}}}
\left(  1+\left\vert \xi\right\vert ^{2}\right)  \left\vert \widehat{T}\left(
\xi\right)  \right\vert ^{2}d\xi<\infty.
\]
$W_{1}(\mathbb{R})$ is a real Hilbert space under the norm $\left\Vert
\cdot\right\Vert _{W_{1}}$. Since the Fourier transform is a unitary map of
$L^{2}(\mathbb{R})$\ into $L^{2}(\mathbb{R})$, and $L^{2}(\mathbb{R}%
)\subset\mathcal{S}^{\prime}(\mathbb{R})$, we have%
\[
W_{1}(\mathbb{R})=\left\{  G\in L^{2}(\mathbb{R});\left(  \sqrt{1+\left\vert
\xi\right\vert ^{2}}\right)  \widehat{T}\left(  \xi\right)  \in L^{2}%
(\mathbb{R})\right\}
\]
The space $W_{1}(\mathbb{R})$ is continuously embedded in $L^{2}(\mathbb{R})$,
and since $\mathcal{S}(\mathbb{R})$ is dense in $L^{2}(\mathbb{R})$, then
$\mathcal{S}(\mathbb{R})$ is dense in $W_{1}(\mathbb{R})$; in addition, it is
a separable Hilbert space.

Set
\[
\partial_{t}G\left(  t\right)  =\mathcal{F}_{\xi\rightarrow t}^{-1}\left(
\sqrt{-1}\xi\mathcal{F}_{t\rightarrow\xi}\left(  G\right)  \right)  ,\text{
for }G\in W_{1}(\mathbb{R}).
\]

\begin{lemma}
\label{Lemma_1} The mapping%
\begin{equation}%
\begin{array}
[c]{ccc}%
W_{1}(\mathbb{R}) & \rightarrow & L^{2}(\mathbb{R})\\
&  & \\
G & \rightarrow & \partial_{t}G
\end{array}
\label{Eq_Derivative_general}%
\end{equation}
is a well-defined linear bounded operator: $\left\Vert \partial_{t}%
G\right\Vert _{L^{2}(\mathbb{R})}\leq\left\Vert G\right\Vert _{W_{1}%
(\mathbb{R})}$. This operator is an extension of the operator defined in
(\ref{Eq_Derivative}).
\end{lemma}

\begin{proof}
By using the Plancherel theorem in $\mathcal{S}(\mathbb{R})$,%
\[
\left\Vert \partial_{t}G\right\Vert _{L^{2}(\mathbb{R})}^{2}=\left\Vert
\widehat{\partial_{t}G}\right\Vert _{L^{2}(\mathbb{R})}^{2}=%
{\displaystyle\int\limits_{\mathbb{R}}}
\left\vert \xi\right\vert ^{2}\left\vert \widehat{G}\left(  \xi\right)
\right\vert ^{2}d\xi\leq\left\Vert G\right\Vert _{W_{1}(\mathbb{R})}^{2}.
\]
Since $\mathcal{S}(\mathbb{R})$ is dense in $W_{1}(\mathbb{R})$, $\partial
_{t}:W_{1}(\mathbb{R})\rightarrow L^{2}(\mathbb{R})$ is a bounded linear operator.
\end{proof}

\subsection{The space of fields}

Let $\left(  \Omega,\Sigma,d\mu\right)  $ be a measure space. We identify
$\Omega$ with the set of neurons of a stochastic continuous NN. A Borel
measurable function $f:\Omega\rightarrow\mathbb{R}$ belong to $L^{2}\left(
\Omega,d\mu\right)  =L^{2}\left(  \Omega\right)  $ if
\[
\left\Vert f\right\Vert _{L^{2}\left(  \Omega\right)  }^{2}:=%
{\displaystyle\int\limits_{\Omega}}
\left\vert f\left(  x\right)  \right\vert ^{2}d\mu\left(  x\right)  <\infty.
\]
$L^{2}\left(  \Omega\right)  $ \ is a real Hilbert space for the inner product%
\[
\left\langle f,g\right\rangle _{L^{2}\left(  \Omega\right)  }:=%
{\displaystyle\int\limits_{\Omega}}
f\left(  x\right)  g(x)d\mu\left(  x\right)  .
\]
We assume that $L^{2}\left(  \Omega\right)  $ is a separable Hilbert space.
Given two real separable Hilbert spaces, $\mathcal{H}_{1}$, $\mathcal{H}_{2}$,
we denote by $\mathcal{H}_{1}%
{\textstyle\bigotimes}
\mathcal{H}_{2}$ its tensor product, which is a separable Hilbert space, see,
e.g., \cite[Section II.4]{Reed-Simon-I}.

\begin{lemma}
\label{Lemma_2}(i) There exists a unitary transformation
\[
U:L^{2}\left(  \Omega,d\mu\right)
{\textstyle\bigotimes}
W_{1}(\mathbb{R})\rightarrow L^{2}\left(  \Omega,d\mu;W_{1}(\mathbb{R}%
)\right)
\]
such that
\[
U\left(  f%
{\textstyle\bigotimes}
g\right)  =fg\text{, for }f\in L^{2}\left(  \Omega,d\mu\right)  \text{, }g\in
W_{1}(\mathbb{R}),
\]
where
\[
L^{2}\left(  \Omega,d\mu;W_{1}(\mathbb{R})\right)  =\left\{  F:\Omega
\rightarrow W_{1}(\mathbb{R})\text{ measurable};\text{ }\left\Vert
F\right\Vert _{L^{2}\left(  \Omega,d\mu;W_{1}(\mathbb{R})\right)  }%
<\infty\right\}  ,
\]%
\[
\left\Vert F\right\Vert _{L^{2}\left(  \Omega,d\mu;W_{1}(\mathbb{R})\right)
}^{2}=%
{\displaystyle\int\limits_{\Omega}}
\left\Vert F(x,\cdot)\right\Vert _{W_{1}(\mathbb{R})}^{2}d\mu\left(  x\right)
.
\]

\noindent(ii) Let $\left\{  e_{n}\left(  x\right)  \right\}  _{n\in\mathbb{N}%
}$ be an orthonormal basis of $L^{2}\left(  \Omega,d\mu\right)  $ and
$\left\{  f_{m}\left(  t\right)  \right\}  _{m\in\mathbb{N}}$ be an
orthonormal basis of $W_{1}(\mathbb{R})$. Then $\left\{  e_{n}\left(
x\right)  f_{m}\left(  t\right)  \right\}  _{n,m\in\mathbb{N}}$ is an
orthonormal basis of $L^{2}\left(  \Omega,d\mu;W_{1}(\mathbb{R})\right)  $.
\end{lemma}

\begin{proof}
The first part follows from the separability of the spaces $L^{2}\left(
\Omega,d\mu\right)  $ and $W_{1}(\mathbb{R})$, by using \cite[Theorem
2.6]{Asao} or \cite[Theorem II.10]{Reed-Simon-I}. The second part is a
well-know result, see \cite[Proposition 2.7]{Asao} or \cite[Section II.4,
Proposition 2]{Reed-Simon-I}.
\end{proof}

\begin{remark}
\label{Nota1A}Using the second part of Lemma \ref{Lemma_2}, the fact that
$\mathcal{S}(\mathbb{R})$ is dense in $W_{1}(\mathbb{R})$, and the
Gram--Schmidt procedure, we can assume that $f_{m}\left(  t\right)
\in\mathcal{S}(\mathbb{R})$. Then any $g(x,t)\in L^{2}\left(  \Omega
,d\mu\right)
{\textstyle\bigotimes}
W_{1}(\mathbb{R})$ admits an expansion of the form
\[
g(x,t)=%
{\displaystyle\sum\limits_{_{n,m\in\mathbb{N}}}}
c_{n,m}e_{n}\left(  x\right)  f_{m}\left(  t\right)  .
\]
For $x$ fixed, $g(x,\cdot)\in W_{1}(\mathbb{R})$, cf. Lemma \ref{Lemma_2}-(i),
now\ by Lemma \ref{Lemma_1},%
\[
\partial_{t}g(x,t)=%
{\displaystyle\sum\limits_{_{n,m\in\mathbb{N}}}}
c_{n,m}e_{n}\left(  x\right)  \partial_{t}f_{m}\left(  t\right)
\]
is a real-valued function.
\end{remark}

Now, since we are assuming that $L^{2}\left(  \Omega,d\mu\right)  $ is
separable, and $L^{2}\left(  \mathbb{R},dt\right)  =L^{2}\left(
\mathbb{R}\right)  $ is separable too, we have%
\begin{equation}
L^{2}\left(  \Omega,d\mu\right)
{\textstyle\bigotimes}
L^{2}\left(  \mathbb{R},dt\right)  \simeq L^{2}\left(  \Omega,d\mu
;L^{2}\left(  \mathbb{R}\right)  \right)  \simeq L^{2}\left(  \Omega
\times\mathbb{R},d\mu dt\right)  , \label{Eq_isomorphism}%
\end{equation}
where $\simeq$ denotes an isomorphism of Hilbert spaces. The inner product of
the last space is given by%
\[
\left\langle f\left(  x,t\right)  ,g\left(  x,t\right)  \right\rangle =%
{\displaystyle\int\limits_{\Omega}}
{\displaystyle\int\limits_{\mathbb{R}}}
f\left(  x,t\right)  ,g\left(  x,t\right)  dtd\mu\left(  x\right)  ,
\]
cf. \cite[Theorem II.10]{Reed-Simon-I}. We denote the corresponding norm as
$\left\Vert \cdot\right\Vert $. there Furthermore, this inclusion is
continuous due to the following result.

\begin{notation}
Along this article, we use the notation $\left\langle \cdot,\cdot\right\rangle
$, respectively $\left\Vert \cdot\right\Vert $, to denote the inner product,
respectively the norm, of $L^{2}\left(  \Omega\times\mathbb{R},d\mu dt\right)
=L^{2}\left(  \Omega\times\mathbb{R}\right)  $. For any other Hilbert space,
we will put a subindex in the $\left\langle \cdot,\cdot\right\rangle $ and in
the $\left\Vert \cdot\right\Vert $.
\end{notation}

\begin{lemma}
\label{Lemma_3}There exits a continuous inclusion%
\[
\mathfrak{i}:\mathcal{H}=L^{2}\left(  \Omega,d\mu;W_{1}(\mathbb{R})\right)
\rightarrow L^{2}\left(  \Omega,d\mu;L^{2}\left(  \mathbb{R}\right)  \right)
\simeq L^{2}\left(  \Omega\times\mathbb{R},d\mu dt\right)  ,
\]
i.e., $\left\Vert \mathfrak{i}\left(  \boldsymbol{h}\right)  \right\Vert
\leq\left\Vert \boldsymbol{h}\right\Vert _{\mathcal{H}}$ for any
$\boldsymbol{h}\in\mathcal{H}$. Furthermore, $\left(  \mathfrak{i}\left(
\mathcal{H}\right)  ,\left\langle \cdot,\cdot\right\rangle \right)  $ is a
dense vector subspace of $L^{2}\left(  \Omega\times\mathbb{R}\right)  $.
\end{lemma}

\begin{proof}
The existence of the inclusion $\mathfrak{i}$ follows from the isomorphism
(\ref{Eq_isomorphism}) and Remark \ref{Nota1A}. Now, by using the fact that
the Fourier transform is a unitary operator in $L^{2}\left(  \mathbb{R}%
\right)  $,
\begin{gather*}
\left\Vert \boldsymbol{h}\right\Vert _{\mathcal{H}}^{2}=%
{\displaystyle\int\limits_{\Omega}}
\left\Vert h\left(  x,\cdot\right)  \right\Vert _{W_{1}\left(  \mathbb{R}%
\right)  }^{2}d\mu\left(  x\right)  =%
{\displaystyle\int\limits_{\Omega}}
{\displaystyle\int\limits_{\mathbb{R}}}
\left(  1+\left\vert \xi\right\vert ^{2}\right)  \left\vert \mathcal{F}%
_{t\rightarrow\xi}\left(  \boldsymbol{h}\left(  x,t\right)  \right)
\right\vert ^{2}dtd\mu\left(  x\right)  \geq\\%
{\displaystyle\int\limits_{\Omega}}
\left\{  \text{ }%
{\displaystyle\int\limits_{\mathbb{R}}}
\left\vert \mathcal{F}_{t\rightarrow\xi}\left(  \boldsymbol{h}\left(
x,t\right)  \right)  \right\vert ^{2}dt\right\}  d\mu\left(  x\right)  =%
{\displaystyle\int\limits_{\Omega}}
\left\{  \text{ }%
{\displaystyle\int\limits_{\mathbb{R}}}
\left\vert \boldsymbol{h}\left(  x,t\right)  \right\vert ^{2}dt\right\}
d\mu\left(  x\right)  =\left\Vert \mathfrak{i}\left(  \boldsymbol{h}\right)
\right\Vert .
\end{gather*}
Finally, the density follows from the fact that $\mathfrak{i}\left(
\mathcal{H}\right)  $ contains an orthonormal basis of $L^{2}\left(
\Omega\times\mathbb{R}\right)  $,
\end{proof}

\begin{remark}
\label{Note_act_fun}Along the article, we assume that $\phi:\mathbb{R}%
\rightarrow\mathbb{R}$ is a bounded, differentiable, Lipschitz function, with
$\phi\left(  0\right)  =0$. We recall that a\ function $\phi:\mathbb{R}%
\rightarrow\mathbb{R}$ is called a Lipschitz function if there exists a real
constant $L(\phi)>0$ such that, for all $x,y\in\mathbb{R}$, $|\phi
(x)-\phi(y)|\leq L(\phi)|x-y|$. We note that $\tanh(x)$ is a function
satisfying all the mentioned requirements.
\end{remark}

\begin{lemma}
\label{Lemma_4}Assume that $\boldsymbol{J}\left(  x,y\right)  \in L^{2}\left(
\Omega^{2},d\mu d\mu\right)  =L^{2}\left(  \Omega^{2}\right)  $, then%
\[%
\begin{array}
[c]{ccc}%
L^{2}\left(  \Omega\times\mathbb{R},d\mu dt\right)  & \rightarrow &
L^{2}\left(  \Omega\times\mathbb{R},d\mu dt\right) \\
&  & \\
g\left(  y,t\right)  & \rightarrow &
{\displaystyle\int\limits_{\Omega}}
\boldsymbol{J}\left(  x,y\right)  \phi\left(  g\left(  y,t\right)  \right)
d\mu\left(  x\right)
\end{array}
\]
is a well-defined linear bounded operator.
\end{lemma}

\begin{proof}
By using $|\phi(y)-\phi(0)|\leq L(\phi)|y|$, and the Cauchy--Schwarz
inequality,
\begin{gather*}
\left\Vert
{\displaystyle\int\limits_{\Omega}}
\boldsymbol{J}\left(  x,y\right)  \phi\left(  g\left(  y,t\right)  \right)
d\mu\left(  y\right)  \right\Vert ^{2}=%
{\displaystyle\int\limits_{\Omega}}
{\displaystyle\int\limits_{\mathbb{R}}}
\left\vert \text{ }%
{\displaystyle\int\limits_{\Omega}}
\boldsymbol{J}\left(  x,y\right)  \phi\left(  g\left(  y,t\right)  \right)
d\mu\left(  y\right)  \right\vert ^{2}dtd\mu\left(  x\right) \\
\leq%
{\displaystyle\int\limits_{\Omega}}
{\displaystyle\int\limits_{\mathbb{R}}}
\left\{  \text{ }%
{\displaystyle\int\limits_{\Omega}}
\left\vert \boldsymbol{J}\left(  x,y\right)  \right\vert \left\vert
\phi\left(  g\left(  y,t\right)  \right)  \right\vert d\mu\left(  y\right)
\right\}  ^{2}dtd\mu\left(  x\right) \\
\leq%
{\displaystyle\int\limits_{\Omega}}
{\displaystyle\int\limits_{\mathbb{R}}}
\left\{  L(\phi)%
{\displaystyle\int\limits_{\Omega}}
\left\vert \boldsymbol{J}\left(  x,y\right)  \right\vert \left\vert g\left(
y,t\right)  \right\vert d\mu\left(  y\right)  \right\}  ^{2}dtd\mu\left(
x\right) \\
\leq L(\phi)^{2}%
{\displaystyle\int\limits_{\Omega}}
{\displaystyle\int\limits_{\mathbb{R}}}
\left\{  \text{ }%
{\displaystyle\int\limits_{\Omega}}
\left\vert \boldsymbol{J}\left(  x,y\right)  \right\vert ^{2}d\mu\left(
y\right)
{\displaystyle\int\limits_{\Omega}}
\left\vert g\left(  y,t\right)  \right\vert ^{2}d\mu\left(  y\right)
\right\}  dtd\mu\left(  x\right) \\
\leq L(\phi)^{2}%
{\displaystyle\int\limits_{\Omega^{2}}}
\left\vert \boldsymbol{J}\left(  x,y\right)  \right\vert ^{2}d\mu\left(
y\right)  d\mu\left(  x\right)
{\displaystyle\int\limits_{\mathbb{R}}}
{\displaystyle\int\limits_{\Omega}}
\left\vert g\left(  y,t\right)  \right\vert ^{2}d\mu\left(  y\right)  dt\\
=L(\phi)^{2}\left\Vert \boldsymbol{J}\right\Vert _{L^{2}\left(  \Omega
^{2}\right)  }^{2}\left\Vert g\right\Vert ^{2}.
\end{gather*}

\end{proof}

\begin{proposition}
\label{Porp1}Assume that $\boldsymbol{J}\left(  x,y\right)  \in L^{2}\left(
\Omega^{2}\right)  $, $\boldsymbol{j}\left(  x,t\right)  $, $\widetilde
{\boldsymbol{j}}\left(  x,t\right)  \in\mathcal{H}$. Set%
\[
S\left[  \boldsymbol{h},\widetilde{\boldsymbol{h}},\boldsymbol{j}%
,\widetilde{\boldsymbol{j}}\right]  :=S_{0}\left[  \boldsymbol{h}%
,\widetilde{\boldsymbol{h}}\right]  -\left\langle \widetilde{\boldsymbol{h}},%
{\displaystyle\int\limits_{\Omega}}
\boldsymbol{J}\left(  x,y\right)  \phi\left(  \boldsymbol{h}\left(
y,t\right)  \right)  d\mu\left(  x\right)  \right\rangle +\left\langle
\boldsymbol{j},\boldsymbol{h}\right\rangle +\left\langle \widetilde
{\boldsymbol{j}},\widetilde{\boldsymbol{h}}\right\rangle ,
\]
where $S_{0}\left[  \boldsymbol{h},\widetilde{\boldsymbol{h}}\right]  $ is
given in (\ref{action}). Then, the mapping%
\[%
\begin{array}
[c]{ccc}%
\mathcal{H\times H} & \rightarrow & \mathbb{R}\\
&  & \\
\left(  \boldsymbol{h},\widetilde{\boldsymbol{h}}\right)  & \rightarrow &
\exp\left(  S\left[  \boldsymbol{h},\widetilde{\boldsymbol{h}},\boldsymbol{j}%
,\widetilde{\boldsymbol{j}}\right]  \right)  ,
\end{array}
\]
is well-defined.
\end{proposition}

\begin{proof}
The result follows from the Lemmas \ref{Lemma_3} and \ref{Lemma_4}.
\end{proof}

\section{\label{Appen_B}Appendix B}

In this section, we review some basic results on white noise calculus required
here, see, e.g., \cite{Gelfand-Vilenkin}, \cite{Hida et al}, \cite{Huang et
al}, \cite{Obata}.

\subsection{Countably-Hilbert nuclear spaces}

Let $\mathcal{X}$ be a locally convex topological space whose topology is
generated by a sequence of consistent seminorms $\left\{  \left\Vert
\cdot\right\Vert _{n}\right\}  _{n\in\mathbb{N}}$, where each seminorm
$\left\Vert \cdot\right\Vert _{n}$ is induced by an inner product
$\left\langle \cdot,\cdot\right\rangle _{n}$ in $\mathcal{X}$. Let
$\mathcal{X}_{n}$ be the completion of $\mathcal{X}$ with respect to
$\left\Vert \cdot\right\Vert _{n}$. Without loss of generality, we assume that
the norms satisfy $\left\Vert \cdot\right\Vert _{n}\leq\left\Vert
\cdot\right\Vert _{m}$ for $m\geq n$; so there exists a continuous and dense
embedding $i_{m,n}:\mathcal{X}_{m}\rightarrow\mathcal{X}_{n}$. The space
$\mathcal{X}$ is the projective limit of $\left\{  \mathcal{X}_{n}\right\}
_{n\in\mathbb{N}}$, more precisely, $\mathcal{X}=\cap_{n\in\mathbb{N}%
}\mathcal{X}_{n}$, and%
\[
\mathcal{X}_{0}\mathcal{\supset X}_{1}\mathcal{\supset\cdots\supset X}%
_{n}\mathcal{\supset X}_{n+1}\cdots\mathcal{\supset X}\text{.}%
\]

The topological space $\left(  \mathcal{X},\tau\right)  $, where $\tau
$\ denotes the projective limit topology is called a countably-Hilbert space.
This space is complete and metrizable, indeed, the topology induced by the
metric%
\[
d(x,y)=%
{\displaystyle\sum\limits_{n\in\mathbb{N}}}
2^{-n}\frac{\left\Vert x-y\right\Vert _{n}}{1+\left\Vert x-y\right\Vert _{n}%
}\text{, }x,y\in\mathcal{X},
\]
agrees with $\tau$.

A countably-Hilbert space $\mathcal{X}$ is called nuclear, if for every $n$
there exists $m\geq n$ such that $i_{m,n}:\mathcal{X}_{m}\rightarrow
\mathcal{X}_{n}$ is trace class, i.e. $i_{m,n}$ \ has the form%
\[
i_{m,n}\left(  \varphi\right)  =%
{\displaystyle\sum\limits_{k=0}^{\infty}}
\lambda_{k}\left\langle \varphi,\varphi_{k}\right\rangle \psi_{k}\text{, for
}\varphi\in\mathcal{X}_{m},
\]
where $\left\{  \varphi_{k}\right\}  _{k\in\mathbb{N}}$ and $\left\{  \psi
_{k}\right\}  _{k\in\mathbb{N}}$ are orthonormal systems of vectors in the
spaces $\mathcal{X}_{m}$, $\mathcal{X}_{n}$ respectively, $\lambda_{k}>0$, and
the series $\sum_{n\in\mathbb{N}}\lambda_{k}$ converges.

Let $\mathcal{X}_{n}^{\prime}$ be the topological dual of $\mathcal{X}_{n}$,
then $\mathcal{X}^{\prime}=\cup_{n\in\mathbb{N}}\mathcal{X}_{n}^{\prime}$ and
\[
\mathcal{X\subset X}_{1}^{\prime}\subset\cdots\subset\mathcal{X}_{n}^{\prime
}\subset\mathcal{X}_{n+1}^{\prime}\subset\cdots\subset\mathcal{X}^{\prime
}\text{.}%
\]
The space $\mathcal{X}^{\prime}$ is the inductive limit of $\left\{
\mathcal{X}_{n}^{\prime}\right\}  _{n\in\mathbb{N}}$, then, it has naturally
the topology of the inductive limit. This space also has other topologies, the
strong and the weak; all these topologies\ are equivalent if $\mathcal{X}$ is
nuclear. The $\sigma$-algebras on $\mathcal{X}^{\prime}$ for all these
topologies are the same. This $\sigma$-algebra is regarded as the Borel
$\sigma$-algebra of $\mathcal{X}^{\prime}$.

\subsection{The Bochner--Minlos theorem}

Let $\mathcal{X}$ be a real nuclear space, $\mathcal{X}^{\prime}$ its
topological dual and $\left(  \cdot,\cdot\right)  $ the canonical bilinear
form (the pairing) on $\mathcal{X}^{\prime}\times\mathcal{X}$. Let
$\mathcal{B}$ be the cylindrical $\sigma$-algebra on $\mathcal{X}^{\prime}$,
i.e. , the smallest $\sigma$-algebra \ such that the functions%
\begin{equation}
T\rightarrow\left(  \left(  T,x_{1}\right)  ,\ldots,\left(  T,x_{n}\right)
\right)  \in\mathbb{R}^{n}\text{, }T\in\mathcal{X}^{\prime}\text{,}
\label{Cylindrical}%
\end{equation}
are measurable for any choice of $x_{1},\ldots,x_{n}\in\mathcal{X}$ and
$n=1,2,\ldots$, where $\mathbb{R}^{n}$ is equipped with the Borel $\sigma
$-algebra. This is the $\sigma$-algebra corresponding to the weak topology on
$\mathcal{X}^{\prime}$.

A $\mathbb{C}$-valued function $\mathcal{C}$ on $\mathcal{X}$ is called a
characteristic functional if (i) $\mathcal{C}$\ is continuous; (ii)
$\mathcal{C}$ is positive definite, i.e.%
\[%
{\displaystyle\sum\limits_{j,k=1}^{m}}
\alpha_{j}\overline{\alpha_{k}}\mathcal{C}\left(  \varphi_{j}-\varphi
_{k}\right)  \geq0
\]
for any choice of $\alpha_{1},\ldots,\alpha_{m}\in\mathbb{C}$, $\varphi
_{1},\ldots,\varphi_{m}\in\mathcal{X}$ and $n=1,2,\ldots$; (iii)
$\mathcal{C}\left(  0\right)  =1$.

The Bochner-Minlos theorem establishes a one-to-one correspondence between
characteristic functions and probability measures on $\left(  \mathcal{X}%
^{\prime},\mathcal{B}\right)  $. If $\boldsymbol{P}$ is a probability measure
on $\mathcal{X}^{\prime}$, its Fourier transform%
\[
\widehat{\boldsymbol{P}}\left(  \varphi\right)  =%
{\displaystyle\int\limits_{\mathcal{X}^{\prime}}}
e^{i\left(  T,\varphi\right)  }d\boldsymbol{P}(T)
\]
is a characteristic function.

\subsection{\label{Some_int_funct}Some classes of integrable functions}

For $\sigma\in\left[  1,\infty\right)  $, we consider the space $L^{\sigma
}(\mathcal{X}^{\prime},\mathcal{B},d\boldsymbol{P})$. Given $\varphi
\in\mathcal{X}$, $T\in\mathcal{X}^{\prime}$, we set $T_{\varphi}=\left\langle
T,\varphi\right\rangle $. Then%
\[
e^{\alpha T_{\varphi}}\in L^{\sigma}(\mathcal{X}^{\prime},\mathcal{B}%
,d\boldsymbol{P})\text{ for any }\alpha\in\mathbb{C}\text{, }\sigma\in\left[
1,\infty\right)  \text{,}%
\]
\cite[Proposition 1.7]{Hida et al}. Let $Q\left(  x_{1},\cdots,x_{n}\right)  $
be a polynomial in $n$ variables with complex coefficients. Then,%
\[
Q\left(  T_{\varphi_{1}},\cdots,T_{\varphi_{n}}\right)  \in L^{\sigma
}(\mathcal{X}^{\prime},\mathcal{B},d\boldsymbol{P})\text{ for any }\alpha
\in\mathbb{C}\text{, }\sigma\in\left[  1,\infty\right)  \text{,}%
\]
\cite[Proposition 1.6]{Hida et al}.

Let $\mathcal{A}$ be the complex algebra generated by $Q\left(  T_{\varphi
_{1}},\cdots,T_{\varphi_{n}}\right)  $ and $e^{\alpha T_{\varphi}}$. Then,
$\mathcal{A}$ is dense in $L^{\sigma}(\mathcal{X}^{\prime},\mathcal{B}%
,d\boldsymbol{P})$, for , $\sigma\in\left[  1,\infty\right)  $, \cite[Theorem
1.9]{Hida et al}.

\section{\label{Appen_C}Appendix C}

\subsection{\label{Section_Gelfand_triplet}A Gel'fand triple}

Let $\mathcal{H}$ \ be a real separable Hilbert space. We assume the existence
of a Gel'fand triple $\left(  \Phi,\mathcal{H},\Phi^{\prime}\right)  $, where
$\Phi$ is real nuclear space, and $\Phi^{\prime}$ is its topological dual. In
addition, $\Phi$ is a dense subspace of $\mathcal{H}$ and the inclusion
$\mathfrak{i}:\Phi\rightarrow\mathcal{H}$ is continuous. By identifying
$\mathcal{H}$ with its topological dual $\mathcal{H}^{\prime}$,\ there is an
inclusion $\mathfrak{i}^{\prime}:\mathcal{H}^{\prime}\rightarrow\Phi^{\prime}%
$. Then, $\mathcal{H}$ is also a dense subspace of $\Phi^{\prime}$ and the
inclusion $\mathfrak{i}^{\prime}:\mathcal{H}\rightarrow\Phi^{\prime}$ is
continuous. The duality pairing between $\Phi$ and$\ \Phi^{\prime}$ is
compatible with the inner product in $\mathcal{H}$, in the sense that
$\Phi^{\prime}\times\Phi$
\[
\left(  v,u\right)  _{\Phi\times\Phi^{\prime}}=\left\langle v,u\right\rangle
\text{, for }v\in\Phi\subset\mathcal{H}\text{, }u\in\mathcal{H}=\mathcal{H}%
^{\prime}\subset\Phi^{\prime}.
\]
Any Countably-Hilbert nuclear space is reflexive. Then the paring $\left(
T,\varphi\right)  _{\Phi\times\Phi^{\prime}}$, $T\in\Phi^{\prime}$,
$\varphi\in\Phi$, gives rise a topological isomorphism between $\Phi$ and
$\Phi^{\prime}$, so $\left(  \cdot,\varphi\right)  \in\Phi^{\prime}$, for any
$\varphi\in\Phi$. Given $T\in\mathcal{X}^{\prime}=\cup_{n\in\mathbb{N}%
}\mathcal{X}_{n}^{\prime}$, it verifies that $T\in\mathcal{X}_{m}^{\prime}$,
and $\left(  T,\varphi\right)  _{\Phi\times\Phi^{\prime}}=\left\langle
T,\mathfrak{i}\left(  \varphi\right)  \right\rangle _{m}$, for some $m$, which
means that $\mathfrak{i}\left(  \mathcal{X}\right)  \subset H$ is dense in
$\Phi^{\prime}$. For further details about this Gel'fand triplet, the reader
may consult \cite[Chapter I, Section 4.3, \ Theorem 1]{Gelfand-Vilenkin}.

We now take $\mathcal{H}=L^{2}\left(  \Omega,d\mu\right)
{\textstyle\bigotimes}
W_{1}(\mathbb{R})$ and assume the existence of a Gel'fand triple $\left(
\Phi,\mathcal{H},\Phi^{\prime}\right)  $ as above. We fix two characteristic
functions $\mathcal{C}$, $\widetilde{\mathcal{C}}:\Phi\rightarrow\mathbb{C}$,
then by the Bocher-Minlos theorem, there exist two probability measures
$d\boldsymbol{P}\left(  \boldsymbol{h}\right)  $, $d\widetilde{\boldsymbol{P}%
}\left(  \widetilde{\boldsymbol{h}}\right)  $ in $\left(  \Phi^{\prime
},\mathcal{B}\right)  $. We restrict each these probability measures to
$\mathcal{H}$. Note that the $\sigma$-algebra of each restriction is the
family of subsets for the form $\mathcal{H}\cap B$, with $B\in\mathcal{B}$. By
abuse of language, we denote this $\sigma$-algebra as $\mathcal{B}$.

\begin{lemma}
\label{Lemma_5A}(i) Assume that $\Phi\hookrightarrow\mathcal{H}\hookrightarrow
\Phi^{\prime}$ is a Gel'fand triplet, and that $\boldsymbol{P}$ is a
probability measure on $\left(  \Phi^{\prime},\mathcal{B}\right)  $, where
$\mathcal{B}$ is the cylindrical $\sigma$-algebra on $\Phi^{\prime}$. Then,
the restriction of $\boldsymbol{P}$\ to $\mathcal{H}$ is a probability
measure, i.e.
\[%
{\displaystyle\int\limits_{\mathcal{H}}}
d\boldsymbol{P}\left(  \boldsymbol{h}\right)  =1.
\]
(ii) (i) Assume that $\Psi\hookrightarrow L^{2}\left(  \Omega\right)
{\textstyle\bigotimes}
L^{2}\left(  \mathbb{R}\right)  \hookrightarrow\Psi^{\prime}$ is a Gel'fand
triplet, and that $\boldsymbol{Q}$ is a probability measure on $\left(
\Psi^{\prime},\mathcal{B}\right)  $, where $\mathcal{B}$ is the cylindrical
$\sigma$-algebra on $\Psi^{\prime}$. Then,%
\[%
{\displaystyle\int\limits_{L^{2}\left(  \Omega\right)  {\bigotimes}%
L^{2}\left(  \mathbb{R}\right)  }}
d\boldsymbol{Q}\left(  \boldsymbol{h}\right)  =%
{\displaystyle\int\limits_{\mathcal{H}}}
d\boldsymbol{Q}\left(  \boldsymbol{h}\right)  =1.
\]
Which means that the restriction of $\mathcal{H}$ of $\boldsymbol{Q}$ is a
probability measure.
\end{lemma}

\begin{proof}
(i) Since the inclusion is continuous and dense, the characteristic function
$1_{\mathcal{H}}$ is continuous and admits a unique extension to $\Phi
^{\prime}$, which is $1_{\Phi^{\prime}}$, thus,
\[%
{\displaystyle\int\limits_{\mathcal{H}}}
d\boldsymbol{P}\left(  \boldsymbol{h}\right)  =%
{\displaystyle\int\limits_{\Phi^{\prime}}}
{\large 1}_{\mathcal{H}}\left(  \boldsymbol{h}\right)  d\boldsymbol{P}\left(
\boldsymbol{h}\right)  =%
{\displaystyle\int\limits_{\Phi^{\prime}}}
{\large 1}_{\Phi^{\prime}}\left(  \boldsymbol{h}\right)  d\boldsymbol{P}%
\left(  \boldsymbol{h}\right)  =1.
\]
(ii) It follows from Lemma \ref{Lemma_3} by the reasoning given in the first part.
\end{proof}

By using the first part of this result, the product measure $d\boldsymbol{P}%
\left(  \boldsymbol{h}\right)  d\widetilde{\boldsymbol{P}}\left(
\widetilde{\boldsymbol{h}}\right)  $, on the product $\sigma$-algebra
$\mathcal{B\times B}$, is a probability measure on $\mathcal{H}\times
\mathcal{H}$.

\subsection{\label{Sect_examples}Some examples of probability measures}

We now use the Gel'fand triplet $\left(  \Phi,\mathcal{H},\Phi^{\prime
}\right)  $ to construct some specific measures. For $\sigma>0$, we set%
\begin{equation}
\mathcal{C}_{\sigma}\left(  \varphi\right)  =\exp\left(  \frac{-\sigma^{2}}%
{2}\left\Vert \varphi\right\Vert _{\mathcal{H}}^{2}\right)  \text{, for
}\varphi\in\Phi. \label{Eq_Char_fun}%
\end{equation}
Then, $\mathcal{C}_{\sigma}\left(  \varphi\right)  $ is a characteristic
functional, cf. \cite[Lemma 2.1.1]{Obata}. The probability measure with this
characteristic functional is called the Gaussian measure with variance
$\sigma^{2}$ and mean zero. Now, let $\tau:\Phi\rightarrow\Phi$ be a linear,
homeomorphism. Then $\mathcal{C}_{\sigma}\left(  \tau\left(  \varphi\right)
\right)  $ is also a characteristic functional. The corresponding probability
measure is non-Gaussian if $\left\Vert \tau\left(  \varphi\right)  \right\Vert
_{\mathcal{H}}^{2}\neq\left\Vert \varphi\right\Vert _{\mathcal{H}}^{2}$.

On the other hand, if $\mathcal{C}_{1}$, $\mathcal{C}_{2}$ are two
characteristic functionals on $\Phi$ not necessarily of the form
(\ref{Eq_Char_fun}), and $\alpha$, $\beta$ are positive real numbers, then
$\alpha\mathcal{C}_{1}+\beta\mathcal{C}_{2}$ is a characteristic functional on
$\Phi$. The characteristic functionals on $\Phi$\ form a cone. So, in
principle, starting with the Gaussian measure corresponding to $\mathcal{C}%
_{\sigma}\left(  \varphi\right)  $, we can construct infinitely many
non-Gaussian measures in $\left(  \Phi^{\prime},\mathcal{B}\right)  $.

\subsection{Generating functional with Cutoff}

We set $S\left[  \boldsymbol{h},\widetilde{\boldsymbol{h}},\boldsymbol{j}%
,\widetilde{\boldsymbol{j}}\right]  =S_{0}\left[  \boldsymbol{h}%
,\widetilde{\boldsymbol{h}}\right]  -S_{int}\left[  \boldsymbol{h}%
,\widetilde{\boldsymbol{h}};\boldsymbol{J}\right]  +\left\langle
\boldsymbol{j},\boldsymbol{h}\right\rangle +\left\langle \widetilde
{\boldsymbol{j}},\widetilde{\boldsymbol{h}}\right\rangle $. The Problem
\ref{P4} posses the question:%
\[
\exp\left(  S\left[  \boldsymbol{h},\widetilde{\boldsymbol{h}},\boldsymbol{j}%
,\widetilde{\boldsymbol{j}}\right]  \right)  \in L^{1}\left(  \mathcal{H}%
\times\mathcal{H},\mathcal{B}\times\mathcal{B},d\boldsymbol{P}\left(
\boldsymbol{h}\right)  d\widetilde{\boldsymbol{P}}\left(  \widetilde
{\boldsymbol{h}}\right)  \right)  ?
\]
To tackle this problem, we consider $\mathcal{H}\times\mathcal{H}$ as a
Hilbert space with respect to the inner product
\begin{equation}
\left\langle \left(  f,g\right)  ,\left(  f_{1},g_{1}\right)  \right\rangle
_{\mathcal{H}}:=\left\langle f,f_{1}\right\rangle _{\mathcal{H}}+\left\langle
g,g_{1}\right\rangle _{\mathcal{H}}, \label{inner_product_1}%
\end{equation}
and construct the algebra $\mathcal{A}$ defined in Section
\ref{Some_int_funct}, which consists of integrable functions. We note that
$\exp\left(  S\left[  \boldsymbol{h},\widetilde{\boldsymbol{h}},\boldsymbol{j}%
,\widetilde{\boldsymbol{j}}\right]  \right)  \notin\mathcal{A}$. So, at the
moment, we do not know if $\exp\left(  S\left[  \boldsymbol{h},\widetilde
{\boldsymbol{h}},\boldsymbol{j},\widetilde{\boldsymbol{j}}\right]  \right)  $
is integrable. For this reason, we are forced to introduce a cutoff in the
generating functional.

We set%
\[
\mathcal{P}_{M}:=\left\{  \left(  \boldsymbol{h},\widetilde{\boldsymbol{h}%
}\right)  \in\mathcal{H}\times\mathcal{H};\left\Vert \boldsymbol{h}\right\Vert
_{\mathcal{H}},\left\Vert \widetilde{\boldsymbol{h}}\right\Vert _{\mathcal{H}%
}<M\right\}  ,
\]
where $M$ is a positive constant, and denote by ${\large 1}_{\mathcal{P}_{M}%
}\left(  \boldsymbol{h},\widetilde{\boldsymbol{h}}\right)  $ the
characteristic function of $\mathcal{P}_{M}$.

\begin{lemma}
\label{Lemma_5}${\large 1}_{\mathcal{P}_{M}}\left(  \boldsymbol{h}%
,\widetilde{\boldsymbol{h}}\right)  \exp\left(  S\left[  \boldsymbol{h}%
,\widetilde{\boldsymbol{h}},\boldsymbol{j},\widetilde{\boldsymbol{j}}\right]
\right)  \in L^{1}\left(  \mathcal{H}\times\mathcal{H},\mathcal{B}%
\times\mathcal{B},d\boldsymbol{h}d\widetilde{\boldsymbol{h}}\right)  $.
\end{lemma}

\begin{proof}
We first show that
\begin{equation}
{\large 1}_{\mathcal{P}_{M}}\left(  \boldsymbol{h},\widetilde{\boldsymbol{h}%
}\right)  \exp\left(  S\left[  \boldsymbol{h},\widetilde{\boldsymbol{h}%
},\boldsymbol{j},\widetilde{\boldsymbol{j}}\right]  \right)  :\left(
\mathcal{H}\times\mathcal{H},\mathcal{B}\times\mathcal{B}\right)
\rightarrow\left(  \mathbb{R},\mathcal{B}\left(  \mathbb{R}\right)  \right)
\label{Function_question}%
\end{equation}
is a measurable function, where $\mathcal{H}\times\mathcal{H}$ is Hilbert
space with inner product (\ref{inner_product_1}), where $\mathcal{B}\left(
\mathbb{R}\right)  $ is the Borel $\sigma$-algebra on $\mathbb{R}$, and
$\mathcal{B}$\ denotes the restriction of the cylindrical $\sigma$-algebra on
$\Phi^{\prime}$ to $\mathcal{H}$; this restriction agrees with the Borel
$\sigma$-algebra on $\mathcal{H}$, cf. \cite[Theorem 4.1]{Kukush}, and thus
$\mathcal{B}\times\mathcal{B}$ is the the Borel $\sigma$-algebra on
$\mathcal{H}\times\mathcal{H}$. Then, it is sufficient to show that function
(\ref{Function_question}) is continuous. The continuity of $\left\langle
\boldsymbol{j},\boldsymbol{h}\right\rangle +\left\langle \widetilde
{\boldsymbol{j}},\widetilde{\boldsymbol{h}}\right\rangle $ follows from the
fact that $\mathcal{H\hookrightarrow}L^{2}(\Omega\times\mathbb{R})$, cf. Lemma
\ref{Lemma_3}, \ an the continuity of the bilinear form $\left\langle
\cdot,\cdot\right\rangle $. The continuity of the term $S_{int}\left[
\boldsymbol{h},\widetilde{\boldsymbol{h}}\right]  $\ follows from Lemmas
\ref{Lemma_3}, \ref{Lemma_4}, and finally the continuity of the term
$S_{0}\left[  \boldsymbol{h},\widetilde{\boldsymbol{h}}\right]  $ follows from
Lemmas \ref{Lemma_1}, \ref{Lemma_3}.

We now show that function (\ref{Function_question}) is bounded. By the
Cauchy-Schwarz inequality and Lemma \ref{Lemma_4},%
\begin{align*}
\left\vert S\left[  \boldsymbol{h},\widetilde{\boldsymbol{h}},\boldsymbol{j}%
,\widetilde{\boldsymbol{j}}\right]  \right\vert  &  \leq\left(  1+\gamma
\right)  \left\Vert \widetilde{\boldsymbol{h}}\right\Vert \left\Vert
\partial_{t}\boldsymbol{h}\right\Vert +\frac{1}{2}\sigma^{2}\left\Vert
\widetilde{\boldsymbol{h}}\right\Vert \left\Vert \boldsymbol{h}\right\Vert \\
&  +L(\phi)\left\Vert \boldsymbol{J}\right\Vert _{L^{2}\left(  \Omega
\times\Omega\right)  }\left\Vert \widetilde{\boldsymbol{h}}\right\Vert
\left\Vert \boldsymbol{h}\right\Vert +\left\Vert \boldsymbol{j}\right\Vert
\left\Vert \boldsymbol{h}\right\Vert +\left\Vert \widetilde{\boldsymbol{j}%
}\right\Vert \left\Vert \widetilde{\boldsymbol{h}}\right\Vert .
\end{align*}
Then, by taking $\Gamma:=3+\gamma+\frac{1}{2}\sigma^{2}+L(\phi)\left\Vert
\boldsymbol{J}\right\Vert _{L^{2}\left(  \Omega\times\Omega\right)  }$, we
have
\[
{\large 1}_{\mathcal{P}_{M}}\left(  \boldsymbol{h},\widetilde{\boldsymbol{h}%
}\right)  \exp\left(  S\left[  \boldsymbol{h},\widetilde{\boldsymbol{h}%
},\boldsymbol{j},\widetilde{\boldsymbol{j}}\right]  \right)  \leq\exp\left(
\Gamma M^{2}\right)  {\large 1}_{\mathcal{P}_{M}}\left(  \boldsymbol{h}%
,\widetilde{\boldsymbol{h}}\right)  .
\]

\end{proof}

As a consequence of Lemma \ref{Lemma_5} and Proposition \ref{Porp1}, we have
the following result:

\begin{theorem}
\label{Porp2}Assume that $\boldsymbol{J}\left(  x,y\right)  \in L^{2}\left(
\Omega\times\Omega\right)  $, $\boldsymbol{j}\left(  x,t\right)  $,
$\widetilde{\boldsymbol{j}}\left(  x,t\right)  \in\mathcal{H}$. Then, the
functional%
\begin{equation}
{\LARGE Z}_{M}\left(  \boldsymbol{j},\widetilde{\boldsymbol{j}};\boldsymbol{J}%
\right)  :=%
{\displaystyle\iint\limits_{\mathcal{H}\times\mathcal{H}}}
{\large 1}_{\mathcal{P}_{M}}\left(  \boldsymbol{h},\widetilde{\boldsymbol{h}%
}\right)  \exp\left(  S\left[  \boldsymbol{h},\widetilde{\boldsymbol{h}%
},\boldsymbol{j},\widetilde{\boldsymbol{j}}\right]  \right)  d\boldsymbol{P}%
\left(  \boldsymbol{h}\right)  d\widetilde{\boldsymbol{P}}\left(
\widetilde{\boldsymbol{h}}\right)  \label{Eq_generrating_func_3}%
\end{equation}
is well-defined.
\end{theorem}

\subsection{\label{Section_Gelfand_triplets}Gel'fand triplets for
Archimedean/non-Archimedean NNs}

As we said before, we need a space of continuous test functions $\mathcal{D}%
(\Omega)$ on the space of neurons $\Omega$\ such that $\mathcal{D}%
(\Omega)\subset L^{2}\left(  \Omega\right)  \subset\mathcal{D}^{\prime}%
(\Omega)$ is a Gel'fand triplet. The fields $\boldsymbol{h}\left(  x,t\right)
$, $\widetilde{\boldsymbol{h}}\left(  x,t\right)  $ are functions from
$L^{2}\left(  \Omega\right)  \otimes W_{1}(\mathbb{R})\subset L^{2}\left(
\Omega\right)  \otimes L^{2}(\mathbb{R})\simeq L^{2}\left(  \Omega
\times\mathbb{R}\right)  $, where $\otimes$ denotes the Hilbert tensor product
of separable Hilbert spaces.\ Using the fact that $\mathcal{S}(\mathbb{R})$ is
a nuclear space, and that the complete projective tensor product (denoted as
$\widehat{%
{\textstyle\bigotimes}
}_{\pi}$)\ of two nuclear spaces is nuclear, \cite[Proposition 1.3.8]{Obata};
it seems reasonable to conjecture that
\begin{equation}
\mathcal{D}(\Omega)\widehat{%
{\textstyle\bigotimes}
}_{\pi}\mathcal{S}(\mathbb{R})\subset L^{2}\left(  \Omega\right)  \otimes
L^{2}(\mathbb{R})\subset\left(  \mathcal{D}(\Omega)\widehat{%
{\textstyle\bigotimes}
}_{\pi}\mathcal{S}(\mathbb{R})\right)  ^{\prime} \label{Gelfand_triplet_1}%
\end{equation}
is a Gel'fand triplet. However, due to the fact that
\[
L^{2}\left(  \Omega\right)  \otimes L^{2}(\mathbb{R})\neq L^{2}\left(
\Omega\right)  \widehat{%
{\textstyle\bigotimes}
}_{\pi}L^{2}(\mathbb{R}),
\]
the demonstration that (\ref{Gelfand_triplet_1}) is a Gel'fand triplet
involves some technicalities. The next step is to show that $\mathcal{D}%
(\Omega)\widehat{%
{\textstyle\bigotimes}
}_{\pi}\mathcal{S}(\mathbb{R})\subset\mathcal{H}=L^{2}\left(  \Omega\right)
\otimes W_{1}(\mathbb{R})$. After this step, using the existence a of
continuous and dense embedding $\mathcal{H\subset}L^{2}\left(  \Omega\right)
\otimes L^{2}(\mathbb{R})$, we get that
\[
\mathcal{D}(\Omega)\widehat{%
{\textstyle\bigotimes}
}_{\pi}\mathcal{S}(\mathbb{R})\subset\mathcal{H}\subset\left(  \mathcal{D}%
(\Omega)\widehat{%
{\textstyle\bigotimes}
}_{\pi}L^{2}(\mathbb{R})\right)  ^{\prime}%
\]
is a Gel'fand triplet. This construction depends heavily on the particular
space of neurons chosen.

\subsubsection{Gel'fand triplets for Archimedean NNs}

It is well-know that, for any positive integer $n$,
\begin{equation}
\mathcal{S}(\mathbb{R}^{n})\hookrightarrow L^{2}(\mathbb{R}^{n}%
)\hookrightarrow\mathcal{S}^{\prime}(\mathbb{R}^{n}) \label{Gelfand_triplet_2}%
\end{equation}
is Gel'fand triplet, where $\mathcal{S}(\mathbb{R}^{n})$ denotes the Schwartz
space on $\mathbb{R}^{n}$, and $\mathcal{S}^{\prime}(\mathbb{R}^{n})$ denotes
the space of tempered distributions, see \cite{Hida et al}, \cite{Huang et
al}. In the Archimedean case, $\mathcal{D}(\Omega)=\mathcal{S}(\mathbb{R})$,
thus $\mathcal{D}(\Omega)\widehat{%
{\textstyle\bigotimes}
}_{\pi}\mathcal{S}(\mathbb{R})=\mathcal{S}(\mathbb{R})\widehat{%
{\textstyle\bigotimes}
}_{\pi}\mathcal{S}(\mathbb{R})\simeq\mathcal{S}(\mathbb{R}^{2})$, and thus we
get the triplet (\ref{Gelfand_triplet_2}) with $n=2$. Now, $\mathcal{H}%
=L^{2}\left(  \mathbb{R}\right)
{\textstyle\bigotimes}
W_{1}(\mathbb{R})\subset L^{2}\left(  \mathbb{R}\right)  \otimes
L^{2}(\mathbb{R})$ is a dense embedding due to the fact that $L^{2}\left(
\mathbb{R}\right)
{\textstyle\bigotimes}
W_{1}(\mathbb{R})$ contains an orthonormal basis of $L^{2}\left(
\mathbb{R}\right)  \otimes L^{2}(\mathbb{R})$, cf. Lemma \ref{Lemma_2} in
Appendix A. Therefore,%
\[
\mathcal{S}(\mathbb{R}^{2})\hookrightarrow L^{2}\left(  \mathbb{R}\right)
{\textstyle\bigotimes}
W_{1}(\mathbb{R}))\hookrightarrow\mathcal{S}^{\prime}(\mathbb{R}^{2}).
\]

\subsubsection{\label{Section_Gelfan_triplets}Gel'fand triplets for
non-Archimedean NNs}

\paragraph{A $p$-adic version of the Schwartz space}

The Bruhat-Schwartz space $\mathcal{D}(\mathbb{Q}_{p}\mathbb{)}$ is not
invariant under the action of pseudo\-differential operators. In
\cite{Zuniga-JFAA}, see also \cite[Chapter 10]{KKZuniga}, the author
introduced a class of countably-Hilbert nuclear spaces which are invariant
under the action of a large class of pseudo-differential operators.

We set $\left[  \xi\right]  _{p}:=\max\left(  1,\left\vert \xi\right\vert
_{p}\right)  $, and denote by $L^{2}\left(  \mathbb{Q}_{p},dx\right)
=L^{2}\left(  \mathbb{Q}_{p}\right)  $, the vector space of all the
real-valued functions $g$ satisfying $\int_{\mathbb{Q}_{p}}\left\vert g\left(
x\right)  \right\vert ^{r}dx<\infty$. We define for $\varphi$, $\varrho$ in
$\mathcal{D}(\mathbb{Q}_{p}\mathbb{)}$ the following scalar product:%
\begin{equation}
\left\langle \varphi,\varrho\right\rangle _{2,l}:=%
{\displaystyle\int\limits_{\mathbb{Q}_{p}}}
\left[  \xi\right]  _{p}^{l}\widehat{\varphi}\left(  \xi\right)
\overline{\widehat{\varrho}}\left(  \xi\right)  d\xi, \label{product_5}%
\end{equation}
for $l\in\mathbb{N=}\left\{  0,1,2,\ldots\right\}  $, where the bar denotes
the complex conjugate and $\widehat{\varphi}\left(  \xi\right)  $ denotes the
Fourier transform of $\varphi\left(  x\right)  $, see \cite[Section
4.8]{A-K-S}. We also set $\left\Vert \varphi\right\Vert _{2,l}^{2}%
:=\left\langle \varphi,\varphi\right\rangle _{2,l}$.

For $l\in\mathbb{Z}$, We define the pseudo-differential operator
\[%
\begin{array}
[c]{cccc}%
\boldsymbol{A}_{l}: & \mathcal{D}(\mathbb{Q}_{p}\mathbb{)} & \rightarrow &
\mathcal{D}(\mathbb{Q}_{p}\mathbb{)}\\
&  &  & \\
& \varphi\left(  x\right)  & \rightarrow & \mathcal{F}_{\mathcal{\xi
\rightarrow}x}^{-1}\left(  \left[  \xi\right]  _{p}^{\frac{l}{2}}%
\mathcal{F}_{x\mathcal{\rightarrow\xi}}\left(  \varphi\right)  \right)  ,
\end{array}
\]
where $\mathcal{F}:\mathcal{D}(\mathbb{Q}_{p}\mathbb{)}%
{\textstyle\bigotimes\nolimits_{\text{alg}}}
\mathbb{C\rightarrow}\mathcal{D}(\mathbb{Q}_{p}\mathbb{)}%
{\textstyle\bigotimes\nolimits_{\text{alg}}}
\mathbb{C}$ denotes the Fourier transform on the space of complex-valued test
functions, see e.g. \cite[Section 4.8]{A-K-S}. By using the fact that the
Fourier transform is an isometry on $L^{2}\left(  \mathbb{Q}_{p}\right)  $,
and taking $l\geq0$,%
\begin{align*}
\left\Vert \boldsymbol{A}_{l}\varphi\right\Vert _{2}^{2}  &  =%
{\displaystyle\int\limits_{\mathbb{Q}_{p}}}
\left\vert \boldsymbol{A}_{l}\varphi\left(  x\right)  \right\vert ^{2}dx=%
{\displaystyle\int\limits_{\mathbb{Q}_{p}}}
\left\vert \mathcal{F}_{x\mathcal{\rightarrow\xi}}\left(  \boldsymbol{A}%
_{l}\varphi\left(  x\right)  \right)  \right\vert ^{2}dx\\
&  =%
{\displaystyle\int\limits_{\mathbb{Q}_{p}}}
\left[  \xi\right]  _{p}^{l}\left\vert \widehat{\varphi}\left(  \xi\right)
\right\vert ^{2}d\xi=\left\Vert \varphi\right\Vert _{2,l}^{2}.
\end{align*}

Let denote by $\mathcal{V}_{l}\left(  \mathbb{Q}_{p}\right)  $ $=\mathcal{V}%
_{l}$ the completion of $\mathcal{D}(\mathbb{Q}_{p}\mathbb{)}$ with respect to
$\left\langle \cdot,\cdot\right\rangle _{l}$. We set%
\[
\mathcal{V}\left(  \mathbb{Q}_{p}\right)  :=%
{\textstyle\bigcap\limits_{l\in\mathbb{N}}}
\mathcal{V}_{l}\left(  \mathbb{Q}_{p}\right)  .
\]
Notice that $\mathcal{V}_{0}=L^{2}\left(  \mathbb{Q}_{p}\right)  $,
$\mathcal{V}\left(  \mathbb{Q}_{p}\right)  \subset L^{2}\left(  \mathbb{Q}%
_{p}\right)  $, and that $\left\Vert \cdot\right\Vert _{l}\leq\left\Vert
\cdot\right\Vert _{m}$ for $l\leq m$. Then $\mathcal{V}_{m}\left(
\mathbb{Q}_{p}\right)  \hookrightarrow\mathcal{V}_{l}\left(  \mathbb{Q}%
_{p}\right)  $ (continuous embedding) for $l\leq m$.

With the topology induced by the family of seminorms $\left\Vert
\cdot\right\Vert _{l\in\mathbb{N}}$, $\mathcal{V}$ becomes a locally convex
space, which is metrizable. Indeed,
\[
d\left(  f,g\right)  :=\max_{l\in\mathbb{N}}\left\{  2^{-l}\frac{\left\Vert
f-g\right\Vert _{2,l}}{1+\left\Vert f-g\right\Vert _{2,l}}\right\}  \text{,
with }f\text{, }g\in\mathcal{V}\text{,}%
\]
is a metric for the topology of $\mathcal{V}\left(  \mathbb{Q}_{p}\right)  $
considered as a convex topological space.

Set $\overline{\left(  \mathcal{D}\left(  \mathbb{Q}_{p}\right)  ,d\right)  }$
for the completion of the metric space $\left(  \mathcal{D}\left(
\mathbb{Q}_{p}\right)  ,d\right)  $. Then $\overline{\left(  \mathcal{D}%
\left(  \mathbb{Q}_{p}\right)  ,d\right)  }=\left(  \mathcal{V}\left(
\mathbb{Q}_{p}\right)  ,d\right)  $, \cite[Lemma 10.4]{KKZuniga}, and since
$\mathcal{D}\left(  \mathbb{Q}_{p}\right)  $ is a nuclear space and the
completion of a such space is also nuclear, \cite[Proposition 50.1]{Treves},
$\mathcal{V}\left(  \mathbb{Q}_{p}\right)  $ is a nuclear space.

For $m\in\mathbb{N}$ and $T\in\mathcal{D}^{\prime}\left(  \mathbb{Q}%
_{p}\right)  $ such that the Fourier transform $\widehat{T}\left(  \xi\right)
$ is a measurable function, we set
\[
\left\Vert T\right\Vert _{2,-m}^{2}:=%
{\textstyle\int\limits_{\mathbb{Q}_{p}}}
\left[  \xi\right]  _{p}^{-m}\left\vert \widehat{T}\left(  \xi\right)
\right\vert ^{2}d\xi.
\]
Then, $\mathcal{V}_{-m}\left(  \mathbb{Q}_{p}\right)  :=\left\{
T\in\mathcal{D}^{\prime}\left(  \mathbb{Q}_{p}\right)  ;\left\Vert
T\right\Vert _{2,-m}^{2}<\infty\right\}  $ is a real Hilbert space. Denote by
$\mathcal{V}_{m}^{\prime}\left(  \mathbb{Q}_{p}\right)  $ the strong dual
space of $\mathcal{V}_{m}\left(  \mathbb{Q}_{p}\right)  $. It is useful to
suppress the correspondence between $\mathcal{V}_{m}^{\prime}\left(
\mathbb{Q}_{p}\right)  $ and $\mathcal{V}_{m}\left(  \mathbb{Q}_{p}\right)  $
given by the Riesz theorem. Instead we identify $\mathcal{V}_{m}^{\prime
}\left(  \mathbb{Q}_{p}\right)  $ and $\mathcal{V}_{-m}\left(  \mathbb{Q}%
_{p}\right)  $ by associating $T\in\mathcal{V}_{-m}\left(  \mathbb{Q}%
_{p}\right)  $ with the functional on $\mathcal{V}_{m}\left(  \mathbb{Q}%
_{p}\right)  $ given by
\begin{equation}
\left(  T,g\right)  :=\int\limits_{\mathbb{Q}_{p}^{n}}\overline{\widehat
{T}\left(  \xi\right)  }\widehat{g}\left(  \xi\right)  d^{n}\xi.
\label{pairing_5}%
\end{equation}
note that $\left\vert \left(  T,g\right)  \right\vert \leq\left\Vert
T\right\Vert _{2,-m}\left\Vert g\right\Vert _{2,m}$ for $m\in\mathbb{N}$. By
using the theory of countable-Hilbert spaces, see, e.g., \cite[Chapter I,
Section 3.1]{Gelfand-Vilenkin}, $\mathcal{V}_{0}^{\prime}\subset
\mathcal{V}_{1}^{\prime}\subset\ldots\subset\mathcal{V}_{m}^{\prime}%
\subset\ldots$ and%
\begin{equation}
\mathcal{V}^{\prime}\left(  \mathbb{Q}_{p}\right)  =\bigcup\limits_{m\in
\mathbb{N}}\mathcal{V}_{-m}\left(  \mathbb{Q}_{p}\right)  =\left\{
T\in\mathcal{D}^{\prime}\left(  \mathbb{Q}_{p}\right)  ;\left\Vert
T\right\Vert _{2,-l}<\infty\text{, for some }l\in\mathbb{N}\right\}
\label{H_infinity_*_5}%
\end{equation}
as vector spaces. Since $\mathcal{V}\left(  \mathbb{Q}_{p}\right)  $ is a
nuclear space, the weak and strong convergence are equivalent in
$\mathcal{V}^{\prime}\left(  \mathbb{Q}_{p}\right)  $, see, e.g.,
\cite[Chapter I, Section 6, Theorem 6.4]{Gel-Shilov}. In conclusion, we have
the following Gel'fand triplet:%
\[
\mathcal{V}\left(  \mathbb{Q}_{p}\right)  \hookrightarrow L^{2}\left(
\mathbb{Q}_{p}\right)  \hookrightarrow\mathcal{V}^{\prime}\left(
\mathbb{Q}_{p}\right)  .
\]

The author and his collaborators have use this triplet and others in the
construction of several types of QFTs over $p$-adic spaces, see, e.g.,
\cite{Arroyo et al}, \cite{Fuquen et al}, and the references therein.

\begin{remark}
\label{Note_Gelfand}By \ replacing $\mathcal{D}(\mathbb{Q}_{p}\mathbb{)}$ with
$\mathcal{D}(\mathbb{Z}_{p}\mathbb{)}$, in the above construction, we obtain
Gel'fand triplet:%
\[
\mathcal{V}\left(  \mathbb{Z}_{p}\right)  \hookrightarrow L^{2}\left(
\mathbb{Z}_{p}\right)  \hookrightarrow\mathcal{V}^{\prime}\left(
\mathbb{Z}_{p}\right)  .
\]

\end{remark}

\paragraph{A reconstruction of the Schwartz space}

Set
\[
\boldsymbol{D}\varphi\left(  x\right)  =\left(  -\frac{d^{2}}{dx^{2}}%
+x^{2}+1\right)  \varphi\left(  x\right)  \text{, \ for }\varphi\in
\mathcal{S}(\mathbb{R}).
\]
Given $l\in\mathbb{N=}\left\{  0,1,2,\ldots\right\}  $, we denote by
$\mathcal{S}_{l}(\mathbb{R})$ the real Hilbert space which the completion of
$\mathcal{S}(\mathbb{R})$ with respect to the norm%
\[
\left\Vert \varphi\right\Vert _{2,l}:=\left\Vert \boldsymbol{D}^{l}%
\varphi\right\Vert _{2},
\]
where $\left\Vert \cdot\right\Vert _{2}$ is the norm in $L^{2}\left(
\mathbb{R}\right)  $. Then%
\[
\mathcal{S}_{l}(\mathbb{R})=\left\{  f\in L^{2}\left(  \mathbb{R}\right)
;\left\Vert f\right\Vert _{2,l}<\infty\right\}  ,
\]
and
\[
L^{2}\left(  \mathbb{R}\right)  =\mathcal{S}_{0}(\mathbb{R})\supset
\mathcal{S}_{1}(\mathbb{R})\supset\cdots\supset\mathcal{S}_{l}(\mathbb{R}%
)\supset\cdots,
\]
and the Schwartz space is the projective limit of the $\mathcal{S}%
_{l}(\mathbb{R})$:%
\[
\mathcal{S}(\mathbb{R})=\underset{l}{\underleftarrow{\lim}}\mathcal{S}%
_{l}(\mathbb{R})=%
{\displaystyle\bigcap\limits_{l\in\mathbb{N}}}
\mathcal{S}_{l}(\mathbb{R}).
\]
We now define in $L^{2}\left(  \mathbb{R}\right)  $ the norm%
\[
\left\Vert f\right\Vert _{2,-l}:=\left\Vert \boldsymbol{D}^{-l}f\right\Vert
_{2}\text{, for }l\in\mathbb{N=}\left\{  0,1,2,\ldots\right\}  \text{,}%
\]
and set $\mathcal{S}_{-l}(\mathbb{R})$ to be the completion of $L^{2}\left(
\mathbb{R}\right)  $ \ with respect to $\left\Vert \cdot\right\Vert _{2,-l}$.
By identifying $L^{2}\left(  \mathbb{R}\right)  $ with its dual, we have
$\mathcal{S}_{-l}(\mathbb{R})=\mathcal{S}_{l}^{\prime}(\mathbb{R})$. In
addition,
\[
L^{2}\left(  \mathbb{R}\right)  =\mathcal{S}_{0}(\mathbb{R})\subset
\mathcal{S}_{-1}(\mathbb{R})\subset\cdots\subset\mathcal{S}_{-l}%
(\mathbb{R})\subset\cdots,
\]
and the inductive limit of the $\mathcal{S}_{-l}(\mathbb{R})$ is the space
tempered distributions:%
\[
\mathcal{S}^{\prime}(\mathbb{R})=\underset{l}{\underrightarrow{\lim}%
}\mathcal{S}_{-l}(\mathbb{R})=%
{\displaystyle\bigcup\limits_{l\in\mathbb{N}}}
\mathcal{S}_{l}^{\prime}(\mathbb{R})\text{.}%
\]
Finally, $\mathcal{S}(\mathbb{R})$, $\mathcal{S}^{\prime}(\mathbb{R})$ are
nuclear spaces and $\mathcal{S}(\mathbb{R})\hookrightarrow L^{2}%
(\mathbb{R})\hookrightarrow\mathcal{S}^{\prime}(\mathbb{R})$.

The above results are well-known; see, e.g., \cite{Hida et al}, \cite{Huang et
al}.

\paragraph{Some Gel'fand triplets for non-Archimedean NNs}

We know now that $\left(  \mathcal{V}(\mathbb{Q}_{p}),\left\{  \left\Vert
\cdot\right\Vert _{2,m}\right\}  _{m\in\mathbb{N}}\right)  $ and $\left(
\mathcal{S}(\mathbb{R}),\left\{  \left\Vert \cdot\right\Vert _{2,l}\right\}
_{l\in\mathbb{N}}\right)  $ are nuclear countably-Hilbert spaces, with
$\mathcal{V}_{m}(\mathbb{Q}_{p})$, respectively $\mathcal{S}_{l}(\mathbb{R})$,
the Hilbert space obtained completing $\mathcal{V}(\mathbb{Q}_{p})$ with
respect to $\left\Vert \cdot\right\Vert _{2,m}$, respectively completing
$\mathcal{S}(\mathbb{R})$ with respect to $\left\Vert \cdot\right\Vert _{2,l}%
$. The Hilbert tensor product of these spaces is denoted as $\mathcal{V}%
_{m}(\mathbb{Q}_{p})%
{\textstyle\bigotimes}
\mathcal{S}_{l}(\mathbb{R})$. We denote by $\mathcal{V}(\mathbb{Q}%
_{p})\widehat{%
{\textstyle\bigotimes\limits_{\pi}}
}\mathcal{S}(\mathbb{R})$ the projective tensor product of nuclear spaces. By
\cite[Proposition 1.3.8]{Obata}, $\left\{  \mathcal{V}_{m}(\mathbb{Q}_{p})%
{\textstyle\bigotimes}
\mathcal{S}_{l}(\mathbb{R})\right\}  _{m,l\in\mathbb{N}}$ is a projective
system of Hilbert spaces and
\[
\mathcal{V}(\mathbb{Q}_{p})\widehat{%
{\textstyle\bigotimes\limits_{\pi}}
}\mathcal{S}(\mathbb{R})=\underset{m,l}{\underleftarrow{\lim}}\mathcal{V}%
_{m}(\mathbb{Q}_{p})%
{\textstyle\bigotimes}
\mathcal{S}_{l}(\mathbb{R)}%
\]
is a nuclear space. Notice that for $m=l=0$, $\mathcal{V}_{m}(\mathbb{Q}_{p})%
{\textstyle\bigotimes}
\mathcal{S}_{l}(\mathbb{R)=}L^{2}\left(  \mathbb{Q}_{p}\right)  \otimes
L^{2}(\mathbb{R})$ and the canonical projection
\begin{equation}
\mathcal{V}(\mathbb{Q}_{p})\widehat{%
{\textstyle\bigotimes\limits_{\pi}}
}\mathcal{S}(\mathbb{R})\rightarrow L^{2}\left(  \mathbb{Q}_{p}\right)
\otimes L^{2}(\mathbb{R}) \label{canonical_proj}%
\end{equation}
is continuous and dense. Alternatively, $\mathcal{D}(\mathbb{Q}_{p})%
{\textstyle\bigotimes\limits_{\text{alg}}}
\mathcal{S}(\mathbb{R})\subset\mathcal{V}(\mathbb{Q}_{p})\widehat{%
{\textstyle\bigotimes\limits_{\pi}}
}\mathcal{S}(\mathbb{R})$, and that $L^{2}\left(  \mathbb{Q}_{p}\right)  $ has
an orthonormal basis $\left\{  \Psi_{rbk}\left(  x\right)  \right\}  _{rbk}$
consisting of functions from $\mathcal{D}(\mathbb{Q}_{p})$, \cite[Theorems
9.4.5 and 8.9.3]{A-K-S} or \cite[Theorem 3.3]{KKZuniga}; the polynomials of
Hermite $\left\{  H_{n}(t)\right\}  _{n\in\mathbb{N}}$, which are elements
from $\mathcal{S}(\mathbb{R})$, form an orthonormal basis of $L^{2}%
(\mathbb{R})$. Then $\left\{  \Psi_{rbk}\left(  x\right)  H_{m}\left(
t\right)  \right\}  _{rbk,m}$ is an orthonormal basis of $L^{2}\left(
\mathbb{Q}_{p}\right)  \otimes L^{2}(\mathbb{R})$, cf. \cite[Proposition
2.7]{Asao} or \cite[Section II.4, Proposition 2]{Reed-Simon-I}, and
consequently, the canonical projection (\ref{canonical_proj}) is dense.
Therefore,%
\[
\mathcal{V}(\mathbb{Q}_{p})\widehat{%
{\textstyle\bigotimes\limits_{\pi}}
}\mathcal{S}(\mathbb{R})\hookrightarrow L^{2}\left(  \mathbb{Q}_{p}\right)
\otimes L^{2}(\mathbb{R})\hookrightarrow\left(  \mathcal{V}(\mathbb{Q}%
_{p})\widehat{%
{\textstyle\bigotimes\limits_{\pi}}
}\mathcal{S}(\mathbb{R})\right)  ^{\prime}%
\]
is a Gel'fand triplet. Furthermore, $\left(  \mathcal{V}(\mathbb{Q}%
_{p})\widehat{%
{\textstyle\bigotimes\limits_{\pi}}
}\mathcal{S}(\mathbb{R})\right)  ^{\prime}\simeq\mathcal{V}^{\prime
}(\mathbb{Q}_{p})\widehat{%
{\textstyle\bigotimes\limits_{\pi}}
}\mathcal{S}^{\prime}(\mathbb{R})$. By using the above construction and Remark
\ref{Note_Gelfand}, we have the following result.

\begin{theorem}
\label{Theorem6}The sequences
\[
\mathcal{V}(\mathbb{Q}_{p})\widehat{%
{\textstyle\bigotimes\limits_{\pi}}
}\mathcal{S}(\mathbb{R})\hookrightarrow L^{2}\left(  \mathbb{Q}_{p}\right)
\otimes L^{2}(\mathbb{R})\hookrightarrow\left(  \mathcal{V}(\mathbb{Q}%
_{p})\widehat{%
{\textstyle\bigotimes\limits_{\pi}}
}\mathcal{S}(\mathbb{R})\right)  ^{\prime},
\]%
\[
\mathcal{V}(\mathbb{Z}_{p})\widehat{%
{\textstyle\bigotimes\limits_{\pi}}
}\mathcal{S}(\mathbb{R})\hookrightarrow L^{2}\left(  \mathbb{Z}_{p}\right)
\otimes L^{2}(\mathbb{R})\hookrightarrow\left(  \mathcal{V}(\mathbb{Z}%
_{p})\widehat{%
{\textstyle\bigotimes\limits_{\pi}}
}\mathcal{S}(\mathbb{R})\right)  ^{\prime},
\]
are Gel'fand triplets.
\end{theorem}

\section{\label{Appen_D}Appendix D}

\subsection{Gaussian measures on $L^{2}\left(  \Omega^{2}\right)  $}

We consider $\left(  L^{2}\left(  \Omega^{2}\right)  ,\mathcal{B}\left(
\Omega^{2}\right)  \right)  $ as a measurable space, where $\mathcal{B}\left(
\Omega^{2}\right)  $ is the Borel $\sigma$-algebra of the space $\Omega^{2}$.
We recall that $L^{2}\left(  \Omega^{2}\right)  $ is a separable Hilbert
space. There is a bijection between symmetric, positive definite, and trace
class operators on $L^{2}\left(  \Omega^{2}\right)  $ and Gaussian probability
measures on $L^{2}\left(  \Omega^{2}\right)  $ with mean zero. More precisely,
given a such operator $\square:L^{2}\left(  \Omega^{2}\right)  \rightarrow
L^{2}\left(  \Omega^{2}\right)  $, there exists a unique Gaussian probability
measure $\boldsymbol{P}_{\square}$ on $L^{2}\left(  \Omega\times\Omega\right)
$ with mean zero, covariance $\square$, and Fourier transform
\begin{equation}%
{\displaystyle\int\limits_{L^{2}\left(  \Omega^{2}\right)  }}
e^{\sqrt{-1}\left\langle f,w\right\rangle _{L^{2}\left(  \Omega^{2}\right)  }%
}d\boldsymbol{P}_{\square}\left(  w\right)  =e^{\frac{-1}{2}\left\langle
\square f,f\right\rangle _{L^{2}\left(  \Omega^{2}\right)  }}\text{, for }f\in
L^{2}\left(  \Omega^{2}\right)  , \label{Minlos-Bochner}%
\end{equation}
see \cite[Theorem 1.12]{Da prato}, or \cite[Theorem 2.3.1]{Bogachev}. We call
$\square$ the covariance operator of the measure $\boldsymbol{P}_{\square}$.

In the case, $\square=r\mathbb{I}$, $r>0$, where $\mathbb{I}$ is the identity
operator, the Gaussian measure is called isonormal with parameter $r$:%
\[%
{\displaystyle\int\limits_{L^{2}\left(  \Omega^{2}\right)  }}
e^{\sqrt{-1}\left\langle f,w\right\rangle }d\boldsymbol{P}_{\square}\left(
w\right)  =e^{\frac{-r^{2}}{2}\left\Vert f\right\Vert _{L^{2}\left(
\Omega\times\Omega\right)  }^{2}}.
\]

\begin{remark}
\label{Nota1}Let $\left(  \mathcal{X},\left\langle \cdot,\cdot\right\rangle
_{\mathcal{X}}\right)  $ be a real separable Hilbert space. Let $B\left(
f,g\right)  $ be a continuous bilinear form, which is positive definite, i.e.,
$B\left(  f,f\right)  \geq0$ for any $f\in\mathcal{X}$. This bilinear form
determines a unique Gaussian probability measure\ on $\left(  \mathcal{X}%
_{,}\mathcal{B}\left(  \mathcal{X}\right)  \right)  $, where $\mathcal{B}%
\left(  \mathcal{X}\right)  $ is the Borel $\sigma$-algebra of $\mathcal{X}$,
with mean zero and correlation functional $B\left(  f,g\right)  $, with
Fourier transform%
\[
f\rightarrow\exp\left(  -\frac{1}{2}B\left(  f,f\right)  \right)  .
\]
Furthermore, there exists $\square:\mathcal{X}\rightarrow\mathcal{X}$ a
symmetric, positive definite, and trace class (or nuclear) operator such that
$B\left(  f,g\right)  :=\left\langle f,\square g\right\rangle _{\mathcal{X}}$.
Conversely, if $\square$ is symmetric, positive definite, and trace class
operator, then $B\left(  f,g\right)  :=\left\langle f,\square g\right\rangle
_{\mathcal{X}}$ is a continuous, positive definite bilinear from, see
\cite[Theorem 2.2.4]{Bogachev}.
\end{remark}

In this article, we assume that $\square$\ is an integral operator, which
means that there exists
\[
K_{\square}\left(  u_{1},u_{2},y_{1},y_{2}\right)  \in L^{2}\left(  \Omega
^{4}\right)  .
\]
such that%
\begin{equation}
\square f\left(  y_{1},y_{2}\right)  =%
{\displaystyle\int\limits_{\Omega}}
{\displaystyle\int\limits_{\Omega}}
K_{\square}\left(  u_{1},u_{2},y_{1},y_{2}\right)  f\left(  u_{1}%
,u_{2}\right)  d\mu\left(  u_{1}\right)  d\mu\left(  u_{2}\right)  .
\label{Kernel_operator}%
\end{equation}

\begin{remark}
\label{Nota_Mercer}Let $\left(  \Omega,\mathcal{B}\left(  \Omega\right)
,d\mu\right)  $ be a measure space, where $\Omega$ is a compact space, and
$\mathcal{B}\left(  \Omega\right)  $\ is the Borel $\sigma$-algebra. Let
$K:\Omega\times\Omega\rightarrow\mathbb{R}$ be a continuous, symmetric,
positive definite function in $L^{2}(\Omega\times\Omega)$. Then the operator%
\[%
\begin{array}
[c]{ccc}%
L^{2}(\Omega) & \rightarrow & L^{2}(\Omega)\\
&  & \\
f & \rightarrow & \left(  T_{K}f\right)  \left(  x\right)  :=%
{\displaystyle\int\limits_{\Omega}}
K(x,y)f(y)d\mu\left(  y\right)
\end{array}
\]
is a positive definite, symmetric and trace class operator. This result is
known as the Mercer theorem, \cite{Koning}, \cite{Berlinet et al}.
\end{remark}

The formula%
\begin{equation}%
{\displaystyle\int\limits_{L^{2}\left(  \Omega^{2}\right)  }}
e^{\lambda\left\langle f,w\right\rangle _{L^{2}\left(  \Omega^{2}\right)  }%
}d\boldsymbol{P}_{\square}\left(  w\right)  =e^{\frac{\lambda^{2}}%
{2}\left\langle \square f,f\right\rangle _{L^{2}\left(  \Omega^{2}\right)  }%
}\text{, for }f\in L^{2}\left(  \Omega^{2}\right)  \text{, }\lambda
\in\mathbb{R},\label{Key_formula}%
\end{equation}
\cite[Proposition 1.15]{Da prato}, shows that (\ref{Minlos-Bochner}) admits a
entire continuation to the complex plane, i.e., we can replace in
(\ref{Minlos-Bochner}) $f$ by $-\sqrt{-1}\lambda f$ and the identity remains
valid if $\left\langle \cdot,\cdot\right\rangle _{L^{2}\left(  \Omega
^{2}\right)  }$ is interpreted as a sesquilinear form; see \cite[Proposition
2.2]{Hida et al} for an analog result. We use this formula as a substitute of
the stationary phase method used in \cite[Chapter 10]{Helias et al}.

\subsection{The averaged partition function}

\begin{remark}
\label{Nota7}By using Fubini's theorem, and Lemma \ref{Lemma_4},%
\begin{gather}
\left\langle \widetilde{\boldsymbol{h}}\left(  x,t\right)  ,%
{\displaystyle\int\limits_{\Omega}}
\boldsymbol{J}\left(  x,y\right)  \phi\left(  \boldsymbol{h}\left(
y,t\right)  \right)  d\mu\left(  y\right)  \right\rangle \nonumber\\
=%
{\displaystyle\int\limits_{\Omega}}
{\displaystyle\int\limits_{\mathbb{R}}}
\left\{  \widetilde{\boldsymbol{h}}\left(  x,t\right)  ,%
{\displaystyle\int\limits_{\Omega}}
\boldsymbol{J}\left(  x,y\right)  \phi\left(  \boldsymbol{h}\left(
y,t\right)  \right)  d\mu\left(  y\right)  \right\}  dtd\mu\left(  x\right)
\nonumber\\
=%
{\displaystyle\int\limits_{\Omega^{2}}}
\boldsymbol{J}\left(  x,y\right)  \left\{
{\displaystyle\int\limits_{\mathbb{R}}}
\widetilde{\boldsymbol{h}}\left(  x,t\right)  \phi\left(  \boldsymbol{h}%
\left(  y,t\right)  \right)  dt\right\}  d\mu\left(  x\right)  d\mu\left(
y\right) \nonumber\\
=\left\langle \boldsymbol{J}\left(  x,y\right)  ,%
{\displaystyle\int\limits_{\mathbb{R}}}
\widetilde{\boldsymbol{h}}\left(  x,t\right)  \phi\left(  \boldsymbol{h}%
\left(  y,t\right)  \right)  dt\right\rangle _{L^{2}\left(  \Omega^{2}\right)
}. \label{Formula_1}%
\end{gather}

\end{remark}

\begin{lemma}
\label{Lemma_7}Let $\mathbb{P}_{\boldsymbol{J}}$ be the Gaussian probability
measure with mean zero and covariance operator $\square_{\boldsymbol{J}}$
attached to the random coupling kernel J. The following formula holds,
\begin{gather*}%
{\displaystyle\int\limits_{L^{2}\left(  \Omega^{2}\right)  }}
\exp\left\langle -\widetilde{\boldsymbol{h}},%
{\displaystyle\int\limits_{\Omega}}
\boldsymbol{J}\left(  x,y\right)  \phi\left(  \boldsymbol{h}\left(
y,t\right)  \right)  d\mu\left(  y\right)  \right\rangle d\mathbb{P}%
_{\boldsymbol{J}}=\\
\exp\left(  \frac{1}{2}%
{\displaystyle\int\limits_{\mathbb{R}^{2}}}
\text{ }%
{\displaystyle\int\limits_{\Omega^{2}}}
\widetilde{\boldsymbol{h}}\left(  u_{1},t_{1}\right)  C_{\phi\phi}\left(
u_{1},y_{1},t_{1},t_{2}\right)  \widetilde{\boldsymbol{h}}\left(  y_{1}%
,t_{2}\right)  d\mu\left(  u_{1}\right)  d\mu\left(  y_{1}\right)
dt_{1}dt_{2}\right)  ,
\end{gather*}
where $C_{\phi\phi}\left(  u_{1},y_{1},t_{1},t_{2}\right)  $ is defined in
(\ref{Covarience_pi(h)phi(h)}).
\end{lemma}

\begin{proof}
By using the formula (\ref{Key_formula}), with $\lambda=-1$, and%
\[
f=%
{\displaystyle\int\limits_{\mathbb{R}}}
\widetilde{\boldsymbol{h}}\left(  x,t\right)  \phi\left(  \boldsymbol{h}%
\left(  y,t\right)  \right)  dt,
\]
and the formula (\ref{Formula_1}),%
\begin{gather*}%
{\displaystyle\int\limits_{L^{2}\left(  \Omega^{2}\right)  }}
\exp\left\langle -\widetilde{\boldsymbol{h}},%
{\displaystyle\int\limits_{\Omega}}
\boldsymbol{J}\left(  x,y\right)  \phi\left(  \boldsymbol{h}\left(
y,t\right)  \right)  d\mu\left(  y\right)  \right\rangle d\mathbb{P}%
_{\boldsymbol{J}}=\\
\exp\frac{1}{2}\left\langle \square_{\boldsymbol{J}}\left(  \text{ }%
{\displaystyle\int\limits_{\mathbb{R}}}
\widetilde{\boldsymbol{h}}\left(  y_{1},t\right)  \phi\left(  \boldsymbol{h}%
\left(  y_{2},t\right)  \right)  dt\right)  ,%
{\displaystyle\int\limits_{\mathbb{R}}}
\widetilde{\boldsymbol{h}}\left(  y_{1},t\right)  \phi\left(  \boldsymbol{h}%
\left(  y_{2},t\right)  \right)  dt\right\rangle _{L^{2}\left(  \Omega
^{2}\right)  }.
\end{gather*}
On the other hand, using the fact that $\square_{\boldsymbol{J}}$ is an
integral operator with kernel $K_{\boldsymbol{J}}$,%
\begin{gather*}
\left\langle \square_{\boldsymbol{J}}\left(  \text{ }%
{\displaystyle\int\limits_{\mathbb{R}}}
\widetilde{\boldsymbol{h}}\left(  y_{1},t_{1}\right)  \phi\left(
\boldsymbol{h}\left(  y_{2},t_{1}\right)  \right)  dt_{1}\right)  ,%
{\displaystyle\int\limits_{\mathbb{R}}}
\widetilde{\boldsymbol{h}}\left(  y_{1},t_{2}\right)  \phi\left(
\boldsymbol{h}\left(  y_{2},t_{2}\right)  \right)  dt_{2}\right\rangle
_{L^{2}\left(  \Omega\times\Omega\right)  }\\
=%
{\displaystyle\int\limits_{\Omega^{2}}}
\text{ }%
{\displaystyle\int\limits_{\mathbb{R}}}
\left\{  \text{ }%
{\displaystyle\int\limits_{\Omega^{2}}}
K_{\boldsymbol{J}}\left(  u_{1},u_{2},y_{1},y_{2}\right)
{\displaystyle\int\limits_{\mathbb{R}}}
\widetilde{\boldsymbol{h}}\left(  u_{1},t_{1}\right)  \phi\left(
\boldsymbol{h}\left(  u_{2},t_{1}\right)  \right)  dt_{1}d\mu\left(
u_{1}\right)  d\mu\left(  u_{2}\right)  \right\}  \times\\
\widetilde{\boldsymbol{h}}\left(  y_{1},t_{2}\right)  \phi\left(
\boldsymbol{h}\left(  y_{2},t_{2}\right)  \right)  dt_{2}d\mu\left(
y_{1}\right)  d\mu\left(  y_{2}\right) \\
=%
{\displaystyle\int\limits_{\mathbb{R}^{2}}}
\text{ }%
{\displaystyle\int\limits_{\Omega^{2}}}
\widetilde{\boldsymbol{h}}\left(  u_{1},t_{1}\right)  C_{\phi\left(
\boldsymbol{h}\right)  \phi\left(  \boldsymbol{h}\right)  }\left(  u_{1}%
,y_{1},t_{1},t_{2}\right)  \widetilde{\boldsymbol{h}}\left(  y_{1}%
,t_{2}\right)  d\mu\left(  u_{1}\right)  d\mu\left(  y_{1}\right)
dt_{1}dt_{2},
\end{gather*}
where%
\begin{gather}
C_{\phi\left(  \boldsymbol{h}\right)  \phi\left(  \boldsymbol{h}\right)
}\left(  u_{1},y_{1},t_{1},t_{2}\right)  :=\label{Covarience_pi(h)phi(h)}\\%
{\displaystyle\int\limits_{\Omega^{2}}}
\phi\left(  \boldsymbol{h}\left(  y_{2},t_{2}\right)  \right)
K_{\boldsymbol{J}}\left(  u_{1},u_{2},y_{1},y_{2}\right)  \phi\left(
\boldsymbol{h}\left(  u_{2},t_{1}\right)  \right)  d\mu\left(  y_{2}\right)
d\mu\left(  u_{2}\right)  .\nonumber
\end{gather}

\end{proof}

\begin{lemma}
\label{Lemma_9A} With the above notation,%
\[
\left\Vert C_{\phi\left(  \boldsymbol{h}\right)  \phi\left(  \boldsymbol{h}%
\right)  }\right\Vert _{L^{2}(\Omega^{2}\times\mathbb{R}^{2})}\leq
L^{2}\left(  \phi\right)  \left\Vert \boldsymbol{h}\right\Vert ^{2}\left\Vert
K_{\boldsymbol{J}}\right\Vert _{L^{2}(\Omega^{4})}.
\]

\end{lemma}

\begin{proof}
By using that $\phi$ is a Lipschitz function, and the Cauchy-Schwarz
inequality in $L^{2}(\Omega^{2})$,%

\begin{gather*}
\left\Vert C_{\phi\left(  \boldsymbol{h}\right)  \phi\left(  \boldsymbol{h}%
\right)  }\right\Vert _{L^{2}(\Omega^{2}\times\mathbb{R}^{2})}^{2}=%
{\displaystyle\iint\limits_{\mathbb{R}\times\mathbb{R}}}
\text{ }%
{\displaystyle\int\limits_{\Omega^{2}}}
\times\\
\left\vert \text{ \ }%
{\displaystyle\int\limits_{\Omega^{2}}}
\phi\left(  \boldsymbol{h}\left(  y_{2},t_{2}\right)  \right)
K_{\boldsymbol{J}}\left(  u_{1},u_{2},y_{1},y_{2}\right)  \phi\left(
\boldsymbol{h}\left(  u_{2},t_{1}\right)  \right)  d\mu\left(  y_{2}\right)
d\mu\left(  u_{2}\right)  \right\vert ^{2}\times\\
d\mu\left(  u_{1}\right)  d\mu\left(  y_{1}\right)  dt_{1}dt_{2}\leq
L^{4}\left(  \phi\right)  \times\\%
{\displaystyle\int\limits_{\mathbb{R}^{2}}}
\text{ }%
{\displaystyle\int\limits_{\Omega^{2}}}
\left\Vert \boldsymbol{h}\left(  \cdot,t_{2}\right)  \boldsymbol{h}\left(
\cdot,t_{1}\right)  \right\Vert _{L^{2}(\Omega^{2})}^{2}\left\Vert
K_{\boldsymbol{J}}\left(  u_{1},\cdot,y_{1},\cdot\right)  \right\Vert
_{L^{2}(\Omega^{2})}^{2}d\mu\left(  u_{1}\right)  d\mu\left(  y_{1}\right)
dt_{1}dt_{2}%
\end{gather*}

\begin{gather*}
=L^{4}\left(  \phi\right)  \left(  \text{ }%
{\displaystyle\int\limits_{\mathbb{R}}}
\left\Vert \boldsymbol{h}\left(  \cdot,t_{2}\right)  \right\Vert
_{L^{2}(\Omega^{2})}^{2}dt_{2}\right)  \left(  \text{ }%
{\displaystyle\int\limits_{\mathbb{R}}}
\left\Vert \boldsymbol{h}\left(  \cdot,t_{1}\right)  \right\Vert
_{L^{2}(\Omega^{2})}^{2}dt_{1}\right)  \times\\
\left(  \text{ }%
{\displaystyle\int\limits_{\Omega^{2}}}
\left\Vert K_{\boldsymbol{J}}\left(  u_{1},\cdot,y_{1},\cdot\right)
\right\Vert _{L^{2}(\Omega^{2})}^{2}d\mu\left(  u_{1}\right)  d\mu\left(
y_{1}\right)  \right)  =L^{4}\left(  \phi\right)  \left\Vert \boldsymbol{h}%
\right\Vert ^{4}\left\Vert K_{\boldsymbol{J}}\right\Vert _{L^{2}(\Omega^{4}%
)}^{2}.
\end{gather*}

\end{proof}

Set%
\[
\left(  \square_{C_{\phi\left(  \boldsymbol{h}\right)  \phi\left(
\boldsymbol{h}\right)  }}\widetilde{\boldsymbol{h}}\right)  \left(
x,t_{1}\right)  :=\left\{
{\displaystyle\int\limits_{\mathbb{R}}}
\text{ }%
{\displaystyle\int\limits_{\Omega}}
\text{ }C_{\phi\left(  \boldsymbol{h}\right)  \phi\left(  \boldsymbol{h}%
\right)  }\left(  x,y,t_{1},t_{2}\right)  \widetilde{\boldsymbol{h}}\left(
y,t_{2}\right)  d\mu\left(  y\right)  dt_{2}\right\}  ,
\]
and%
\begin{gather}
\left\langle \widetilde{\boldsymbol{h}},C_{\phi\left(  \boldsymbol{h}\right)
\phi\left(  \boldsymbol{h}\right)  }\widetilde{\boldsymbol{h}}\right\rangle
_{\left(  L^{2}(\Omega\times\mathbb{R})\right)  ^{2}}%
:=\label{Formula_C_phi_phi}\\%
{\displaystyle\int\limits_{\mathbb{R}^{2}}}
\text{ \ }%
{\displaystyle\int\limits_{\Omega^{2}}}
\text{ }\widetilde{\boldsymbol{h}}\left(  x,t_{1}\right)  C_{\phi\left(
\boldsymbol{h}\right)  \phi\left(  \boldsymbol{h}\right)  }\left(
x,y,t_{1},t_{2}\right)  \widetilde{\boldsymbol{h}}\left(  y,t_{2}\right)
d\mu\left(  x\right)  d\mu\left(  y\right)  dt_{1}dt_{2}.\nonumber
\end{gather}
Then
\[
\left\langle \widetilde{\boldsymbol{h}},C_{\phi\left(  \boldsymbol{h}\right)
\phi\left(  \boldsymbol{h}\right)  }\widetilde{\boldsymbol{h}}\right\rangle
_{\left(  L^{2}(\Omega\times\mathbb{R})\right)  ^{2}}=\left\langle
\widetilde{\boldsymbol{h}},\square_{C_{\phi\left(  \boldsymbol{h}\right)
\phi\left(  \boldsymbol{h}\right)  }}\widetilde{\boldsymbol{h}}\right\rangle
.
\]

\begin{lemma}
\label{Lemma_9}Take $\left(  \widetilde{\boldsymbol{h}},\boldsymbol{h}\right)
\in\left(  L^{2}(\Omega\times\mathbb{R})\right)  ^{2}$, $K_{\boldsymbol{J}}\in
L^{2}(\Omega^{4})$. Then%
\[
\left\vert \left\langle \widetilde{\boldsymbol{h}},C_{\phi\left(
\boldsymbol{h}\right)  \phi\left(  \boldsymbol{h}\right)  }\widetilde
{\boldsymbol{h}}\right\rangle _{\left(  L^{2}(\Omega\times\mathbb{R})\right)
^{2}}\right\vert \leq L^{2}\left(  \phi\right)  \left\Vert K_{\boldsymbol{J}%
}\right\Vert _{L^{2}(\Omega^{4})}\left\Vert \boldsymbol{h}\right\Vert
^{2}\left\Vert \widetilde{\boldsymbol{h}}\right\Vert ^{2}.
\]

\end{lemma}

\begin{proof}
By using the Cauchy-Schwarz inequality in $L^{2}(\Omega\times\mathbb{R})$,%
\[
\left\vert \left\langle \widetilde{\boldsymbol{h}},C_{\phi\left(
\boldsymbol{h}\right)  \phi\left(  \boldsymbol{h}\right)  }\widetilde
{\boldsymbol{h}}\right\rangle _{\left(  L^{2}(\Omega\times\mathbb{R})\right)
^{2}}\right\vert =\left\vert \left\langle \widetilde{\boldsymbol{h}}%
,\square_{C_{\phi\left(  \boldsymbol{h}\right)  \phi\left(  \boldsymbol{h}%
\right)  }}\widetilde{\boldsymbol{h}}\right\rangle \right\vert \leq\left\Vert
\widetilde{\boldsymbol{h}}\right\Vert \left\Vert \square_{C_{\phi\left(
\boldsymbol{h}\right)  \phi\left(  \boldsymbol{h}\right)  }}\widetilde
{\boldsymbol{h}}\right\Vert .
\]
Now, by using the Cauchy-Schwarz inequality in $L^{2}(\Omega\times\mathbb{R}%
)$, and Lemma \ref{Lemma_9A},%
\begin{align*}
\left\Vert \square_{C_{\phi\left(  \boldsymbol{h}\right)  \phi\left(
\boldsymbol{h}\right)  }}\widetilde{\boldsymbol{h}}\right\Vert ^{2}  &  =%
{\displaystyle\int\limits_{\Omega}}
\text{ }%
{\displaystyle\int\limits_{\mathbb{R}}}
\left\vert
{\displaystyle\int\limits_{\Omega}}
\text{ }%
{\displaystyle\int\limits_{\mathbb{R}}}
C_{\phi\left(  \boldsymbol{h}\right)  \phi\left(  \boldsymbol{h}\right)
}\left(  x,y,t_{1},t_{2}\right)  \widetilde{\boldsymbol{h}}\left(
y,t_{2}\right)  d\mu\left(  y\right)  dt_{2}\right\vert ^{2}dt_{1}d\mu\left(
x\right) \\
&  \leq\left\Vert \widetilde{\boldsymbol{h}}\right\Vert ^{2}%
{\displaystyle\int\limits_{\Omega}}
\text{ }%
{\displaystyle\int\limits_{\mathbb{R}}}
\left\Vert C_{\phi\left(  \boldsymbol{h}\right)  \phi\left(  \boldsymbol{h}%
\right)  }\left(  x,\cdot,t_{1},\cdot\right)  \right\Vert ^{2}dt_{1}%
d\mu\left(  x\right) \\
&  =\left\Vert \widetilde{\boldsymbol{h}}\right\Vert ^{2}\left\Vert
C_{\phi\left(  \boldsymbol{h}\right)  \phi\left(  \boldsymbol{h}\right)
}\right\Vert _{L^{2}(\Omega^{2}\times\mathbb{R}^{2})}^{2}\leq L^{4}\left(
\phi\right)  \left\Vert \widetilde{\boldsymbol{h}}\right\Vert ^{2}\left\Vert
\boldsymbol{h}\right\Vert ^{4}\left\Vert K_{\boldsymbol{J}}\right\Vert
_{L^{2}(\Omega^{4})}^{2}.
\end{align*}
Therefore $\left\Vert \widetilde{\boldsymbol{h}}\right\Vert \left\Vert
\square_{C_{\phi\left(  \boldsymbol{h}\right)  \phi\left(  \boldsymbol{h}%
\right)  }}\widetilde{\boldsymbol{h}}\right\Vert \leq L^{2}\left(
\phi\right)  \left\Vert \widetilde{\boldsymbol{h}}\right\Vert ^{2}\left\Vert
\boldsymbol{h}\right\Vert ^{2}\left\Vert K_{\boldsymbol{J}}\right\Vert
_{L^{2}(\Omega^{4})}.$
\end{proof}

\begin{remark}
\label{Nota2}The bilinear form $B(\widetilde{\boldsymbol{h}},\widetilde
{\boldsymbol{h}}):=$ $\left\langle \widetilde{\boldsymbol{h}},C_{\phi\left(
\boldsymbol{h}\right)  \phi\left(  \boldsymbol{h}\right)  }\widetilde
{\boldsymbol{h}}\right\rangle _{\left(  L^{2}(\Omega\times\mathbb{R})\right)
^{2}}$ is the correlation functional of a Gaussian noise, with mean zero.
Indeed, by using (\ref{Formula_C_phi_phi})-(\ref{Covarience_pi(h)phi(h)}),%
\begin{gather*}
B(\widetilde{\boldsymbol{h}},\widetilde{\boldsymbol{h}})=\left\langle
\widetilde{\boldsymbol{h}},C_{\phi\left(  \boldsymbol{h}\right)  \phi\left(
\boldsymbol{h}\right)  }\widetilde{\boldsymbol{h}}\right\rangle _{\left(
L^{2}(\Omega\times\mathbb{R})\right)  ^{2}}=\\%
{\displaystyle\int\limits_{\mathbb{R}^{2}}}
\text{ \ }%
{\displaystyle\int\limits_{\Omega^{2}}}
\text{ }\widetilde{\boldsymbol{h}}\left(  x,t_{1}\right)  \left\{
{\displaystyle\int\limits_{\Omega^{2}}}
\phi\left(  \boldsymbol{h}\left(  y_{2},t_{2}\right)  \right)
K_{\boldsymbol{J}}\left(  x,u_{2},y,y_{2}\right)  \phi\left(  \boldsymbol{h}%
\left(  u_{2},t_{1}\right)  \right)  d\mu\left(  y_{2}\right)  d\mu\left(
u_{2}\right)  \right\}  \times\\
\widetilde{\boldsymbol{h}}\left(  y,t_{2}\right)  d\mu\left(  x\right)
d\mu\left(  y\right)  dt_{1}dt_{2}=\\
=\text{ \ }%
{\displaystyle\int\limits_{\Omega^{2}}}
\left\{
{\displaystyle\int\limits_{\mathbb{R}}}
\widetilde{\boldsymbol{h}}\left(  x,t_{1}\right)  \phi\left(  \boldsymbol{h}%
\left(  u_{2},t_{1}\right)  \right)  dt_{1}\right\}  \times\\
\left\{
{\displaystyle\int\limits_{\Omega^{2}}}
K_{\boldsymbol{J}}\left(  x,u_{2},y,y_{2}\right)  \left\{
{\displaystyle\int\limits_{\mathbb{R}}}
\widetilde{\boldsymbol{h}}\left(  y,t_{2}\right)  \phi\left(  \boldsymbol{h}%
\left(  y_{2},t_{2}\right)  \right)  dt_{2}\right\}  d\mu\left(  y\right)
d\mu\left(  y_{2}\right)  \right\}  d\mu\left(  x\right)  d\mu\left(
u_{2}\right) \\
=\left\langle
{\displaystyle\int\limits_{\mathbb{R}}}
\widetilde{\boldsymbol{h}}\left(  x,t_{1}\right)  \phi\left(  \boldsymbol{h}%
\left(  u_{2},t_{1}\right)  \right)  dt_{1},\square\left\{
{\displaystyle\int\limits_{\mathbb{R}}}
\widetilde{\boldsymbol{h}}\left(  \cdot,t_{2}\right)  \phi\left(
\boldsymbol{h}\left(  -,t_{2}\right)  \right)  dt_{2}\right\}  \left(
x,u_{2}\right)  \right\rangle _{L^{2}(\Omega^{2})}.
\end{gather*}
We now use the fact that operator $\square:L^{2}(\Omega^{2})\rightarrow
L^{2}(\Omega^{2})$ is positive definite to conclude that $B(\widetilde
{\boldsymbol{h}},\widetilde{\boldsymbol{h}})\geq0$. By Lemma \ref{Lemma_9},
$B(\widetilde{\boldsymbol{h}},\widetilde{\boldsymbol{h}})$ is a continuous,
positive definite bilinear form, then, there exists a Gaussian measure
$\boldsymbol{P}\left(  \zeta\right)  $, with mean zero, satisfying
\begin{gather*}%
{\displaystyle\int\limits_{L^{2}(\Omega\times\mathbb{R})}}
\exp\sqrt{-1}\left\langle \widetilde{\boldsymbol{h}},C_{\phi\left(
\boldsymbol{h}\right)  \phi\left(  \boldsymbol{h}\right)  }\zeta\right\rangle
_{\left(  L^{2}(\Omega\times\mathbb{R})\right)  ^{2}}d\boldsymbol{P}\left(
\zeta\right)  =\\%
{\displaystyle\int\limits_{L^{2}(\Omega\times\mathbb{R})}}
\exp\sqrt{-1}\left\langle \widetilde{\boldsymbol{h}},\square_{C_{\phi\left(
\boldsymbol{h}\right)  \phi\left(  \boldsymbol{h}\right)  }}\widetilde
{\boldsymbol{h}}\zeta\right\rangle _{L^{2}(\Omega\times\mathbb{R}%
)}d\boldsymbol{P}\left(  \zeta\right)  =\\
\exp\frac{1}{2}\left\langle \widetilde{\boldsymbol{h}},C_{\phi\left(
\boldsymbol{h}\right)  \phi\left(  \boldsymbol{h}\right)  }\widetilde
{\boldsymbol{h}}\right\rangle _{\left(  L^{2}(\Omega\times\mathbb{R})\right)
^{2}},
\end{gather*}
see Remark \ref{Nota1}.
\end{remark}

\begin{theorem}
\label{Porp3}The averaged partition function with cutoff is given by%
\begin{gather}
\overline{{\LARGE Z}}_{M}=%
{\displaystyle\iint\limits_{\mathcal{H}\times\mathcal{H}}}
{\large 1}_{\mathcal{P}_{M}}\left(  \boldsymbol{h},\widetilde{\boldsymbol{h}%
}\right)  \exp\left(  S_{0}\left[  \widetilde{\boldsymbol{h}},\boldsymbol{h}%
\right]  +\left\langle \widetilde{\boldsymbol{h}},C_{\phi\left(
\boldsymbol{h}\right)  \phi\left(  \boldsymbol{h}\right)  }\widetilde
{\boldsymbol{h}}\right\rangle _{\left(  L^{2}(\Omega\times\mathbb{R})\right)
^{2}}\right)  \times\label{Eq_Z_M_averaged}\\
d\boldsymbol{P}\left(  \boldsymbol{h}\right)  d\widetilde{\boldsymbol{P}%
}\left(  \widetilde{\boldsymbol{h}}\right)  ,\nonumber
\end{gather}
where $S_{0}\left[  \widetilde{\boldsymbol{h}},\boldsymbol{h}\right]
=-\left\langle \widetilde{\boldsymbol{h}},\left(  \partial_{t}+\gamma\right)
\boldsymbol{h}\right\rangle +\frac{1}{2}\sigma^{2}\left\langle \widetilde
{\boldsymbol{h}},\widetilde{\boldsymbol{h}}\right\rangle $. The term
\[
\left\langle \widetilde{\boldsymbol{h}},C_{\phi\left(  \boldsymbol{h}\right)
\phi\left(  \boldsymbol{h}\right)  }\widetilde{\boldsymbol{h}}\right\rangle
_{\left(  L^{2}(\Omega\times\mathbb{R})\right)  ^{2}}%
\]
is defined for any $\left(  \boldsymbol{h},\widetilde{\boldsymbol{h}}\right)
\in\left(  L^{2}(\Omega\times\mathbb{R})\right)  ^{2}$; consequently, it is
not affected by the support of the cutoff function ${\large 1}_{\mathcal{P}%
_{M}}\left(  \boldsymbol{h},\widetilde{\boldsymbol{h}}\right)  $. Furthermore,
the bilinear form $B(\widetilde{\boldsymbol{h}},\widetilde{\boldsymbol{h}})=$
$\left\langle \widetilde{\boldsymbol{h}},C_{\phi\left(  \boldsymbol{h}\right)
\phi\left(  \boldsymbol{h}\right)  }\widetilde{\boldsymbol{h}}\right\rangle
_{\left(  L^{2}(\Omega\times\mathbb{R})\right)  ^{2}}$ is the correlation
functional of a Gaussian noise, with mean zero.
\end{theorem}

\begin{proof}
The result follows from Lemmas \ref{Lemma_7}-\ref{Lemma_9}, by using Fubini's
theorem, and Remark \ref{Nota2}.
\end{proof}

\section{\label{Appen_E}Appendix E}

\subsection{A formula for $\overline{{\protect\LARGE Z}}_{M}$}

We now set%
\[
\boldsymbol{z}\left(  y,t\right)  =\left[
\begin{array}
[c]{c}%
\boldsymbol{h}\left(  y,t\right) \\
\\
\widetilde{\boldsymbol{h}}\left(  y,t\right)
\end{array}
\right]  \text{, \ }\Xi\left(  x,y,t_{1},t_{2}\right)  =\left[
\begin{array}
[c]{ccc}%
\Xi_{\boldsymbol{hh}} &  & \Xi_{\boldsymbol{h}\widetilde{\boldsymbol{h}}}\\
&  & \\
\Xi_{\widetilde{\boldsymbol{h}}\boldsymbol{h}} &  & \Xi_{\widetilde
{\boldsymbol{h}}\widetilde{\boldsymbol{h}}}%
\end{array}
\right]  ,
\]
with%
\[
\Xi_{\boldsymbol{hh}}=0\text{, }\Xi_{\boldsymbol{h}\widetilde{\boldsymbol{h}}%
}=\delta\left(  y-x\right)  \delta\left(  t_{2}-t_{1}\right)  \left(
\partial_{t_{2}}-\gamma\right)  \text{,}%
\]%
\[
\Xi_{\widetilde{\boldsymbol{h}}\boldsymbol{h}}=-\delta\left(  y-x\right)
\delta\left(  t_{2}-t_{1}\right)  \left(  \partial_{t_{2}}+\gamma\right)
\text{,}%
\]%
\[
\Xi_{\widetilde{\boldsymbol{h}}\widetilde{\boldsymbol{h}}}=C_{\phi\left(
\boldsymbol{h}\right)  \phi\left(  \boldsymbol{h}\right)  }\left(
x,y,t_{1},t_{2}\right)  +\sigma^{2}\delta\left(  y-x\right)  \delta\left(
t_{2}-t_{1}\right)  .
\]

\begin{theorem}
\label{Theorem2} Assume that $\mathcal{D}(\Omega)\subset L^{2}(\Omega
)\subset\mathcal{D}^{\prime}(\Omega)$ is is a Gel'fand triplet, where
$\mathcal{D}(\Omega)$ is a nuclear space of continuous functions, and $\Omega$
is contained in additive group. Set%
\[
S\left[  \boldsymbol{h},\widetilde{\boldsymbol{h}};\Xi\right]  :=\frac{1}{2}%
{\displaystyle\int\limits_{\Omega^{2}}}
\text{ }%
{\displaystyle\int\limits_{\mathbb{R}^{2}}}
\text{\ }\boldsymbol{z}\left(  x,t_{1}\right)  ^{T}\Xi\left(  x,y,t_{1}%
,t_{2}\right)  \boldsymbol{z}\left(  y,t_{2}\right)  d\mu\left(  x\right)
d\mu\left(  y\right)  dt_{1}dt_{2}.
\]
Then%
\begin{equation}
\overline{{\LARGE Z}}_{M}=%
{\displaystyle\iint\limits_{\mathcal{H}\times\mathcal{H}}}
{\large 1}_{\mathcal{P}_{M}}\left(  \boldsymbol{h},\widetilde{\boldsymbol{h}%
}\right)  \exp\left(  S\left[  \boldsymbol{h},\widetilde{\boldsymbol{h}}%
;\Xi\right]  \right)  d\boldsymbol{P}\left(  \boldsymbol{h}\right)
d\widetilde{\boldsymbol{P}}\left(  \widetilde{\boldsymbol{h}}\right)  ,
\label{Formula_6}%
\end{equation}
and the functions $\Xi$, $S\left[  \boldsymbol{h},\widetilde{\boldsymbol{h}%
};\Xi\right]  $ are well-defined in $L^{2}\left(  \Omega\times\mathbb{R}%
\right)  \times L^{2}\left(  \Omega\times\mathbb{R}\right)  $, so they do not
depend on the cutoff function ${\large 1}_{\mathcal{P}_{M}}\left(
\boldsymbol{h},\widetilde{\boldsymbol{h}}\right)  $.
\end{theorem}

\begin{proof}
Result follows from Theorem \ref{Porp3} and (\ref{Formula_6}). To verify
(\ref{Formula_6}), we proceed as follows. First, we note that%
\begin{gather*}
\mathcal{J}\left(  y,x,t_{1},t_{2}\right)  :=\left[
\begin{array}
[c]{c}%
\boldsymbol{h}\left(  x,t_{1}\right)  \\
\\
\widetilde{\boldsymbol{h}}\left(  x,t_{1}\right)
\end{array}
\right]  ^{T}\text{\ }\Xi\left(  y,x,t_{1},t_{2}\right)  \left[
\begin{array}
[c]{c}%
\boldsymbol{h}\left(  y,t_{2}\right)  \\
\\
\widetilde{\boldsymbol{h}}\left(  y,t_{2}\right)
\end{array}
\right]  =\\
\boldsymbol{h}\left(  x,t_{1}\right)  \delta\left(  y-x\right)  \delta\left(
t_{2}-t_{1}\right)  \left(  \partial_{t_{2}}-\gamma\right)  \widetilde
{\boldsymbol{h}}\left(  y,t_{2}\right)  \\
-\widetilde{\boldsymbol{h}}\left(  x,t_{1}\right)  \delta\left(  x-y\right)
\delta\left(  t_{2}-t_{1}\right)  \left(  \partial_{t_{2}}+\gamma\right)
\boldsymbol{h}\left(  y,t_{2}\right)  +\widetilde{\boldsymbol{h}}\left(
x,t_{1}\right)  \Xi_{\widetilde{\boldsymbol{h}}\widetilde{\boldsymbol{h}}%
}\widetilde{\boldsymbol{h}}\left(  y,t_{2}\right)  .
\end{gather*}
Set%
\begin{multline*}
\mathcal{J}_{0}\left(  y,x,t_{1},t_{2}\right)  =\boldsymbol{h}\left(
x,t_{1}\right)  \delta\left(  y-x\right)  \delta\left(  t_{2}-t_{1}\right)
\left(  \partial_{t_{2}}-\gamma\right)  \widetilde{\boldsymbol{h}}\left(
y,t_{2}\right)  \\
-\widetilde{\boldsymbol{h}}\left(  x,t_{1}\right)  \delta\left(  y-x\right)
\delta\left(  t_{2}-t_{1}\right)  \left(  \partial_{t_{2}}+\gamma\right)
\boldsymbol{h}\left(  y,t_{2}\right)  ,
\end{multline*}%
\[
\mathcal{J}_{1}\left(  y,x,t_{1},t_{2}\right)  =\widetilde{\boldsymbol{h}%
}\left(  x,t_{1}\right)  \Xi_{\widetilde{\boldsymbol{h}}\widetilde
{\boldsymbol{h}}}\widetilde{\boldsymbol{h}}\left(  y,t_{2}\right)  ,
\]
so $\mathcal{J}\left(  y,x,t_{1},t_{2}\right)  =\mathcal{\mathcal{J}}%
_{0}\left(  y,x,t_{1},t_{2}\right)  +\mathcal{J}_{1}\left(  y,x,t_{1}%
,t_{2}\right)  $. Now%
\begin{align*}
&  \frac{1}{2}%
{\displaystyle\int\limits_{\Omega^{2}}}
\text{ }%
{\displaystyle\int\limits_{\mathbb{R}^{2}}}
\text{\ }\mathcal{J}_{0}\left(  y,x,t_{1},t_{2}\right)  d\mu\left(  x\right)
d\mu\left(  y\right)  dt_{1}dt_{2}\\
&  =\frac{1}{2}%
{\displaystyle\int\limits_{\Omega^{2}}}
\text{ }%
{\displaystyle\int\limits_{\mathbb{R}^{2}}}
\left\{  \boldsymbol{h}\left(  x,t_{1}\right)  \delta\left(  y-x\right)
\delta\left(  t_{2}-t_{1}\right)  \left(  \partial_{t_{2}}-\gamma\right)
\widetilde{\boldsymbol{h}}\left(  y,t_{2}\right)  \right\}  d\mu\left(
x\right)  d\mu\left(  y\right)  dt_{1}dt_{2}\\
&  -\frac{1}{2}%
{\displaystyle\int\limits_{\Omega^{2}}}
\text{ }%
{\displaystyle\int\limits_{\mathbb{R}^{2}}}
\left\{  \widetilde{\boldsymbol{h}}\left(  x,t_{1}\right)  \delta\left(
y-x\right)  \delta\left(  t_{2}-t_{1}\right)  \left(  \partial_{t_{2}}%
+\gamma\right)  \boldsymbol{h}\left(  y,t_{2}\right)  \right\}  d\mu\left(
x\right)  d\mu\left(  y\right)  dt_{1}dt_{2}%
\end{align*}%
\begin{align*}
&  =\frac{1}{2}%
{\displaystyle\int\limits_{\Omega^{2}}}
\text{ }\left\{
{\displaystyle\int\limits_{\mathbb{R}}}
\boldsymbol{h}\left(  x,t_{1}\right)  \partial_{t_{1}}\widetilde
{\boldsymbol{h}}\left(  x,t_{1}\right)  dt_{1}\right\}  d\mu\left(  x\right)
\\
&  -\frac{1}{2}%
{\displaystyle\int\limits_{\Omega^{2}}}
\left\{  \text{ }%
{\displaystyle\int\limits_{\mathbb{R}}}
\widetilde{\boldsymbol{h}}\left(  x,t_{1}\right)  \partial_{t_{2}%
}\boldsymbol{h}\left(  x,t_{1}\right)  dt_{1}\right\}  d\mu\left(  x\right)
\\
&  -\gamma%
{\displaystyle\int\limits_{\Omega^{2}}}
\text{ }%
{\displaystyle\int\limits_{\mathbb{R}}}
\boldsymbol{h}\left(  x,t_{1}\right)  \widetilde{\boldsymbol{h}}\left(
x,t_{1}\right)  d\mu\left(  x\right)  dt_{1}.
\end{align*}
Integrating by parts,%
\begin{align*}
&
{\displaystyle\int\limits_{\Omega}}
\text{ }\left\{
{\displaystyle\int\limits_{\mathbb{R}}}
\boldsymbol{h}\left(  x,t_{1}\right)  \partial_{t_{1}}\widetilde
{\boldsymbol{h}}\left(  x,t_{1}\right)  dt_{1}\right\}  d\mu\left(  x\right)
\\
&  =%
{\displaystyle\int\limits_{\Omega}}
\text{ }\left(  \left.  \boldsymbol{h}\left(  x,t_{1}\right)  \widetilde
{\boldsymbol{h}}\left(  x,t_{1}\right)  \right\vert _{-\infty}^{\infty
}\right)  d\mu\left(  x\right)  \\
&  -%
{\displaystyle\int\limits_{\Omega}}
\left\{  \text{ }%
{\displaystyle\int\limits_{\mathbb{R}}}
\widetilde{\boldsymbol{h}}\left(  x,t_{1}\right)  \partial_{t_{1}%
}\boldsymbol{h}\left(  x,t_{1}\right)  dt_{1}\right\}  d\mu\left(  y\right)  .
\end{align*}
We now use that \ for fixed $x$, $y$ outside of a set of measure zero, the
fact that $\boldsymbol{h}\left(  x,\cdot\right)  \widetilde{\text{,
}\boldsymbol{h}}\left(  y,\cdot\right)  \in L^{2}(\mathbb{R})$ implies that%
\[
\lim_{t_{1}\rightarrow\pm\infty}\boldsymbol{h}\left(  x,t_{1}\right)
=\lim_{t_{1}\rightarrow\pm\infty}\widetilde{\boldsymbol{h}}\left(
x,t_{1}\right)  =0,
\]
therefore%
\[
\frac{1}{2}%
{\displaystyle\iint\limits_{\Omega\times\Omega}}
\text{ }%
{\displaystyle\int\limits_{\mathbb{R}^{2}}}
\text{\ }\mathcal{J}_{0}\left(  y,x,t_{1},t_{2}\right)  d\mu\left(  x\right)
d\mu\left(  y\right)  dt_{1}dt_{2}=-\left\langle \widetilde{\boldsymbol{h}%
},\left(  \partial_{t}+\gamma\right)  \boldsymbol{h}\right\rangle .
\]
On the other hand,%
\begin{align*}
&  \frac{1}{2}%
{\displaystyle\int\limits_{\Omega^{2}}}
\text{ }%
{\displaystyle\iint\limits_{\mathbb{R}\times\mathbb{R}}}
\text{\ }\mathcal{J}_{1}\left(  y,x,t_{1},t_{2}\right)  d\mu\left(  x\right)
d\mu\left(  y\right)  dt_{1}dt_{2}\\
&  =\frac{1}{2}%
{\displaystyle\int\limits_{\Omega^{2}}}
\text{ }%
{\displaystyle\int\limits_{\mathbb{R}^{2}}}
\widetilde{\boldsymbol{h}}\left(  x,t_{1}\right)  C_{\phi\left(
\boldsymbol{h}\right)  \phi\left(  \boldsymbol{h}\right)  }\left(
x,y,t_{1},t_{2}\right)  \widetilde{\boldsymbol{h}}\left(  y,t_{2}\right)
d\mu\left(  x\right)  d\mu\left(  y\right)  dt_{1}dt_{2}+\\
&  \frac{1}{2}%
{\displaystyle\int\limits_{\Omega^{2}}}
\text{ }%
{\displaystyle\int\limits_{\mathbb{R}^{2}}}
\sigma^{2}\left\{  \widetilde{\boldsymbol{h}}\left(  x,t_{1}\right)
\delta\left(  y-x\right)  \delta\left(  t_{2}-t_{1}\right)  \widetilde
{\boldsymbol{h}}\left(  y,t_{2}\right)  \right\}  d\mu\left(  x\right)
d\mu\left(  y\right)  dt_{1}dt_{2},
\end{align*}
and%
\begin{align*}
&
{\displaystyle\int\limits_{\Omega^{2}}}
\text{ }%
{\displaystyle\int\limits_{\mathbb{R}^{2}}}
\left\{  \widetilde{\boldsymbol{h}}\left(  x,t_{1}\right)  \delta\left(
y-x\right)  \delta\left(  t_{2}-t_{1}\right)  \widetilde{\boldsymbol{h}%
}\left(  y,t_{2}\right)  \right\}  d\mu\left(  x\right)  d\mu\left(  y\right)
dt_{1}dt_{2}\\
&  =\frac{1}{2}\sigma^{2}%
{\displaystyle\int\limits_{\Omega}}
{\displaystyle\int\limits_{\mathbb{R}}}
\left\{  \widetilde{\boldsymbol{h}}\left(  x,t_{1}\right)  \widetilde
{\boldsymbol{h}}\left(  x,t_{1}\right)  \right\}  d\mu\left(  x\right)
dt_{1}=\frac{1}{2}\sigma^{2}\left\langle \widetilde{\boldsymbol{h}}%
,\widetilde{\boldsymbol{h}}\right\rangle .
\end{align*}

\end{proof}

\subsection{A Remark on Gaussian random variables}

\begin{lemma}
\label{Lemma_16}There exists\ a subset $\mathcal{M}\subset\Omega
\times\mathbb{R}$ with $d\mu dt$-measure zero, such that for any $\left(
x,t\right)  \in\Omega\times\mathbb{R\smallsetminus}\mathcal{M}$, the mapping%
\[%
\begin{array}
[c]{cccc}%
\operatorname{eval}_{\left(  x,t\right)  } & \mathcal{H} & \rightarrow &
\mathbb{R}\\
&  &  & \\
& \boldsymbol{h} & \rightarrow & \boldsymbol{h}\left(  x,t\right)
\end{array}
\]
is an ordinary Gaussian random variable with mean zero.
\end{lemma}

\begin{proof}
We first show that $\operatorname{eval}_{\left(  x,t\right)  }$ is a
well-defined linear functional on $\mathcal{H}$. Indeed, by Lemma
\ref{Lemma_2}-(ii) in Appendix A, given orthonormal bases $\left\{
e_{n}\left(  x\right)  \right\}  _{n\in\mathbb{N}}$ in $L^{2}\left(
\Omega,d\mu\right)  $ and $\left\{  f_{m}\left(  t\right)  \right\}
_{m\in\mathbb{N}}$ in $W_{1}(\mathbb{R})$, $\left\{  e_{n}\left(  x\right)
f_{m}\left(  t\right)  \right\}  _{n,m\in\mathbb{N}}$ is an orthonormal basis
of $L^{2}\left(  \Omega,d\mu;W_{1}(\mathbb{R})\right)  \simeq L^{2}\left(
\Omega\right)
{\textstyle\bigotimes}
W_{1}(\mathbb{R})$. Since $\mathcal{D}\left(  \Omega\right)  \subset
L^{2}\left(  \Omega\right)  $ is a dense space of continuous functions in
$L^{2}\left(  \Omega\right)  $, which is a separable space, then
$\mathcal{D}\left(  \Omega\right)  $ is separable, and by using the
Gram--Schmidt procedure, we can assume that all the $e_{n}\left(  x\right)
\in\mathcal{D}\left(  \Omega\right)  $, and also that all the $f_{m}\left(
t\right)  \in\mathcal{S}(\mathbb{R})$, see Remark \ref{Nota1A} in Appendix A.
Therefore, $\operatorname{eval}_{\left(  x,t\right)  }\left(  \boldsymbol{h}%
\right)  \in\mathbb{R}$ for any $\boldsymbol{h}\in\mathcal{H}$.

We now show the existence of\ a subset $\mathcal{M}\subset\Omega
\times\mathbb{R}$ with $d\mu dt$-measure zero, such that $\operatorname{eval}%
_{\left(  x,t\right)  }$ is a continuous linear functional for any $\left(
x,t\right)  \in\Omega\times\mathbb{R\smallsetminus}\mathcal{M}$. Take a
sequence $\left\{  \boldsymbol{h}_{n}\right\}  _{n\in\mathbb{N}}$ converging
to $\boldsymbol{0}$ in $\mathcal{H}$. We must show that $\operatorname{eval}%
_{\left(  x,t\right)  }\left(  \boldsymbol{h}_{n}\right)  \rightarrow
\boldsymbol{0}$. By using that $\left\Vert \cdot\right\Vert _{\mathcal{H}}%
\geq\left\Vert \cdot\right\Vert $, cf. Lemma \ref{Lemma_3} in Appendix A,\ we
have $\boldsymbol{h}_{n}\rightarrow\boldsymbol{0}$ in $L^{2}\left(
\Omega\times\mathbb{R}\right)  $. Then, by passing to a subsequence, there
exists\ a subset $\mathcal{M}\subset\Omega\times\mathbb{R}$ with $d\mu
dt$-measure zero, such that $\boldsymbol{h}_{n}\left(  y,s\right)
\rightarrow\boldsymbol{0}$ uniformly in $\Omega\times\mathbb{R}$
$\mathbb{\smallsetminus}$ $\mathcal{M}$, cf. \cite[Theorem 2.5.1, 2.5.3]{Ash}.
Now, by the Riesz representation theorem, for $\left(  x,t\right)  \in
\Omega\times\mathbb{R}$ $\mathbb{\smallsetminus}$ $\mathcal{M}$, there exists
a unique $\boldsymbol{v}=\boldsymbol{v}\left(  x,t\right)  \in\mathcal{H}$
\ such that $\operatorname{eval}_{\left(  x,t\right)  }\left(  \boldsymbol{h}%
\right)  =\left\langle \boldsymbol{h},\boldsymbol{v}\right\rangle
_{\mathcal{H}}$. Finally, we use that $\boldsymbol{h}\rightarrow\left\langle
\boldsymbol{h},\boldsymbol{v}\right\rangle _{\mathcal{H}}$ is an ordinary
Gaussian random variable with mean zero and variance $\left\Vert
\boldsymbol{v}\right\Vert _{\mathcal{H}}$, cf. \cite[Lemmas 2.1.3 and
2.1.5]{Obata}.
\end{proof}

\section{\label{Appen_F}Appendix F}

\begin{lemma}
\label{Lemma_14}Let $\mathbb{P}_{\boldsymbol{J}}$ be the Gaussian probability
measure with mean zero and covariance operator $\square_{\boldsymbol{J}}$
attached to the random coupling kernel $\boldsymbol{J}$. With the above notation,%

\begin{gather*}%
{\displaystyle\int\limits_{L^{2}\left(  \Omega\times\Omega\right)  }}
\exp\left(  -S_{int}\left(  \widetilde{\boldsymbol{h}}^{\alpha},\boldsymbol{h}%
^{\alpha},\boldsymbol{J}\right)  \right)  d\mathbb{P}_{\boldsymbol{J}}=\\
\exp\Bigg(\frac{1}{2}%
{\displaystyle\sum\limits_{\alpha=1}^{2}}
\text{ }%
{\displaystyle\sum\limits_{\beta=1}^{2}}
\text{\ }%
{\displaystyle\int\limits_{\mathbb{R}^{2}}}
\text{ }%
{\displaystyle\int\limits_{\Omega^{2}}}
\widetilde{\boldsymbol{h}}^{\alpha}\left(  x,t_{1}\right)  C_{\phi\left(
\widetilde{\boldsymbol{h}}^{\alpha}\right)  \phi\left(  \widetilde
{\boldsymbol{h}}^{\beta}\right)  }^{\alpha\beta}\left(  x,y,t_{1}%
,t_{2}\right)  \widetilde{\boldsymbol{h}}^{\beta}\left(  y,t_{2}\right)
\times\\
\\
d\mu\left(  x\right)  d\mu\left(  y\right)  dt_{1}dt_{2}\Bigg),
\end{gather*}
where $C_{\phi\left(  \boldsymbol{h}^{\alpha}\right)  \phi\left(
\boldsymbol{h}^{\beta}\right)  }^{\alpha\beta}\left(  x,y,t_{1},t_{2}\right)
$ is determined by the kernel $K_{\boldsymbol{J}}$, see (\ref{Eq_8}).
\end{lemma}

\begin{proof}
We first note that%
\[
-S_{int}\left(  \widetilde{\boldsymbol{h}}^{\alpha},\boldsymbol{h}^{\alpha
},\boldsymbol{J}\right)  =\left\langle \boldsymbol{J}(x,y),\left(  -1\right)
{\displaystyle\sum\limits_{\alpha=1}^{2}}
\text{ }%
{\displaystyle\int\limits_{\mathbb{R}}}
\widetilde{\boldsymbol{h}}^{\alpha}\left(  x,t\right)  \phi\left(
\boldsymbol{h}^{\alpha}\left(  y,t\right)  \right)  dt\right\rangle
_{L^{2}\left(  \Omega^{2}\right)  },
\]
and thus by (\ref{Key_formula}),%
\begin{gather*}%
{\displaystyle\int\limits_{L^{2}\left(  \mathbb{Z}_{p}^{2}\right)  }}
\exp\left(  S_{int}\left(  \widetilde{\boldsymbol{h}}^{\alpha},\boldsymbol{h}%
^{\alpha},\boldsymbol{J}\right)  \right)  d\boldsymbol{P}_{\boldsymbol{J}%
}=\exp\frac{1}{2}\times\\
\left\langle \square_{\boldsymbol{J}}\left(
{\displaystyle\sum\limits_{\alpha=1}^{2}}
{\displaystyle\int\limits_{\mathbb{R}}}
\widetilde{\boldsymbol{h}}^{\alpha}\left(  x,t\right)  \phi\left(
\boldsymbol{h}^{\alpha}\left(  y,t\right)  \right)  dt\right)  ,%
{\displaystyle\sum\limits_{\beta=1}^{2}}
{\displaystyle\int\limits_{\mathbb{R}}}
\widetilde{\boldsymbol{h}}^{\beta}\left(  x,t\right)  \phi\left(
\boldsymbol{h}^{\beta}\left(  y,t\right)  \right)  dt\right\rangle
_{L^{2}\left(  \Omega^{2}\right)  }.
\end{gather*}

Now,%
\begin{gather*}
\left\langle \square_{\boldsymbol{J}}\left(
{\displaystyle\sum\limits_{\alpha=1}^{2}}
\text{ }%
{\displaystyle\int\limits_{\mathbb{R}}}
\widetilde{\boldsymbol{h}}^{\alpha}\left(  x,t\right)  \phi\left(
\boldsymbol{h}^{\alpha}\left(  y,t\right)  \right)  dt\right)  ,%
{\displaystyle\sum\limits_{\beta=1}^{2}}
\text{ }%
{\displaystyle\int\limits_{\mathbb{R}}}
\widetilde{\boldsymbol{h}}^{\beta}\left(  x,t\right)  \phi\left(
\boldsymbol{h}^{\beta}\left(  y,t\right)  \right)  dt\right\rangle
_{L^{2}\left(  \Omega^{2}\right)  }=\\%
{\displaystyle\sum\limits_{\alpha=1}^{2}}
\text{ }%
{\displaystyle\sum\limits_{\beta=1}^{2}}
\left\langle \square_{\boldsymbol{J}}\left(
{\displaystyle\int\limits_{\mathbb{R}}}
\widetilde{\boldsymbol{h}}^{\alpha}\left(  x,t\right)  \phi\left(
\boldsymbol{h}^{\alpha}\left(  y,t\right)  \right)  dt\right)  ,%
{\displaystyle\int\limits_{\mathbb{R}}}
\widetilde{\boldsymbol{h}}^{\beta}\left(  x,t\right)  \phi\left(
\boldsymbol{h}^{\beta}\left(  y,t\right)  \right)  dt\right\rangle
_{L^{2}\left(  \Omega^{2}\right)  }=\\%
{\displaystyle\sum\limits_{\alpha=1}^{2}}
\text{ }%
{\displaystyle\sum\limits_{\beta=1}^{2}}
\text{ \ }%
{\displaystyle\int\limits_{\mathbb{R}^{2}}}
\text{ \ }%
{\displaystyle\int\limits_{\Omega^{2}}}
\widetilde{\boldsymbol{h}}^{\alpha}\left(  x,t_{1}\right)  C_{\phi\left(
\boldsymbol{h}^{\alpha}\right)  \phi\left(  \boldsymbol{h}^{\beta}\right)
}^{\alpha\beta}\left(  x,y,t_{1},t_{2}\right)  \widetilde{\boldsymbol{h}%
}^{\beta}\left(  y,t_{2}\right)  d\mu\left(  x\right)  d\mu\left(  y\right)
dt_{1}dt_{2},
\end{gather*}
where%
\begin{gather}
C_{\phi\left(  \boldsymbol{h}^{\alpha}\right)  \phi\left(  \boldsymbol{h}%
^{\beta}\right)  }^{\alpha\beta}\left(  x,y,t_{1},t_{2}\right)  =
\label{Eq_8}\\%
{\displaystyle\int\limits_{\Omega^{2}}}
\phi\left(  \boldsymbol{h}^{\alpha}\left(  z,t_{2}\right)  \right)
K_{\boldsymbol{J}}\left(  x,u_{2},y,z\right)  \phi\left(  \boldsymbol{h}%
^{\beta}\left(  u_{2},t_{1}\right)  \right)  d\mu\left(  z\right)  d\mu\left(
u_{2}\right)  .\nonumber
\end{gather}

\end{proof}

\begin{remark}
\label{Nota_C_alpha_beta}We note that
\[
C_{\phi\left(  \boldsymbol{h}^{\alpha}\right)  \phi\left(  \boldsymbol{h}%
^{\beta}\right)  }^{\alpha\beta}=C_{\phi\left(  \boldsymbol{h}^{\alpha
}\right)  \phi\left(  \boldsymbol{h}^{\beta}\right)  }^{\beta\alpha}\text{,
for }\alpha\neq\beta,
\]
and
\[
C_{\phi\left(  \boldsymbol{h}^{1}\right)  \phi\left(  \boldsymbol{h}%
^{1}\right)  }^{11}=C_{\phi\left(  \boldsymbol{h}^{1}\right)  \phi\left(
\boldsymbol{h}^{1}\right)  }^{22}.
\]

\end{remark}

We set
\[
\left(  \square_{C_{\phi\left(  \boldsymbol{h}^{\alpha}\right)  \phi\left(
\boldsymbol{h}^{\beta}\right)  }^{\alpha\beta}}\widetilde{\boldsymbol{h}%
}^{\beta}\right)  \left(  x,t_{1}\right)  :=%
{\displaystyle\int\limits_{\mathbb{R}}}
\text{ \ }%
{\displaystyle\int\limits_{\Omega}}
C_{\phi\left(  \boldsymbol{h}^{\alpha}\right)  \phi\left(  \boldsymbol{h}%
^{\beta}\right)  }^{\alpha\beta}\left(  x,y,t_{1},t_{2}\right)  \widetilde
{\boldsymbol{h}}^{\beta}\left(  y,t_{2}\right)  d\mu\left(  y\right)  dt_{2}.
\]
Then%
\begin{gather*}
\left\langle \widetilde{\boldsymbol{h}}^{\alpha},C_{\phi\left(  \boldsymbol{h}%
^{\alpha}\right)  \phi\left(  \boldsymbol{h}^{\beta}\right)  }^{\alpha\beta
}\widetilde{\boldsymbol{h}}^{\beta}\right\rangle _{\left(  L^{2}(\Omega
\times\mathbb{R})\right)  ^{2}}:=\\%
{\displaystyle\int\limits_{\mathbb{R}^{2}}}
\text{ \ }%
{\displaystyle\int\limits_{\Omega^{2}}}
\widetilde{\boldsymbol{h}}^{\alpha}\left(  x,t_{1}\right)  C_{\phi\left(
\boldsymbol{h}^{\alpha}\right)  \phi\left(  \boldsymbol{h}^{\beta}\right)
}^{\alpha\beta}\left(  x,y,t_{1},t_{2}\right)  \widetilde{\boldsymbol{h}%
}^{\beta}\left(  y,t_{2}\right)  d\mu\left(  x\right)  d\mu\left(  y\right)
dt_{1}dt_{2}\\
=\left\langle \widetilde{\boldsymbol{h}}^{\alpha},\square_{C_{\phi\left(
\boldsymbol{h}^{\alpha}\right)  \phi\left(  \boldsymbol{h}^{\beta}\right)
}^{\alpha\beta}}\widetilde{\boldsymbol{h}}^{\beta}\right\rangle .
\end{gather*}
We note that Lemma \ref{Lemma_9A}\ are valid for $C_{\phi\left(
\boldsymbol{h}^{\alpha}\right)  \phi\left(  \boldsymbol{h}^{\beta}\right)
}^{\alpha\beta}$, and that Lemma \ref{Lemma_9} is valid for $\left\langle
\widetilde{\boldsymbol{h}}^{\alpha},C_{\phi\left(  \boldsymbol{h}^{\alpha
}\right)  \phi\left(  \boldsymbol{h}^{\beta}\right)  }^{\alpha\beta}%
\widetilde{\boldsymbol{h}}^{\beta}\right\rangle _{\left(  L^{2}(\Omega
\times\mathbb{R})\right)  ^{2}}$.

\begin{remark}
\label{Nota2A}The bilinear form%
\[
B^{\alpha\beta}\left(  \widetilde{\boldsymbol{h}}^{\alpha},\widetilde
{\boldsymbol{h}}^{\beta}\right)  =\left\langle \widetilde{\boldsymbol{h}%
}^{\alpha},C_{\phi\left(  \boldsymbol{h}^{\alpha}\right)  \phi\left(
\boldsymbol{h}^{\beta}\right)  }^{\alpha\beta}\widetilde{\boldsymbol{h}%
}^{\beta}\right\rangle _{\left(  L^{2}(\Omega\times\mathbb{R})\right)  ^{2}}%
\]
is the correlation functional of a Gaussian noise with mean zero. The
verification of this assertion uses the argument given in Remark \ref{Nota2}.
\end{remark}

Using the above observations, we can rewrite Lemma \ref{Lemma_14} as follows:

\begin{theorem}
\label{Lemma_15}With the above notation,%
\begin{gather*}
\overline{{\LARGE Z}_{M}^{\left(  2\right)  }}=%
{\displaystyle\iint\limits_{\mathcal{H}\times\mathcal{H}}}
\text{ \ }%
{\displaystyle\iint\limits_{\mathcal{H}\times\mathcal{H}}}
\exp\left(  S_{0}\left(  \widetilde{\boldsymbol{h}}^{\alpha},\boldsymbol{h}%
^{\alpha}\right)  +\frac{\sigma^{2}}{2}\left\langle \widetilde{\boldsymbol{h}%
}^{1},\widetilde{\boldsymbol{h}}^{2}\right\rangle \right)  \times\\
\exp\left(  \frac{1}{2}%
{\displaystyle\sum\limits_{\alpha=1}^{2}}
\text{ }%
{\displaystyle\sum\limits_{\beta=1}^{2}}
\text{ \ }\left\langle \widetilde{\boldsymbol{h}}^{\alpha},C_{\phi\left(
\boldsymbol{h}^{\alpha}\right)  \phi\left(  \boldsymbol{h}^{\beta}\right)
}^{\alpha\beta}\widetilde{\boldsymbol{h}}^{\beta}\right\rangle _{\left(
L^{2}(\Omega\times\mathbb{R})\right)  ^{2}}\right)  \times\\%
{\displaystyle\prod\limits_{\alpha=1}^{2}}
{\large 1}_{\mathcal{P}_{M}}\left(  \boldsymbol{h}^{\alpha},\widetilde
{\boldsymbol{h}}^{\alpha}\right)  d\boldsymbol{P}\left(  \boldsymbol{h}%
^{\alpha}\right)  d\widetilde{\boldsymbol{P}}\left(  \widetilde{\boldsymbol{h}%
}^{\alpha}\right)  ,
\end{gather*}
where the bilinear forms $\left\langle \widetilde{\boldsymbol{h}}^{\alpha
},C_{\phi\left(  \boldsymbol{h}^{\alpha}\right)  \phi\left(  \boldsymbol{h}%
^{\beta}\right)  }^{\alpha\beta}\widetilde{\boldsymbol{h}}^{\beta
}\right\rangle _{\left(  L^{2}(\Omega\times\mathbb{R})\right)  ^{2}}$ are
correlation functionals of a Gaussian noises with mean zero.
\end{theorem}

We now set%
\begin{align*}
\boldsymbol{z}^{\alpha}\left(  y,t\right)   &  =\left[
\begin{array}
[c]{c}%
\boldsymbol{h}^{\alpha}\left(  y,t\right) \\
\\
\widetilde{\boldsymbol{h}}^{\alpha}\left(  y,t\right)
\end{array}
\right]  \text{, }\\
\text{\ }\Xi^{\alpha\beta}\left(  x,y,t_{1},t_{2}\right)   &  =\left[
\begin{array}
[c]{ccc}%
\Xi_{\boldsymbol{hh}}^{\alpha\beta}\left(  x,y,t_{1},t_{2}\right)  &  &
\Xi_{\boldsymbol{h}\widetilde{\boldsymbol{h}}}^{\alpha\beta}\left(
x,y,t_{1},t_{2}\right) \\
&  & \\
\Xi_{\widetilde{\boldsymbol{h}}\boldsymbol{h}}^{\alpha\beta}\left(
x,y,t_{1},t_{2}\right)  &  & \Xi_{\widetilde{\boldsymbol{h}}\widetilde
{\boldsymbol{h}}}^{\alpha\beta}\left(  x,y,t_{1},t_{2}\right)
\end{array}
\right]  ,
\end{align*}
with%
\[
\Xi_{\boldsymbol{hh}}^{\alpha\beta}=0\text{, \ \ }\Xi_{\boldsymbol{h}%
\widetilde{\boldsymbol{h}}}^{\alpha\beta}=\delta_{\alpha\beta}\delta\left(
y-x\right)  \delta\left(  t_{2}-t_{1}\right)  \left(  \partial_{t_{2}}%
-\gamma\right)  \text{,}%
\]%
\[
\Xi_{\widetilde{\boldsymbol{h}}\boldsymbol{h}}^{\alpha\beta}=-\delta
_{\alpha\beta}\delta\left(  y-x\right)  \delta\left(  t_{2}-t_{1}\right)
\left(  \partial_{t_{2}}+\gamma\right)  \text{,}%
\]%
\[
\Xi_{\widetilde{\boldsymbol{h}}\widetilde{\boldsymbol{h}}}^{\alpha\beta
}=C_{\phi\left(  \boldsymbol{h}^{\alpha}\right)  \phi\left(  \boldsymbol{h}%
^{\beta}\right)  }^{\alpha\beta}\left(  x,y,t_{1},t_{2}\right)  +\sigma
^{2}\delta\left(  y-x\right)  \delta\left(  t_{2}-t_{1}\right)  .
\]

\begin{theorem}
\label{Theorem3}With the hypotheses of Theorem \ref{Theorem2}. Set%
\begin{gather*}
S\left[  \boldsymbol{h},\widetilde{\boldsymbol{h}};\left\{  \Xi^{\alpha\beta
}\right\}  _{\alpha,\beta\in\left\{  1,2\right\}  }\right]  :=\\
\frac{1}{2}%
{\displaystyle\int\limits_{\Omega^{2}}}
\text{ }%
{\displaystyle\int\limits_{\mathbb{R}^{2}}}
\boldsymbol{z}^{\alpha}\left(  x,t_{1}\right)  ^{T}\Xi^{\alpha\beta}\left(
x,y,t_{1},t_{2}\right)  \boldsymbol{z}^{\beta}\left(  y,t_{2}\right)
d\mu\left(  x\right)  d\mu\left(  y\right)  dt_{1}dt_{2}.
\end{gather*}
Then%
\begin{multline*}
\overline{{\LARGE Z}_{M}^{\left(  2\right)  }}=%
{\displaystyle\iint\limits_{\mathcal{H}\times\mathcal{H}}}
\text{ \ }%
{\displaystyle\iint\limits_{\mathcal{H}\times\mathcal{H}}}
\exp\left(  S\left[  \boldsymbol{h},\widetilde{\boldsymbol{h}};\left\{
\Xi^{\alpha\beta}\right\}  _{\alpha,\beta\in\left\{  1,2\right\}  }\right]
\right)  \times\\%
{\displaystyle\prod\limits_{\alpha=1}^{2}}
{\large 1}_{\mathcal{P}_{M}}\left(  \boldsymbol{h}^{\alpha},\widetilde
{\boldsymbol{h}}^{\alpha}\right)  d\boldsymbol{P}\left(  \boldsymbol{h}%
^{\alpha}\right)  d\widetilde{\boldsymbol{P}}\left(  \widetilde{\boldsymbol{h}%
}^{\alpha}\right)  ,
\end{multline*}
and the functions $\left\{  \Xi^{\alpha\beta}\right\}  _{\alpha,\beta
\in\left\{  1,2\right\}  }$, $S\left[  \boldsymbol{h},\widetilde
{\boldsymbol{h}};\left\{  \Xi^{\alpha\beta}\right\}  _{\alpha,\beta\in\left\{
1,2\right\}  }\right]  $ are well-defined in $L^{2}\left(  \Omega
\times\mathbb{R}\right)  \times L^{2}\left(  \Omega\times\mathbb{R}\right)  $,
so they do not depend on the cutoff function.
\end{theorem}

\begin{proof}
The proof is a variation of the one given for Theorem \ref{Theorem2}.
\end{proof}

\section{\label{Appen_G}Appendix G}

In this section, we study the following distributions from $\mathcal{S}%
^{\prime}(\mathbb{R}^{2})$:%
\begin{align}
&
{\displaystyle\int\limits_{\Omega^{2}}}
G_{\boldsymbol{hh}}^{\alpha\beta}\left(  x,y,t_{1},t_{2}\right)  \theta\left(
x,y\right)  d\mu\left(  x\right)  d\mu\left(  y\right)  \text{, }%
\label{Correlations_1}\\
&
{\displaystyle\int\limits_{\Omega}}
G_{\boldsymbol{hh}}^{\alpha\beta}\left(  x,x,t_{1},t_{2}\right)  \theta\left(
x,x\right)  d\mu\left(  x\right)  \text{,}\label{Correlations_2}%
\end{align}
where $\theta\left(  x,y\right)  \in\mathcal{D}(\Omega^{2})$. In our view,
these integrals play a central role in the understanding of self-averaging
property of the continuous random NNs. We propose to interpret the above
integrals as spatial-averages of the correlators of the network.

\subsection{Spatial averages of the first type}

We now study the averages of the type (\ref{Correlations_1}). We recall that
we are using the notation $C_{\phi\phi}^{\alpha\beta}=C_{\phi\left(
\boldsymbol{h}^{\alpha}\right)  \phi\left(  \boldsymbol{h}^{\beta}\right)
}^{\alpha\beta}$.

\begin{lemma}
\label{Lemma_12}With the above notation the following assertions hold true:

\noindent(i) $C_{\phi\phi}^{\alpha\beta}\left(  x,y,t_{1},t_{2}\right)  \in
L^{2}\left(  \Omega^{2}\times\mathbb{R}^{2}\right)  $;

\noindent(ii) $C_{\phi\phi}^{\alpha\beta}\left(  \cdot,\cdot,t_{1}%
,t_{2}\right)  \in L^{2}\left(  \Omega^{2}\right)  $, for almost every
$\left(  t_{1},t_{2}\right)  \in\mathbb{R}^{2}$;

\noindent(iii) $C_{\phi\phi}^{\alpha\beta}\left(  x,y,\cdot,\cdot\right)  \in
L^{2}\left(  \mathbb{R}^{2}\right)  $, for almost every $\left(  x,y\right)
\in\Omega^{2}$.
\end{lemma}

\begin{proof}
The parts (ii)-(iii) follows from the first part by the Fubini theorem. To
show the first statement, we proceed as follows. By using the fact that $\phi$
is Lipschitz and the Cauchy-Schwarz inequality in $L^{2}\left(  \Omega
^{2}\right)  $,%
\[
\left\vert C_{\phi\phi}^{\alpha\beta}\left(  x,y,t_{1},t_{2}\right)
\right\vert ^{2}\leq L\left(  \phi\right)  ^{4}\left\Vert \boldsymbol{h}%
^{\alpha}\left(  \cdot,t_{2}\right)  \boldsymbol{h}^{\beta}\left(  \cdot
,t_{1}\right)  \right\Vert _{L^{2}\left(  \Omega^{2}\right)  }^{2}\left\Vert
K_{\boldsymbol{J}}\left(  x,\cdot,y,\cdot\right)  \right\Vert _{L^{2}\left(
\Omega^{2}\right)  }^{2},
\]
see formula (\ref{Eq_8}). Now, the first assertion follows from%
\begin{align*}
\left\Vert C_{\phi\phi}^{\alpha\beta}\right\Vert _{L^{2}\left(  \Omega
^{2}\times\mathbb{R}^{2}\right)  }^{2}  &  \leq L\left(  \phi\right)
^{4}\left(  \text{ }%
{\displaystyle\int\limits_{\mathbb{R}^{2}}}
\left\Vert \boldsymbol{h}^{\alpha}\left(  \cdot,t_{2}\right)  \boldsymbol{h}%
^{\beta}\left(  \cdot,t_{1}\right)  \right\Vert _{L^{2}\left(  \Omega
^{2}\right)  }^{2}dt_{1}dt_{2}\right)  \times\\
&
{\displaystyle\int\limits_{\Omega^{2}}}
\left\Vert K_{\boldsymbol{J}}\left(  x,\cdot,y,\cdot\right)  \right\Vert
_{L^{2}\left(  \Omega^{2}\right)  }^{2}d\mu\left(  x\right)  d\mu\left(
y\right) \\
&  =L\left(  \phi\right)  ^{4}\left\Vert \boldsymbol{h}^{\alpha}\right\Vert
_{L^{2}\left(  \Omega\times\mathbb{R}\right)  }^{2}\left\Vert \boldsymbol{h}%
^{\beta}\right\Vert _{L^{2}\left(  \Omega\times\mathbb{R}\right)  }%
^{2}\left\Vert K_{\boldsymbol{J}}\right\Vert _{L^{2}\left(  \Omega^{4}\right)
}^{2}.
\end{align*}

\end{proof}

We denote by $\mathcal{D}^{\prime}\left(  \Omega^{2}\right)
{\textstyle\bigotimes\nolimits_{\text{alg}}}
\mathcal{S}^{\prime}\left(  \mathbb{R}^{2}\right)  $, the $\mathbb{R}$-vector
space spanned by terms of the form $aT(x,y)G(t_{1},t_{2})$,\ with
$a\in\mathbb{R}$, $T(x,y)\in\mathcal{D}^{\prime}\left(  \Omega^{2}\right)  $,
$G(t_{1},t_{2})\in\mathcal{S}^{\prime}\left(  \mathbb{R}^{2}\right)  $.
Furthermore,
\begin{equation}
\left(  aT(x,y)G(t_{1},t_{2}),\varphi\left(  x,y,t_{1},t_{2}\right)  \right)
=a\left(  T(x,y),\left(  G(t_{1},t_{2}),\varphi\left(  x,y,t_{1},t_{2}\right)
\right)  \right)  ; \label{Pairing}%
\end{equation}
thus $T(x,y)G(t_{1},t_{2})$ is the direct product of the distributions
$T(x,y)$, $G(t_{1},t_{2})$. The space of test functions for the distributions
from $\mathcal{D}^{\prime}\left(  \Omega^{2}\right)
{\textstyle\bigotimes\nolimits_{\text{alg}}}
\mathcal{S}^{\prime}\left(  \mathbb{R}^{2}\right)  $ is $\mathcal{D}\left(
\Omega^{2}\right)
{\textstyle\bigotimes\nolimits_{\text{alg}}}
\mathcal{S}\left(  \mathbb{R}^{2}\right)  $, which means that in
(\ref{Pairing}) the function $\varphi\left(  x,y,t_{1},t_{2}\right)  $ has the
form%
\[
\varphi\left(  x,y,t_{1},t_{2}\right)  =%
{\displaystyle\sum\limits_{i=1}^{m}}
\theta_{i}\left(  x,y\right)  \psi_{i}\left(  t_{1},t_{2}\right)  \text{, }%
\]
with $\theta_{i}\left(  x,y\right)  \in\mathcal{D}\left(  \Omega^{2}\right)
$, $\psi_{i}\left(  t_{1},t_{2}\right)  \in\mathcal{S}\left(  \mathbb{R}%
^{2}\right)  $. Now, by Lemma \ref{Lemma_12}-(i),
\[
C_{\phi\phi}^{\alpha\beta}\left(  x,y,t_{1},t_{2}\right)  \in\mathcal{D}%
^{\prime}\left(  \Omega^{2}\right)
{\textstyle\bigotimes\nolimits_{\text{alg}}}
\mathcal{S}^{\prime}\left(  \mathbb{R}^{2}\right)
\]
because $C_{\phi\phi}^{\alpha\beta}\in L_{\text{loc}}^{1}\left(  \Omega
^{2}\times\mathbb{R}^{2}\right)  $, and then
\[
C_{\phi\phi}^{\alpha\beta}\left(  x,y,t_{1},t_{2}\right)  +\sigma^{2}%
\delta\left(  x-y\right)  \delta\left(  t_{1}-t_{2}\right)  \in\mathcal{D}%
^{\prime}\left(  \Omega^{2}\right)
{\textstyle\bigotimes\nolimits_{\text{alg}}}
\mathcal{S}^{\prime}\left(  \mathbb{R}^{2}\right)  ,
\]
now by the equation (\ref{Eq_15A}),%
\[
\left(  \partial_{1}+\gamma\right)  \left(  \partial_{2}+\gamma\right)
G_{\boldsymbol{hh}}^{\alpha\beta}\left(  x,y,t_{1},t_{2}\right)
\in\mathcal{D}^{\prime}\left(  \Omega^{2}\right)
{\textstyle\bigotimes\nolimits_{\text{alg}}}
\mathcal{S}^{\prime}\left(  \mathbb{R}^{2}\right)  .
\]
For a fixed $\theta\left(  x,y\right)  \in\mathcal{D}(\Omega^{2})$, we set%
\begin{multline*}
\overline{G}_{\boldsymbol{hh}}^{\alpha\beta}\left(  t_{1},t_{2}\right)
=\overline{G}_{\boldsymbol{hh}}^{\alpha\beta}\left(  t_{1},t_{2}%
;\theta\right)  :=\\%
{\displaystyle\int\limits_{\Omega^{2}}}
G_{\boldsymbol{hh}}^{\alpha\beta}\left(  x,y,t_{1},t_{2}\right)  \theta\left(
x,y\right)  d\mu\left(  x\right)  d\mu\left(  y\right)  ,
\end{multline*}%
\[
\overline{C}_{\phi\phi}^{\alpha\beta}\left(  t_{1},t_{2}\right)  =\overline
{C}_{\phi\phi}^{\alpha\beta}\left(  t_{1},t_{2};\theta\right)  :=%
{\displaystyle\int\limits_{\Omega^{2}}}
C_{\phi\phi}^{\alpha\beta}\left(  x,y,t_{1},t_{2}\right)  \theta\left(
x,y\right)  d\mu\left(  x\right)  d\mu\left(  y\right)  ,
\]%
\[
\sigma_{\theta}^{2}:=\sigma^{2}%
{\displaystyle\int\limits_{\Omega^{2}}}
\delta\left(  x-y\right)  \theta\left(  x,y\right)  d\mu\left(  x\right)
d\mu\left(  y\right)  =\sigma^{2}%
{\displaystyle\int\limits_{\Omega}}
\theta\left(  y,y\right)  d\mu\left(  y\right)  ,
\]
in $\mathcal{S}^{\prime}\left(  \mathbb{R}^{2}\right)  $.

\begin{theorem}
\label{Theorem4}With the above notation,
\begin{equation}
\left(  \partial_{1}+\gamma\right)  \left(  \partial_{2}+\gamma\right)
\overline{G}_{\boldsymbol{hh}}^{\alpha\beta}\left(  t_{1},t_{2}\right)
=\overline{C}_{\phi\phi}^{\alpha\beta}\left(  t_{1},t_{2}\right)
+\sigma_{\theta}^{2}\delta\left(  t_{1}-t_{2}\right)  \text{ in }%
\mathcal{S}^{\prime}\left(  \mathbb{R}^{2}\right)  .\nonumber
\end{equation}

\end{theorem}

\begin{proof}
Since the equation (\ref{Eq_15A}) holds in $\mathcal{D}^{\prime}\left(
\Omega^{2}\right)
{\textstyle\bigotimes\nolimits_{\text{alg}}}
\mathcal{S}^{\prime}\left(  \mathbb{R}^{2}\right)  $, by applying these
distributions to $\theta\left(  x,y\right)  \psi\left(  t_{1},t_{2}\right)  $,
we have%
\begin{multline*}
\left(  \left(  \partial_{1}+\gamma\right)  \left(  \partial_{2}%
+\gamma\right)  G_{\boldsymbol{hh}}^{\alpha\beta}\left(  x,y,t_{1}%
,t_{2}\right)  ,\theta\left(  x,y\right)  \psi\left(  t_{1},t_{2}\right)
\right)  =\\
\left(  C_{\phi\left(  \boldsymbol{h}^{\alpha}\right)  \phi\left(
\boldsymbol{h}^{\beta}\right)  }^{\alpha\beta}\left(  x,y,t_{1},t_{2}\right)
,\theta\left(  x,y\right)  \psi\left(  t_{1},t_{2}\right)  \right)  +\\
\left(  \sigma^{2}\delta\left(  x-y\right)  \delta\left(  t_{1}-t_{2}\right)
,\theta\left(  x,y\right)  \psi\left(  t_{1},t_{2}\right)  \right)  .
\end{multline*}
The announced formula, following by computing then following expressions:%
\begin{align*}
\mathfrak{G}  &  \mathfrak{:}=\left(  \left(  \partial_{1}+\gamma\right)
\left(  \partial_{2}+\gamma\right)  G_{\boldsymbol{hh}}^{\alpha\beta}\left(
x,y,t_{1},t_{2}\right)  ,\theta\left(  x,y\right)  \psi\left(  t_{1}%
,t_{2}\right)  \right)  ,\\
\mathfrak{C}  &  \mathfrak{:}=\left(  C_{\phi\left(  \boldsymbol{h}^{\alpha
}\right)  \phi\left(  \boldsymbol{h}^{\beta}\right)  }^{\alpha\beta}\left(
x,y,t_{1},t_{2}\right)  ,\theta\left(  x,y\right)  \psi\left(  t_{1}%
,t_{2}\right)  \right)  ,\\
\mathfrak{D}  &  \mathfrak{:}=\left(  \sigma^{2}\delta\left(  x-y\right)
\delta\left(  t_{1}-t_{2}\right)  ,\theta\left(  x,y\right)  \psi\left(
t_{1},t_{2}\right)  \right)  .
\end{align*}
We begin with $\mathfrak{G}$, by integrating by parts twice in the temporal
variables, and using Fubini's theorem,%

\begin{multline*}
\mathfrak{G=}%
{\displaystyle\int\limits_{\Omega^{2}}}
\left\{  \text{ }%
{\displaystyle\int\limits_{\mathbb{R}^{2}}}
\left(  \partial_{1}+\gamma\right)  \left(  \partial_{2}+\gamma\right)
G_{\boldsymbol{hh}}^{\alpha\beta}\left(  x,y,t_{1},t_{2}\right)  \psi\left(
t_{1},t_{2}\right)  dt_{1}dt_{2}\right\}  \times\\
\theta\left(  x,y\right)  d\mu\left(  x\right)  d\mu\left(  y\right)
\end{multline*}%
\begin{multline*}
=%
{\displaystyle\int\limits_{\Omega^{2}}}
\left\{  \text{ }%
{\displaystyle\int\limits_{\mathbb{R}^{2}}}
G_{\boldsymbol{hh}}^{\alpha\beta}\left(  x,y,t_{1},t_{2}\right)  \left\{
\left(  \partial_{1}+\gamma\right)  \left(  \partial_{2}+\gamma\right)
\psi\left(  t_{1},t_{2}\right)  \right\}  dt_{1}dt_{2}\right\}  \times\\
\theta\left(  x,y\right)  d\mu\left(  x\right)  d\mu\left(  y\right)
\end{multline*}%
\begin{multline*}
=%
{\displaystyle\int\limits_{\mathbb{R}^{2}}}
\left\{
{\displaystyle\int\limits_{\Omega^{2}}}
G_{\boldsymbol{hh}}^{\alpha\beta}\left(  x,y,t_{1},t_{2}\right)  \theta\left(
x,y\right)  d\mu\left(  x\right)  d\mu\left(  y\right)  \right\}  \times\\
\left\{  \left(  \partial_{1}+\gamma\right)  \left(  \partial_{2}%
+\gamma\right)  \psi\left(  t_{1},t_{2}\right)  \right\}  dt_{1}dt_{2}%
\end{multline*}%
\begin{align*}
&  =%
{\displaystyle\int\limits_{\mathbb{R}^{2}}}
\overline{G}_{\boldsymbol{hh}}^{\alpha\beta}\left(  t_{1},t_{2}\right)
\left\{  \left(  \partial_{1}+\gamma\right)  \left(  \partial_{2}%
+\gamma\right)  \psi\left(  t_{1},t_{2}\right)  \right\}  dt_{1}dt_{2}\\
&  =%
{\displaystyle\int\limits_{\mathbb{R}^{2}}}
\left(  \partial_{1}+\gamma\right)  \left(  \partial_{2}+\gamma\right)
\overline{G}_{\boldsymbol{hh}}^{\alpha\beta}\left(  t_{1},t_{2}\right)
\psi\left(  t_{1},t_{2}\right)  dt_{1}dt_{2}.
\end{align*}
For $\mathfrak{C}$,%
\begin{align*}
\mathfrak{C}  &  =%
{\displaystyle\int\limits_{\mathbb{R}^{2}}}
\left\{  \text{ }%
{\displaystyle\int\limits_{\Omega^{2}}}
C_{\phi\left(  \boldsymbol{h}^{\alpha}\right)  \phi\left(  \boldsymbol{h}%
^{\beta}\right)  }^{\alpha\beta}\left(  x,y,t_{1},t_{2}\right)  \theta\left(
x,y\right)  d\mu\left(  x\right)  d\mu\left(  y\right)  \right\}  \psi\left(
t_{1},t_{2}\right)  dt_{1}dt_{2}\\
&  =%
{\displaystyle\int\limits_{\mathbb{R}^{2}}}
\overline{C}_{\phi\phi}^{\alpha\beta}\left(  t_{1},t_{2}\right)  \psi\left(
t_{1},t_{2}\right)  dt_{1}dt_{2},
\end{align*}
and for $\mathfrak{D}$,
\[
\mathfrak{D=}\sigma_{\theta}^{2}%
{\displaystyle\int\limits_{\mathbb{R}^{2}}}
\delta\left(  t_{1}-t_{2}\right)  \psi\left(  t_{1},t_{2}\right)  dt_{1}%
dt_{2}.
\]

\end{proof}

\begin{remark}
We now take $\Omega=\mathbb{Z}_{p}$, $\theta\left(  x,y\right)  $\ as the
characteristic function of $\mathbb{Z}_{p}\times\mathbb{Z}_{p}$, and use that
$\mu\left(  \mathbb{Z}_{p}\right)  =1$. In this case, the
\begin{gather*}
\left(  \partial_{1}+\gamma\right)  \left(  \partial_{2}+\gamma\right)
\left\{  \text{ }%
{\displaystyle\iint\limits_{\mathbb{Z}_{p}\times\mathbb{Z}_{p}}}
G_{\boldsymbol{hh}}^{\alpha\beta}\left(  x,y,t_{1},t_{2}\right)  dxdy\right\}
=\\
\left\{  \text{ }%
{\displaystyle\iint\limits_{\mathbb{Z}_{p}\times\mathbb{Z}_{p}}}
C_{\phi\phi}^{\alpha\beta}\left(  x,y,t_{1},t_{2}\right)  dxdy\right\}
+\sigma^{2}\delta\left(  t_{1}-t_{2}\right)
\end{gather*}
in $\mathcal{S}^{\prime}\left(  \mathbb{R}^{2}\right)  $, where $dxdy$ is the
normalized Haar measure of $\mathbb{Z}_{p}\times\mathbb{Z}_{p}$.
\end{remark}

\subsection{Spatial averages of the second type}

For a fixed $\theta\left(  x,y\right)  \in\mathcal{D}(\Omega^{2})$, we set%
\begin{equation}
\underline{G}_{\boldsymbol{hh}}^{\alpha\beta}\left(  t_{1},t_{2}\right)
=\underline{G}_{\boldsymbol{hh}}^{\alpha\beta}\left(  t_{1},t_{2}%
;\theta\right)  :=%
{\displaystyle\int\limits_{\Omega}}
G_{\boldsymbol{hh}}^{\alpha\beta}\left(  y,y,t_{1},t_{2}\right)  \theta\left(
y,y\right)  d\mu\left(  y\right)  , \label{Distribution_1}%
\end{equation}%
\[
\underline{C}_{\phi\phi}^{\alpha\beta}\left(  t_{1},t_{2}\right)
=\underline{C}_{\phi\phi}^{\alpha\beta}\left(  t_{1},t_{2};\theta\right)  :=%
{\displaystyle\int\limits_{\Omega}}
C_{\phi\phi}^{\alpha\beta}\left(  y,y,t_{1},t_{2}\right)  \theta\left(
y,y\right)  d\mu\left(  y\right)  ,
\]%
\[
\sigma_{\theta}^{2}:=\sigma^{2}%
{\displaystyle\int\limits_{\Omega^{2}}}
\delta\left(  x-y\right)  \theta\left(  x,y\right)  d\mu\left(  x\right)
d\mu\left(  y\right)  =\sigma^{2}%
{\displaystyle\int\limits_{\Omega}}
\theta\left(  y,y\right)  d\mu\left(  y\right)  ,
\]

in $\mathcal{S}^{\prime}\left(  \mathbb{R}^{2}\right)  $. The first step is to
show that (\ref{Distribution_1}) is a well-defined distribution from
$\mathcal{S}^{\prime}\left(  \mathbb{R}^{2}\right)  $ under a suitable hypothesis.

\begin{remark}
\label{Nota3}We assume that $\mu\left(  K\right)  >0$ for nay compact
$K\subset\Omega$, and that $\varphi\left(  y,y\right)  \in\mathcal{D}\left(
\Omega\right)  $ for any $\varphi\left(  x,y\right)  \in\mathcal{D}\left(
\Omega^{2}\right)  $. Given
\[
T\left(  x,y,t_{1},t_{2}\right)  \in\mathcal{D}^{\prime}\left(  \Omega
^{2}\right)
{\textstyle\bigotimes\nolimits_{\text{alg}}}
\mathcal{S}^{\prime}\left(  \mathbb{R}^{2}\right)  ,
\]
we want to define its restriction to the line $y=x$, so
\[
T\left(  x,x,t_{1},t_{2}\right)  \in\mathcal{D}^{\prime}\left(  \Omega\right)
%
{\textstyle\bigotimes\nolimits_{\text{alg}}}
\mathcal{S}^{\prime}\left(  \mathbb{R}^{2}\right)  .
\]
We assume that $T\left(  \cdot,\cdot,t_{1},t_{2}\right)  \in L^{2}\left(
\Omega^{2}\right)  $, and define%
\[
\left(  T\left(  y,y,t_{1},t_{2}\right)  ,\varphi\left(  y,y\right)  \right)
=%
{\displaystyle\int\limits_{\Omega}}
T\left(  y,y,t_{1},t_{2}\right)  \varphi\left(  y,y\right)  d\mu\left(
y\right)  ,
\]
for $\varphi\in\mathcal{D}\left(  \Omega^{2}\right)  $. In order to show that
this formula defines a distribution on $\mathcal{D}\left(  \Omega^{2}\right)
$, it is sufficient to show that $T\left(  \cdot,\cdot,t_{1},t_{2}\right)  \in
L_{\text{loc}}^{1}\left(  \Omega\right)  $. Indeed, take $K\subset\Omega$ be a
compact subset, and set
\[
L=\left\{  \left(  x,y\right)  \in\Omega\times\Omega;x=y\right\}  .
\]
Then
\[
\left\vert T\left(  y,y,t_{1},t_{2}\right)  \right\vert =1_{L}\left(
x,y\right)  \left\vert T\left(  x,y,t_{1},t_{2}\right)  \right\vert
\leq\left\vert T\left(  x,y,t_{1},t_{2}\right)  \right\vert \text{ for any
}x,y\in\Omega,
\]
and by using the Cauchy-Schwarz inequality, and the fact that $\mu\left(
K\right)  >0$,
\[%
{\displaystyle\iint\limits_{K\times K}}
\left\vert T\left(  y,y,t_{1},t_{2}\right)  \right\vert d\mu\left(  x\right)
d\mu\left(  y\right)  \leq%
{\displaystyle\iint\limits_{K\times K}}
\left\vert T\left(  x,y,t_{1},t_{2}\right)  \right\vert d\mu\left(  x\right)
d\mu\left(  y\right)  ,
\]
i.e.,%
\[%
{\displaystyle\int\limits_{K}}
\left\vert T\left(  y,y,t_{1},t_{2}\right)  \right\vert d\mu\left(  y\right)
\leq\left\Vert T\left(  \cdot,\cdot,t_{1},t_{2}\right)  \right\Vert
_{L^{2}\left(  \Omega^{2}\right)  }.
\]

\end{remark}

Since $C_{\phi\phi}^{\alpha\beta}\left(  \cdot,\cdot,t_{1},t_{2}\right)  \in
L^{2}\left(  \Omega^{2}\right)  $, for almost every $\left(  t_{1}%
,t_{2}\right)  \in\mathbb{R}^{2}$, Remark \ref{Nota3} implies that
$C_{\phi\phi}^{\alpha\beta}\left(  y,y,t_{1},t_{2}\right)  \in\mathcal{S}%
^{\prime}\left(  \mathbb{R}^{2}\right)  $. Intuitively, the restriction of
distribution $C_{\phi\phi}^{\alpha\beta}\left(  x,y,t_{1},t_{2}\right)  $ to
the line $\left\{  x=y\right\}  $ is given by%
\begin{multline*}%
{\displaystyle\int\limits_{\Omega^{2}}}
C_{\phi\phi}^{\alpha\beta}\left(  x,y,t_{1},t_{2}\right)  \delta\left(
x-y\right)  \theta\left(  x,y\right)  d\mu\left(  x\right)  d\mu\left(
y\right)  =\\%
{\displaystyle\int\limits_{\Omega}}
C_{\phi\phi}^{\alpha\beta}\left(  y,y,t_{1},t_{2}\right)  \theta\left(
y,y\right)  d\mu\left(  y\right)  .
\end{multline*}

We note that $\sigma^{2}\delta\left(  x-y\right)  \delta\left(  t_{1}%
-t_{2}\right)  $ is by definition a distribution restricted to the line
$\left\{  x=y\right\}  $. Therefore, the restriction of the distribution
$C_{\phi\phi}^{\alpha\beta}\left(  x,y,t_{1},t_{2}\right)  $ $\ +\sigma
^{2}\delta\left(  x-y\right)  \delta\left(  t_{1}-t_{2}\right)  $ to the line
$y=x$ is
\[
C_{\phi\phi}^{\alpha\beta}\left(  x,x,t_{1},t_{2}\right)  +\sigma^{2}%
\delta\left(  x-y\right)  \delta\left(  t_{1}-t_{2}\right)  \in\mathcal{D}%
^{\prime}\left(  \Omega^{2}\right)
{\textstyle\bigotimes\nolimits_{\text{alg}}}
\mathcal{S}^{\prime}\left(  \mathbb{R}^{2}\right)  .
\]
Now, since the equation (\ref{Eq_15A}) holds in $\mathcal{D}^{\prime}\left(
\Omega^{2}\right)
{\textstyle\bigotimes\nolimits_{\text{alg}}}
\mathcal{S}^{\prime}\left(  \mathbb{R}^{2}\right)  $, the restriction to the
distribution $\left(  \partial_{1}+\gamma\right)  \left(  \partial_{2}%
+\gamma\right)  G_{\boldsymbol{hh}}^{\alpha\beta}\left(  x,y,t_{1}%
,t_{2}\right)  $ to the line $\left\{  x=y\right\}  $ is well-defined. We
denote such distribution as $\left(  \partial_{1}+\gamma\right)  \left(
\partial_{2}+\gamma\right)  G_{\boldsymbol{hh}}^{\alpha\beta}\left(
y,y,t_{1},t_{2}\right)  $. We now use the reasoning used in the proof of
Theorem \ref{Theorem4} to get the following result:

\begin{theorem}
\label{Theorem5}With the above notation,
\begin{equation}
\left(  \partial_{1}+\gamma\right)  \left(  \partial_{2}+\gamma\right)
\underline{G}_{\boldsymbol{hh}}^{\alpha\beta}\left(  t_{1},t_{2}\right)
=\underline{C}_{\phi\phi}^{\alpha\beta}\left(  t_{1},t_{2}\right)
+\sigma_{\theta}^{2}\delta\left(  t_{1}-t_{2}\right)  \text{ in }%
\mathcal{S}^{\prime}\left(  \mathbb{R}^{2}\right)  .\nonumber
\end{equation}

\end{theorem}

\begin{remark}
We now take $\Omega=\mathbb{Z}_{p}$, $\theta\left(  x,y\right)  $\ as the
characteristic function of $\mathbb{Z}_{p}\times\mathbb{Z}_{p}$, and use that
$\mu\left(  \mathbb{Z}_{p}\right)  =1$. We consider restriction to lines of
the form $x=ay+b$, where $a$, $b\in\mathbb{Z}_{p}$. In this case, the Theorem
\ref{Theorem5} has the form
\begin{gather*}
\left(  \partial_{1}+\gamma\right)  \left(  \partial_{2}+\gamma\right)
\left\{  \text{ }%
{\displaystyle\iint\limits_{\mathbb{Z}_{p}\times\mathbb{Z}_{p}}}
G_{\boldsymbol{hh}}^{\alpha\beta}\left(  ay+b,y,t_{1},t_{2}\right)
dxdy\right\}  =\\
\left\{  \text{ }%
{\displaystyle\iint\limits_{\mathbb{Z}_{p}\times\mathbb{Z}_{p}}}
C_{\phi\phi}^{\alpha\beta}\left(  ay+b,y,t_{1},t_{2}\right)  dxdy\right\}
+\frac{\sigma^{2}}{\left\vert a\right\vert _{p}}\delta\left(  t_{1}%
-t_{2}\right)
\end{gather*}
in $\mathcal{S}^{\prime}\left(  \mathbb{R}^{2}\right)  $.
\end{remark}

\section{\label{Appen_H}Appendix H}

In this appendix, we collect some basic results from the $p$-adic analysis.
For a detailed exposition on $p$-adic analysis, the reader may consult
\cite{A-K-S}, \cite{V-V-Z}, \cite{Taibleson}, \cite{Zuniga-Textbook}. Our
presentation here is based on the book \cite{Zuniga-Textbook}.

\subsection{The field of $p$-adic numbers}

Let $p$ be a fixed prime number. The field of $p-$adic numbers $\mathbb{Q}%
_{p}$ is defined as the completion of the field of rational numbers
$\mathbb{Q}$ with respect to the $p-$adic norm $|\cdot|_{p}$, which is defined
as
\[
|x|_{p}=%
\begin{cases}
0 & \text{if }x=0\\
p^{-\gamma} & \text{if }x=p^{\gamma}\dfrac{a}{b},
\end{cases}
\]
where $a$ and $b$ are integers coprime with $p$. The integer $\gamma
=ord_{p}(x):=ord(x)$, with $ord(0):=+\infty$, is called the\textit{\ }$p-$adic
order of $x$. We extend the $p-$adic norm to $\mathbb{Q}_{p}^{N}$ by taking%
\[
||x||_{p}:=\max_{1\leq i\leq N}|x_{i}|_{p},\qquad\text{for }x=(x_{1}%
,\dots,x_{N})\in\mathbb{Q}_{p}^{N}.
\]
We define $ord(x)=\min_{1\leq i\leq N}\{ord(x_{i})\}$, then $||x||_{p}%
=p^{-ord(x)}$.\ The metric space $\left(  \mathbb{Q}_{p}^{N},||\cdot
||_{p}\right)  $ is a complete ultrametric space. As a topological space
$\mathbb{Q}_{p}$\ is homeomorphic to a Cantor-like subset of the real line,
see, e.g., \cite{V-V-Z}, \cite{A-K-S}.

Any $p-$adic number $x\neq0$ has a unique expansion of the form
\[
x=p^{ord(x)}\sum_{j=0}^{\infty}x_{j}p^{j},
\]
where $x_{j}\in\{0,1,2,\dots,p-1\}$ and $x_{0}\neq0$. By using this expansion,
we define \textit{the fractional part }$\{x\}_{p}$\textit{ of }$x\in
\mathbb{Q}_{p}$ as the rational number
\[
\{x\}_{p}=%
\begin{cases}
0 & \text{if }x=0\text{ or }ord(x)\geq0\\
p^{ord(x)}\sum_{j=0}^{-ord(x)-1}x_{j}p^{j} & \text{if }ord(x)<0.
\end{cases}
\]
In addition, any $x\in\mathbb{Q}_{p}^{N}\smallsetminus\left\{  0\right\}  $
can be represented uniquely as $x=p^{ord(x)}v$, where $\left\Vert v\right\Vert
_{p}=1$.

\subsection{Topology of $\mathbb{Q}_{p}^{N}$}

For $r\in\mathbb{Z}$, denote by $B_{r}^{N}(a)=\{x\in\mathbb{Q}_{p}%
^{N};||x-a||_{p}\leq p^{r}\}$ the ball of radius $p^{r}$ with center at
$a=(a_{1},\dots,a_{N})\in\mathbb{Q}_{p}^{N}$, and take $B_{r}^{N}%
(0):=B_{r}^{N}$. Note that $B_{r}^{N}(a)=B_{r}(a_{1})\times\cdots\times
B_{r}(a_{N})$, where $B_{r}(a_{i}):=\{x\in\mathbb{Q}_{p};|x_{i}-a_{i}|_{p}\leq
p^{r}\}$ is the one-dimensional ball of radius $p^{r}$ with center at
$a_{i}\in\mathbb{Q}_{p}$. The ball $B_{0}^{N}$ equals the product of $N$
copies of $B_{0}=\mathbb{Z}_{p}$, the ring of\textit{ }$p-$adic integers. We
also denote by $S_{r}^{N}(a)=\{x\in\mathbb{Q}_{p}^{N};||x-a||_{p}=p^{r}\}$ the
sphere of radius\textit{ }$p^{r}$ with center at $a=(a_{1},\dots,a_{N}%
)\in\mathbb{Q}_{p}^{N}$, and take $S_{r}^{N}(0):=S_{r}^{N}$. We notice that
$S_{0}^{1}=\mathbb{Z}_{p}^{\times}$ (the group of units of $\mathbb{Z}_{p}$),
but $\left(  \mathbb{Z}_{p}^{\times}\right)  ^{N}\subsetneq S_{0}^{N}$. The
balls and spheres are both open and closed subsets in $\mathbb{Q}_{p}^{N}$. In
addition, two balls in $\mathbb{Q}_{p}^{N}$ are either disjoint or one is
contained in the other.

As a topological space $\left(  \mathbb{Q}_{p}^{N},||\cdot||_{p}\right)  $ is
totally disconnected, i.e., the only connected \ subsets of $\mathbb{Q}%
_{p}^{N}$ are the empty set and the points. A subset of $\mathbb{Q}_{p}^{N}$
is compact if and only if it is closed and bounded in $\mathbb{Q}_{p}^{N}$,
see, e.g., \cite[Section 1.3]{V-V-Z}, or \cite[Section 1.8]{A-K-S}. The balls
and spheres are compact subsets. Thus $\left(  \mathbb{Q}_{p}^{N}%
,||\cdot||_{p}\right)  $ is a locally compact topological space.

\subsection{The Haar measure}

Since $(\mathbb{Q}_{p}^{N},+)$ is a locally compact topological group, there
exists a Haar measure $d^{N}x$, which is invariant under translations, i.e.,
$d^{N}(x+a)=d^{N}x$, \cite{Halmos}. If we normalize this measure by the
condition $\int_{\mathbb{Z}_{p}^{N}}dx=1$, then $d^{N}x$ is unique.

\begin{notation}
We use $\Omega\left(  p^{-r}||x-a||_{p}\right)  $ to denote the characteristic
function of the ball $B_{r}^{N}(a)=a+p^{-r}\mathbb{Z}_{p}^{N}$, where
\[
\mathbb{Z}_{p}^{N}=\left\{  x\in\mathbb{Q}_{p}^{N};\left\Vert x\right\Vert
_{p}\leq1\right\}
\]
is the $N$-dimensional unit ball. For more general sets, we will use the
notation $1_{A}$ for the characteristic function of set $A$.
\end{notation}

\subsection{The Bruhat-Schwartz space}

A complex-valued function $\varphi$ defined on $\mathbb{Q}_{p}^{N}$ is
\textit{called locally constant} if for any $x\in\mathbb{Q}_{p}^{N}$ there
exist an integer $l(x)\in\mathbb{Z}$ such that%
\begin{equation}
\varphi(x+x^{\prime})=\varphi(x)\text{ for any }x^{\prime}\in B_{l(x)}^{N}.
\label{local_constancy}%
\end{equation}
A function $\varphi:\mathbb{Q}_{p}^{N}\rightarrow\mathbb{C}$ is called a
Bruhat-Schwartz function\textit{ }(or a test function) if it is locally
constant with compact support. Any test function can be represented as a
linear combination, with complex coefficients, of characteristic functions of
balls. The $\mathbb{C}$-vector space of Bruhat-Schwartz functions is denoted
by $\mathcal{D}(\mathbb{Q}_{p}^{N})$. We denote by $\mathcal{D}_{\mathbb{R}%
}(\mathbb{Q}_{p}^{N})$\ the $\mathbb{R}$-vector space of Bruhat-Schwartz
functions. For $\varphi\in\mathcal{D}(\mathbb{Q}_{p}^{N})$, the largest number
$l=l(\varphi)$ satisfying (\ref{local_constancy}) is called the exponent of
local constancy (or the parameter of constancy) of $\varphi$.

We warn the reader that in most of this work only real-valued functions are
needed. However, in this appendix, we adopt a more general formulation.

\subsection{$L^{\rho}$ spaces}

Given $\rho\in\lbrack1,\infty)$, we denote by $L^{\rho}\left(
\mathbb{Q}
_{p}^{N}\right)  :=L^{\rho}\left(
\mathbb{Q}
_{p}^{N},d^{N}x\right)  ,$ the $\mathbb{C}-$vector space of all the complex
valued functions $g$ satisfying
\[
\left\Vert g\right\Vert _{\rho}=\left(  \text{ }%
{\displaystyle\int\limits_{\mathbb{Q}_{p}^{N}}}
\left\vert g\left(  x\right)  \right\vert ^{\rho}d^{N}x\right)  ^{\frac
{1}{\rho}}<\infty,
\]
where $d^{N}x$ is the normalized Haar measure on $\left(  \mathbb{Q}_{p}%
^{N},+\right)  $. The corresponding $\mathbb{R}$-vector spaces are denoted as
$L_{\mathbb{R}}^{\rho}\left(
\mathbb{Q}
_{p}^{N}\right)  =L_{\mathbb{R}}^{\rho}\left(
\mathbb{Q}
_{p}^{N},d^{N}x\right)  $, $1\leq\rho<\infty$.

If $U$ is an open subset of $\mathbb{Q}_{p}^{N}$, $\mathcal{D}(U)$ denotes the
$\mathbb{C}$-vector space of test functions with supports contained in $U$,
then $\mathcal{D}(U)$ is dense in
\[
L^{\rho}\left(  U\right)  =\left\{  \varphi:U\rightarrow\mathbb{C};\left\Vert
\varphi\right\Vert _{\rho}=\left\{
{\displaystyle\int\limits_{U}}
\left\vert \varphi\left(  x\right)  \right\vert ^{\rho}d^{N}x\right\}
^{\frac{1}{\rho}}<\infty\right\}  ,
\]
for $1\leq\rho<\infty$, see, e.g., \cite[Section 4.3]{A-K-S}.

\subsection{The Fourier transform}

Set $\chi_{p}(y)=\exp(2\pi i\{y\}_{p})$ for $y\in\mathbb{Q}_{p}$. The map
$\chi_{p}(\cdot)$ is an additive character on $\mathbb{Q}_{p}$, i.e., a
continuous map from $\left(  \mathbb{Q}_{p},+\right)  $ into $S$ (the unit
circle considered as multiplicative group) satisfying $\chi_{p}(x_{0}%
+x_{1})=\chi_{p}(x_{0})\chi_{p}(x_{1})$, $x_{0},x_{1}\in\mathbb{Q}_{p}$.\ The
additive characters of $\mathbb{Q}_{p}$ form an Abelian group which is
isomorphic to $\left(  \mathbb{Q}_{p},+\right)  $. The isomorphism is given by
$\kappa\rightarrow\chi_{p}(\kappa x)$, see, e.g., \cite[Section 2.3]{A-K-S}.

Given $\xi=(\xi_{1},\dots,\xi_{N})$ and $x=(x_{1},\dots,x_{N})\allowbreak
\in\mathbb{Q}_{p}^{N}$, we set $\xi\cdot x:=\sum_{j=1}^{N}\xi_{j}x_{j}$. The
Fourier transform of $\varphi\in\mathcal{D}(\mathbb{Q}_{p}^{N})$ is defined
as
\[
\mathcal{F}\varphi(\xi)=%
{\displaystyle\int\limits_{\mathbb{Q} _{p}^{N}}}
\chi_{p}(\xi\cdot x)\varphi(x)d^{N}x\quad\text{for }\xi\in\mathbb{Q}_{p}^{N},
\]
where $d^{N}x$ is the normalized Haar measure on $\mathbb{Q}_{p}^{N}$. The
Fourier transform is a linear isomorphism from $\mathcal{D}(\mathbb{Q}_{p}%
^{N})$ onto itself satisfying
\begin{equation}
(\mathcal{F}(\mathcal{F}\varphi))(\xi)=\varphi(-\xi), \label{Eq_FFT}%
\end{equation}
see, e.g., \cite[Section 4.8]{A-K-S}. We will also use the notation
$\mathcal{F}_{x\rightarrow\kappa}\varphi$ and $\widehat{\varphi}$\ for the
Fourier transform of $\varphi$.

The Fourier transform extends to $L^{2}$. If $f\in L^{2}\left(  \mathbb{Q}%
_{p}^{N}\right)  $, its Fourier transform is defined as
\[
(\mathcal{F}f)(\xi)=\lim_{k\rightarrow\infty}%
{\displaystyle\int\limits_{||x||_{p}\leq p^{k}}}
\chi_{p}(\xi\cdot x)f(x)d^{N}x,\quad\text{for }\xi\in%
\mathbb{Q}
_{p}^{N},
\]
where the limit is taken in $L^{2}\left(  \mathbb{Q}_{p}^{N}\right)  $. We
recall that the Fourier transform is unitary on $L^{2}\left(  \mathbb{Q}%
_{p}^{N}\right)  ,$ i.e. $||f||_{2}=||\mathcal{F}f||_{2}$ for $f\in L^{2}$ and
that (\ref{Eq_FFT}) is also valid in $L^{2}$, see, e.g., \cite[Chapter III,
Section 2]{Taibleson}.

\subsection{Distributions}

The $\mathbb{C}$-vector space $\mathcal{D}^{\prime}\left(  \mathbb{Q}_{p}%
^{N}\right)  $ of all continuous linear functionals on $\mathcal{D}%
(\mathbb{Q}_{p}^{N})$ is called the Bruhat-Schwartz space of distributions.
Every linear functional on $\mathcal{D}(\mathbb{Q}_{p}^{N})$ is continuous,
i.e. $\mathcal{D}^{\prime}\left(  \mathbb{Q}_{p}^{N}\right)  $\ agrees with
the algebraic dual of $\mathcal{D}(\mathbb{Q}_{p}^{N})$, see, e.g.,
\cite[Chapter 1, VI.3, Lemma]{V-V-Z}.

We endow $\mathcal{D}^{\prime}\left(  \mathbb{Q}_{p}^{N}\right)  $ with the
weak topology, i.e. a sequence $\left\{  T_{j}\right\}  _{j\in\mathbb{N}}$ in
$\mathcal{D}^{\prime}\left(  \mathbb{Q}_{p}^{N}\right)  $ converges to $T$ if
$\lim_{j\rightarrow\infty}T_{j}\left(  \varphi\right)  =T\left(
\varphi\right)  $ for any $\varphi\in\mathcal{D}(\mathbb{Q}_{p}^{N})$. The
map
\[%
\begin{array}
[c]{lll}%
\mathcal{D}^{\prime}\left(  \mathbb{Q}_{p}^{N}\right)  \times\mathcal{D}%
(\mathbb{Q}_{p}^{N}) & \rightarrow & \mathbb{C}\\
\left(  T,\varphi\right)  & \rightarrow & T\left(  \varphi\right)
\end{array}
\]
is a bilinear form which is continuous in $T$ and $\varphi$ separately. We
call this map the pairing between $\mathcal{D}^{\prime}\left(  \mathbb{Q}%
_{p}^{N}\right)  $ and $\mathcal{D}(\mathbb{Q}_{p}^{N})$. From now on we will
use $\left(  T,\varphi\right)  $ instead of $T\left(  \varphi\right)  $.

Every $f$\ in $L_{loc}^{1}$ defines a distribution $f\in\mathcal{D}^{\prime
}\left(  \mathbb{Q}_{p}^{N}\right)  $ by the formula
\[
\left(  f,\varphi\right)  =%
{\displaystyle\int\limits_{\mathbb{Q}_{p}^{N}}}
f\left(  x\right)  \varphi\left(  x\right)  d^{N}x.
\]

\subsection{\label{SEction Fourier Transform}The Fourier transform of a
distribution}

The Fourier transform $\mathcal{F}\left[  T\right]  $ of a distribution
$T\in\mathcal{D}^{\prime}\left(  \mathbb{Q}_{p}^{N}\right)  $ is defined by%
\[
\left(  \mathcal{F}\left[  T\right]  ,\varphi\right)  =\left(  T,\mathcal{F}%
\left[  \varphi\right]  \right)  \text{ for all }\varphi\in\mathcal{D}\left(
\mathbb{Q}_{p}^{N}\right)  \text{.}%
\]
The Fourier transform $T\rightarrow\mathcal{F}\left[  T\right]  $ is a linear
and continuous isomorphism from $\mathcal{D}^{\prime}\left(  \mathbb{Q}%
_{p}^{N}\right)  $\ onto $\mathcal{D}^{\prime}\left(  \mathbb{Q}_{p}%
^{N}\right)  $. Furthermore, $T=\mathcal{F}\left[  \mathcal{F}\left[
T\right]  \left(  -\xi\right)  \right]  $.

Let $T\in\mathcal{D}^{\prime}\left(  \mathbb{Q}_{p}^{n}\right)  $ be a
distribution. Then \textrm{supp}$T\subset B_{L}^{N}$ if and only if
$\mathcal{F}\left[  T\right]  $ is a locally constant function, and the
exponent of local constancy of $\mathcal{F}\left[  T\right]  $ is $\geq-L$. In addition%

\[
\mathcal{F}\left[  T\right]  \left(  \xi\right)  =\left(  T\left(  y\right)
,\Omega\left(  p^{-L}\left\Vert y\right\Vert _{p}\right)  \chi_{p}\left(
\xi\cdot y\right)  \right)  ,
\]
see, e.g., \cite[Section 4.9]{A-K-S}.

\bigskip


\bigskip

\begin{thebibliography}{99}                                                                                               %


\bibitem {A-K-S}S. Albeverio, A. Yu. Khrennikov, V. M. Shelkovich, Theory of
$p$-adicdistributions: linear and nonlinear models. Cambridge University
Press, 2010.

\bibitem {Altland et al}A. Altland and B. Simons, Condensed Matter Field
Theory (Cambridge University Press, Cambridge, England, 2010).

\bibitem {Amari}Amari S., Dynamics of pattern formation in lateral inhibition
type neural fields, Biol. Cybern., 27 (1977), pp. 77-87.

\bibitem {Arroyo et al}Edilberto Arroyo-Ortiz, W. A. Z\'{u}\~{n}iga-Galindo,
Construction of p-adic covariant quantum fields in the framework of white
noise analysis.(English summary)Rep. Math. Phys.84(2019), no.1, 1--34.

\bibitem {Ash}Ash, Robert B., Measure, integration, and functional analysis,
Academic Press, New York-London, 1972.

\bibitem {Asao}Asao Arai, Analysis on Fock spaces and mathematical theory of
quantum fields. An introduction to mathematical analysis of quantum
fieldsWorld Scientific Publishing Co. Pte. Ltd., Hackensack, NJ, 2018.

\bibitem {Berlinet et al}A. Berlinet, Christine Thomas-Agnan, Reproducing
kernel Hilbert spaces in probability and statistics. Springer Science+Business
Media, New York, 2004.

\bibitem {Bogachev}Vladimir I. Bogachev, Gaussian measures, Math. Surveys
Monogr., 62. American Mathematical Society, Providence, RI, 1998.

\bibitem {Bressloff}Paul C. Bressloff, Stochastic neural field theory and the
system-size expansion, SIAM J. Appl. Math.70(2009/10), no.5, 1488--1521.

\bibitem {Buice-Cowean-2007}Michael A. Buice, Jack D. Cowan, Field-theoretic
approach to fluctuation effects in neural networks Phys. Rev. E (3) 75 (2007),
no. 5, 051919, 14 pp.

\bibitem {Buice and Cowan}M.A. Buice, J.D. Cowan, Statistical mechanics of the
neocortex. Progress in Biophysics and Molecular Biology. 2009
Feb-Apr;99(2-3):53-86. DOI: 10.1016/j.pbiomolbio.2009.07.003.

\bibitem {Buice-Cowan-Chow}Michael A. Buice, Jack D. Cowan, Carson C. Chow,
Systematic fluctuation expansion for neural network activity equations, Neural
Comput. 22 (2010), no. 2, 377--426.

\bibitem {Chialvo}D. R. Chialvo, Emergent complex neural dynamics. Nature
Physics. 6 (10): 744--750 (2010). arXiv:1010.2530. doi:10.1038/nphys1803.

\bibitem {Chow et al}C. Chow and M. Buice, Path Integral Methods for
Stochastic Differential Equations, J. Math. Neurosci. 5, 8 (2015).

\bibitem {Chua-Tamas}Chua Leon O, Roska, Tamas, Cellular neural networks and
visual computing: foundations and applications. Cambridge university press, 2002.

\bibitem {Neural-Fields}Stephen Coobes, Peter Baim Graben, Roland Potthast ,
and James Wright, Editors. (2014). Neural fields. Theory and applications.
Springer, Heidelberg.

\bibitem {Da prato}Giuseppe Da Prato, An introduction to infinite-dimensional
analysis. Universitext. Springer-Verlag, Berlin, 2006.

\bibitem {Demirtsas et al}Mehmet Demirtas, James Halverson, Anindita Maiti,
Matthew D Schwartz and Keegan Stoner, Neural network field theories:
non-Gaussianity, actions, and locality, Mach. Learn.: Sci. Technol. 5 015002
(2024), https://doi.org/10.1088/2632-2153/ad17d3.

\bibitem {Erbin et al}H. Erbin, V. Lahoche and D. Ousmane Samary,
Non-perturbative renormalization for the neural network-QFT correspondence,
Mach. Learn.: Sci. Technol. 3 015027 (2022), https://doi.org/10.1088/2632-2153/ac4f69.

\bibitem {Fuquen et al}A. R. Fuquen-Tibat\'{a}, H. Garc\'{\i}a-Compe\'{a}n, W.
A. Z\'{u}\~{n}iga-Galindo, Euclidean quantum field formulation of p -adic open
string amplitudes, Nuclear Phys. B 975 (2022), Paper No. 115684, 27 pp.

\bibitem {Gel-Shilov}I. M. Gel'fand, G. E. Shilov , Generalized functions.
Vol. 2. Spaces of fundamental and generalized functions. AMS Chelsea
publishing, 2010.

\bibitem {Gelfand-Vilenkin}I.M. Gel'fand, N.Y. Vilenkin, Generalized
Functions. Applications of Harmonic Analysis, vol. 4. Academic Press, New
York, 1964.

\bibitem {Grosvenor-Jefferson}K. T. Grosvenor, R. Jefferson, The edge of
chaos: quantum field theory and deep neural networks, SciPost Phys., 12, [81]
(2022). https://doi.org/10.21468/SciPostPhys.12.3.081.

\bibitem {Halmos}P. Halmos, Measure Theory\textit{.} D. Van Nostrand Company
Inc., New York, 1950.

\bibitem {Halverson et al}J. Halverson, A. Maiti and K. Stoner, Neural
networks and quantum field theory, Mach. Learn.: Sci. Technol. 2, 035002
(2021), doi:10.1088/2632-2153/abeca3.

\bibitem {Helias et al}Moritz Helias, David Dahmen, Statistical Field Theory
for Neural Networks. Lecture Notes in Physics 970, Springer Cham, 2020. DOI: https://doi.org/10.1007/978-3-030-46444-8.

\bibitem {Hesse et al}Janina Hesse, Thilo Gross, Self-organized criticality as
a fundamental property of neural systems, Front. Syst. Neurosci., 22 September
2014, Volume 8 - 2014, https://doi.org/10.3389/fnsys.2014.00166.

\bibitem {Hida et al}Takeyuki Hida, Hui-Hsiung Kuo, J\"{u}rgen Potthoff,
LudwigStreit, White noise. An Infinite Dimensional Calculus. Math. Appl., 253.
Kluwer Academic Publishers Group, Dordrecht, 1993.

\bibitem {Hilgetag et al}C. C. Hilgetag, A. Goulas, Is the brain really a
small-world network?. Brain Struct Funct 221, 2361--2366 (2016). https://doi.org/10.1007/s00429-015-1035-6.

\bibitem {Huang et al}Zhi-yuan Huang, Jia-an Yan, Introduction to infinite
dimensional stochastic analysis. Math. Appl., 502. Kluwer Academic Publishers,
Dordrecht; Science Press Beijing, Beijing, 2000.

\bibitem {KKZuniga}Andrei Khrennikov, Sergei Kozyrev, W. A. Z{\'{u}}%
{\~{n}}iga-Galindo, Ultrametric Equations and its Applications. Encyclopedia
of Mathematics and its Applications (168), Cambridge University Press, 2018.

\bibitem {Koning}H. K\"{o}nig, Eigenvalue distribution of compact operators.
Operator Theory: Advances and Applications, 16. Birkh%
\"{}%
auser Verlag, Basel, 1986.

\bibitem {Kukush}Alexander Kukush, Gaussian Measures in Hilbert Space:
Construction and Properties. Wiley-ISTE, 2020.

\bibitem {Laing}C. Laing, W. Troy, B. Gutkin, G. Ermentrout, Multiple bumps in
a neuronal model of working memory, SIAM J. Appl. Math. 63, 62 (2002).

\bibitem {Lecun et al}Y. LeCun, Y. Bengio, G. Hinton, Deep learning, Nature
521, 436--444 (2015). https://doi.org/10.1038/nature14539.

\bibitem {Lee et al}Jaehoon Lee, Yasaman Bahri, Roman Novak, Samuel S.
Schoenholz, Jeffrey Pennington, Jascha Sohl-Dickstein, Deep Neural Networks as
Gaussian Processes, International Conference on Learning Representations
(ICLR), 2018.

\bibitem {Martin et al}P. Martin, E. Siggia, and H. Rose, Statistical Dynamics
of Classical Systems, Phys. Rev. A 8, 423 (1973).

\bibitem {de Domicis et al}C. De Dominicis and L. Peliti, Field-Theory
Renormalization and Critical Dynamics above $T_{c}$: Helium, Antiferromagnets,
and Liquid-Gas Systems, Phys. Rev. B 18, 353 (1978).

\bibitem {Neal}Radford M. Neal, Bayesian Learning for Neural Networks. PhD
thesis, University of Toronto, Dept.of Computer Science, 1994.

\bibitem {Obata}Nobuaki Obata, White noise calculus and Fock space. Lecture
Notes in Math., 1577 Springer-Verlag, Berlin, 1994.

\bibitem {Pleiss et al}Geoff Pleiss, John P. Cunningham, The limitations of
large width in neural networks: a deep Gaussian process perspective, NIPS'21:
Proceedings of the 35th International Conference on Neural Information
Processing Systems,Article No.: 256, Pages 3349 - 336.

\bibitem {Price}R. Price, A useful theorem for nonlinear devices having
gaussian inputs, IRE Transactions Info. Th. 4, 69 (1958), doi:10.1109/TIT.1958.1057444.

\bibitem {Papoulis et al}A. Papoulis and S. U. Pillai, Probability, random
variables, and stochastic processes. McGraw-Hill, New York, ISBN 9780071226615 (1991).

\bibitem {Poole et al}Ben Poole, Subhaneil Lahiri, Maithra Raghu, Jascha
Sohl-Dickstein, and Surya Ganguli. Exponential expressivity in deep neural
networks through transient chaos. In Advances In Neural Information Processing
Systems, pp. 3360--3368, 2016.

\bibitem {Reed-Simon-I}M. Reed, B. Simon, Methods of Modern Mathematical
Physics: Functional Analysis I. Academic Press, Inc. [Harcourt Brace
Jovanovich, Publishers], New York, 1980.

\bibitem {Reed-Simon-II}M. Reed, B. Simon, Methods of Modern Mathematical
Physics: Functional Analysis II. Academic Press [Harcourt Brace Jovanovich,
Publishers], New York-London, 1975.

\bibitem {Roberts et al}Daniel A. Roberts, Sho Yaida, The Principles of Deep
Learning Theory: An Effective Theory Approach to Understanding Neural
Networks. Cambridge University Press, 2022.

\bibitem {Schoenholz et al}Samuel S. Schoenholz, Jeffrey Pennington, Jascha
Sohl-Dickstein, A Correspondence Between Random Neural Networks and
Statistical Field Theory, \ https://doi.org/10.48550/arXiv.1710.06570.

\bibitem {Schuecker et al}J. Schuecker, S. Goedeke, M. Helias, Optimal
Sequence Memory in Driven Random Networks. Physical review X 8(4), 041029 (2018).

\bibitem {Segadlo et al}Kai Segadlo, Bastian Epping, Alexander van Meegen,
David Dahmen, Michael Kr\"{a}mer and Moritz Helias, Unified field theoretical
approach to deep and recurrent neuronal networks, J. Stat. Mech. (2022)
103401. DOI 10.1088/1742-5468/ac8e57.

\bibitem {Slavova}Angela Slavova, Cellular neural networks: dynamics and
modelling. Mathematical Modelling: Theory and Applications, 16. Kluwer
Academic Publishers, Dordrecht, 2003.

\bibitem {Sompolinsky et al}H. Sompolinsky, A. Crisanti, and H.\thinspace J.
Sommers, Chaos in Random Neural Networks, Phys. Rev. Lett. 61, 259 (1988).

\bibitem {Sporns}O. Sporns, Small-world connectivity, motif composition, and
complexity of fractal neuronal connections, Biosystems\textit{, } 85(1),
(2006) 55-64.

\bibitem {Taibleson}M. H. Taibleson, Fourier analysis on local
fields\textit{.} Princeton University Press, 1975.

\bibitem {Treves}Fran\c{c}ois Tr\`{e}ves, Topological vector spaces,
distributions and kernels. Academic Press, New York-London, 1967.

\bibitem {V-V-Z}V. S. Vladimirov, I. V. Volovich, E. I. Zelenov,
$p$\textit{-Adic analysis and mathematical physics. }Singapore, World
Scientific, 1994.

\bibitem {Weil}Andr\'{e} Weil, Basic number theory, Classics in Mathematics,
Berlin, Heidelberg, 1995.

\bibitem {Williams}C. K. I. Williams, Computation with Infinite Neural
Networks, in Neural Computation, vol. 10, no. 5, pp. 1203-1216, 1 July 1998,
doi: 10.1162/089976698300017412.

\bibitem {Wilson-Cowan-1}H. R. Wilson, J. D. Cowan, Excitatory and inhibitory
interactions in localized populations of model neurons, Biophys. J., 12
(1972), pp. 1-23.

\bibitem {Wilson-Cowan-2}H. R. Wilson, J. D. Cowan, A mathematical theory of
the functional dynamics of cortical and thalamic nervous tissue, Kybernetik,
13 (1973), pp. 55-80.

\bibitem {Zambrano-Zuniga-1}B. A. Zambrano-Luna, W. A. Z\'{u}\~{n}iga-Galindo,
$p$-adic cellular neural networks, J. Nonlinear Math. Phys. 30 (2023), no. 1, 34--70.

\bibitem {Zambrano-Zuniga-2}B. A. Zambrano-Luna, W. A. Z\'{u}\~{n}iga-Galindo,
$p$-adic cellular neural networks: applications to image processing, Phys. D
446 (2023), Paper No. 133668, 11 pp.

\bibitem {Zinn-Justin}J. Zinn-Justin, Quantum field theory and critical
phenomena. Internat. Ser. Monogr. Phys., 85 Oxford Sci. Publ. The Clarendon
Press, Oxford University Press, New York, 1993.

\bibitem {Zuniga-ATMP}W. A. Z\'{u}\~{n}iga-Galindo, A correspondence between
deep Boltzmann machines and $p$-adic statistical field theories, Adv. Theor.
Math. Phys. Volume 28, Number 2, pp. 679-741 (2024).

\bibitem {Zuniga et al}W. A. Z\'{u}\~{n}iga-Galindo, C. He, B. A.
Zambrano-Luna, p-adic statistical field theory and convolutional deep
Boltzmann machines, PTEP. Prog. Theor. Exp. Phys.(2023), no. 6, Paper No.
063A01, 17 pp.

\bibitem {Zuniga-PhyA}W. A. Z\'{u}\~{n}iga-Galindo, $p$-adic statistical field
theory and deep belief networks, Phys. A 612 (2023), Paper No. 128492, 23 pp.

\bibitem {Zuniga-JFAA}W. A. Z\'{u}\~{n}iga-Galindo, Non-Archimedean white
noise, pseudodifferential stochastic equations, and massive Euclidean fields,
J. Fourier Anal. Appl. 23 (2017), no. 2, 288--323.

\bibitem {Zuniga-RIM}W. A. Z\'{u}\~{n}iga-Galindo, Non-Archimedean statistical
field theory, Rev. Math. Phys. 34 (2022), no. 8, Paper No. 2250022, 41 pp.

\bibitem {Zuniga-Textbook}W. A. Z\'{u}\~{n}iga-Galindo, $p$-Adic Analysis:
Stochastic Processes and Pseudo-Differential Equations, De Gruyter, 2025.

\bibitem {Zuniga-images}W. A. Z\'{u}\~{n}iga-Galindo, B. A. Zambrano-Luna,
Baboucarr Dibba, Hierarchical neural networks, $p$-adic PDEs, and applications
to image processing, J. Nonlinear Math. Phys. 31 (2024), no. 1, Paper No. 63,
40 pp.

\bibitem {Zuniga-Entropy}W. A. Z\'{u}\~{n}iga-Galindo, B. A. Zambrano-Luna,
Hierarchical Wilson-Cowan models and connection matrices, Entropy 25 (2023),
no. 6, Paper No. 949, 20 pp.
\end{thebibliography}
\end{document}